\RequirePackage{fixltx2e}
\documentclass[11pt]{article}  
\usepackage[utf8]{inputenc}
\usepackage[T1]{fontenc}
\DeclareFontFamily{U}{mathx}{}
\DeclareFontShape{U}{mathx}{m}{n}{<-> mathx10}{}
\DeclareSymbolFont{mathx}{U}{mathx}{m}{n}
\DeclareMathAccent{\widehat}{0}{mathx}{"70}
\DeclareMathAccent{\widecheck}{0}{mathx}{"71}  
\usepackage{tikz} 

\usepackage{mathrsfs}
\usepackage{amsbsy} 
\usepackage{amsmath, amsfonts, amsthm, amssymb, amscd}
\usepackage{accents, verbatim}
\usepackage{pdfpages}
\input{epsf}
\usepackage{graphicx}
\usepackage{xcolor}
\usepackage{tikz} 
\usepackage{microtype}
\usepackage[nottoc]{tocbibind} 
\usepackage[normalem]{ulem}
\usepackage{zref-perpage} \zmakeperpage{footnote}
\usepackage{hyperref}
\usepackage{a4wide}
\usepackage{titlesec}

\usepackage{hyperref,bookmark,footnotebackref}
\hypersetup{linktoc = all}                
\hypersetup{hidelinks}
\hypersetup{bookmarksnumbered}
\makeatletter
\let\oldaligned\aligned
\def\@alignedenvir{aligned}
\def\aligned{\ifx\@currenvir\@alignedenvir\expandafter\@firstoftwo\fi\oldaligned\relax}
\makeatother
\newcommand{\abs}[1]{\lvert#1\rvert}
\newcommand{\Ric}{\operatorname{Ric}}
\usepackage{aliascnt}
\theoremstyle{plain}
\newtheorem{theorem}{Theorem}[section]
\newcommand{\mynewtheorem}[2]{
  \newaliascnt{#1}{theorem}
  \newtheorem{#1}[#1]{#2}
  \aliascntresetthe{#1}
  \expandafter\providecommand\csname #1autorefname\endcsname{#2}
}
\mynewtheorem{definition}{Definition}
\mynewtheorem{proposition}{Proposition}
\mynewtheorem{lemma}{Lemma}
\mynewtheorem{corollary}{Corollary}
\mynewtheorem{remark}{Remark}
\mynewtheorem{claim}{Claim}

\numberwithin{equation}{section}
\DeclareMathAlphabet\mathbfcal{OMS}{cmsy}{b}{n} 

\newcommand \bse {\begin{subequations}}
\newcommand \ese {\end{subequations}}

\renewcommand \div {\operatorname{div}}
\newcommand \divN {\operatorname{div}_{\Ncal}}
\newcommand \curl {\operatorname{curl}}

\newcommand \Jperp  {\mathbb J^\perp}
\newcommand \Jpar {J^\parallel{}}

\newcommand \Jhat {\widehat{\mathbb J}}
\newcommand \Jhatpar { \widehat{J}^\parallel}
\newcommand \Kpar {K}
\newcommand \Xpar {X}

\newcommand \lot {\textnormal{l.o.t.}}
\newcommand \geo {\textnormal{g.m.t.}}

\newcommand \Jbb {\mathbb J} 
\newcommand \Lbb {\mathbb L}
 
\newcommand \Pbb {\mathbb P}
\newcommand \Qbb {\mathbb Q}
\newcommand \Vbb {\mathbb V} 
\newcommand \Xbb {\mathbb X} 
\newcommand \Ybb {\mathbb Y} 

\newcommand \Ucal {\mathcal U}

\newcommand \sign {\operatorname{sgn}}

\renewcommand \ln \log
\newcommand \Hess {\operatorname{Hess}}

\definecolor{myviolet}{RGB}{148,0,211}

\newcommand \la 	\langle
\newcommand \ra 	\rangle

\newcommand \Ncal 	{\mathcal N}
\newcommand \Vcal 	{\mathcal V}
  
\newcommand \Dcal 	{\mathcal D}
\newcommand \Gcal 	{\mathcal G}

\newcommand \Ebf 	E            

\newcommand \bei 	{\begin{itemize}}
\newcommand \eei 	{\end{itemize}}
\newcommand \del	\partial
\newcommand \auth 	{\textsc}   
\newcommand \Mcal 	{\mathcal M}   
\newcommand \Hcal 	{\mathcal H}

\newcommand \loc   	{\textnormal{loc}} 
\newcommand \RR 	{\mathbb R}   
\newcommand \CC 	{\mathbb C}   

\newcommand \Tbb 	{\mathbb T}   
\newcommand \Sbb 	{\mathbb S}   
\newcommand \eps 	\epsilon  
\newcommand \be 	{\begin{equation}}
\newcommand \ee 	{\end{equation}} 
\newcommand \bel 	{\be \label}
\let\oldmarginpar\marginpar
\renewcommand\marginpar[1]{\ifhmode\unskip\fi\- \oldmarginpar[\raggedleft\footnotesize #1]%
{\raggedright\footnotesize #1}}
\newcommand \lam 	\lambda

\newcommand{\WARNONCE}{\gdef\WARNONCE{}\GenericWarning{}{There are still BLF/PLF comments in this file!}}

\renewcommand \th 	\theta 
\newcommand \sgn 	{\operatorname{sgn}}

\newcommand \ZZ       {\mathbb{Z}} 

\renewcommand \geq \geqslant
\renewcommand \ge \geqslant
\renewcommand \leq \leqslant 
\renewcommand \le \leqslant

\newcommand \bea  {\begin{eqnarray}}
\newcommand \eea  {\end{eqnarray}}

\newcommand \Ccal {\mathcal C}

\newcommand \Jtilde {\widetilde J}

\newcommand \dive {\operatorname{div}}

\newcommand \rhofluid \rho

\newcommand \lbrac \llbracket 
\newcommand \rbrac \rrbracket

\newcommand \Lie {\mathcal{L}}

\let\Re\relax\let\Im\relax
\DeclareMathOperator{\Im}{\mathfrak{Im}}
\DeclareMathOperator{\Re}{\mathfrak{Re}}

\newcommand \coloneqq {\mathrel{\mathop :}\mathrel{\mkern-1.2mu}=} 
\newcommand \eqqcolon {=\mathrel{\mkern-1.2mu}\mathrel{\mathop :}} 

\makeatletter
\newcommand*\smallbullet{\mathpalette\smallbullet@{.5}}
\newcommand*\smallbullet@[2]{\mathbin{\vcenter{\hbox{\scalebox{#2}{$\m@th#1\bullet$}}}}}
\makeatother

  \newcommand \amdeux  \lambda  
  
  \newcommand \Mbf   M  

\newcommand \Vol {\mathrm{vol} \hskip.03cm}
\newcommand \guntrois {g^{(1+3)}}
\newcommand \dVuntrois {\mathrm{d}\Vol^{(1+3)}}
\newcommand \divuntrois {\div^{(1+3)}}
\newcommand \curluntrois {\curl^{(1+3)}}

\newcommand \gtrois {g^{(3)}}
\newcommand \dVtrois {\mathrm{d}\Vol^{(3)}}

\newcommand \dVundeux {\mathrm{d}\Vol^{(1+2)}}

\newcommand \dVdeux {\mathrm{d}\Vol^{(2)}}

\newcommand \suit \eps

\newcommand \Fcomp {F} 
\newcommand \Fvee {F^\vee} 
\newcommand \Fbb {\mathbb F}

\newcommand \Msource {\Mbf_{\textnormal{\textbf{sour}}}}
\newcommand \Mwave {\Mbf_{\textnormal{\textbf{wave}}}}

\newcommand \Ebfscri {\boldsymbol{\mathscr{E}}}
\newcommand \Tbfscri {\boldsymbol{\mathscr{T}}}

\newcommand \Nbfscri {\boldsymbol{\mathscr{N}}}
\newcommand \Lbfscri {\boldsymbol{\mathscr{L}}_{\textnormal{conf}}}
\newcommand \mathringLbfscri {\mathring{\boldsymbol{\mathscr{L}}}_{\textnormal{conf}}}
\newcommand \Lmin {\Lbfscri^{\textnormal{\!min}}}

\newcommand \JKL {\texorpdfstring{\ensuremath{\Jbb K\Lbb}}{JKL}}

\newcommand \Numb {N}
\newcommand \hNumb {\widehat{N}}
\newcommand \vNumb {{\mathbb N}}

\newcommand{\Obig}{\mathcal{O}}

\newcommand \Gammazero  {\widetilde K}      
\newcommand \Gammaun     {\widetilde J}         
\newcommand \Gammaunpar     {\widetilde J^\parallel}         
\newcommand \Gammadeux  {\widetilde I}  
\newcommand \Gammadeuxpar  {\widetilde I^\parallel}  
\newcommand \funmu {\underline{\smash{\ensuremath{\mu}}}}

\newcommand \BV {\ensuremath{BV}}
\newcommand \BVac {\ensuremath{BV_{\textnormal{ac}}}}
\newcommand \Meas {M}
\newcommand \Measac {M_{\textnormal{ac}}}

\newcommand \Var {\textnormal{Var}}
\newcommand \Err {\textnormal{Error}}

\newcommand \muref {\mu_{\textnormal{ref}}}
\newcommand \Fmeasure {\nu_\Fbb{}}
\newcommand \opL {\mathfrak{L}}

\newcommand \Deltafluid {\Delta_{\textnormal{fluid}}}
\newcommand \Deltatwist {\Delta_{\textnormal{twist}}}
\newcommand \Deltageom {\Delta_{\textnormal{geom}}}

\newcommand \Hterm {\mathfrak{H}}
\newcommand \Mterm {\mathfrak{M}}

\newcommand \tPhi {\widetilde \Phi}
\newcommand \tPsi {\widetilde \Psi}

\newcommand \Interval {I}
\newcommand \Iopen {(t_0,t_1)}
\newcommand \Ihalf {[t_0,t_1)}

\newcommand \Lip {\textnormal{Lip}}
\newcommand \Rep {\operatorname{Rep}}
\newcommand \undt {\underline{t}}



\titleformat{\subsubsection}[runin]{\normalfont\bfseries}{}{.5em}{}[\ifnum\spacefactor<1001.\fi\ \ ]
\titlespacing{\subsubsection}{0pt}{2ex plus .1ex minus .2ex}{0pt}
\titleformat{\paragraph}[runin]{\normalfont\itshape}{}{.5em}{\hspace{-1pt}}[\ifnum\spacefactor<1001.\fi\ \ ]
\titlespacing{\paragraph}{0pt}{1ex plus .1ex minus .2ex}{0pt}

\usepackage{float} 
\usepackage{calligra}
\usepackage[T1]{fontenc} 
\usepackage[mathscr]{euscript}
\usepackage{bbm}
\numberwithin{table}{section}

\newcommand \Jac {\textbf{Jac}}

\begin{document}

\title{\bf
Stability and instability of torus-symmetric Einstein spacetimes  
with square-integrable connection}

\author{Bruno Le Floch\thanks{Laboratoire de Physique Théorique et Hautes Énergies, Sorbonne Universit\'e \& Centre National de la Recherche Scientifique, 4 Place Jussieu, 75252 Paris, France. Email: {\tt bruno@le-floch.fr}.}
\, 
and Philippe G. LeFloch\thanks{Laboratoire Jacques-Louis Lions, Centre National de la Recherche Scientifique, Sorbonne Universit\'e,  4~Place Jussieu, 75252 Paris, France. Email: {\tt contact@philippelefloch.org}. 
\newline
\textit{2000 AMS Classification.} 83C05, 58K25, 35L65, 76L05. 
}}
\date{May 2026}
\maketitle

\begin{center}
\vspace*{-0.5cm}  
{\scshape \small À Nounette, à l'occasion de ses vingt années cérésiennes}
\end{center}


\begin{abstract}
We study the global evolution problem for the Einstein equations under $\mathbb{T}^2$ symmetry on $\mathbb{T}^3$, allowing vacuum, scalar-field, and compressible-fluid matter models, governed by a general equation of state including isothermal and polytropic fluids. Under this symmetry, we obtain the first non-perturbative, global existence and stability theory with connection coefficients being merely square-integrable, which allows both \emph{impulsive gravitational waves} and \emph{shock waves}. In areal gauge, we introduce new fluid and geometric variables and reformulate the Einstein--Euler system as a first-order system of nonlinear balance laws with constraints and an entropy structure. The resulting formulation exhibits hyperbolicity, null forms, entropy currents, div-curl structure, maximum principles, and spacetime estimates. This leads to a notion of \emph{tame Einstein--Euler flow} for which the essential geometric and fluid variables are square-integrable (finite energy), and the secondary variables are absolutely continuous (or, more generally, of bounded variation). In this non-perturbative and weak regularity setting, the equations remain meaningful even when the Weyl curvature concentrates into Dirac masses along timelike hypersurfaces, and the Ricci curvature remains only integrable. Our main results are a global existence theorem for areal foliations, a nonlinear stability theorem for well-prepared initial data, and a nonlinear instability theorem for geometrically oscillatory data, the latter producing measure corrections to the stress energy tensor. In the future-contracting regime, the areal foliation reaches a geometric singularity where the volume of $\mathbb{T}^3$ spatial slices degenerates to zero. The areal function reaches zero \emph{generically} in the non-vacuum Gowdy-symmetric and vacuum torus-symmetric cases. In the future-expanding regime, the areal foliation is complete, with both the spatial volume and the areal function tending to infinity at the future boundary.
\end{abstract}


\



{
 
\setcounter{secnumdepth}{2}
\setcounter{tocdepth}{2}
\tableofcontents

} 


\

\section{Introduction}
\label{section=1}

\subsection{Global evolution in the non-perturbative regime} 
\label{section=1-1}

\subsubsection{Matter spacetimes with $\Tbb^2$ symmetry} 

\paragraph{Non-perturbative regime.}

We study a class of globally hyperbolic spacetimes satisfying the (vacuum or non-vacuum) Einstein equations.
The existing mathematical literature on global solutions most often focuses on small-data results, namely results restricted to spacetimes that are \emph{close to known backgrounds} such as Minkowski or Kerr spacetimes.
Our goal is to encompass large-amplitude configurations and singular spacetimes.
To this end, we consider the class of torus symmetric spacetimes, that is, spacetimes endowed with two commuting spacelike Killing fields.

The natural questions are to parametrize such spacetimes by their initial data sets, solve the corresponding initial value problem and investigate the stability or instability of sequences of spacetimes. Our results in this paper\footnote{An overview first appeared in 2019 as~\emph{ArXiv:1912.1298} and was later published in~\cite{LeFlochLeFloch-port}.} describe the dynamics of matter fields evolving under their own gravity in a \emph{non-perturbative regime}.

More precisely, our methodology applies to a broad class of matter fields including vacuum solutions, scalar fields, stiff fluids, as well as isothermal, polytropic, and more general perfect fluids. Under the assumption of $\Tbb^2$~symmetry on a spatial three-torus~$\Tbb^3$, we establish a global theory of evolution and stability. 


\paragraph{Beyond the formation of impulsive waves and shock waves.}

Studying such a broad class of spacetimes requires overcoming several major challenges, which we first outlined in~\cite{LeFlochLeFloch-port}. These challenges arise from both the Weyl part of the Riemann curvature (which is defined only in the sense of distributions) and the Ricci part (which is determined by the matter content).

\bei 

\item We allow for the presence of \textit{impulsive gravitational waves} ---a phenomenon that is suppressed in the spherically symmetric setting. These waves propagate at the speed of light and may lead to concentration phenomena throughout the spacetime. Our framework encompasses such waves, for which the Weyl curvature may be defined only in the sense of distributions.

\item In addition, we allow for the presence of \textit{shock waves} ---a phenomenon that arises generically in compressible fluid flows, even when the initial data are smooth. This is a classical consequence of the nonlinear nature of the Euler equations, which makes shock formation unavoidable in general. Accordingly, we must seek a global Cauchy development that remains meaningful beyond the onset of shocks, while the Euler equations (derived from the contracted second Bianchi identity satisfied by the geometry) must be understood in a weak sense. 

\item Furthermore, the presence of \emph{fluid vacuum} regions cannot in general be excluded, and the stability theory must therefore encompass the \emph{transition} between vacuum and non-vacuum regions of spacetime. This is a delicate issue, since the mathematical techniques required in these two regimes are quite different, and the Euler equations may fail to be strictly hyperbolic at vacuum.

\eei
\noindent The mathematical analysis of spacetimes containing both gravitational waves and shock waves requires dealing with \emph{weak solutions}, as first recognized 
by Rendall and St\"ahl \cite{Rendall-S} (on shock formation in plane symmetry) and 
LeFloch and Rendall~\cite{LeFlochRendall-2011} (on global foliations of Gowdy-symmetric Einstein--Euler spacetimes); see also the lectures \cite{LeFloch-lectures}. In the present paper, relying on the strategy that we presented earlier by the authors in~\cite{LeFlochLeFloch-port}, we propose a framework for defining Einstein--Euler spacetimes with weak regularity, and investigating their global geometry and their stability properties under \emph{weak convergence}. 


\subsubsection{Non-vacuum spacetimes with symmetry}  

\paragraph{Scalar fields in the non-perturbative regime.} 

We briefly review the works that have motivated us to study Einstein–Euler spacetimes with symmetry in the large-data regime. (\autoref{section=1-4}, below, includes further related references). 

As far as the global evolution problem is concerned, all currently available mathematical results in the non-perturbative regime rely on some symmetry assumption. The first major breakthrough in this direction was achieved by Christodoulou in his pioneering work on self-gravitating \emph{scalar fields in spherical symmetry}~\cite{Christo86,Christo92}. In his early studies of the Einstein--scalar field system in Bondi coordinates, he introduced a notion of generalized solution that is $C^2$ regular except possibly on the axis of symmetry. He subsequently developed a much broader weak-regularity framework in which the metric coefficients have \emph{bounded variation}. This constituted one of the first systematic investigations of spacetimes with weak regularity, including the analysis of the future boundary of global Cauchy developments, and it culminated in his proof of Penrose's weak cosmic censorship conjecture for scalar fields in spherical symmetry for generic data. Christodoulou's pioneering ideas have had a lasting influence and have inspired a large body of subsequent research.


\paragraph{Kinetic matter in the non-perturbative regime.}

The study of spherically symmetric spacetimes with kinetic matter governed by the \emph{Vlasov equation} was initiated by Rendall~\cite{Rendall-1992}, Andreasson~\cite{Andreasson-1999}, Andreasson and Rein \cite{Andreasson-Rein}, and followers. While the existence of global spacetime foliations can be established by arguments similar to those developed for vacuum spacetimes, it is significantly more challenging to understand their global geometry, since these spacetimes may exhibit properties very different from those of their vacuum counterparts. In particular, the coupling between the geometry and the kinetic distribution function introduces additional analytical and structural difficulties. As a result, even in spherical symmetry, the kinetic-matter case with \emph{large data} remains considerably less developed than the corresponding vacuum and fluid settings.


\paragraph{Compressible fluids in the non-perturbative regime.} 

Despite the importance\footnote{For instance, in astrophysical or cosmological contexts.} of the evolution of self-gravitating \textit{compressible perfect fluids}, the existing literature contains mainly formal constructions of special classes of solutions, such as static fluids, together with formal asymptotic analyses. No rigorous mathematical treatment of the non-perturbative problem was available until recent years. Even under Gowdy symmetry, that is, when two-planes orthogonal to symmetry orbits are integrable, only partial results have been obtained so far; cf.~the pioneering work by LeFloch and Rendall~\cite{LeFlochRendall-2011}. 

In particular, the results established in the present paper are \emph{new even in the Gowdy-symmetric setting.} Our methodology is to work within a class of \textit{weak solutions} for which the Einstein equations remain meaningful, admit large-data solutions, and can be analyzed by weak convergence techniques such as compensated compactness (see below). This provides a first step toward a rigorous analysis of the global geometry of such spacetimes and suggests a methodology of interest well beyond the specific framework studied here.

On the other hand, on the Euler equations without symmetry restriction, major contributions due to Christodoulou \cite{Christodoulou-book} and Speck et al.\ \cite{AbbresciaSpeck-2020,DisconziLuoMazzoneSpeck-2022,LukSpeck-2018,LukSpeck-2020,Speck-2019} 
are concerned with the \emph{formation} of shocks, whereas our focus in the present paper is on the evolution \emph{beyond} shock formation (and without specifically tracking the shock structure).


\subsection{A new formulation of the Einstein–Euler system}
\label{section=1-2}

\subsubsection{Einstein--Euler system}

By definition, a spacetime is a $(3+1)$-dimensional Lorentzian manifold $(\Mcal, \guntrois)$ satisfying the Einstein equations\footnote{We normalize the standard constant to unity.}
\bel{eq:44}
G = T,
\ee
which relate the Einstein tensor $G \coloneqq \Ric - (R/2)\guntrois$, representing the geometry of spacetime, to its matter content encoded by the stress-energy tensor~$T$. By convention, the Lorentzian signature is $(-,+,+,+)$; here $\Ric$ denotes the Ricci curvature tensor associated with the spacetime metric $\guntrois$, and $R$ denotes the scalar curvature. Whenever necessary, we introduce indices, with Greek indices ranging over $\{0,\ldots,3\}$. Indices are raised and lowered using the metric, and we adopt the Einstein summation convention over repeated indices.

A perfect compressible fluid is governed by the stress-energy tensor
\bel{eq:45}
T = (\mu + p(\mu))\, \guntrois(u,\cdot)\otimes \guntrois(u,\cdot) + p(\mu)\, \guntrois,
\ee
where $\mu \geq 0$ denotes the mass-energy density of the fluid, and $u$ its future-oriented timelike velocity field, normalized so that $\guntrois(u,u)=-1$. By the second contracted Bianchi identities for $\guntrois$, the Einstein equations imply the \textit{Euler equations}:
\bel{eq:46}
\divuntrois T = 0,
\ee
where $\divuntrois$ denotes the divergence operator associated with the Levi-Civita connection of the spacetime metric $\guntrois$. We consider perfect fluids with stress-energy tensor~\eqref{eq:45}, governed by an equation of state $p=p(\mu)$ satisfying hyperbolicity conditions (stated in~\eqref{hyperbolic-eos}, below).


\subsubsection{Global Cauchy developments}

When the initial data are \emph{sufficiently regular} and a suitable gauge is chosen (for instance the areal gauge in our $\Tbb^2$-symmetry setup), the Einstein--Euler equations form a locally well-posed evolution system of partial differential equations (of hyperbolic type) with constraints. For any such initial data set, the Cauchy problem admits a maximal globally hyperbolic Cauchy development~\cite{Choquet-book, Wald-book}. 
An initial data set for the Einstein--Euler system~\eqref{eq:44}–\eqref{eq:46} consists of a Riemannian manifold \( (\Mcal_0, g_0) \) together with a symmetric two-tensor~$k_0$ defined on~\( \Mcal_0 \), as well as a mass-energy field~$\mu_0 \geq 0$ and a vector field~$v_0$ defined on~\( \Mcal_0 \), satisfying suitable constraints, referred to as the Hamiltonian and momentum constraints~\cite{Choquet-book}. Solving the initial value problem for the Einstein equations from a given initial data set consists of finding a Lorentzian manifold \( (\Mcal, \guntrois) \) together with a scalar field \( \mu \colon \Mcal \to [0, + \infty) \) and a vector field $u$ defined on \( \Mcal \), such that the Einstein equations, and hence the Euler equations, are satisfied. 
The future boundary of the Cauchy development may exhibit singular behavior in the geometry, including impulsive gravitational waves, while compressible fluids may develop singularities in finite time, namely shock waves. The purpose of the present work is to formulate a weak regularity framework in which the Einstein--Euler evolution remains meaningful beyond both types of singular behavior.


\subsubsection{Proposed formulation}

The central theme throughout the present study is the interplay between the properties of the spacetime metric and those of the fluid unknowns. We proceed in several steps, each step providing certain intermediate results of independent interest. Our first objective is to exhibit certain hidden structural properties of the field equations. Remarkably, as we show in~\autoref{section=2}, an analogue of the formulation we found in~\cite{LeFlochLeFloch-1} for vacuum Gowdy spacetimes is available for the general class of \emph{$\Tbb^2$ symmetric matter spacetimes}, even when the fluid is governed by a general equation of state. A difficulty we overcome is that the so-called twist variables need not be constants, in contrast with the `Gowdy matter' or `$\Tbb^2$ vacuum' cases. Interestingly, our choice of variables and exponents is rigidly fixed by requiring cubic terms not to appear. In particular, for the description of the fluid we introduce suitably metric-weighted variables, needed to allow for the mass-energy density to vanish in some spacetime region ---which is necessary since vacuum regions can arise from non-vacuum initial data.

In the gauge under consideration, the field equations~\eqref{eq:44}--\eqref{eq:46} take the form of a system of partial differential equations and include nonlinear hyperbolic equations whose solutions generically develop singularities from initially regular data. The field equations couple metric and matter unknowns together, which makes the whole analysis particularly challenging. More precisely, in areal gauge the spacetime metric $\guntrois$ depends upon six metric coefficients and satisfies a system of nonlinear wave equations coupled to differential equations. The solutions typically become singular at some final time in the contracting direction when the area function approaches zero. On the other hand, the matter variables, i.e.~the mass-energy density $\mu \geq 0$ and the velocity vector $u$ satisfy a nonlinear hyperbolic system of two coupled balance laws. The main challenge to be overcome is to elucidate the nature of the coupling between the geometric equations of wave map type and nonlinear hyperbolic equations. 
The wave map structure has attracted a lot of attention in recent years; see~\cite{KST,SterbenzTataru,Tataru} and the references therein, as well as Kenig~\cite{Kenig-2017} for a study of the blow-up mechanism in focusing wave equations.


\subsection{Main results}
\label{section=1-3}

\subsubsection{Areal flows with weak regularity} 

\paragraph{Geometric setup.}

The \(\Tbb^2\)-symmetric spacetimes under consideration admit a foliation by compact spacelike hypersurfaces \(\Mcal_t\), each diffeomorphic to \(\Tbb^3\), that is, 
\be
\Mcal=\bigcup_{t\in\Interval}\Mcal_t,
\ee
where the time function \(t:\Mcal\to\Interval\subset\RR\) parametrizes the leaves of the foliation and agrees, up to sign, with the area of the \(\Tbb^2\)-orbits of symmetry. This choice of time variable is called the \emph{areal gauge}. An initial data set, denoted by \((\mathring g,\mathring k,\mathring\mu,\mathring v)\), consists of \(\Tbb^2\)-invariant data on \(\Tbb^3\) for which all symmetry orbits have the same area. A corresponding spacetime solution consists of a \(\Tbb^2\)-invariant Lorentzian metric \(\guntrois\), a scalar field \(\mu\ge0\), and a unit timelike vector field \(u\). The quantity \(\mathring\mu\,\mathring v\) is interpreted as the \emph{projection} onto the initial hypersurface of the matter momentum \(\mu u\), and is tangent to that hypersurface. 

The density \(\mathring\mu\) is allowed to vanish; in that case, the velocity field \(u\) itself is no longer intrinsically defined, and similarly for \(\mathring v\). This vacuum degeneracy is a major technical difficulty in the analysis of the Euler equations. We may describe the solution to the Einstein equation as a $\Tbb^2$ symmetric \emph{flow on the torus} $\Tbb^3$ consisting of 
\bei 

\item a Riemannian $3$-metric $g(t)$ and a two-tensor $k(t)$ ---representing the second fundamental form of each spatial slice, and 

\item a scalar field $\mu(t) \geq 0$, and a vector field $v(t)$ ---representing the projection of the physical velocity field. 

\eei


\paragraph{Weak regularity setup.}

The weak regularity conditions will be introduced in the course of our analysis: roughly speaking, $k(t)$~and the first-order space derivatives of~$g(t)$ must be \emph{square-integrable on spacelike slices,} while the mass-energy density~$\mu(t)$ is integrable and velocity components parallel to the $\Tbb^2$~orbits of symmetry obey suitable uniform bounds.
Consequently, the Weyl curvature includes derivatives of square-integrable functions, which are defined only in the sense of distributions, while the Ricci curvature is integrable.
While in the present paper we will state the weak regularity conditions in areal gauge, our weak regularity conditions can be also expressed in a different gauge by working out the regularity of the fields~\cite{LeFlochLeFloch-next}.


\subsubsection{Global future evolution}

The remarkable structure of the Einstein equations (in $\Tbb^2$-symmetry) exhibited in \autoref{section=2}, below, plays a central role in order to investigate the nonlinear geometry-fluid coupling and tackle the global evolution problem for the Einstein equations. 

Our main statements, presented in this introduction in a preliminary form, apply to compressible perfect fluids governed by a general pressure law 
(satisfying hyperbolicity and convexity conditions~\eqref{hyperbolic-eos}), which we assume to be asymptotically polytropic or isothermal in the vicinity of the vacuum state, and to be asymptotic isothermal for large values of the mass-energy density. In \autoref{def-pression}, below, we define this class referred to as an \textit{admissible equation of state}. 

Moreover, in the presence of shock waves, the Euler equations are no longer
time-reversible. It is therefore physically meaningless to evolve fluid
solutions backward across shocks. This naturally leads us to distinguish
between \emph{two classes of initial data sets}, corresponding respectively
to future-expanding and future-contracting spacetimes.  
As in the standard theory of compressible fluids, this irreversibility is encoded by an entropy inequality. This feature plays a fundamental role in the present theory and is addressed in \autoref{section=3}, through the introduction of entropy currents, and in \autoref{section=4}, through the definition of \emph{weak solutions} and \emph{entropy inequalities}. For clarity of exposition, we refer to our class of weak solutions as \emph{tame solutions}; cf.~\autoref {def-weaksolu}
for the notion of \emph{Einstein--Euler flow}, together with \autoref{def-weaksolu-deux} for the definition of a \emph{tame Einstein--Euler flow}. 

By definition, the tame regularity assumptions impose Sobolev \(H^1\)-type
bounds on the principal metric coefficients
(therefore $L^2$ integrability at the level of the connection coefficients), suitable Lebesgue integrability
conditions on the fluid variables, and absolute continuity together with
finite variation bounds on the auxiliary metric coefficients and on a rescaled
version of the momentum component parallel to the symmetry orbits.

\begin{theorem}[Global foliation of $\Tbb^2$-symmetric spacetimes with weak regularity]
\label{thm:1.1}
Consider the initial value problem associated with the Einstein--Euler
system~\eqref{eq:44}--\eqref{eq:46}, describing the global evolution of a
compressible perfect fluid governed by an admissible equation of state
(in the sense of~\autoref{def-pression} below). Suppose that the initial data
set \((\mathring g,\mathring k,\mathring \mu,\mathring v)\) is \(\Tbb^2\)-symmetric on the
three-torus \(\Tbb^3\)with orbits of constant area \(\abs{t_0}\), where the sign of \(t_0\neq0\) determines whether the data are of future-contracting or future-expanding type.
Assume moreover that the data have tame regularity
(cf.~\autoref{section=4}) and satisfies
Einstein's constraint equations in the weak sense. 

Then this initial data set generates a future Cauchy development which is a
tame solution\footnote{In the same regularity class as the initial data set.}
\((\Mcal,\guntrois,\mu,u)\) of the Einstein--Euler equations, understood in
the sense of distributions\footnote{In particular, the Ricci tensor is
Lebesgue integrable, while the Weyl tensor is at most a first-order
distribution.}. This solution also satisfies suitable entropy inequalities
and enjoys the following global geometric properties. It may be viewed as a
\(\Tbb^2\)-flow on \(\Tbb^3\), endowed with a foliation by a time function
\(t\colon\Mcal\to \Interval\subset\RR\), whose leaves all have topology
\(\Tbb^3\) and are \(\Tbb^2\)-symmetric, with symmetry orbits of area
\(\abs{t}\).

\bei

\item {\bf Future-expanding regime (\(t_0>0\)):}
one has $\Mcal=[t_0,+\infty)\times\Tbb^3$,
and the areal foliation is complete toward the future, in the sense that both
the volume of the \(\Tbb^3\)-slices and the area of the \(\Tbb^2\)-orbits
tend to infinity.

\item {\bf Future-contracting regime (\(t_0<0\)):}
one has $\Mcal=[t_0,t_*)\times\Tbb^3$, where the areal foliation extends up to some time \(t_*\in(t_0,0]\) such that
the volume of the \(\Tbb^3\)-slices tends to zero as \(t\to t_*\). Moreover,
either \(t_*<0\), in which case the conformal length of the quotient manifold
\(\Tbb^3/\Tbb^2\) tends to zero, or \(t_*=0\), in which case the area of the
\(\Tbb^2\)-orbits tends to zero. For vacuum \(\Tbb^2\)-symmetric spacetimes,
or for non-vacuum Gowdy-symmetric spacetimes, the latter alternative
\(t_*=0\) holds \emph{generically}\footnote{The relevant notions of
genericity were first introduced in \cite{IsenbergWeaver-2003} and
\cite{GrubicLeFloch-2015}, respectively.}.

\eei
\end{theorem}


\subsubsection{Nonlinear stability and instability} 

Our second main result concerns sequences of spacetimes. The statement below is relevant (and new) even for vacuum spacetimes and generalizes our earlier study of (vacuum) Gowdy spacetimes~\cite{LeFlochLeFloch-1}, while supplementing \cite{LeFlochRendall-2011} (for the Euler system with Gowdy symmetry). The notions of \emph{flow with finite energy} and \emph{corrector stress tensor} will be introduced later in~\autoref{section=4}, and a detailed statement will be provided in \autoref{thm:exactsol-compactness}, below. 

\begin{theorem}[Nonlinear stability and instability of $\Tbb^2$ symmetric spacetimes]
\label{thm:stability0}
Under the conditions in Theorem~\ref{thm:1.1}, consider a sequence of $\Tbb^2$ symmetric, either future-expanding or future-contracting initial data sets $(g_0^\suit, k_0^\suit, \mu_0^\suit, v_0^\suit)$ (for $\suit \in (0,1]$) defined on the torus $\Tbb^3$, with the same initial areal time~$t_0$, and assume that their natural (energy, BV) norms are uniformly bounded with respect to~$\suit$, and that the conformal length density and parallel momentum per particle converge at initial time\footnote{Remarkably, no convergence assumptions are made on any of the remaining geometric and fluid variables. The norms and convergence are explicited in~\eqref{exactsol-initbounds-all}, below.}. By Theorem~\ref{thm:1.1} there exist flows $(g^{(4)\suit},\mu^\suit,u^\suit)$ defined on time intervals $[t_0,t_*^\suit)$ for some sequence of times $t_*^\suit>t_0$. 

\bei 

\item For future-expanding data sets, assume that the sequence of conformal lengths at time $t_0$ is uniformly bounded below, i.e.
\be
\inf_{\suit \in (0,1]} \Lbfscri^\eps(t_0) >0. 
\ee

\item For future-contracting data sets, assume that the time of existence are uniformly bounded below\footnote{For instance, it is sufficient to assume an upper bound $\sup_{\suit \in (0,1]} \Ebfscri^\eps(t_0)/\Lbfscri^\eps(t_0) < +\infty$ on the ratio energy/length at the \emph{initial time} $t_0$.}, i.e. 
\be
t_0 < \inf_{\suit \in (0,1]} t_*^\suit \leq 0
\ee
and that, at a uniform time $t_1\in(t_0,\inf t_*^\suit)$ the sequence of conformal lengths at time $t_1$ is uniformly bounded below, i.e.
\be
\inf_{\suit \in (0,1]} \Lbfscri^\eps(t_1) >0. 
\ee
\eei 
Then, up to extracting a subsequence, the flows converge weakly (in natural norms) on every compact time interval to a limit $(g^{(4)\sharp}, \mu^\sharp, u^\sharp)$, which is defined on $[t_0,+\infty)$ in the future-expanding regime and at least on $[t_0,t_1)$ in the future-contracting regime.

\bei 

\item {\bf Stability for well-prepared initial data sets.}
If the geometric initial data $(g_0^\suit, k_0^\suit)$ \emph{converges strongly} in natural norms, the limit $(g^{(4)\sharp}, \mu^\sharp, u^\sharp)$ defines an \emph{Einstein--Euler flow}, satisfying the Einstein--Euler system~\eqref{eq:44}--\eqref{eq:46} in a weak sense. 

\item {\bf Instability for general initial data sets.}
More generally, the limit $(g^{(4)\sharp}, \mu^\sharp, u^\sharp)$ defines an {\rm Einstein--Euler flow with corrector}, satisfying a suitable \emph{extension of the Einstein--Euler system}, defined by supplementing the evolution equations with a \emph{corrector stress-tensor}~$\Pi$ (as we call it), which a symmetric traceless measure field, is orthogonal to the $\Tbb^2$~orbits of symmetry, and satisfies a divergence inequality.
\eei 
\end{theorem} 

A full statement will be presented only later on in~\autoref{section=5}, after introducing the notion of weak solution (\autoref{section=4}). Importantly, we also state and prove a version of \autoref{thm:stability0} for \emph{sequences of approximate solutions}, which opens a way to analyze vanishing viscosity approximations as well as discrete schemes. The proof will be presented in Sections~\ref{section=8} and~\ref{section=9} after establishing a priori regularity and integrability estimates on exact solutions in Sections~\ref{section=6} and~\ref{section=7}.


\subsubsection{Classes of matter models} 

There remains to list our conditions on the equation of state $p=p(\mu)$ of the fluid under consideration, which is assumed to be smooth for $\mu>0$ and at least once differentiable at the origin.
\bse \label{hyperbolic-eos} 
\bei 

\item \textbf{Hyperbolicity condition}
\bel{hyperbolic-eos-1} 
\aligned
& 0 < p'(\mu) < 1 \quad  
 (\text{for all } \mu > 0), 
 &&
 \qquad p(0) =0, \quad 0 \leq p'(0) < 1. 
\endaligned
\ee
\item \textbf{Convexity condition}
\bel{hyperbolic-eos-3} 
p''(\mu) + 2 \, {1 - p'(\mu) \over \mu+p(\mu)} \, p'(\mu) > 0 \quad  
 (\text{for all } \mu > 0).  
\ee
\eei
\ese
The upper bound $p'(\mu)<1$ rules out the possibility that the sound speed approach the speed of light.
The condition $0 < p'(\mu)$ for $\mu>0$ implies that the Euler equations form a strictly hyperbolic system \emph{away} from vacuum, while the assumptions $p(0) =0$ and $p'(0) < 1$ ensure that \(p(\mu) < \mu\) for $\mu>0$.
Meanwhile, the lower bound in~\eqref{hyperbolic-eos} is a mild restriction ensuring that the pressure is not too small, thereby preventing the fluid from concentrating at arbitrarily large density values.
Furthermore, the condition \eqref{hyperbolic-eos-3}  implies that the characteristic fields of the relativistic Euler equations are \emph{genuinely nonlinear} (in the sense of Lax). 

For instance, our results apply to isothermal fluids with $p= k^2 \mu$, which clearly satisfy these conditions, provided the speed of sound $k \in (0,1)$ remains below the speed of light (normalized to unity).  They also apply to equations of states that asymptote to a polytropic or isothermal law $p=\mu^\gamma$ with exponent $\gamma\geq 1$ at low densities $\mu\to 0$ and to the isothermal case at large densities.
More precisely, in addition to \eqref{hyperbolic-eos}, mild asymptotic conditions at $\mu\to 0,+\infty$ are required.
\begin{itemize}
\item \textbf{Asymptotically polytropic or isothermal at the vacuum}
\bse\label{equa-asym}  
\be
\lim_{\mu\to 0} \frac{p(\mu)}{\mu^\gamma} = c_\gamma \in (0, +\infty) ,
\qquad
\text{or}
\qquad
\lim_{\mu\to 0} \frac{p(\mu)}{\mu} = k_0^2 \in (0, 1)
\ee
for some constants $\gamma>1$ and $c_\gamma \in (0,+\infty)$, or for some constant $k_0 \in (0,1)$.

\item \textbf{Asymptotically isothermal at large density:} for some constant $k_\infty \in (0,1)$,
\bel{hyperbolic-eos-2}
\lim_{\mu\to+\infty} p'(\mu) = k_\infty^2 \in (0, 1) .
\ee
\ese
\end{itemize}

\begin{definition}
\label{def-pression}
A pressure function $p=p(\mu)$ is called an \textbf{admissible equation of state}  provided the hyperbolicity and convexity conditions \eqref{hyperbolic-eos}, together with\footnote{Strictly speaking, in our companion paper~\cite{LeFlochLeFloch-next} devoted to relativistic fluids, we also require a very mild integrability condition on the deviation $p(\mu)-c_\gamma \mu^\gamma$ near vacuum, and likewise near infinity.} the vacuum and large density asymptotic condition~\eqref{equa-asym}.
\end{definition}
 
Throughout this paper we work under the assumptions  \eqref{hyperbolic-eos} and \eqref{equa-asym}, while the extension to other matter models (scalar fields, stiff fluids, and non-zero cosmological constant) is considered in \autoref{section=10}, below. 
 

\subsection{Further background and related works}
\label{section=1-4}

\subsubsection{Sufficiently regular and vacuum spacetimes with \(\Tbb^2\) symmetry on \(\Tbb^3\)}

Over the past thirty years, significant progress has been made on the initial value problem for the Einstein equations in the \emph{non-perturbative} regime, provided the gravitational and matter fields are assumed to be \emph{sufficiently regular}. In addition to the references already cited in \autoref{section=1-1}, we would like to emphasize the following results for the class of spacetimes with \(\Tbb^2\) symmetry on \(\Tbb^3\). Moncrief~\cite{Moncrief-1981,EardleyMoncrief} first provided a full treatment of the Cauchy problem for in the polarized Gowdy case. A series of works followed (for which we refer to Rendall's textbook~\cite{Rendall-book}), establishing the existence of a global foliation by spacelike hypersurfaces in the areal gauge. After work by Berger et al. \cite{Berger-1997}, the existence problem was finally closed by Isenberg and Weaver~\cite{IsenbergWeaver-2003}: they established that the corresponding areal time function covers the entire range $(0, + \infty)$, describing an evolution originating from a Big Bang and dispersing indefinitely toward the future. 

We do not attempt here to review the vast literature on this class of spacetimes.  Beyond the issue of the existence of global foliations, and even for vacuum spacetimes, many challenging questions remain open concerning the global geometric behavior of \(\Tbb^2\)-symmetric vacuum spacetimes, including geodesic completeness, curvature blow-up, and the Penrose conjecture. It is only for the restricted class of Gowdy spacetimes that the behavior near the cosmological singularity is now well understood, thanks to pioneering work by~Ringström~\cite{Ringstrom-2004a,Ringstrom-2004b,Ringstrom-2006}. In the expanding direction, we refer to \cite{LeFlochSmulevici-2015} for partial results on the geometry of the future boundary.  


\subsubsection{Metrics with weak regularity}

\paragraph{General metrics.}

A definition of the Einstein equations for general metrics (without symmetry restriction) in the sense of distributions was proposed by several authors, including~\cite{LeFlochMardare-2007,Lott-2016} (and the references therein). However, no existence result is available in such a setup, and only junction conditions can be derived; cf.~for instance~\cite{Israel,Darmois,MarsSenovilla,LeFlochLeFloch-PhilTrans}. 

\paragraph{Bounded variation.}

As already pointed out in this introduction, earlier on, Christodoulou defined the class of spherically symmetric spacetimes with bounded variation~\cite{Christo86,Christo92}, and proved that singularities in spherically symmetric spacetimes arise at the center of symmetry; see also \cite{LeFlochSormani}. In the bounded variation class, the coupling with the Euler equations was studied first in spherical symmetry (and in Bondi coordinates) in~\cite{GroahTemple-2024,Burtscher-1}.  
 
In contrast, gravitational singularities in $\Tbb^2$  symmetry propagate and require an evolution system of their own: in Gowdy symmetry and areal coordinates by~\cite{Barnes-2004}, solutions with bounded variation were investigated in \cite{GrubicLeFloch-2015}. 


\paragraph{Sobolev regularity.}

Solutions in a suitable Sobolev regularity class were considered more recently. The first work on solutions with shock waves and impulsive waves to the Einstein--Euler system was LeFloch and Rendall~\cite{LeFlochRendall-2011} who constructed a global areal foliation. 
Using Penrose's spinor formalism, LeFloch and Stewart \cite{LeFloch-Stewart-2005,LeFlochStewart-2011} analyzed the characteristic initial value problem for plane-symmetric spacetimes with weak regularity. 
Next, weakly regular vacuum spacetimes with $\Tbb^2$ symmetry were defined and studied in~\cite{LeFlochSmulevici-2015}; in addition, the class of Gowdy spacetimes was particularly investigated in \cite{LeFlochLeFloch-1}. 


\subsubsection{Dynamics of compressible fluids with shocks} 
 
Importantly, our work deals with weak solutions that may have arbitrary large amplitude and may contain vacuum regions, by extending known techniques to the Einstein--Euler spacetimes under consideration. This requires new observations concerning the \emph{nonlinear coupling} between the fluid and the geometric parts for the Cauchy problem. In particular, 
we must consider balance laws and cope with geometric corrections to the fluid dynamics, which are of lower order in the fluid variables and thus turn out to be
controlled at the level of weak regularity under consideration.
 
The mathematical analysis of compressible fluid flows was initiated by DiPerna in his breakthrough work~\cite{DiPerna-1983a,DiPerna-1983b}. While bounded solutions were treated by DiPerna, the generalization to the broader class of solutions enjoying natural energy bounds (as we do in the present paper) was discovered only more recently by LeFloch and Westdickenberg~\cite{LeFlochWestdickenberg-2007} and extended by Germain and LeFloch~\cite{PG-PLF}.

Let us recall that several challenges arise with fluid flows containing shock waves. Due to the generic phenomenon of shock formation from smooth initial data, only weak solutions understood in the distributional sense can be sought for. Moreover, vacuum regions with $\mu=0$ (or $J_0=0$) may appear in finite time from non-vacuum initial data; in such a vacuum region, the Euler equations are trivially satisfied since all 
flux terms are proportional to~$J_0$. This property is consistent with the physical intuition that the velocity is ill-defined in the absence of a fluid, but must also be taken into account in designing a notion of solution. 
 In this regime, the equations (stated in~\eqref{eq:T2-fluid23-def-a}--\eqref{eq:T2-fluid23-def-b}, below) satisfied by the parallel momentum are trivially satisfied while the equations for the orthogonal momentum reduce to energy equations for the metric coefficients only. This difficulty is unavoidable when dealing with compressible matter, since vacuum regions may arise from non-vacuum initial data. 
 

\subsubsection{Back-reaction phenomenon and oscillatory limits}

The instability case of \autoref{thm:stability0} above may be viewed as a realization of Einstein's intuition that matter terms\footnote{The authors first became aware of this phenomenon through discussions with Luc Tartar, who understood it as arising naturally from the theory of compensated compactness and as expressing, at the physical level, the emergence of effective matter from oscillatory gravitational degrees of freedom.}  could arise as ``singularities'' of the gravitational field. In the physics literature, matter terms arising from oscillatory limits of \emph{smooth vacuum spacetimes} are discussed under the terminology of \emph{back-reaction} from high-frequency waves. In particular, small-scale inhomogeneities and their effective contribution to the large-scale dynamics have been investigated in cosmology by Green and Wald~\cite{Green-Wald} and the references therein. 

Rigorous mathematical statements were established by the authors \cite{LeFlochLeFloch-1} for the class of Gowdy spacetimes and, without symmetry assumptions, in the work of Lott~\cite{Lott-2018a}, Huneau and Luk~\cite{HuneauLuk-2018b,HuneauLuk-2020, HuneauLuk-2025}, and Touati~\cite{Touati-2023,Touati-2025}, by developping argument of (trilinear) compensated compactness for \emph{sufficiently regular} spacetimes ---without symmetry restriction. See also the historical review by Kiessling and Tahvildar-Zadeh~\cite{KiesslingTahvildar}. 

Our result fits within this broad perspective and, on the present symmetry class, is the first mathematical result in a \emph{weak regularity setting}, proving that oscillations in the geometry generate effective matter terms in the limit, which we encode by a corrector stress-energy tensor. We also mention that other interesting asymptotic regimes for matter models have recently been studied, for instance by Petropoulos et al.\ in connection with the so-called Carrollian limit~\cite{Petropoulos-2024}. 
  

\subsection{Overview of \autoref{part-one}: the theory of $\Tbb^2$-symmetric Einstein spacetimes with weak regularity}
\label{section=1-5}

\subsubsection{The  \JKL\ first-order formulation (\autoref{section=2})} 

The formulation introduced in this paper is presented in \autoref{def-quadr}. Its main unknowns are the fields $\Jperp$, $\Jpar$, $\Kpar$, and $\Lbb$ on the quotient manifold $\Ncal=\Mcal/\Tbb^2$, together with the quotient metric~$g_{\Ncal}$ described by the lapse function~$\Omega$ and the conformal length density~$\amdeux$ associated with an areal foliation; cf.~\autoref{section=2-1}. More precisely:
\bei

\item $\Jperp$ is a real-valued one-form describing the component of the fluid velocity orthogonal to the $\Tbb^2$~orbits, pushed forward to~$\Ncal$ and suitably weighted by the fluid density and pressure.

\item $\Jpar$ is a complex-valued scalar field encoding the weighted fluid velocity parallel to the $\Tbb^2$~orbits.

\item $\Kpar$ is a complex-valued scalar field describing the two twist parameters.

\item $\Lbb=i\,d\tau/\Im\tau$ is a complex-valued one-form encoding the modular parameter $\tau\in\CC$, $\Im\tau>0$, of the $\Tbb^2$~orbits.

\eei

\noindent Complex-valued fields arise naturally when parametrizing the $\Tbb^2$~orbits as discrete quotients of~$\CC$. With the notation
\be
\Jbb=(\Jperp,\Jpar), \qquad \Lbb=\Pbb+i\,\Qbb,
\ee
our first-order system of evolution and constraint equations for $(\Jbb,\Kpar,\Lbb)$ consists of

\bei

\item four nonlinear hyperbolic balance laws for~$\Jbb$,

\item four evolution and constraint equations for~$\Kpar$, 

\item four nonlinear wave equations for~$\Lbb$, 

\item together with two evolution equations for the conformal length density~$\amdeux$ and lapse function~$\Omega$, and one constraint equation for~$\Omega$.

\eei
\noindent These equations are stated in \eqref{eq:T2-1234-def}--\eqref{eq:theconstraints-def}, below, and their geometric and analytic structure is described in \autoref{theo-struct}. This choice of variables is especially well suited to the weak regularity theory developed in this paper, and already reveals several key features of the Einstein--Euler system, including the entropy current and quasi-current structures discussed in \autoref{section=3}.


\subsubsection{Hidden entropy structure of the Euler equations (\autoref{section=3})}

The next section uncovers one important feature: the Euler equations in torus symmetry possess a rich hidden entropy structure. This structure is the basis of our weak solution framework and of the compactness arguments used later in the global existence and stability theory. Starting from the geometric notion of entropy current introduced in \autoref{def-entropcurrent}, we present several nontrivial families of balance laws for the Euler system. These include the particle number current~\eqref{equa--Nbb-1}, whose main properties are summarized in \eqref{equa-2311}--\eqref{equa--Nbb-17}, as well as entropy currents associated with the parallel momentum variables through \eqref{equa-entrop23} and~\eqref{equa-entrop23-cons}.

A key point is that the standard notion of entropy current is not sufficient for the compactness analysis, except in situations where the parallel fluid components are assumed to vanish. This leads us to introduce a broader class of objects, called \emph{quasi-currents}; see \autoref{def-quasi-cur}. In curved geometry these satisfy  quasi-balance laws (derived in~\eqref{equa-balancelaws}, below) and, together with a convexity notion (cf.~\autoref{lem-convexM}) they provide the flexible entropy framework needed for compensated compactness and for the continuation of the Einstein--Euler evolution beyond shock formation.


\subsubsection{The class of tame Einstein--Euler flows (\autoref{section=4})} 

The aim is to formulate the Einstein--Euler system at the weak regularity level naturally associated with the evolution and constraint equations. 
First of all, the notion of $\Phi\Psi$ flow with finite energy, given in \autoref{weakdefinitionT2}, distinguishes between the $L^2$-type variables
\be
\Phi=(\Jperp,\Jpar, \Pbb,\Qbb)
\ee
and the \BV-type variables
\be
\Psi=(\ell,\log\Omega,\Kpar),
\ee
which, by definition, satisfy integrability and weak regularity conditions (cf.~\eqref{equa-integrable-BV}) relative to the natural spacelike and timelike volume forms, denoted by $\dVtrois$ and $\dVundeux$ (cf.~\eqref{equaDVtous}). These assumptions are dictated by the nonlinear structure of the first-order system. At first, all equations except the wave equation for the lapse are required to hold in the sense of distributions through a weak formulation of the Einstein equations and the fluid equations equations; cf.~\autoref{def-weaksolu}. We show the Einstein constraint equations propagate throughout the whole time interval, as stated in \autoref{weakdefinitionT2-2-propo}.  

Next, building on the entropy structure derived in \autoref{section=3}, we introduce our central notion, namely the class of \emph{tame solutions}. Specifically, we strengthen the weak formulation by requiring certain entropy balance laws to hold in the sense of distributions; see \autoref{def-weaksolu-deux}. This framework includes quasi-balance inequalities (weak analogues of~\eqref{equa-balancelaws}) associated with quasi-currents dominated by the reference current \(\Mbf^{0\bullet}\).

Thus, the hidden entropy structure uncovered in \autoref{section=3} becomes an essential ingredient in the definition of a robust class of weak solutions, one that will be shown to be \emph{stable under weak convergence}. In this sense, \autoref{section=4} forms the backbone of the present work: it introduces a notion of solution \emph{weak enough} to ensure existence for the Cauchy problem, yet \emph{strong enough} to provide the compactness properties needed for the analysis of the Einstein--Euler system.


\subsubsection{Compactness, stability, and instability (\autoref{section=5})}

A central objective of the paper is to analyze the nonlinear stability
properties of Einstein flows in \(\Tbb^2\) symmetry. To this end, in \autoref{section=5}, we
study suitably bounded sequences of tame Einstein--Euler flows, namely by assuming uniform
bounds on the natural \(L^2\) and \(\BV\) norms, and we determine their
possible weak limits. This yields a compactness framework broad enough to
encompass perturbations of initial data, vanishing-viscosity approximations,
and finite volume schemes.

The analysis relies on the decomposition into \(L^2\)-type variables
\(\Phi=(\Jperp,\Jpar,\Pbb,\Qbb)\) and \(\BV\)-type variables
\(\Psi=(\ell,\log\Omega,\Kpar)\), together with uniform bounds in the natural
norms. It is sufficient to assume that the field equations hold (possibly) only approximately, up to
negligible \(H^{-1}\) and measure error terms, and the entropy inequalities
are imposed in generalized form as well. This leads to a
comprehensive notion of tame sequence of Einstein--Euler flows, 
as introduced in \autoref{def-810}. 

The section begins by stating a general \emph{compactness result} for sequences of flows with corrector.
More precisely, \autoref{theorem-493} shows that every tame sequence
of (approximate) Einstein--Euler flows with corrector admits, after extraction
of a subsequence, a limiting tame Einstein--Euler flow with corrector.
This yields the nonlinear instability mechanism in its most general
form, since generally speaking the sequence of corrector measures \emph{does not converge} to the corrector measure of the limiting flow. 

 The notion of \emph{initially well-prepared sequence} is then introduced
in \autoref{def-well-p}, and the general instability theorem is refined into a
\emph{stability statement} for the corrector itself; see \autoref{theorem-494}. By
specializing to sequences without corrector, we obtain the \emph{nonlinear
stability theorem} for Einstein--Euler flows stated in
\autoref{theorem-492}. Finally, \autoref{corollary-493-init} extends these
compactness results to include (approximate) initial data sets.

Hence, \autoref{section=5} identifies both the stability mechanism and the instability mechanism governing weak limits of Einstein--Euler flows in the non-perturbative regime.


\subsection{Overview of \autoref{part-two}: the finite energy method for $\Tbb^2$-symmetric Einstein spacetimes} 
\label{section=1-6}

\subsubsection{A priori estimates for Einstein--Euler flows (\autoref{section=6})}

The second part of this paper is devoted to the proof of the main results. An important task is the derivation of \emph{a priori estimates for weak solutions}, which provide the analytic backbone of the whole weak regularity and compactness theory. We begin by introducing a geometric energy of an Einstein--matter spacetime $(\Mcal,\guntrois,\mu,u)$ with $\Tbb^2$ symmetry on $\Tbb^3$, namely the function $\Ebfscri$ of the areal time~$t$ defined by
\bel{intro-energy}
\Ebfscri(t) \coloneqq \frac{2}{|t|} \int_{\Tbb^3} \Lie_{e_0}\bigl(\dVtrois\bigr). 
\ee
Here, $\gtrois$ denotes the Riemannian metric induced on the slices of constant areal time~$t$, $\dVtrois$ is its volume form, $e_0$ is the future-pointing timelike unit vector field orthogonal to these slices\footnote{We do not need to distinguish explicitly the frame vectors $e_0,e_1$ on~$\Ncal$ and the corresponding frame vectors on~$\Mcal$ orthogonal to the $\Tbb^2$ symmetry orbits, since they are identified by push-forward along the quotient map.}, and $\Lie$ denotes the Lie derivative. Thus, $\Ebfscri(t)$ measures how the volume of a slice of constant areal time varies under shifts by constant proper time. Once the Einstein equations are imposed, this functional becomes a natural square-integrability norm for the first-order geometric variables and the matter fields.

In \autoref{section=6}, starting from an arbitrary tame Einstein--Euler flow, possibly with corrector,  and prescribed initial data, we derive a collection of properties: \emph{monotonicity, integrability,} and \emph{pointwise boundedness} on both the geometric and fluid variables. A first group of estimates concerns the spacelike volume~$\Vcal(t)$, the conformal length functional~$\Lbfscri(t)$, and the energy~$\Ebfscri(t)$ itself (defined in \eqref{equa-333}, \eqref{equa-conf-vol}, and \eqref{equa--energy}, below). The corresponding estimates are stated in Lemmas~\ref{lemma-31}, \ref{lem-conf-vol}, and~\ref{prop:energy}, below. These quantities enjoy monotonicity properties along the areal foliation, yielding a quantitative control of the geometry in both the expanding and contracting regimes.

A second group of estimates provides two-sided bounds for the lapse~$\Omega$, control of the conformal length measure~$d_x\ell$, and pointwise estimates for the twist variables~$\Kpar$, which are stated in~Lemmas~\ref{lemma-sup-Omega}, \ref{lem-speed}, and~\ref{lem-apriori--twist}, below. We also derive spacetime and timelike-slice estimates for the main first-order variables.   
Altogether, these estimates show that the weak regularity assumptions introduced in \autoref{section=4} \emph{are propagated} by the Einstein--Euler evolution and provide the uniform bounds needed later in the nonlinear stability and instability theory.


\subsubsection{A priori estimates for tame Einstein--Euler flows (\autoref{section=7})}

The next section complements the \emph{a priori} estimates of \autoref{section=6} by deriving maximum principles and additional integrability properties for the fluid and geometric variables. A first objective is to refine our control of the parallel momentum and the twist variables in sup norm. For sufficiently regular solutions, this can be seen directly from the transport equations satisfied by the metric-weighted quantities $\Gammaunpar$ and $\Gammazero$, which are propagated along the fluid characteristics. This leads to the maximum principle stated in \autoref{proposition-JJJ}, showing that the metric-weighted parallel momentum and twist functions remain bounded pointwise in terms of their initial values. A second, more robust, argument is then developed at the level of weak solutions. It relies on the entropy structure introduced in \autoref{section=3}, and more precisely on the $H$-divergence law (as we call it) satisfied by currents of the form $H(\Gammaunpar)\vNumb$ and $H(\Gammazero)\vNumb$; cf.~\eqref{equa-mom-twists0}. This \emph{entropy-based method} yields the same maximum principles without appealing to characteristics, and therefore applies in the weak regularity setting relevant to the global theory.

The section then turns to an additional integrability property for the Einstein--Euler system. By combining the Euler equations with the estimates established earlier for the lapse, the twists, and the energy, we derive a nonlinear integral estimate involving suitable quadratic and quartic expressions built from $\Pbb$, $\Qbb$, $\Jperp$, and $\Jpar$ (see \autoref{theo:te-Euler-explicit}, below). In particular, we obtain spacetime control of null forms such as $\Xbb\cdot\Ybb$ and $\Xbb\wedge\Ybb$, of $|\Jbb\cdot\Jbb|$, and of the parallel momentum in $L^4$. Together with the maximum principles, these estimates provide the integrability properties needed to justify the entropy inequalities, to control nonlinear products, and ultimately to carry out the compensated compactness argument. 


\subsubsection{Compactness properties of metric components (\autoref{section=8})}

Next, we establish the compactness properties needed in the stability and instability theorems of \autoref{section=5}. Our strategy is hierarchical. Starting from a bounded sequence of tame flows, we first analyze the geometric variables: the conformal length and the lapse, then the essential geometry variables \(\Pbb,\Qbb\), and finally the twist functions. These variables satisfy either linear transport equations or div-curl systems whose source terms are controlled by the estimates derived in \autoref{section=6} and \autoref{section=7}. In the well-prepared regime, this yields strong compactness of the geometry; for general data, it yields a weaker compactness statement, still sufficient to pass to the limit, but allowing the appearance of a corrector in the lapse equations.

A central motivation for the compactness analysis is the following general mechanism, which goes back to compensated compactness in the sense of Tartar~\cite{Tartar1,Tartar2} and Murat~\cite{Murat-1974}. In our formulation, Einstein's equations couple nonlinear hyperbolic balance laws with generalized wave-map equations, and the nonlinear terms are organized so as to exhibit a div-curl structure. The relevant phenomenon is already visible in the following model, in arbitrary spatial dimension \(N\geq1\):
\bse
\be
\Box \phi_a = Q_a(\partial \phi,\partial \phi;\phi,t,x),
\ee
where \(\Box\) is the wave operator in \(\RR^{N+1}\), \(\phi=(\phi_a)_{a=1,\dots,A}\) is the unknown, and each quadratic term \(Q_a(\partial\phi,\partial\phi;\phi,t,x)\) is a linear combination, with coefficients depending on \(\phi,t,x\), of the null forms
\be
Q_{0}(\partial \phi_a,\partial \phi_b)
= - \partial_t \phi_a\,\partial_t \phi_b + \sum_j \partial_j \phi_a\,\partial_j \phi_b,
\qquad
Q_{jk}(\partial \phi_a,\partial \phi_b)
= - \partial_j \phi_a\,\partial_k \phi_b + \partial_k \phi_a\,\partial_j \phi_b.
\ee
If \(\phi^\suit=(\phi_a^\suit)\) is a sequence of solutions to
\[
\Box \phi_a^\suit = Q_a(\partial\phi^\suit,\partial\phi^\suit;\phi^\suit,t,x),
\]
uniformly bounded in \(H^1\), and if the sequence of right-hand sides
\[
Q_a(\partial\phi^\suit,\partial\phi^\suit;\phi^\suit,t,x)
\]
is relatively compact in \(H^{-1}\), then any weak limit \(\phi^\sharp=(\phi_a^\sharp)\) is again a weak solution:
\[
\Box \phi_a^\sharp = Q_a(\partial\phi^\sharp,\partial\phi^\sharp;\phi^\sharp,t,x).
\]

Indeed, the first derivatives of \(\phi^\suit\) have the structure of gradients, hence satisfy the corresponding curl-free compatibility relations, while the wave equations control their divergences. The div-curl lemma therefore applies, and the quadratic null forms are stable under weak convergence. In the framework of the present paper, we further establish spacetime integrability properties ensuring that the relevant quadratic expressions are bounded in \(L^1\); by Murat's lemma, this is sufficient to obtain the required \(H^{-1}\) compactness.
\ese
This is the fundamental compactness mechanism underlying the whole paper. It hints at how the entropy structure, the maximum principles, the additional integrability estimates, and the div-curl structure interact so as to allow passage to the limit in the full Einstein--Euler system.


\subsubsection{Compactness properties of fluid variables (\autoref{section=9})}

The fluid variables require a separate argument. The general methodology for dealing with weak solutions to the compressible Euler equations goes back to the pioneering work of DiPerna~\cite{DiPerna-1983a,DiPerna-1983b} and its subsequent developments. 
In the present paper, we are especially motivated by the extension of DiPerna's theory  to (possibly unbounded) \emph{weak solutions with finite energy}, a class of solutions first defined and investigated by LeFloch and Westdickenberg~\cite{LeFlochWestdickenberg-2007}. A major obstacle is the construction of a sufficiently rich class of nonlinear balance laws and entropy inequalities, together with an analysis of how geometric and matter terms are coupled within these laws. This is precisely the purpose of both Sections~\ref{section=3} and~\ref{section=9}, which are also supplemented with our follow-up paper \cite{LeFlochLeFloch-next}, which present a comprehensive theory of relativistic fluid flows with symmetry. 

For the orthogonal momentum, we exploit the entropy quasi-currents introduced in \autoref{section=3}. Their divergences are shown to be $H^{-1}$-compact, which allows us to apply a metric-weighted div-curl lemma and derive a Tartar commutation relation for the associated Young measure. Since the available class of quasi-currents is sufficiently large, this commutation relation forces the Young measure to collapse to a Dirac mass, and therefore yields strong convergence of the orthogonal fluid momentum. The parallel momentum is treated separately through the transport structure and maximum principles established in \autoref{section=7}, together with compactness derived from a bounded variation bound. 

In turn, in combination with the results in \autoref{section=8}, we prove both the nonlinear stability result for well-prepared data and the instability statement for general data, the latter through the appearance of a stress-energy corrector in the limiting geometry.
 

\subsubsection{Concluding observations (\autoref{section=10})}

\paragraph{Another choice of reference entropy.}

The final section is devoted to several extensions and complements of the theory developed in this paper. Its first purpose is to compare two natural entropy frameworks for the Euler equations in the present geometric setting. In the formulation adopted throughout the core of the paper, the particle number current~\(\vNumb\) is kept as an exact divergence equality, while the stress-energy current \(\Mbf^{0\bullet}\) is treated as the reference entropy current and therefore satisfies an entropy inequality. In \autoref{section=10}, we also discuss the reverse choice, in which \(\Mbf^{0\bullet}\) is kept in exact balance form and the particle number current is taken as the reference entropy current; cf.~\eqref{equa-entr-rrr}. Although both formulations agree at the level of regular solutions, they lead to different theories of weak solutions once shocks are present.

The two entropy frameworks do not have the same scope. The formulation developed in the present paper, based on the reference entropy current \(\Mbf^{0\bullet}\), allows for \emph{nontrivial parallel momentum}, namely \(J_2\not\equiv0\) and \(J_3\not\equiv0\), and is therefore adapted to the full Einstein--Euler system in torus symmetry. By contrast, the alternative formulation based on the particle number current is restricted\footnote{due to limitations of the method of compensated compactness and the lack of a sufficiently rich family of suitable entropy currents} to the reduced regime \(J_2=J_3=0\), in which the quasi-currents introduced earlier reduce to standard entropy currents. This comparison clarifies the structural role of our entropy choice and explains why the framework based on \(\Mbf^{0\bullet}\) is the one best suited to the non-perturbative regime considered in this paper.


\paragraph{Other matter models.}

Beyond compressible fluids, \autoref{section=10} indicates that the methods developed here apply more broadly. Scalar fields fit naturally into the present first-order framework, while a cosmological constant $\Lambda$ introduces only lower-order modifications and leaves the main hyperbolic, entropy, and compactness structures unchanged. More generally, the discussion encompasses vacuum, scalar fields, and several fluid models, and points toward further extensions to other symmetry classes.

More conceptually, under symmetry, the Einstein equations retain a robust nonlinear structure even at weak regularity, provided one works with suitable first-order variables and an adapted entropy framework. This structure is flexible enough to accommodate both fluid shocks and impulsive gravitational waves, while still allowing one to formulate a global evolution and compactness theory.

The paper also includes an appendix containing technical material: expressions of the Vielbein and spin connection (\autoref{appendix=A}), simplified expressions for isothermal fluids (\autoref{appendix=B}), structure of the Einstein--Euler system (\autoref{appendix=C}), 
further properties of the Euler equations (\autoref{appendix=D}), and a density property (\autoref{appendix=E}).


\subsection{Summary of notation}
\label{section=1-7}

For the reader's convenience, we collect here the main notation used throughout the paper. The quotient manifold is denoted by $\Ncal=\Mcal/\Tbb^2$, and the basic geometric unknowns are the lapse~$\Omega$, the conformal length measure $d_x\ell=\amdeux\,dx$, 
the twist variable
$\Kpar=K_2+i\,K_3$,
and the essential geometric variables
$\Pbb=(P_0,P_1)$ and $\Qbb=(Q_0,Q_1)$ collected as $\Lbb=\Pbb+i\,\Qbb$.
The fluid momentum is written $\Jbb=(\Jperp,\Jpar)$ with $\Jperp=(J_0,J_1)$ and $\Jpar=(J_2,J_3)$,
with
\be
J_0 \leq 0 ,
\qquad
-\Jbb\cdot\Jbb = J_0^2-J_1^2-J_2^2-J_3^2 \geq 0,
\qquad
\mu+p(\mu)=-\frac12\Jbb\cdot\Jbb.
\ee
The thermodynamical quantities are the pressure law \(p=p(\mu)\), the pressure ratio
\be
q(\mu)=\frac{\mu-p(\mu)}{\mu+p(\mu)},
\ee
the particle number density \(\Numb(\mu)\), defined by
\be
\frac{d\Numb}{\Numb}=\frac{d\mu}{\mu+p(\mu)},
\ee
and the metric-weighted particle density
\be
\hNumb(\mu)=\Numb(\mu)\bigl(2(\mu+p(\mu))\bigr)^{-1/2}.
\ee
The particle number current is then
$\vNumb=\hNumb_\Jbb\,\Jbb$. 
We also use the parallel momentum per particle \(\Jhatpar=(\hNumb_\Jbb)^{-1}\Jpar\), together with its metric-weighted version \(\Gammaunpar\), and similarly the metric-weighted twist variables \(\Gammazero\).

Several quadratic expressions occur repeatedly. We use
\be
\Ebf_0(X_0,X_1)=\frac12(X_0^2+X_1^2),
\qquad
\Ebf_1(X_0,X_1)=X_0X_1,
\ee
and the stress-energy coefficients
\be
\aligned
\Mbf^{00}(\Jbb,\Kpar,\Lbb)
&=
\Ebf_0(\Jperp,\Pbb,\Qbb)
+\Ebf_0(\Jpar,\Kpar)
+\frac{q_\Jbb}{2}\,(-\Jbb\cdot\Jbb),
\\
\Mbf^{01}(\Jbb,\Kpar,\Lbb)
&= - \Ebf_1(\Pbb) - \Ebf_1(\Qbb) - \Ebf_1(\Jperp) ,
\\
\Mbf^{11}(\Jbb,\Kpar,\Lbb)
&=
\Ebf_0(\Jperp,\Pbb,\Qbb)
-\Ebf_0(\Jpar,\Kpar)
-\frac{q_\Jbb}{2}\,(-\Jbb\cdot\Jbb),
\endaligned
\ee
where \(q_\Jbb=q(\mu)\). Throughout the paper we assume
\be
q\in(0,1],
\qquad
-\Jbb\cdot\Jbb\geq 0,
\ee
with $q=1$ allowed only for $-\Jbb\cdot\Jbb=0$, namely $\mu=0$.
The main geometric functionals are the spacelike volume \(\Vcal(t)\), the conformal length \(\Lbfscri(t)\), and the energy \(\Ebfscri(t)\); cf.~\eqref{equa-333}, \eqref{equa-conf-vol}, and \eqref{intro-energy}. Finally, the main notions of solution introduced in the paper includes: flows with finite energy, (tame) Einstein--Euler flows, and their extensions to flows with corrector. For the reader's convenience, a summary of the main notation and terminology is given in Tables~\ref{table-notation-1}, \ref{table-notation-2}, and~\ref{table-def}.


\begin{table}[H]
\centering
\begin{tabular}{|l|l|}
\hline
\textbf{Notation} & \textbf{Meaning} \\
\hline
\(\Ncal=\Mcal/\Tbb^2\) & quotient manifold \\
\hline
$\Phi = (\Jbb, \Lbb) = (\Jperp, \Jpar, \Pbb, \Qbb)$ & $L^2$-type fields \\
\hline
$\Psi = (\ell, \log\Omega, \Kpar)$ & $\BV$-type fields \\
\hline
\(\Omega\) & lapse function \\
\hline
\(d_x\ell=\amdeux\,dx\), \(\ell(t,x)=\int_0^x \amdeux(t,x')\,dx'\) & conformal length \\
\hline
\(\Jbb=(\Jperp,\Jpar)\), \(\Jperp=(J_0,J_1)\), \(\Jpar=(J_2,J_3)\) & fluid momentum variables \\
\hline
\(\Kpar=K_2+i\,K_3\) & twist variables \\
\hline
\(\Lbb=\Pbb+i\,\Qbb\), \(\Pbb=(P_0,P_1)\), \(\Qbb=(Q_0,Q_1)\) & essential geometric variables \\
\hline
\(\mu+p(\mu)=-\Jbb\cdot\Jbb/2\) & relation between enthalpy and momentum \\
\hline
\(p(\mu)\), \(q(\mu)=\dfrac{\mu-p}{\mu+p}\vphantom{\biggl|}\) & pressure law and pressure ratio \\
\hline
\rule{0pt}{12pt}\(\hNumb(\mu)\) & enthalpy-normalized particle number density \\
\hline
\end{tabular}
\caption{Main geometric and fluid notation.}
\label{table-notation-1}
\end{table}


\begin{table}[H]
\centering
\begin{tabular}{|l|l|}
\hline
\textbf{Notation} & \textbf{Meaning} \\
\hline
\rule{0pt}{12pt}\(\vNumb = \hNumb_\Jbb\,\Jbb\) & particle number current \\
\hline
\rule{0pt}{13pt}\(\Jhatpar=(\hNumb_\Jbb)^{-1}\Jpar\) & parallel momentum per particle \\
\hline
\rule{0pt}{13pt}\(\Gammaunpar\), \(\Gammazero\) & metric-weighted parallel momentum and twists \\
\hline
\(\Ebf_0,\Ebf_1\) & basic quadratic forms~\eqref{eq:E0E1-T2-3lignes} \\
\hline
\(\Mbf^{00},\Mbf^{01},\Mbf^{11}\) & total energy-momentum tensor~\eqref{eq:T2-Mdef-0} \\
\hline
\rule{0pt}{12pt}\(\dVtrois,\dVundeux,\dVuntrois\) & spacelike, timelike, and spacetime volume forms \\
\hline
\(\Vcal(t)\), \(\Lbfscri(t)\), \(\Ebfscri(t)\) & volume, conformal length, and energy functionals \\
\hline
\end{tabular}
\caption{Derived quantities, quadratic expressions, and main functionals.}
\label{table-notation-2}
\end{table}

\begin{table}[H]
\centering
\begin{tabular}{|l|l|}
\hline
\textbf{Terminology} & \textbf{Reference} \\
\hline
first-order \((\Jbb,\Kpar,\Lbb)\) formulation & \autoref{def-quadr}\\
\hline
\(\Phi\Psi\) flow with finite energy & \autoref{weakdefinitionT2} \\
\hline
Einstein--Euler flow & \autoref{def-weaksolu} \\
\hline 
tame Einstein--Euler flow & \autoref{def-weaksolu-deux} \\
\hline
\(\Phi\Psi\Pi\)  flow with finite energy & \autoref{weakdefinitionT2-relaxed} \\
\hline
tame Einstein--Euler flow with corrector & \autoref{weakdefinitionT2-deux2} \\
\hline
initial data set & \autoref{def-initialdata} \\
\hline
\end{tabular}
\caption{Main formulations and notions of weak solutions.}
\label{table-def}
\end{table}


\clearpage 

\part{The theory of $\Tbb^2$-symmetric Einstein spacetimes with weak regularity} 
\label{part-one}

\section{The  \JKL\ first-order formulation}
\label{section=2} 

\subsection{Spacetime metrics in areal gauge} 
\label{section=2-1}

\subsubsection{Global areal time}

Throughout, we assume a $\Tbb^2$ symmetric spacetime with $\Tbb^3$ spatial topology, and we foliate the spacetime $(\Mcal, \guntrois)$ under consideration with spacelike hypersurfaces of constant areal time denoted by~$t$. Namely, we introduce a global time coordinate $t\colon\Mcal \to \Interval \subset \RR \setminus \{0\}$ that coincides (up to a sign) with the area of the $\Tbb^2$-orbits of symmetry. Here, $\Interval$ denotes a (compact or non-compact) interval that does not contain~$0$. Most of our definitions and results will be stated on a compact interval $\Interval=[t_0, t_1]$; however, our definition extends immediately to semi-open intervals $\Interval = [t_0, t_*) \subset (0, + \infty)$ or $\Interval= [t_0, t_*) \subset (- \infty, 0)$. Einstein's constraint equations, along with the positive energy condition satisfied by the matter model, imply that the gradient $\nabla\abs{t}$ of the area function is a \emph{timelike vector field} (cf.~\cite{Rendall-book} and the references cited therein). We define the sign of~$t$ so that $\nabla t$ is \emph{future-oriented}. Hence, positive~$t$ and negative~$t$ correspond to \emph{future-expanding} and \emph{future-contracting} spacetimes, respectively. Furthermore, in this section, we assume that all objects under consideration are sufficiently regular. 


\subsubsection{Modular parameter}

Besides the area~$|t|$, the metric induced on each $\Tbb^2$~orbit is characterized by the \emph{modular parameter}, denoted $\tau\in\CC$ with $\Im\tau>0$, which describes the conformal structure of the discrete quotient $\CC/(\ZZ+\tau\ZZ)$. The modular parameter is defined up to discrete transformations 
\be
\tau\mapsto (a\tau+b)/(c\tau+d) \text{ with }
\begin{pmatrix}a&b\\c&d\end{pmatrix} \in SL(2,\ZZ), 
\ee
which can be gauge-fixed by choosing a pair of reference one-cycles inside~$\Tbb^2$ (with intersection number~$1$). We focus our attention on spacetimes whose spacelike hypersurfaces have $\Tbb^3$ topology, ensuring a global choice of such cycles; thus the modular parameter $\tau$ is periodic in space. For the remainder of this paper, we decompose $\tau=Q+ie^{-P}$ into its real and (positive) imaginary parts. Although we assume $\Tbb^3$ topology, our method is general and should also apply to the case where Cauchy surfaces are realized as certain non-trivial $\Tbb^2$ bundles over the circle~$\Sbb^1$, including the comparatively simpler case of (vacuum) \emph{twisted Gowdy spacetimes} introduced in~\cite{Rendall-twisted}. In that setting, the spatial periodicity of~$\tau$ is replaced by periodicity up to an $SL(2,\ZZ)$ transformation.


\subsubsection{Lapse function and conformal length density}

We consider the (two-dimensional) \textbf{quotient manifold} ${\Ncal \coloneqq \Mcal/ \Tbb^2}$ and the associated \emph{quotient metric}~$g_{\Ncal}$, defined from the four-dimensional metric via the assumed action by the $\Tbb^2$ group of symmetries.  By construction, the areal time~$t$ remains constant along the $\Tbb^2$ orbits, and thus serves as a time coordinate on this $(1+1)$-dimensional Lorentzian manifold. We define the \emph{conformal factor} $\Omega>0$ via the squared norm $g_{\Ncal}^{-1}(dt,dt)$, namely 
\be
\Omega^{-2} \coloneqq - g_{\Ncal}^{-1}(dt,dt). 
\ee 
In a spatial coordinate~$x$ chosen so that its level sets are orthogonal to those of~$t$, the quotient metric $g_{\Ncal}$ on~$\Ncal$ takes the form
\bel{quotient-metric-gN}
g_{\Ncal} = \Omega^2 (- dt^2 + \amdeux^2 dx^2).
\ee 
More precisely, we define the coordinate~$x$ (up to an overall constant shift) along the initial data surface $\{t=t_0\}\simeq\Tbb^3/\Tbb^2$ by prescribing\footnote{We may choose $\amdeux |_{t=t_0}$ to be a constant.  In the vacuum and under Gowdy symmetry, this choice is natural as $\amdeux$ then remains constant throughout the spacetime. For general spacetimes, however, this choice depends on the initial time slice.} the function $\amdeux|_{t=t_0}$ such that $x$ is periodic with period~$1$. The results we establish in this paper ensure the existence of a coordinate~$x$ for which the quotient metric takes the form~\eqref{quotient-metric-gN}. 

The conformal coefficient $\Omega$ is commonly referred to as the \textbf{lapse function} of the foliation, while we will dub the function $\amdeux$ the \textbf{conformal length density}.  It is an (inverse) speed of light relating the null slopes in the quotient metric $g_{\Ncal}=\Omega^2(-dt^2+\amdeux^2 dx^2)$ and therefore will appear as a speed in key differential operators relevant for the geometry variables.


\subsubsection{Adapted moving frame}

We now construct an \emph{adapted moving frame,} namely an orthonormal basis of vector fields on~$\Mcal$ 
\be
e_0,e_1,e_2,e_3 \text{ that remain invariant under the $\Tbb^2$ symmetry}.
\ee
First, $e_0$~is the future-directed unit vector normal to constant-$t$ hypersurfaces, while $e_1$~is a (unit) vector orthogonal to $e_0$ and to the $\Tbb^2$~symmetry orbits. Next, $e_2$~is a unit vector tangent to $\Tbb^2$ orbits (hence orthogonal to $e_0$ and $e_1$) and aligned with the first reference one-cycle mentioned above. Finally, $e_3$ is a unit vector orthogonal to the other three basis vectors. For any given spacetime with $\Tbb^2$~symmetry on~$\Tbb^3$, the choice of adapted moving frame depends only on discrete data: the cycle of~$\Tbb^2$ along which $e_2$~is aligned (with possible choices differing by $SL(2,\ZZ)$ transformations), and the orientations of $e_3$ and~$e_1$.
The modular parameter~$\tau$ depends additionally on the second reference cycle.


\subsubsection{Metric decomposition}

In the so-called areal gauge, with the coordinates $t,x$ introduced above, any $\Tbb^2$ symmetric spacetime metric $\guntrois$ on the torus $\Tbb^3$ can be written in the form 
\bel{metric:areal}
\guntrois = \Omega^2 (- dt^2 + \amdeux^2 dx^2) + \abs{t} e^P \bigl( dy + Q\,dz + (G+QH) \, dx\bigr)^2 + \abs{t} e^{-P} (dz + H\,dx)^2, 
\ee  
where the metric coefficients $P,Q,\Omega, \amdeux, G, H$, with $\Omega,\amdeux>0$, depend only on $t \in \Interval$ and $x \in \Sbb^1 \simeq [0,1]$, whereas the remaining `parallel' variables $y,z$ describe $\Tbb^2 \simeq [0,1]^2$. Observe that the vector fields $\del/\del y$ and $\del / \del z$ are Killing fields; for the existence of this metric decomposition, we refer to Rendall~\cite{Rendall-book}. In our analysis, we will sometimes find it convenient to specialize the general results to the following important cases. 
\bei 

\item \emph{Gowdy-symmetric spacetimes}, by definition, are characterized by the condition $G=H=0$ throughout the spacetime. Geometrically, this is equivalent, up to an affine coordinate change, to the vanishing of the  \emph{twist variables} $K_2, K_3$, which will be introduced below. This condition implies also the vanishing\footnote{This follows from \eqref{eq:221a}--\eqref{eq:221b}, below.} of the parallel fluid variables $J_2,J_3$. In addition, \emph{polarized} Gowdy-symmetric spacetimes are further characterized by $Q=0$. 

\item \emph{Vacuum $\Tbb^2$ symmetric spacetimes} are defined by setting $\mu = 0$ throughout the spacetime, in which case the Euler equations are automatically satisfied for an arbitrary (and physically irrelevant) velocity field~$u$. 

\eei


\subsubsection{Vielbein and spin connection}

In order to express the field equations for the metric~\eqref{metric:areal}, careful attention is required to obtain a tractable system suitable for analyzing the nonlinearities. First of all, we rewrite the metric and its inverse as 
\bel{equa-metric-inverse}
\guntrois_{\alpha\beta} = - e^0_\alpha e^0_\beta + e^1_\alpha e^1_\beta + e^2_\alpha e^2_\beta + e^3_\alpha e^3_\beta,
\qquad 
\guntrois{}^{\alpha\beta} 
 = - e_0^\alpha e_0^\beta + e_1^\alpha e_1^\beta + e_2^\alpha e_2^\beta + e_3^\alpha e_3^\beta, 
\ee 
namely in terms of the adapted moving frame defined above, also called \emph{Vielbein} and inverse Vielbein, respectively. We then introduce the notation 
\bel{eq:definevar-01}
\aligned 
P_0 & \coloneqq e_0(P) = \Omega^{-1} P_t, 
\,
& Q_0 & \coloneqq e^P e_0(Q) = \Omega^{-1} e^P Q_t, 
\\
P_1 & \coloneqq e_1(P) = \Omega^{-1} \amdeux^{-1} \, P_x,  
\,
& Q_1 & \coloneqq e^P e_1(Q) = \Omega^{-1}\amdeux^{-1} \, e^P Q_x,  
\\
K_2 & \coloneqq \Omega^{-2}   \amdeux^{-1}  \, |t|^{1/2} e^{P/2} \bigl ( G_t + QH_t \bigr),  
\, \qquad
& K_3 & \coloneqq \Omega^{-2}  \amdeux^{-1} \, |t|^{1/2} e^{-P/2} H_t.  
\endaligned
\ee  
The expressions of the Vielbein vectors and one-forms in~\eqref{equa-metric-inverse} can be found in~\autoref{appendix=A}, together with the spin connection.  


\subsection{A choice of first-order variables}
\label{section=2-2}

\subsubsection{Geometry variables}

The geometric unknowns $(\Kpar,\Lbb) = (\Kpar, \Pbb, \Qbb)$ are defined as follows. The Einstein equations are conveniently expressed in terms of the lapse~$\Omega$ and conformal length density~$\amdeux$ (which describe the quotient metric~\eqref{quotient-metric-gN}), the two one-forms $\Pbb$ and $\Qbb$ (whose components in the Vielbein basis are $(P_0, P_1)$ and $(Q_0, Q_1)$), and a complex scalar $\Kpar$ (associated with $K_2$ and $K_3$). Specifically, the \textbf{essential geometry variables} are defined as 
\bel{eq:definevar}
\aligned
\Pbb & \coloneqq (P_0, P_1), 
\quad
& \Qbb & \coloneqq (Q_0, Q_1),
\endaligned
\ee
as well as the twists functions
\be
\Kpar \coloneqq K_2 + i K_3.
\ee
Throughout, we tacitly assume the vanishing integrability conditions\footnote{These conditions make sense for weak solutions (defined later), at least for almost every (a.e.) time.} 
\bel{equa-P1Q1}
\int_{\Tbb^3} P_1 \, \dVtrois = \int_{\Tbb^3} e^{-P} Q_1 \, \dVtrois = 0 \quad \text{  at  } t \in \Interval,  
\ee 
where $\dVtrois = |t|\,\Omega\,\amdeux\,dx\,dy\,dz$,
which are required in view of~\eqref{eq:definevar-01}. Using the evolution equations derived below, one easily checks that these conditions hold for all times provided they hold on an initial Cauchy surface. We may also rely on complex-valued geometry variables, consisting of a \emph{complex-valued one-form}~$\Lbb$ and a \emph{complex-valued scalar field}~$\Kpar$ on~$\Ncal$:
\be
\Lbb = \Pbb + i\,\Qbb = \frac{i\,d\tau}{\Im\tau}, \qquad
\Kpar = K_2 + iK_3 = \guntrois\bigl(e_2+ie_3,[e_1,e_0]\bigr). 
\ee 
Here, $\tau=Q+ie^{-P}$ is the modular parameter which characterizes the conformal geometry of the $\Tbb^2$  orbits, while $[\, \cdot\,,\, \cdot\,]$ denotes the commutator of two vector fields on~$\Mcal$. 


\subsubsection{Fluid variables}

\bse
Throughout, the matter flow is assumed to enjoy the same $\Tbb^2$~symmetry as the spacetime, so that the fluid unknowns depend on the variables $t,x$, only. However, all four components of the velocity vector itself may be non-vanishing. Using the orthonormal frame geometrically defined above $(e_0, e_1, e_2, e_3)$ (linked to the assumed symmetry), we decompose the \emph{fluid velocity} $u$ as $u=u^m e_m$, with frame components \((u^0,u^1,u^2,u^3)\) and we define the \emph{momentum one-form} (or material current) $\Jbb$ 
by lowering the index with the metric~$\guntrois$, namely
\bel{appendixJmatter}
\Jbb \coloneqq (J_0, J_1, J_2, J_3), 
\qquad
J_m \coloneqq \sqrt{2(\mu+p)} \, u_m, 
\ee
Hence, we have 
\bel{equa-211b}
- \Jbb \cdot \Jbb = J_0^2 - J_1^2 - J_2^2 - J_3^2 = 2 (\mu+p(\mu)). 
\ee
Since $\mu+p(\mu) >0$, this field obeys the \textbf{fluid causality inequalities}
\bel{eq:J0J} 
J_0 \leq 0,
\qquad 
- \Jbb \cdot \Jbb \coloneqq J_0^2 - J_1^2 - J_2^2 - J_3^2 \geq 0, 
\ee
namely, $\Jbb$ is either timelike or null and future-oriented, or else vanishes.
Hence, our theory allows the momentum vector to be \emph{a non-zero null vector} ---this amounts to a vanishing density $\mu=0$. This is necessary if we want our notion of (weak) solution to be \emph{stable under (weak) convergence}.

In contrast to the geometric variables, which primarily reside on the quotient manifold~$\Ncal$, the one-form $\Jbb$ is naturally defined on the full spacetime~$\Mcal$. It is useful to decompose $\Jbb$ into components that are \emph{transverse} or \emph{parallel} to the $\Tbb^2$--orbits of symmetry, namely a one-form $\Jperp$ and a complex field $\Jpar$, defined as 
\be
\Jperp \coloneqq (J_0, J_1), \qquad \Jpar \coloneqq J_2+iJ_3. 
\ee
\ese
The formulation of the Einstein--Euler system proposed here is based on the set of fluid and metric variables $(\Jbb,\Kpar,\Lbb)$ and $(\Omega, \amdeux)$. 


\subsubsection{Quadratic forms}

We now introduce nonlinear quantities that play a crucial role in our presentation. For any one-form $\Xbb$ chosen among $\Pbb, \Qbb, \Jperp$ we define the quadratic forms 
\bse
\label{eq:E0E1-T2-3lignes}
\bel{eq:E0E1-T2}
\aligned
\Ebf_0(\Xbb) & \coloneqq \frac{1}{2} (X_0^2 + X_1^2),
& \qquad
\Ebf_1(\Xbb) & \coloneqq X_0 X_1. 
\endaligned
\ee
It is convenient to also set $\Ebf_0(\Lbb)=\Ebf_0(\Pbb,\Qbb)$ and use the short-hand notation $\Ebf_0(\Xbb, \Ybb) \coloneqq \Ebf_0(\Xbb) + \Ebf_0(\Ybb)$ ---and similarly for additional arguments and for the quadratic form~$\Ebf_1$. For $\Xpar=X_2+iX_3$ chosen among $\Kpar, \Jpar$ we use an analogous notation 
\bel{equa-Ezero-etc}
\Ebf_0(\Xpar) \coloneqq \frac{1}{2}(X_2^2 + X_3^2), \qquad
\Re(\Xpar^2)=X_2^2-X_3^2, \qquad \Im(\Xpar^2)=2 X_2X_3.
\ee
We also write
\be
\Ebf_0(\Jbb) \coloneqq \Ebf_0(\Jperp) + \Ebf_0(\Jpar) = \frac{1}{2}(J_0^2 + J_1^2 + J_2^2 + J_3^2).
\ee
We collectively refer to these expressions as energy and energy-flux terms.
\ese
\bse 
On the other hand, we will also use the following quadratic forms in $\Xbb, \Ybb$, chosen among $\Pbb, \Qbb, \Jperp$, 
\be 
- \Xbb \cdot \Ybb \coloneqq X_0 Y_0 - X_1 Y_1, 
\qquad 
\Xbb \wedge \Ybb \coloneqq X_0 Y_1 - X_1 Y_0.
\ee 
Finally, for the material current $\Jbb$ and its orthogonal part $\Jperp$, we denote
\bel{equa-2122}
- \Jperp \cdot \Jperp = J_0^2 - J_1^2 \geq - \Jbb \cdot \Jbb 
= J_0^2 - J_1^2 - J_2^2 - J_3^2 \geq 0.
\ee 
\ese


\subsubsection{Divergence and curl operators}

\bse
\label{equa-div-curl}
To proceed, we need a notation for several operators defined on $(\Mcal, \guntrois)$. We consider one-form fields $\Xbb= (X_0,X_1)$ on the quotient manifold~$\Ncal$,  and we introduce their \emph{divergence with respect to the quotient metric}~$g_{\Ncal}$ 
\bel{equa-divN} 
\divN(\Xbb)
\coloneqq \amdeux^{-1} \Omega^{-2}  \bigl(-( \amdeux \Omega X_0)_t + (\Omega X_1)_x\bigr).  
\ee
We emphasize that the divergence operator depends only on the metric coefficients $(\Omega,\amdeux)$. This formula differs slightly from the divergence with respect to the four-dimensional metric. For simplicity, we use the same notation $\Xbb$ (by a mild abuse) for the one-form field on~$\Mcal$ with the same components $X_0,X_1$ and $X_2=X_3=0$ in the Vielbein basis. The \emph{divergence with respect to the full spacetime metric} reads 
\bel{equa-def-div13}
\divuntrois \Xbb \coloneqq |t|^{-1} \divN(|t| \, \Xbb), 
\ee
where as usual we implicitly identify scalar fields on~$\Ncal$ with $\Tbb^2$-invariant scalar fields on~$\Mcal$. We recognize here the factor of $|t|$ due to the $\Tbb^2$ directions, with an absolute value to account for both time orientations.
\ese

We identify vector fields and one-forms on the quotient manifold~$\Ncal$ using the metric~$g_{\Ncal}$. This identification maps a one-form $\Xbb$ with components $(X_0,X_1)$ in the inverse Vielbein basis to a vector with the same components in the Vielbein basis ---up to a sign, namely ${X^0=-X_0}$ and ${X^1=X_1}$. Hence, the divergence operator can be applied equally well to vector fields or one-forms. 
\bse\label{equa-def-curl-0}
By differentiation of the one-form field $\Xbb$, we associate to it the two-form field $d\Xbb$ which has a single non-vanishing component, namely 
\be
d\Xbb= \bigl( \curl_\Ncal  \Xbb \bigr) \, e^0\wedge e^1, 
\ee
which leads us to define $\curl_\Ncal (\Xbb)= (d\Xbb)_{01}$. This is equivalent to applying the Hodge star operator ($X_0\leftrightarrow -X_1$) before taking the divergence, and yields the \emph{curl operator:}
\be
\curl_\Ncal\Xbb 
\coloneqq \amdeux^{-1} \Omega^{-2}  \bigl(( \amdeux \Omega X_1)_t - (\Omega X_0)_x\bigr),  
\ee
which, by analogy with the expression of the divergence above, leads us to set 
\bel{equa-def-curl}
\curluntrois \Xbb \coloneqq |t|^{-1} \curl_\Ncal( |t| \, \Xbb) .
\ee
In the following, it will be convenient to use the spacetime notation $\divuntrois $ and $\curluntrois$. 
\ese
%


\subsubsection{Ricci and matter components in the adapted frame}

We denote by $(\eta_{mn})$ the Minkowski metric with signature $(-1, 1, 1, 1)$. In the adapted frame introduced above, the Einstein equations take the following form for $m,n = 0, 1, 2, 3$: 
\bel{eq:Einsframe}
\aligned
  e_m^\alpha R_{\alpha\beta} e_n^\beta 
& = 
e_m^\alpha T_{\alpha\beta} e_n^\beta - \frac{1}{4-2}\bigl(T_{\alpha\beta} g^{(4)\alpha\beta}\bigr)  \eta_{mn} 
\\
& = (\mu+p) \, u_m u_n + \frac{1}{2} (\mu-p) \eta_{mn}, 
\endaligned
\ee
where $R_{\alpha\beta}$ denotes the components of the Ricci curvature $\Ric$. Rewriting the first term $(\mu+p) \, u_m u_n$ in terms of the vector~$\Jbb$, we obtain $\frac{1}{2} J_m J_n$, which \emph{does not} explicitly depend on the equation of state, whereas the second term does. For the second term, we regard $(\mu-p(\mu))/(\mu+p(\mu))$ as a bounded function $q\colon (0,+ \infty) \to (0,1)$ of $\mu$
\bel{eq:defq} 
q \coloneqq {\mu-p(\mu) \over \mu+p(\mu)} \in (0,1), \qquad \mu > 0,
\ee
referred to as the \textbf{pressure ratio}; in our first-order formulation, we find it convenient to denote it by 
\be
 q_\Jbb = q(\mu).
\ee
We point out that the dependence of $q$ with respect to $\mu$ (namely, through $-\Jbb \cdot \Jbb$) is essentially irrelevant at large densities. Namely, thanks to our mild assumptions in \eqref{hyperbolic-eos}--\eqref{equa-asym}, $p(\mu)/\mu$ is bounded away from the squared speed of light (metric-weighted to~$1$), hence the function $q_\Jbb \in (0,1)$ is bounded away from zero: for some constant~$q_{\min}$,
\bel{equa-qminqmax-1}
0 < q_{\min} \leq q(\mu)  \quad \text{ for all } \mu\in(0,+\infty) .
\ee
For large densities, $p(\mu)/\mu$ is also assumed to be bounded away from zero, which implies that there exists a constant $q_\infty^+$ such that
\bel{equa-qminqmax-2}
q(\mu) \leq q_\infty^+ < 1 \quad \text{for sufficiently large } \mu.
\ee
Finally, in the small density limit $\mu\to 0$, the function has a limit denoted by $q(0)$, with
\be
q(0) = \lim_{\mu \to 0} q(\mu) = \frac{1-p'(0)}{1+p'(0)} \in (0,1] .
\ee
In particular, for an equation of state with $p'(0)=0$ (e.g., in the polytropic case), one has $q(0)=1$.


Moreover, we define the (total) \textbf{energy-momentum tensor} 
\be
\Mbf=\Mbf(\Jbb, \Kpar, \Lbb) = \Mbf(\Jperp, \Jpar, \Kpar, \Pbb, \Qbb), 
\ee
which incorporates both geometric and matter contributions, via its components\footnote{Observe that $\Mbf^{01}(\Jbb, \Kpar,\Lbb)$ does not depend on~$\Kpar$, hence this argument will be freely omitted.}
\bel{eq:T2-Mdef-0}
\begin{alignedat}{2}
\Mbf^{00} & \coloneqq \Ebf_0(\Jbb, \Kpar, \Lbb) + {q_\Jbb \over 2}  \, (- \Jbb \cdot \Jbb)
&& = \Ebf_0(\Jperp, \Pbb, \Qbb) + \Ebf_0(\Jpar, \Kpar) + {q_\Jbb \over 2}  \, (- \Jbb \cdot \Jbb), 
\\
\Mbf^{01} & \coloneqq - \Ebf_1(\Jperp, \Lbb) 
&& = - \Ebf_1(\Jperp) - \Ebf_1(\Pbb) - \Ebf_1(\Qbb),
\\
\Mbf^{11} & \coloneqq \Ebf_0(\Jperp, \Lbb) - \Ebf_0(\Jpar, \Kpar) - {q_\Jbb \over 2}  \, (- \Jbb \cdot \Jbb)
&& = \Ebf_0(\Jperp, \Pbb, \Qbb) - \Ebf_0(\Jpar, \Kpar) - {q_\Jbb \over 2}  \, (- \Jbb \cdot \Jbb),
\end{alignedat}
\ee
together with $\Mbf^{10} = \Mbf^{01}$. Observe that $\Mbf^{00} \geq 0$, and these functions obey the inequalities
\bel{eq:T2-Mdef-1}
2 \, |\Mbf^{01}|\leq \Mbf^{00}+\Mbf^{11}, \qquad \quad - \Ebf_0(\Kpar)\leq \Mbf^{11}\leq \Mbf^{00}.
\ee 
To write the Euler equations we will also need the notation
\bel{eq:T2-Euler-perp2-def0}
\aligned
\Msource(\Jbb, \Kpar, \Lbb)
& \coloneqq \Ebf_0(\Jpar) + 5 \, \Ebf_0(\Kpar)
- \Pbb \cdot \Pbb -  \Qbb \cdot \Qbb - \Jperp \cdot \Jperp 
+ {q_\Jbb \over 2}  \, \Jbb \cdot \Jbb, 
\endaligned
\ee
which obeys the inequality (proven in \autoref{appendix=D-1})
\be
\abs{\Msource}\leq 5 \, \Mbf^{00}.
\ee
 
\begin{remark}
  The components of the stress energy tensor in the orthonormal frame~$(e_m)_{m=0,1,2,3}$ are given by $T^{mn}=\frac{1}{2}J^m J^n + (1/4) (1-q_\Jbb)(-\Jbb\cdot\Jbb) g^{mn}$, and in particular for $m,n=0,1$ they are restrictions of $(1/2)\Mbf^{mn}$ to its fluid part:
  \bel{eq:T2-Texpr}
  \begin{alignedat}{2}
    T^{00} & = \frac{1}{2} \Mbf^{00}(\Jbb, 0, 0) = \frac{1}{2} J_0^2 - \frac{1-q_\Jbb}{4} (- \Jbb \cdot \Jbb) , \qquad
    & \Mbf^{00}(\Jbb, \Kpar, \Lbb) & = 2 T^{00} + \Ebf_0(\Kpar, \Lbb) ,
    \\
    T^{01} & = \frac{1}{2} \Mbf^{01}(\Jbb, 0, 0) = - \frac{1}{2} J_0 J_1 , \qquad
    & \Mbf^{01}(\Jbb, \Kpar, \Lbb) & = 2 T^{01} - \Ebf_1(\Lbb) ,
    \\
    T^{11} & = \frac{1}{2} \Mbf^{11}(\Jbb, 0, 0) = \frac{1}{2} J_1^2 + \frac{1-q_\Jbb}{4} (- \Jbb \cdot \Jbb) , \qquad
    & \Mbf^{11}(\Jbb, \Kpar, \Lbb) & = 2 T^{11} + \Ebf_0(\Lbb) - \Ebf_0(\Kpar) .
  \end{alignedat}
  \ee
  Throughout this section, the fluid contributions in all equations could be replaced by components of the stress-tensor.  This provides an avenue to extend our work to other matter models.
\end{remark}


\subsection{A first-order formulation} 
\label{section=2-3}

\subsubsection{Definition}

Relying on the proposed variables in \autoref{section=2-2}, we now exhibit the basic algebraic and differential structure of the Einstein and Euler equations. Since the fluid mass-energy density may vanish in part of the spacetime, it is essential for our theory to include such vacuum regions. 

Specifically, we claim that the field equations take the form of the following coupled system of partial differential equations\footnote{A single equation is not included in this list (namely \eqref{eq:waveconffac}, below) since it is a consequence of the other equations \emph{and} will not be part of our first-order formulation.}:
\begin{subequations}\label{eq:T2-1234-def}
\begin{align}
\opL_1 & \coloneqq
\divuntrois(\Pbb)
-  \Qbb \cdot \Qbb + {1 \over 2} \Re(\Kpar^2 + \Jpar^2) =0
&
\opL_2 & \coloneqq \curluntrois (|t|^{-1} \Pbb)
= 0,
\label{eq:T2-1234-def-a}
\\
\opL_3 & \coloneqq \divuntrois(\Qbb)
+ \Pbb \cdot \Qbb + {1 \over 2} \Im(\Kpar^2 + \Jpar^2) = 0,
&
\opL_4 & \coloneqq \curluntrois(|t|^{-1} \Qbb)
- |t|^{-1} \Pbb \wedge \Qbb = 0 ,
\label{eq:T2-1234-def-b}
\end{align}
\end{subequations}
\vspace{-.7cm}
\begin{subequations}
\label{eq:T2-Euler-perp-def}
\begin{align}
\opL_5 & \coloneqq \divuntrois \bigl( \Omega \, \Mbf^{0 \bullet}(\Jbb, \Kpar, \Lbb) \bigr)
+ {1 \over 2\, t} \Msource(\Jbb, \Kpar, \Lbb) = 0,
\label{eq:T2-Euler-perp-def-a}
\\
\opL_6 & \coloneqq \divuntrois \bigl( \Omega \, \Mbf^{1 \bullet}(\Jbb, \Kpar, \Lbb) \bigr)
= 0,
\label{eq:T2-Euler-perp-def-b}
\\
\opL_7 & \coloneqq \divuntrois \bigl(|t|^{1/2} J_2 \Jperp\bigr)
+ \curluntrois  \bigl(|t|^{1/2} K_2 \Pbb / 2 \bigr) =0,
\label{eq:T2-fluid23-def-a}
\\
\opL_8 & \coloneqq \divuntrois  \bigl(|t|^{1/2} J_3 \Jperp\bigr)
+ \curluntrois  \bigl(|t|^{1/2} ( K_2 \Qbb - K_3 \Pbb / 2)\bigr) = 0.
\label{eq:T2-fluid23-def-b}
\end{align}
\end{subequations}
\vspace{-.7cm}
\begin{subequations}\label{eq:T2-all-suite-def}
\begin{align}
\opL_9 & \coloneqq (\amdeux)_t
 - {t\over 2}(\Mbf^{00}- \Mbf^{11})(\Jbb, \Kpar, \Lbb)\,\Omega^2 \amdeux = 0, 
\label{eq:T2-8-def}
\\
\opL_{10} & \coloneqq (\log\Omega)_t
 + \frac{1}{4t} - {t \over 2}\Mbf^{11}(\Jbb, \Kpar, \Lbb)\,\Omega^2 =0,
\label{eq:evollambda-def}   
\\
\opL_{11} & \coloneqq |t|^{-3/2} \bigl(  |t|^{3/2} K_2 \bigr)_t
 - \bigl( J_1 J_2 -  P_0 K_2/2 \bigr)\Omega = 0,
\label{eq:T2-9101112-evol-def-a}
\\
\opL_{12} & \coloneqq |t|^{-3/2} \bigl( |t|^{3/2} K_3 \bigr)_t
 - \bigl( J_1 J_3 + P_0 K_3/2 - Q_0 K_2 \bigr)\Omega, 
\label{eq:T2-9101112-evol-def-b} 
\end{align}
\end{subequations}
\vspace{-.7cm}
\begin{subequations}
\label{eq:theconstraints-def}
\begin{align}
\opL_{13} & \coloneqq (\log\Omega)_x
 + {t \over 2 } \, \Mbf^{01}(\Jperp, \Lbb)\, \Omega^2 \amdeux = 0, 
\label{eq:T2-567-def}
\\
\opL_{14} & \coloneqq ( K_2 )_x
 - \bigl( J_0 J_2 - P_1 K_2/2 \bigr) \, \Omega \amdeux = 0, 
\label{eq:221a}
\\
\opL_{15} & \coloneqq ( K_3 )_x
 - \bigl( J_0 J_3  - Q_1 K_2 + P_1 K_3/2 \bigr) \, \Omega \amdeux = 0. 
\label{eq:221b}
\end{align}
\end{subequations}
We formalize our standpoint as follows. Observe that the vacuum condition $\mu=0$ is equivalent to the condition $\Jbb \cdot\Jbb=0$, namely the fluid momentum vector is a null vector. 

\begin{definition}
\label{def-quadr}
With the notation above, the \textbf{first-order formulation} of the Einstein--Euler system (in $\Tbb^2$ symmetry on the torus $\Tbb^3$) consists of the set of evolution equations~\eqref{eq:T2-1234-def}, \eqref{eq:T2-Euler-perp-def}, and \eqref{eq:T2-all-suite-def}, together with the constraint equations \eqref{eq:theconstraints-def}, expressed in terms of the fluid variables $\Jbb=(\Jperp, \Jpar)$, the geometry variables $(\Kpar, \Lbb)$ with $\Lbb=\Pbb+i\,\Qbb$, the lapse ${\Omega}$, and the conformal length density~$\amdeux$, satisfying at each relevant time the positivity conditions    
\bel{eq:causal}
\amdeux > 0, \quad \Omega > 0, 
\qquad
J_0 \leq 0,
\qquad 
- \Jbb \cdot \Jbb = J_0^2 - J_1^2 - J_2^2 - J_3^2 \geq 0, 
\ee 
which are referred to as the (geometry-matter) \textbf{causality inequalities}, 
as well as the \textbf{periodicity constraints}
\bel{equa-pericons}
\int_{\Tbb^3} P_1 \, \dVtrois = \int_{\Tbb^3} e^{-P} Q_1 \, \dVtrois = 0 
\ee 
(which, as pointed out around \eqref{equa-P1Q1}, propagate from any Cauchy hypersurface).
\end{definition}

\begin{remark}
By integrating the constraints~\eqref{eq:theconstraints-def}, one observes the periodicity conditions
\bel{equa-peri-OmK}
\int \Mbf^{01}(\Jbb, \Lbb) \, \dVtrois
= \int e^{P/2} J_2 J_0 \dVtrois
= \int \bigl( Q e^{P/2} J_2 + e^{-P/2} J_3 \bigr) J_0 \dVtrois = 0. 
\ee 
\end{remark}

We also sometimes use the terminology \textbf{\JKL\ formulation} of the Einstein--Euler system to denote our first-order formulation, also written as 
\bel{equa-JKLoperators}
\opL_a(\Jbb, \Kpar, \Lbb, \Omega, \ell) = 0, \qquad a = 1, \ldots, 15.
\ee


\subsubsection{Structural properties}

The proof of the following properties is postponed to~\autoref{appendix=C-1}. We emphasize that our formulation is meaningful, in particular, in vacuum regions. 

\begin{proposition}[Structure of the first-order formulation]
\label{theo-struct}
Consider the Einstein--Euler system under the assumption of $\Tbb^2$ symmetry on $\Tbb^3$ for a general equation of state satisfying~\eqref{hyperbolic-eos}. After expressing the spacetime metric $\guntrois$ in areal gauge~\eqref{metric:areal}, consider the metric unknowns~$\amdeux, \Omega, P, Q, G, H$, together with the fluid unknowns $\mu, u$ arising in the stress-energy tensor \eqref{eq:45}. Assume furthermore that the solutions under consideration are sufficiently regular.
\begin{enumerate}

\item \emph{Hyperbolicity.}

\bei 

\item \emph{Under the non-vacuum condition $- \Jbb \cdot \Jbb>0$,} in terms of the unknowns $(\Jbb, \Kpar, \Lbb)$ defined in~\eqref{eq:definevar} and $(\amdeux, \Omega)$ and subjected to the causality inequalities~\eqref{eq:causal}, the evolution system consisting of~\eqref{eq:T2-1234-def}, \eqref{eq:T2-Euler-perp-def}, and~\eqref{eq:T2-all-suite-def} is a {first-order} {hyperbolic} system of $12$ nonlinear balance laws, which is supplemented with the $3$ constraints~\eqref{eq:theconstraints-def}.  (In addition, the wave equation for the lapse~\eqref{eq:waveconffac}, below, follows from this set of equations and, in vacuum, it is a consequence of the geometry equations alone.)

\item If the sound speed $p'(0)>0$ is positive, the system remains hyperbolic even in vacuum.

\item On the other hand, if the sound speed $p'(0)=0$ vanishes, the system fails to be hyperbolic in the vacuum $- \Jbb \cdot \Jbb=0$, and it is then only weakly hyperbolic.

\eei 

\item \emph{Nonlinearities.} 

\bei 

\item Away from vacuum, the double characteristic family associated with the speed
${-\amdeux^{-1}\frac{J_1}{J_0}}$ is \emph{linearly degenerate}. Moreover, the two acoustic characteristic families associated with the speeds \(\xi_\pm\) are \emph{genuinely nonlinear} in the sense of Lax if and only if the pressure law satisfies the convexity condition~\eqref{hyperbolic-eos-3}.

\item The metric evolution equations~\eqref{eq:T2-1234-def} and~\eqref{eq:T2-all-suite-def} and the fluid evolution equations~\eqref{eq:T2-Euler-perp-def} contain only (at most) quadratic terms in $(\Jbb, \Kpar, \Lbb)$, with the exception of $q$-dependent terms which are sub-quadratic.

\item With the exception of the lapse equations in \eqref{eq:T2-567-def} and \eqref{eq:evollambda-def}, all quadratic source terms involving only the variables $\Jperp, \Lbb$ are null forms ---while this is not the case for nonlinearities involving the variables $\Jpar, \Kpar$.

\eei  

\item \emph{Equivalence.}  

\bei 

\item Given a solution $(\Jbb, \Kpar, \Lbb)$ and $(\amdeux, \Omega)$ to the equations~\eqref{eq:T2-1234-def}, \eqref{eq:T2-Euler-perp-def}, and~\eqref{eq:T2-all-suite-def}, and provided the conditions~\eqref{eq:theconstraints-def} hold on a slice of constant area $t_0$, one recovers a unique solution\footnote{up to irrelevant shifts of $P,Q$ by additive constants and $G,H$ by functions of~$x$} to the Einstein--Euler equations:  the metric coefficients $P,Q$ are determined by direct integration from~\eqref{eq:definevar-01}, while $G,H$ are immediately recovered from~\eqref{eq:definevar-01}. 

\item Next, the fluid variables $\mu$ and $u$ are determined from~\eqref{appendixJmatter} in all non-vacuum regions, as well as the constraints~\eqref{eq:theconstraints-def}. In vacuum regions, one simply has $\mu=0$ and the velocity field $u$ is ill-defined. (In addition, the wave equation~\eqref{eq:waveconffac}, stated below, holds for all areal times.) 

\eei 
\item \emph{Propagation property.} Once the equations~\eqref{eq:theconstraints-def} on $\log\Omega, K_2, K_3$ are satisfied on an initial hypersurface, they also hold for all future times. 

\end{enumerate}
\end{proposition}

The example of isothermal flows is considered in \autoref{appendix=B}. 


\subsection{Derivation of the \JKL\ formulation }
\label{section=2-4}

\subsubsection{Einstein's evolution equations}

We proceed and compute suitably chosen linear combinations of the $10$ frame components~\eqref{eq:Einsframe} in terms of the geometry variables~\eqref{eq:definevar}. Together with two compatibility equations already stated in~\eqref{eq:T2-1234-def}, we obtain a total of $12$ equations, which we  present in the following remarkable form. 
\bei

\item \emph{Evolution of the essential geometry.} Using our notation, after elementary calculations, the evolution equations for $\Lbb=\Pbb+i\,\Qbb$ with $\Pbb = (P_0, P_1)$ and $\Qbb = (Q_0, Q_1)$ read \eqref{eq:T2-1234-def}, as stated previously, where the operators were defined in~\eqref{equa-div-curl}. The two divergence equations are the components $R_{22} - R_{33}$ and $2\, R_{23}$ of the Einstein equations, while the curl equations are compatibility conditions satisfied by the derivatives of $P$ and~$Q$ and are needed to reach a closed evolution system for the first-order unknowns.

\item \emph{Evolution of the conformal length density.} Notably, in our formulation, only the equations \emph{related to the quotient geometry} explicitly depend on the equation of state through the definition~\eqref{eq:T2-Mdef-0} of the tensor $\Mbf$. The component $R_{22}+R_{33}$ of the Einstein equations is found to be as stated\footnote{Since $\del_t = \Omega e_0$ and $\del_x=\Omega\amdeux e_1 + G\del_y+H\del_z$, the differential equations in \eqref{eq:T2-all-suite-def}--\eqref{eq:theconstraints-def} could equivalently be expressed in terms of the frame derivatives $e_0$ and $e_1$. We prefer not to adopt this notation here.} in \eqref{eq:T2-8-def}. The function $\amdeux$ determines the conformal geometry on the quotient manifold $\Ncal = \Mcal/\Tbb^2$, and will play a central role when analyzing the global foliation. In \autoref{section=3}, we will see that the equation~\eqref{eq:T2-8-def}  stated for $\amdeux$ determines the volume of spacelike slices (up to the multiplicative factor $|t| \, \Omega$).

\item \emph{Evolution of the lapse.} The metric coefficient $\Omega$ is nothing but the lapse function in the metric decomposition~\eqref{metric:areal}, and its time derivative is given by the Ricci component $R_{00}+R_{11}-R_{22}-R_{33}$ of the Einstein equations, as stated in \eqref{eq:evollambda-def}. 

\item \emph{Evolution of the twists.} The components $R_{12}$ and $R_{13}$ provide us with the time derivatives of the twist coefficients $K_2, K_3$, namely the differential equations \eqref{eq:T2-9101112-evol-def-a}--\eqref{eq:T2-9101112-evol-def-b}. 

\item \emph{Additional equation for the geometry.} Finally, the component $R_{00}-R_{11}+R_{22}+R_{33}$ provides us with another geometric equation (for the lapse), but its statement is postponed to~\eqref{eq:waveconffac}, below. 
\eei 


\subsubsection{Einstein's constraint equations}

We derive separately the remaining equations for the geometry, which are automatically satisfied once they hold on any given hypersurface of constant areal time.
\bei 

\item \emph{Constraint on the lapse.} While time derivatives of $\amdeux$ and $\Omega$ obey the equations~\eqref{eq:T2-8-def} and~\eqref{eq:evollambda-def}, respectively, only the space derivative of $\Omega$ obeys an equation. The component $R_{01}$ of the Einstein equation relates this derivative to an energy flux term, as stated in \eqref{eq:T2-567-def}. 

\item \emph{Constraints on the twists.} The components $R_{02}$ and $R_{03}$ of the Einstein equations give the space derivatives of the twists as stated in \eqref{eq:221a}--\eqref{eq:221b}. 
\eei 


\subsubsection{The missing equation: wave equation for the lapse}

The final equation, now stated, will play a distinct role in our theory of weak solutions, as we will explain in the course of our analysis. The component $R_{00}-R_{11}+R_{22}+R_{33}$ provides us with the last Einstein equation, which is a wave equation for the lapse and will be used in the derivation of the Euler equations, below. It reads\footnote{We state~\eqref{eq:waveconffac} in a form that will make sense for weak solutions, namely when $(\log \Omega)_t$ and ${\amdeux_t \over \amdeux}$ and ${\Omega_x \over \amdeux}$ are solely defined as integrable functions, while $\amdeux$ denotes the conformal length density.}
\bel{eq:waveconffac}
\aligned
& \divuntrois \Bigl( t^{-1}  \Omega^{-1} (\log(\Omega \amdeux) )_t, \ t^{-1}  \Omega^{-1} \amdeux^{-1} (\log\Omega)_x \Bigr) = - \frac{1}{4t^3} \Omega^{-2} +  \frac{1}{4t} \, \Mwave(\Jbb, \Kpar, \Lbb),
\endaligned
\ee
with 
\bel{eq:waveconffac-2}
\aligned
\Mwave(\Jbb, \Kpar, \Lbb)
& \coloneqq 6 \, \Ebf_0(\Kpar) + 2 \, \Ebf_0(\Jpar) - \Pbb \cdot \Pbb - \Qbb \cdot \Qbb - \Jperp \cdot \Jperp. 
\endaligned
\ee 
Here the left-hand side is the divergence of the one-form
\be
t^{-1} \bigl(d\log\Omega + e_0(\log \amdeux)e^0\bigr), 
\ee
expressed explicitly as ${t^{-1}  \Omega^{-2} \amdeux^{-1} \bigl( (\Omega \amdeux )_t, \Omega_x \bigr)}$ in the Vielbein basis $e_0,e_1$. We separate~\eqref{eq:waveconffac} from the other evolution equations~\eqref{eq:T2-1234-def} and~\eqref{eq:T2-all-suite-def}, since it plays a different role in the sense that it is equivalent to  the first Euler equation stated in~\eqref{eq:T2-Euler-perp-def-a}, and thus is automatically satisfied once that equation holds. Observe also that, for (sufficiently regular and) \emph{vacuum} spacetimes,~\eqref{eq:waveconffac}~is a consequence of the evolution equations~\eqref{eq:T2-1234-def} and~\eqref{eq:T2-all-suite-def} (but this is not true for matter spacetimes). 


\subsubsection{Euler equations}

At this stage, we have (by including~\eqref{eq:waveconffac}) twelve equations corresponding to the ten components of the Ricci curvature tensor, supplemented with the two compatibility equations (for $\curl_\Ncal \Pbb = 0$ and $\curl_\Ncal \Qbb$) appearing in~\eqref{eq:T2-1234-def}. From this set of Einstein equations, we can immediately \emph{deduce} the Euler equations. However, in our first-order presentation, it is important to express these Euler equations explicitly, as follows. We distinguish between two sets of equations for the transverse and parallel components. We observe that the propagation equations for the transverse momentum \eqref{eq:T2-Euler-perp-def-a}--\eqref{eq:T2-Euler-perp-def-b} are \emph{coupled} to the remaining Euler equations \eqref{eq:T2-fluid23-def-a}--\eqref{eq:T2-fluid23-def-b}.
\bei

\item \emph{Evolution of the transverse momentum.} We find \eqref{eq:T2-Euler-perp-def-a}--\eqref{eq:T2-Euler-perp-def-b} stated earlier. Here, $\Mbf^{m\bullet} = (\Mbf^{m0},\Mbf^{m1})$ (for $m=0,1$) is the vector with components $(\Mbf^{m0},\Mbf^{m1})$ defined in~\eqref{eq:T2-Mdef-0}, $\divuntrois$ is the divergence of this vector in the sense~\eqref{equa-def-div13}, and the source term is given in \eqref{eq:T2-Euler-perp2-def0} and obeys the inequality $\abs{\Msource}\leq 5 \, \Mbf^{00}$ (proven in \autoref{appendix=D-1}, below). Interestingly, the first fluid equation coincides with the wave equation~\eqref{eq:waveconffac},
 whereas the second equation is nothing but the compatibility equation $\bigl(\log \Omega \bigr)_{xt} - \bigl(\log \Omega \bigr)_{tx} = 0$, in which we replaced the first-order derivatives of the lapse by energy and energy fluxes (namely, using~\eqref{eq:evollambda-def} and~\eqref{eq:T2-567-def}). 

\item \emph{Evolution of the parallel momentum.} The remaining two fluid equations are deduced from the compatibility equation $\bigl( |t|^{3/2} K_2 \bigr)_{xt} - \bigl( |t|^{3/2} K_2 \bigr)_{tx} = 0$ and the analogue for $K_3$. Namely, from~\eqref{eq:221a}--\eqref{eq:221b} and~\eqref{eq:T2-9101112-evol-def-a}--\eqref{eq:T2-9101112-evol-def-b} we obtain  \eqref{eq:T2-fluid23-def-a}--\eqref{eq:T2-fluid23-def-b}.  
Since $P_0, P_1, K_2, K_3$ and the lapse and conformal length density already have their own evolution equations, it is natural to view~\eqref{eq:T2-fluid23-def-a}--\eqref{eq:T2-fluid23-def-b} as evolution equations for the parallel momentum $(J_2,J_3)$ rather than for the orthogonal momentum~$\Jperp$. 
\eei   
  

\section{Hidden entropy structure of the Euler equations}
\label{section=3}

\subsection{Methodology}
\label{section=3-1}

\subsubsection{The notion of entropy current}

In order to define weak solutions to nonlinear hyperbolic problems \cite{Lax-1971,LeFloch-book} such as the Euler system derived in~\autoref{section=2}, it is essential to identify a sufficiently large class of entropy currents, and their associated balance laws, implied by the Euler equations, and to deduce from them suitable entropy inequalities for weak solutions. In this section, we first explain the relevance of the notion of entropy current in a geometric setting on a general foliated spacetime, and then specialize this notion to the Euler system in torus symmetry. The discussion of entropy inequalities for weak solutions is postponed to \autoref{section=4}. 

Temporarily only, let $(\Mcal,g)$ be a time-oriented $n$-dimensional Lorentzian manifold endowed with a foliation by spacelike hypersurfaces $(\Sigma_t)_{t\in I}$. (The notion of entropy current can also be defined independently of the foliation; however, the convexity condition introduced below is formulated relative to a chosen foliation.) Let us briefly discuss the notion of entropy current for a system of $N$~coupled balance laws\footnote{The notation $u$ here is not related to the fluid velocity vector field.}
\bel{eq-general}
\dive_g f(u) = s(u), 
\ee
where the unknown $u$ is a collection of tensor fields on~$\Mcal$, namely a section of some $N$-dimensional vector bundle $E\to\Mcal$, the source term $s(u)$ is a section of some $N$-dimensional vector bundle $E'\to\Mcal$, and the current $f(u)$ is a section of $T\Mcal\otimes E' \to \Mcal$.
Here, $\dive_g$ denotes the divergence operator on $(\Mcal,g)$, with the index contraction acting on the $T\Mcal$ factor of $T\Mcal\otimes E'$.

Both the source and the current depend on~$u$ without derivatives (and on the metric~$g$), and we assume that
\bel{equa-324}
\text{the time component map $u \mapsto f^0(u)$ is one-to-one on its image.}
\ee
Here, we decompose any vector field $X$ as $X= X^0 e_0 + X_\Sigma$, where $e_0$ denotes the future-pointing unit normal vector field to the foliation and $X_\Sigma$ is tangent to the foliation slices, and we extend this decomposition of $T\Mcal$ to $T\Mcal\otimes E'$ by a tensor product.

In local coordinates, the balance law~\eqref{eq-general} takes the form
\be
\nabla_\alpha f(u)^{\alpha J}=s(u)^J,
\ee
that is,
\be
\partial_\alpha f(u)^{\alpha J}
+\Gamma^\alpha{}_{\alpha\beta} f(u)^{\beta J}
+\widehat\Gamma^J{}_{\alpha K} f(u)^{\alpha K}
=s(u)^J.
\ee
Here, $\Gamma^\alpha{}_{\alpha\beta}$ are the Christoffel symbols of the spacetime metric, while $\widehat\Gamma^J{}_{\alpha K}$ denotes the connection induced on the target bundle $E'\to \Mcal$. We set
\be
\Jac_I{}^J(u)\coloneqq \frac{\partial f(u)^{0J}}{\partial u^I},
\ee
with $I$ an abstract index on the bundle $E\to\Mcal$,
and, in agreement with \eqref{equa-324}, we assume that this Jacobian is invertible. Solving for the time derivative of \(u\), we find
\be
\partial_0 u^I
=
(\Jac^{-1})_J{}^I
\Bigl(
s(u)^J
- \partial_i f(u)^{iJ}
- \Gamma^\alpha{}_{\alpha\beta} f(u)^{\beta J}
- \widehat\Gamma^J{}_{\alpha K} f(u)^{\alpha K}
\Bigr) ,
\ee
with $i=1,\dots,n-1$ an index on the leaves of the foliation.

\begin{definition}
\label{def-entropcurrent}
Consider the system of balance laws~\eqref{eq-general}.

\bei

\item A spacetime vector field $F=F(u)$ which depends on~$u$ (possibly non-covariantly) but not on its derivatives, is called an \textbf{entropy current} provided there exists a field $V=V(u)\in (E')^*$,
called the \textbf{entropy multiplier}, such that
\bel{eq-general-entropy-compatibility}
\frac{\partial F^\alpha}{\partial u^I}(u)
=
V_J(u)\,\frac{\partial f^{\alpha J}}{\partial u^I}(u),
\qquad \alpha=0,\dots,n-1,
\ee
and every sufficiently regular solution to~\eqref{eq-general} also satisfies the scalar balance law
\bel{eq-general-entropy}
\dive_g F(u) = \Sigma_F(u)
\quad \text{(regular solutions),} 
\ee
where the associated \textbf{entropy source} is given by
\bel{eq-general-entropy-production}
\Sigma_F(u)
=
V_J(u)\,s(u)^J
+
\Gamma^\alpha{}_{\alpha\beta}
\Bigl(
F^\beta(u)-V_J(u)\,f(u)^{\beta J}
\Bigr)
-
V_J(u)\,\widehat\Gamma^J{}_{\alpha K}\,f(u)^{\alpha K}.
\ee

\item In addition, an entropy current is called \textbf{future-convex} if the time component of the current $F^0(u)$ is strictly convex as a function of the principal variables~$f^0(u)$.
\eei
\end{definition}


\subsubsection{Basic observations}

In a flat setting, the entropy source reduces to the familiar expression
\be
\Sigma_F(u)=\frac{\del F^0}{\del f^0}(u)\cdot s(u).
\ee
In the present geometric framework, our definition above captures the lower-order contributions arising from the divergence operator, the foliation, and the tensorial nature of the unknowns and currents.
Indeed, for any sufficiently regular solution, we compute
\be
\aligned
\dive_g F(u)
&=
\partial_\alpha F^\alpha(u)+\Gamma^\alpha{}_{\alpha\beta}F^\beta(u)
=
\frac{\partial F^\alpha}{\partial u^I}(u)\,\partial_\alpha u^I
+\Gamma^\alpha{}_{\alpha\beta}F^\beta(u).
\endaligned
\ee
Using the expression for \(\partial_0u\) above and the compatibility conditions~\eqref{eq-general-entropy-compatibility}, all derivative terms cancel out, and we obtain precisely~\eqref{eq-general-entropy}--\eqref{eq-general-entropy-production}. Thus, the entropy production contains not only the source term $s(u)$, but also geometric contributions.


\subsubsection{An example of regularization}

Observe that the notion of future convexity \emph{depends} on both the selection of balance laws in~\eqref{eq-general} and the chosen foliation of the spacetime. 
The standard methodology is based on a regularization of~\eqref{eq-general} obtained by adding a second-order diffusion term acting tangentially to the leaves of the foliation and with a sign chosen to make the problem \emph{parabolic along future timelike directions}. 

Suppose that \(F\) is a future-convex entropy current. Denoting by \(\Delta_\Sigma=\dive_\Sigma\circ\nabla_\Sigma\) the Laplace operator on the hypersurfaces \(\Sigma_t\), we consider a family of smooth approximate solutions \(u^\eps\) satisfying\footnote{Strictly speaking, this is not a physically appropriate regularization, as such a parabolic regularization implies ``propagation at infinite speed''.  A fully covariant regularization involving only subluminal propagation will be studied in~\cite{LeFlochLeFloch-next}.}
\bel{eq-general-parabolic}
\dive_g f(u^\eps) = s(u^\eps) + \eps \, \Delta_\Sigma\bigl(f^0(u^\eps)\bigr).
\ee
Applying the entropy structure to the regularized system yields
\bel{eq-general-entropy-eps}
\dive_g F(u^\eps)
= \Sigma_F(u^\eps)
+ \eps \, {\del F^0 \over \del f^{0J}} (u^\eps) \Delta_\Sigma\bigl(f^{0J}(u^\eps)\bigr).
\ee
\bse
The parabolic contribution can then be decomposed as
\bel{eq-general-entropy-eps-2}
\eps \, {\del F^0 \over \del f^{0J}} (u^\eps) \Delta_\Sigma\bigl(f^{0J}(u^\eps)\bigr)
= \eps \, \dive_\Sigma \Psi_F[u^\eps]
- \eps \, \Dcal_F[u^\eps],
\ee
where 
\be
\Psi_F [u^\eps] \coloneqq {\del F^0 \over \del f^{0J}} (u^\eps) \nabla_\Sigma\bigl(f^{0J}(u^\eps)\bigr)
\ee
is referred to as the viscous tangential flux, and the so-called entropy dissipation reads
\be
\aligned
\Dcal_F[u^\eps] 
& \coloneqq  \nabla_\Sigma \Bigl( {\del F^0 \over \del f^{0J}} (u^\eps) \Bigr) \nabla_\Sigma\bigl(f^{0J}(u^\eps)\bigr)
\\
& = \nabla_\Sigma\bigl(f^{0J}(u^\eps)\bigr) \cdot {\del^2F^0 \over \del f^{0J} \del f^{0K}} (u^\eps) \cdot \nabla_\Sigma\bigl(f^{0K}(u^\eps)\bigr). 
\endaligned
\ee
\ese
Thanks to future convexity, the regularization produces in the entropy equation a dissipative term with a \emph{sign}, that is, \(\Dcal_F[u^\eps]\geq0\). In turn, we obtain
\bel{eq-general-entropy-eps-3}
\dive_g F(u^\eps)
\leq
\Sigma_F(u^\eps) + \eps \, \dive_\Sigma \Psi_F[u^\eps].
\ee
The flux contribution \(\eps \, \dive_\Sigma \Psi_F[u^\eps]\) can be shown to tend to zero in the sense of distributions, under natural assumptions on the convergence of the sequence \(\{u^\eps\}_\eps\); cf., for instance, the textbooks \cite{Dafermos-book,LeFloch-book}. 

\subsubsection{Aim of this section}

Passing \emph{formally} to the limit \(\eps\to0\) leads to an entropy inequality for weak solutions:
\bel{eq-general-entropy-ineq}
\dive_g F(u) \leq \Sigma_F(u)
\quad \text{(weak solutions),} 
\ee
as we will consider in \autoref{section=4}.
In addition, the entropy dissipation, namely the difference $-\mathopen{}\dive_g F(u) + \Sigma_F(u)$, can be shown to define a nonnegative bounded Radon measure, and this property will be built into our notion of tame solution in \autoref{section=4}.
Our aim in the rest of this section is to investigate the entropy structure associated with the Euler equations in torus symmetry. In particular, we introduce the particle number current, as well as general entropy currents. Each class of entropy currents will serve a different purpose in formulating entropy inequalities in \autoref{section=4}.


\subsection{Structure associated with the particle number density} 
\label{section=3-2}

\subsubsection{A relevant function of the mass-energy density} 

Throughout this section, we continue to work with sufficiently regular solutions. We begin with a definition. 

\begin{definition}
Consider the \textbf{particle number density} $\Numb=\Numb(\mu)$ defined on $(0,+\infty)$, up to a positive multiplicative constant, by
\be
\frac{d\Numb}{\Numb} = \frac{d\mu}{\mu+p(\mu)},
\ee
and extended continuously to $\mu=0$ by setting $\Numb(0)=0$. 
\end{definition}

Interestingly under our hyperbolicity condition, it is an increasing and concave function, since
\bel{equa-390}
\frac{d^2 \Numb}{d\mu^2} = -\frac{\Numb p'(\mu)}{(\mu + p(\mu))^2} <0 \quad \text{for } \mu>0, 
\ee
under our assumption that $p'$ is positive away from vacuum. A direct computation shows that, for all sufficiently regular solutions, the vector field 
\bel{equa-dhj3}
\vNumb \coloneqq \Numb(\mu) \, u 
\ee
automatically satisfies the divergence law
\bel{equa-entr}
\divuntrois \vNumb = 0 \quad \text{(regular solutions).} 
\ee 
We refer to $\vNumb$ as the \textbf{particle number current}. Clearly, $\vNumb$ is \emph{timelike} and \emph{future-oriented}. 


\subsubsection{Normalization based on the relativistic enthalpy} 

Let us refer to the function $\mu+p(\mu)$ as the \textbf{relativistic enthalpy.}
In view of the relation $2(\mu+p(\mu)) = - \Jbb \cdot \Jbb$, we express the mass-energy density as a function of our main variable $\Jbb$. Namely, thanks to~\eqref{hyperbolic-eos}, there exists some (increasing) bijection $\funmu$ of $[0,+\infty)$ such that $\mu \eqqcolon \funmu( - \Jbb \cdot \Jbb)$. Next, recalling that \(\Jbb=\sqrt{-\Jbb\cdot\Jbb}\,u=\sqrt{2(\mu+p)}\,u\) and defining the \textbf{enthalpy-normalized particle density} $\hNumb$ by 
\be
\hNumb(\mu) \coloneqq \Numb(\mu) \bigl(2(\mu+p(\mu))\bigr)^{-1/2} 
             = {\Numb(\mu) \over (-\Jbb \cdot \Jbb)^{1/2}}, 
\ee 
we rewrite the entropy current \eqref{equa-dhj3} as  
\bel{equa--Nbb} 
\vNumb = \hNumb(\mu) \, \Jbb, 
\ee
in which $\mu = \funmu( - \Jbb \cdot \Jbb)$ is also expressed in terms of the main unknown $\Jbb$. Observe that \eqref{equa--Nbb} \emph{makes sense in vacuum regions} where the velocity vector $u$ is \emph{ill-defined}. The function $\hNumb$ can be viewed as a function of $\mu$ or $-\Jbb\cdot\Jbb$ via the identity
\bse
\bel{Sdef-scaled}
d (\log \hNumb) = \frac{1}{2} (\mu+p)^{-1} d\bigl((\mu+p)q\bigr)
= {1-p'(\mu) \over 2(\mu + p(\mu))} \, d\mu, 
\ee
where the pressure ratio $q= (\mu-p)/(\mu+p)$ was defined in~\eqref{eq:defq}. 

We also observe that the hyperbolicity condition~\eqref{hyperbolic-eos} ensures that $\hNumb$~is an \emph{increasing function} of~$\mu$ ---or equivalently an increasing function of $(-\Jbb\cdot\Jbb)$. In addition, the condition ${p'(0)<1}$ implies that $\hNumb$ vanishes when $\mu$ vanishes. More precisely, from~\eqref{equa-qminqmax-1} we deduce that 
\be
\frac{d\log\Numb}{d\log\mu} = \frac{\mu}{\mu+p(\mu)} \in \Bigl[ \frac{1+q_{\min}}{2}, 1 \Bigr), \qquad \mu\in(0,+\infty),
\ee
which gives bounds in the near-vacuum limit, namely
\bel{near-vacuum-Numb}
\mu \lesssim \Numb(\mu) \lesssim \mu^{(1+q_{\min})/2}, \qquad
\mu^{1/2} \lesssim \hNumb(\mu) \lesssim \mu^{q_{\min}/2}, \qquad \mu \to 0 .
\ee
At large densities, the inequalities are reversed, and improved due to the upper bound on $q(\mu)$ in~\eqref{equa-qminqmax-2},
\be
\mu^{(1+q_{\min})/2} \lesssim \Numb(\mu) \lesssim \mu^{(1+q_\infty^+)/2}, \qquad
\mu^{q_{\min}/2} \lesssim \hNumb(\mu) \lesssim \mu^{q_\infty^+/2}, \qquad \mu \to +\infty .
\ee
\ese

Since these properties will play a role in our analysis, we summarize them in the following lemma. 
Furthermore, to emphasize the dependency of the metric-weighted particle density in $(-\Jbb \cdot \Jbb)$ and to suppress the variable $\mu$ which need not appear in our first order formulation, we find it convenient to write 
$\hNumb_\Jbb := \hNumb(\mu)$ and rewrite~\eqref{equa--Nbb} as
\bel{equa--Nbb-1}
\vNumb = \hNumb_\Jbb \, \Jbb.
\ee

\begin{proposition}
The particle number density satisfies
\bel{equa-2311}
{d \over d\mu} \hNumb(\mu) >0, \qquad \lim_{\mu \to 0} \hNumb(\mu) = 0, 
\qquad 
\hNumb(\mu) \lesssim \mu^{q_{\min}/2} + \mu^{q_\infty^+/2} .
\ee 
Consequently, the particle number current $\vNumb=(\Numb_0, \Numb_1)$ grows sub-quadratically:
\bel{equa--Nbb-17}  
|\Numb_0| \lesssim |\Jbb \cdot \Jbb|^{q_\infty^+/2} J_0 \ll (J_0)^2,
\quad 
|\Numb_1| \lesssim |\Jbb \cdot \Jbb|^{q_\infty^+/2} |J_1| \ll J_0 |J_1|
\quad \text{for large~$J_0$.}
\ee 
Moreover, under the hyperbolicity assumption~\eqref{hyperbolic-eos-1} on the equation of state of the fluid, the (negative) particle number current $-\vNumb$ can be checked to be future-convex (cf.~\autoref{appendix=D-3}). 
\end{proposition}


\subsubsection{Application to isothermal fluids}

\bse
We now consider the isothermal case $p = k^2\mu$, for which the function $q$ is constant and, as it turns out, the particle number current can be written in closed form. First of all, from
\eqref{Sdef-scaled} we obtain
\be
\hNumb(\mu) \simeq \mu^\kappa, \quad
\Numb(\mu) \simeq \mu^{\kappa+1/2},
\qquad
\kappa = \frac{q}{2} = \frac{1}{1+k^2} - \frac{1}{2} \in (0,1/2).
\ee
More precisely, up to harmless positive multiplicative constants, we may set
\be
\Numb(\mu) = \bigl(2(1+k^2)\bigr)^{\kappa+1/2} \mu^{\kappa+1/2}
\quad \text{(isothermal case),}
\ee
and therefore
\be
\hNumb(\mu) = \bigl(2(1+k^2)\bigr)^{\kappa} \mu^{\kappa}
\quad \text{(isothermal case).}
\ee
Consequently, in view of $-\Jbb\cdot\Jbb = 2(\mu+p)= 2 (1+k^2) \mu$, the particle number current becomes, up to a positive multiplicative constant,
\bel{interesting-entropy}
\vNumb = (-\Jbb\cdot\Jbb)^{\kappa} \Jbb
\quad \text{(isothermal case).}
\ee
\ese
%


\subsection{Structure associated with the parallel momentum}
\label{section=3-3}

\subsubsection{Formulation in terms of first-order variables}

We proceed first by avoiding the use the \emph{undifferentiated} metric variables $P$ and~$Q$, which are suppressed in our first-order \JKL\ formulation. We need the following notation.

\begin{definition}
\label{defJhat}
The \textbf{parallel momentum per particle} $\Jhatpar$ is defined by 
\bel{equa-jhat}
\Jhatpar \coloneqq \widehat J_2 + i \, \widehat J_3 
\coloneqq \bigl( \hNumb_\Jbb\bigr)^{-1} \Jpar .
\ee 
\end{definition}

It is straightforward to rewrite the Euler equations for the parallel momentum as \emph{linear transport equations} 
\bel{equa-jhat-transp}
\aligned
\bigl( \widehat J_2 \bigr)_t - \amdeux^{-1} {J_1\over J_0} \, \bigl( \widehat J_2 \bigr)_x 
& = \Omega {\widehat J_2 \over 2 J_0} \, \Pbb\cdot\Jperp - \frac{1}{2t} \widehat J_2 ,
\\
\bigl( \widehat J_3 \bigr)_t -\amdeux^{-1}  {J_1\over J_0} \, \bigl( \widehat J_3 \bigr)_x 
& = \Omega {\widehat J_2 \over J_0} \, \Qbb\cdot\Jperp - \Omega {\widehat J_3 \over 2 J_0} \, \Pbb\cdot\Jperp
- \frac{1}{2t} \widehat J_3 .
\endaligned
\ee 
These equations have a drawback: they do not have a divergence form, and so apply to sufficiently regular solutions only. By combining them with the divergence law~\eqref{equa-entr} for the particle number and choosing an arbitrary function~$H$, we arrive at the balance law (for regular solutions)
\bel{equa-entrop23}
\aligned
\divuntrois \Bigl( H(\widehat J_2,\widehat J_3) \vNumb \Bigr)
& = {1 \over 2} \bigl( - \widehat J_2\del_{\widehat J_2} H + \widehat J_3\del_{\widehat J_3} H \bigr) (\Pbb\cdot\vNumb)
- \bigl(\widehat J_2 \del_{\widehat J_3} H \bigr) (\Qbb\cdot\vNumb) \\
& \quad + \frac{1}{2t\Omega} \bigl(\widehat J_2 \del_{\widehat J_2} H + \widehat J_3 \del_{\widehat J_3} H\bigr) \Numb_0 .
\endaligned
\ee
The right-hand side features no spacetime derivatives, hence the vector field $H(\widehat J_2,\widehat J_3) \vNumb$ is an entropy current in the sense of \autoref{def-entropcurrent} for any function~$H$. Observe that the two Euler equations for $\Jpar$ are easily recovered from \eqref{equa-entrop23} by choosing $H(\widehat J_2,\widehat J_3)=\widehat J_2$ and $H(\widehat J_2,\widehat J_3)=\widehat J_3$, respectively. 

It should be noted that $\Jhatpar$ only depends on fluid variables, and that all the equations exhibited until now are stated in terms of the \emph{first-order unknowns} $\Lbb=(P_0, P_1, Q_0, Q_1)$, and do not involve the metric coefficients $P,Q$ themselves.


\subsubsection{Formulation in homogeneous divergence form}

Remarkably, it turns out that the $\Jpar$ equations, presented above, can be expressed as divergence laws without source terms. However, to proceed we must break our first-order formulation\footnote{We expect the $\Jhat$ formulation to be more convenient for numerical purposes.}
and \emph{reconstruct~$P,Q$} from our first-order variables by integrating
\bel{reconstruct-P-Q}
\aligned
P_t & = P_0 \Omega, \qquad  & Q_t & = e^{-P} Q_0  \Omega, 
\\
P_x & = P_1 \Omega \amdeux,   \qquad  & Q_x & = e^{-P} Q_1 \Omega \amdeux. 
\endaligned
\ee
We then work with variables normalized using $P$ and~$Q$, as follows: 
\bel{equawidetildeJ-0}
\Gammaun_2 \coloneqq \abs{t}^{1/2} e^{P/2}  \widehat J_2, \qquad 
\Gammaun_3 \coloneqq \abs{t}^{1/2} \bigl( Qe^{P/2} \widehat J_2+e^{-P/2} \widehat J_3 \bigr) .
\ee
Observe that $\widehat J_2, \widehat J_3$ are the components of the one-form $(\hNumb_\Jbb)^{-1}\Jpar$ along the frame vectors $e_2,e_3$, while $\Gammaun_2$ and $\Gammaun_3$ are its components along the Killing vectors $\del_y,\del_z$. In other words, we introduce the following definition.

\begin{definition}
\label{defGammaun}
The \textbf{metric-weighted parallel momentum} $\Gammaunpar$ is defined by 
\bel{equawidetildeJ} 
\Gammaun_2 \coloneqq \abs{t}^{1/2} {1 \over \hNumb_\Jbb} \, e^{P/2}  J_2
\qquad 
\Gammaun_3 \coloneqq \abs{t}^{1/2} {1 \over \hNumb_\Jbb} \, \bigl( Qe^{P/2} J_2+e^{-P/2} J_3 \bigr).
\ee
\end{definition}

It is straightforward to rewrite~\eqref{equa-jhat-transp} as \emph{homogeneous transport equations:}
\bel{equa-jhat-transp-cons}
\aligned
\bigl( \Gammaun_2 \bigr)_t - \amdeux^{-1} {J_1\over J_0} \, \bigl( \Gammaun_2 \bigr)_x 
 = 0,  
\qquad
\bigl( \Gammaun_3 \bigr)_t - \amdeux^{-1} {J_1\over J_0} \, \bigl( \Gammaun_3 \bigr)_x 
 = 0. 
\endaligned
\ee 
This leads us to a divergence formulation of~\eqref{equa-entrop23}, namely the (as we call it)
\bel{equa-entrop23-cons}
\textbf{$H$-divergence law:} \quad 
\divuntrois \bigl( H\bigl( \Gammaunpar \bigr) \, \vNumb \bigr)
= 0
\, \text{ for any function } H.  
\ee 
This structure will be used below in order to derive a \emph{maximum principle} on $\Gammaun_m$.
With the choice $H(\Gammaunpar) = \Gammaun_m$ for $m= 2,3$ we recover the Euler evolution equations
\bel{equa-mom-twists-basic}
\aligned
\divuntrois \bigl( \Gammaun_m \, \vNumb \bigr) &= 0, \quad m=2,3. 
\endaligned
\ee
For sufficiently regular solutions, these divergence identities are equivalent to the homogeneous transport equations satisfied by $\Gammaun_2,\Gammaun_3$.


\subsubsection{Spatial derivative of the metric-weighted momentum current}

Interestingly, we may also \emph{differentiate} the homogeneous transport equations satisfied by the metric-weighted parallel momentum and, then by a suitable rescaling in $\Numb$ reach a new divergence law, as follows.  For clarity, we also introduce a metric-weighted version of the parallel momentum derivative as 
\bel{equawidetildeJ-bis}
\Gammadeux_m \coloneqq  {1 \over |t|\,|\Numb_0|} e_1\bigl( \Gammaun_{m} \bigr)
= {1 \over |t| \Omega \amdeux |\Numb_0|} \, \Gammaun_{m,x}, 
\qquad 
m=2,3, 
\ee
which is regarded as a temporary notation used only in this paragraph. 
Hence, up to a normalization, the coefficients $\Gammadeuxpar$ are spatial derivatives of~$\Gammaunpar$ (which motivates our notation). By differentiating the equation for the parallel momentum in~\eqref{equa-mom-twists-basic} and proceeding as before, 
we find the \textbf{first-order $H$-divergence law} (as we call it)
\bel{equa-kSK1-0}
\divuntrois \bigl( H\bigl(\Gammadeuxpar, \Gammaunpar\bigr) \, \vNumb \bigr) = 0, \quad m=2,3,
\quad \text{ for any function } H. 
\ee 
Observe that (from the standpoint of our first-order formulation, at least) this divergence law is \emph{not associated with an entropy current}, since it involves a (spatial) derivative of our main unknowns. 

We point out that, for \emph{sufficiently regular} solutions, a maximum principle holds for $\Gammadeux_m$ (by applying similar arguments as above for $\Gammaunpar$). However, in our theory, we will seek only an \emph{integrability condition} on $\Gammadeuxpar$; we \emph{will not} require \eqref{equa-kSK1} for arbitrary functions $H$. Instead, we consider~\eqref{equa-kSK1} with the choice $H=|\Gammadeux_m|$: for sufficiently regular solutions we find 
\bse
\bel{equa-kSK1}
\divuntrois \biggl(  {| e_1\bigl(\Gammaun_{m}\bigr) | \over |t|} \, {\Jbb\over J_0} \biggr) = 0, \quad m=2,3. 
\ee 
In coordinates, this reads as follows,
\bel{equa-Jx}
\del_t |\Gammaun_{m,x}| + \del_x\Bigl( \frac{J_1}{J_0} |\Gammaun_{m,x}| \Bigr) = 0 .
\ee
By integration over a hypersurface, this implies an $L^1$ estimate for $e_1\bigl(\Gammaun_{m} \bigr)$, hence an estimate for the total variation $\Var$ of $\Gammaun_m$ (with respect to the spatial variable), namely 
\bel{equa-Jx-Var}
\Var \bigl( \Gammaun_m(t_3,\cdot) \bigr)
= \int_{\Sbb^1} | \Gammaun_{m,x} |\,dx
= {1 \over |t_3|} \int_{\Tbb^3} \bigl| e_1\bigl(\Gammaun_{m}\bigr) \bigr|  \, \dVtrois  \Big|_{t=t_3} 
=
\Var \bigl( \Gammaun_m(t_2,\cdot) \bigr) 
\ee
\ese
for any two times $t_2 \leq t_3$.
On the other hand, for weak solutions, \eqref{equa-Jx} would involve an ill-defined product of a bounded function $J_1/J_0$ by a measure $|\Gammaun_{m,x}|$, thus it will \emph{not be stated} in our framework. The $L^1$ estimate \eqref{equa-Jx-Var} on the spatial derivative of the variable $\Gammaunpar$
will nevertheless make sense as an inequality on the total variation of bounded variation functions.


\subsubsection{The class of entropy currents}

Finally, it is tedious (but elementary) to check the following property. 

\begin{proposition}[Entropy currents for the Euler system]
The linear combinations (with coefficients depending upon the geometry only) of 
\be
\Mbf^{0 \bullet}, \quad \Mbf^{1 \bullet}, \quad \Mbf^{2 \bullet}, \quad \Mbf^{3 \bullet}, \quad H(\Gammaunpar) \vNumb
\ee 
constitute the \emph{complete list}\footnote{excluding the trivial entropy currents with \emph{constant} components.}
of entropy currents for the Euler system in torus symmetry. 
\end{proposition}
 
Furthermore, it can also be checked that this collection of entropy currents is \emph{insufficient} to apply the method of compensated compactness and establish a nonlinear stability result. This motivates us to introduce a generalization of this notion, as follows.  


\subsection{The notion of quasi-current}
\label{section=3-4}

\subsubsection{Entropy structure in Minkowski spacetime}

It is convenient to consider first the Euler system in Minkowski spacetime, which is obtained  by suppressing all lower-order terms as well as replacing $t,\amdeux,\Omega$ by $1$. We arrive at  
\bel{eq-EulerMin}
\begin{aligned}  
\Big( J_0^2 - {1 - q_\Jbb \over 2}  \, (- \Jbb \cdot \Jbb) \Big)_t -  ( J_0 J_1 )_x
& = 0, 
\\
( J_0 J_1 )_t
- \Big( J_1^2 + {1 - q_\Jbb \over 2}  \, (- \Jbb \cdot \Jbb) \Big)_x
& = 0, 
\qquad \qquad 
\\
\big( J_0 \Jpar  \big)_t  - \big( J_1 \Jpar \big)_x & = 0.
\end{aligned}
\ee
Here, the main unknown $\Jbb= ( J_0,  J_1,  J_2,  J_3) = (\Jperp, \Jpar)$ satisfies the constraint $\Jbb \cdot \Jbb \leq 0$, with 
$\Jbb \cdot \Jbb =  -  J_0^2 +  J_1^2 +  J_2^2 +  J_3^2$.
The function \(q_\Jbb\) takes values in \((0,1)\) and is bounded away from both \(0\) and \(1\) for non-zero $\Jbb\cdot\Jbb$; for isothermal fluids, it coincides with the constant $2 \kappa \in (0,1)$.
As pointed out earlier, the particle number current~\eqref{equa--Nbb} provides us with a particular example of entropy current, that is, 
\be
\Numb^0 = - \Numb_0(\Jbb) = - \hNumb_\Jbb \, J_0, 
\qquad 
\Numb^1 = \Numb_1(\Jbb) = \hNumb_\Jbb \, J_1. 
\ee

By treating the transported variables $\Jhatpar = (\hNumb_\Jbb)^{-1} \Jpar$ (introduced in \autoref{defJhat}) as constant background parameters, and considering entropy currents for the remaining two evolution equations, we arrive at a collection of currents that remain useful for non-constant~$\Jhatpar$.

\begin{definition} 
\label{def-quasi-cur}
A regular map $\Fbb = (\Fcomp^0, \Fcomp^1)(\Jbb)$ is called a \textbf{quasi-current} for the Euler system in $\Tbb^2$ symmetry provided every regular solution to the Euler equations in Minkowski geometry~\eqref{eq-EulerMin} satisfies 
\be
\Fcomp^0(\Jbb)_t + \Fcomp^1(\Jbb)_x
= \Fvee_2 (\Jbb) \widehat J_{2,x} + \Fvee_3(\Jbb) \widehat J_{3,x}
\quad \text{ (regular solutions),}
\ee
which is referred to as a \textbf{quasi-balance law}. Here the functions $\Fvee_m$ are called the \textbf{shear-induced factors} associated with $\Fbb$, while $\widehat J_m = (\hNumb_\Jbb)^{-1} J_m$, $m=2,3$.
\end{definition}

Clearly, when the shear-induced factors vanish identically, or when the fluid velocity is orthogonal to $\Tbb^2$~symmetry orbits, a quasi-current reduces to an ordinary entropy current.
The notion of future convexity introduced in \autoref{def-entropcurrent} for currents applies likewise to quasi-currents.
Interestingly, there exist infinitely many quasi-currents, whether future-convex or not, and this collection will be sufficient to apply compensated compactness arguments ---in contrast with the (too narrow) set of standard entropies. Moreover, the quasi-currents  $\Fbb = \Fbb(\Jbb)$ can be conveniently parametrized by solving a second-order wave equation in phase space; cf.~\cite{LeFlochLeFloch-next}.

Our terminology ``quasi-current'' is justified by the fact that the associated balance law is not purely conservative: besides the divergence part, it contains the products $\Fvee_m(\Jbb) \widehat J_{m,x}$. Entropy quasi-currents behave like entropy currents modulo non-conservative corrections, which will be well-defined under our regularity and integrability conditions.


\subsubsection{Auxiliary fluid variables}

To avoid any ambiguity in the notation, we now reintroduce the variables (mass density, fluid velocity) associated with the momentum current.
\be
\Sigma\coloneqq(\Sigma_0,\Sigma_1)=(-\Jbb\cdot\Jbb, J_1/J_0), \quad \text{together with } \Jhatpar,
\ee
which are subject to the causality conditions 
\be
\Sigma_0\geq 0 \, \text{ and } \, -1\leq\Sigma_1\leq 1.
\ee
The Euler system~\eqref{eq-EulerMin} reduces to the form
\bel{Sigmat}
\Sigma_t = A\bigl(\Sigma,\Jhatpar\bigr) \, \Sigma_x, \qquad \Jhatpar_t = (J_1/J_0) \Jhatpar_x,
\ee
for a suitable $2\times 2$ matrix~$A$.  The map from the original variables $\Jbb$ to $(\Sigma,\Jhatpar)$ is one-to-one, and with this set of variables, the following explicit expressions are available (as we prove in \autoref{appendix=C-3}).  We recall the expressions~\eqref{eq:T2-Texpr} of stress tensor components $T^{ab}=\frac{1}{2}\Mbf^{ab}(\Jbb,0,0)$.

\begin{lemma}[Expressions of the shear-induced factors]
\label{lem:Fvee-expr}
With the notation in \autoref{def-quasi-cur}, the shear-induced factors can be written as derivatives of the entropy current with respect to the parameters~$\Jhatpar$, keeping $-\Jbb\cdot\Jbb$ and $J_1/J_0$ fixed:
\bel{Fvee-expr}
\Fvee_m(\Jbb) = \frac{\del(\Fcomp^1 + (J_1/J_0)\Fcomp^0)}{\del\widehat J_m} \biggr|_{-\Jbb\cdot\Jbb,J_1/J_0}  \qquad (m=2,3).
\ee
The entropy multiplier~$\Vbb$ in the sense of \autoref{def-entropcurrent}, whose components $(V_0,V_1)$ are partial derivatives of~$\Fcomp^0$ seen as a function of $(2T^{00},2T^{01},\widehat J_2,\widehat J_3)$ with respect to the first two variables, obeys the identity (where $V_0=-V^0$)
\be
V_0 \mp V_1 = \frac{1}{J_0\pm J_1} \frac{d(\Fcomp_0\pm\Fcomp_1)}{d(J_0\pm J_1)}\Bigr|_{-\Jperp\cdot\Jperp,\Jpar} .
\ee
\end{lemma} 


\subsubsection{Extension to a curved spacetime}

The metric components  and the coupling with the Einstein equations are now taken into account. We introduce the notation 
\bel{equa-GrandDelta}
\aligned
\Delta_\Fbb(\Jbb, \Kpar, \Pbb,\Qbb;\Omega,t)
& = \sign(t) \Bigl( \Fcomp^0 - 2 \sum_a V_a T^{0a} \Bigr) \Bigl(\Ebf_0(\Jpar, \Kpar) + \frac{q_\Jbb}{2}  \, (- \Jbb \cdot \Jbb)\Bigr)
\\
& \quad + \sign(t) \sum_{a=0,1} V_a \Deltafluid^a(\Jbb,\Kpar,\Pbb,\Qbb) \\
& \quad + \frac{2}{|t| \Omega} \frac{\Vbb\cdot\Jperp}{(-\Jperp\cdot\Jperp)} R^\parallel
+ \frac{1}{|t|\Omega} \sum_{m=2,3} \frac{\del\Fcomp^0}{\del\widehat J_m}\biggr|_{\Sigma,\Jhatpar} \widehat R_m
\endaligned
\ee
in which $\Deltafluid^a$, $R^\parallel$ and $\widehat R_m$ are certain quartic (and lower) expressions in the variables $(\Jbb,\Kpar,\Lbb)$, given explicitly in the appendix in~\eqref{explicit-widecheck-R}, \eqref{Rparallel-def}, and~\eqref{widehatR-def}, respectively.
The proof of the following identity is postponed to \autoref{appendix=C-3}. 
Recall that the vector $e_1$ is the unit, spatial derivative vector~\eqref{eq:ourframe}, hence $e_1(\widehat J_m) = \Omega^{-1}\amdeux^{-1}\widehat J_{m,x}$.

\begin{proposition}[Entropy quasi-balance laws] 
\label{lem-gen-entro}
All sufficiently regular solutions to the Einstein--Euler system in torus symmetry satisfy the \textbf{quasi-balance laws}
\bel{equa-balancelaws}
\divuntrois \Bigl( |t|^{-1} \Omega^{-1} \, \Fbb(\Jbb) \Bigr)
= |t|^{-1} \Omega^{-1} \Bigl( \Fvee_2 (\Jbb) e_1(\widehat J_2) + \Fvee_3(\Jbb) e_1(\widehat J_3) \Bigr) + \Delta_\Fbb(\Jbb, \Kpar, \Pbb,\Qbb;\Omega,t) ,
\ee
in which $\Fbb$ is an arbitrary quasi-current and the source term~$\Delta_\Fbb$ is defined in~\eqref{equa-GrandDelta}.
\end{proposition} 

We will use the following terminology.  For regular solutions these terms are smooth, but in the context of tame Einstein--Euler flows, they will be relaxed to signed Radon measures.

\begin{definition}
The terms $|t|^{-1}\Omega^{-1}\Fvee_m (\Jbb) e_1(\widehat J_m)$ ($m=2,3$) in \eqref{equa-balancelaws} are referred to as the \textbf{shear-induced measures}, and represent the loss of conservation of the balance law under consideration. 
\end{definition}


\subsection{The notion of dominated quasi-current}
\label{section=3-5}

\subsubsection{Equation of state for isentropic fluids}

Weak solutions cannot satisfy all of the entropy balance laws we exhibited in \autoref{section=3}: a selection of only four balance laws shall serve as the main set of evolution equations for the fluid variable $\Jbb$, which we may write as equalities.  Additionally, other future-convex entropy currents can also be used, but only as entropy inequalities. 

We treat compressible fluids whose equation of state depends on a single thermodynamic variable, namely the proper energy-density; the physical entropy variable is thus assumed to be a constant. As long as the solutions remain sufficiently regular, for both relativistic and non-relativistic flows, the physical entropy (usually denoted by $S$) is constant; however for weak solutions, the constancy of the physical entropy is regarded as an approximation of a flow described by two thermodynamical variables. 
On the other hand, a drawback of the method of compensated compactness for fluids, already pointed out by DiPerna in the 80's, that it does not apply to the Euler equation when the pressure depends upon two thermodynamical variable. On the other hand, this method has the distinguished advantage to apply to \emph{fluid flow with large data} and thus to cover realistic hyperbolic models from continuum physics. 

\subsubsection{Our proposal in this paper}

Here, we extend to relativistic flows the choice that was made by DiPerna in his pioneering work on the compensated compactness method for non-relativistic fluid flows and, in order for the theory proposed in the present paper to be consistent with DiPerna's theory in the non-relativistic limit (cf.~\autoref{appendix=D-2}), it is natural to treat the stress-energy equation associated with the vector field $\Mbf^{0 \bullet}$ as a \emph{mathematical entropy} inequality. To proceed, we need a convexity property, stated now. With this standpoint, the variable
$\Numb^0 = - \hNumb_\Jbb J_0$ and stress-tensor components $T^{0m}=(-1/2)J_0J_m$, $m=1,2,3$,
play the role of  the \emph{principal variables}\footnote{A different choice of entropy inequality can be made, which will be explored in \autoref{section=10}, below.}. 
 As a result, we will need the following property,
whose proof is postponed to \autoref{appendix=D-4}. 

\begin{proposition}[Convexity of the mass-energy of the fluid]
\label{lem-convexM}
The scalar field $\Jbb \mapsto \Mbf^{00}(\Jbb, 0, 0)$ can be expressed in terms of the principal variables
\bel{equa-hdhgd4}
Z = (Z_l)_{0 \leq j \leq 3} := \Bigl(\Numb^0, \ - \frac{1}{2} J_0 J_1, \ - \frac{1}{2} J_0 J_2, \ - \frac{1}{2} J_0 J_3 \Bigr)
\ee 
and is a \emph{convex function} in these variables.
\end{proposition}

We will then focus on the list of quasi-currents that, in some quantitative sense, are dominated by $\Mbf^{0\bullet}$. More precisely we introduce the following definition.

\begin{definition}
\label{def:quasi-current-dominated}
1. A \textbf{weak-entropy quasi-current} is a current that vanishes on the vacuum, namely 
\bel{equa-zero-vacuum}
\lim_{\Jbb \cdot \Jbb \to 0} \Fbb(\Jbb) = 0
\qquad \text{at fixed } (J_1/J_0 , \widehat J_2, \widehat J_3) \in (-1,1) \times \RR^2 .
\ee

2. Consider $\Fbb^*= \Mbf^{0\bullet}(\Jbb, 0, 0)$, which is known to be a future-convex, weak-entropy quasi-current, referred to as the \textbf{reference entropy current}. Another quasi-current~$\Fbb$ is said  to be \textbf{dominated} by~$\Fbb^*$ provided the Hessian matrix is bounded above and below by the (non-negative) Hessian matrix of the reference entropy current, as follows: 
\bel{equa-condi4}
\pm \Hess_Z( \Fcomp^0) \lesssim \Hess_Z( F^{*0}), 
\ee
where $Z$ denotes the principal variables \eqref{equa-hdhgd4}.
\end{definition}

In~\cite{LeFlochLeFloch-next}, we will provide a parametrization and establish the following properties of weak-entropy quasi-currents. 

\begin{lemma}
\label{lemma:quasi-current-dominated}
Consider the reference weak-entropy current $\Fbb^*= \Mbf^{0\bullet}$. Any quasi-current~$\Fbb$ that is \textbf{dominated} by~$\Fbb^*$ also satisfy the following estimates. 
\bei 

\item $\Fbb$ is a weak-entropy, in the sense \eqref{equa-zero-vacuum}.

\item More precisely, $\Fbb$ is bounded by $\Fbb^*$ in the sense that
\bel{equa-condi1}
\aligned 
\bigl|\Fcomp^0(\Jbb) \pm \Fcomp^1(\Jbb)\bigr| & \lesssim \Fcomp^{*0} (\Jbb)\pm\Fcomp^{*1} (\Jbb) .
\endaligned
\ee

\item Its shear-induced factors are bounded as 
\bel{equa-condi2}
|\Fvee_2 (\Jbb)| + |\Fvee_3 (\Jbb)| \lesssim -\Jbb\cdot\Jbb .
\ee

\item It obeys bounds on its gradient, stated in terms of the entropy multiplier~$\Vbb$ (cf.\ \autoref{lem:Fvee-expr}):
\be
\aligned
\Bigl| \Fcomp^0 - 2 V_0 T^{00} - 2 V_1 T^{01} \Bigr|
+ \bigl| \Vbb\cdot\Jperp \bigr| & \lesssim  -\Jperp\cdot\Jperp ,
\\
|V^0| + |V^1| + \Biggl| \frac{\del\Fcomp^0(\Sigma,\Jhatpar)}{\del\widehat J_2} \Biggr| + \Biggl| \frac{\del\Fcomp^0(\Sigma,\Jhatpar)}{\del\widehat J_3} \Biggr| & \lesssim 1 .
\endaligned
\ee

\eei
\end{lemma}
 

\section{The class of tame Einstein--Euler flows}
\label{section=4}

\subsection{Volume forms and integrability}
\label{section=4-1}

\subsubsection{The  conformal length}

We now introduce the notion of weak solution to the Einstein--Euler equations. In the first part of the section, the lapse function $\Omega>0$ is assumed to be absolutely continuous; cf.~\autoref{def-weaksolu}. A weaker notion of solution, requiring only $\BV$ regularity for $\Omega$, is introduced later in \autoref{section=4-5}.  Throughout, all functions and tensor fields under consideration are assumed to be $\Tbb^2$-symmetric on~$\Tbb^3$. Our formulation relies in an essential way on the hyperbolic and nonlinear structure of the field equations exhibited in~\autoref{section=2}, on the entropy structure uncovered in~\autoref{section=3}, and, later on, on the \emph{a priori} estimates established in Sections~\ref{section=6} and~\ref{section=7}.

We begin with the quotient $\Tbb^3/\Tbb^2\simeq \Sbb^1$ at fixed areal time $t\in\Interval$. Fixing a coordinate~$x$ on~$\Sbb^1$, we write the conformal length measure in the form\footnote{Here, $d_x$ denotes the $dx$ component of the exterior derivative in the basis $dt,dx$.}
\bel{equaDell}
d_x\ell=\amdeux\,dx, \qquad \amdeux>0.
\ee
Equivalently, $\amdeux$ is the density of the measure $d_x\ell$ with respect to the Lebesgue measure~$dx$. We also introduce the \emph{conformal length} function on the universal cover of $\Sbb^1=\RR/\ZZ$,
\bel{equa-ell}
\ell(t,x):=\int_0^x \amdeux(t,y)\,dy, \qquad x\in\RR .
\ee
For each fixed $t\in\Interval$, the function $x\mapsto \ell(t,x)$ is increasing, and is not periodic on~$\Sbb^1$; instead, $\ell(t,x+1)=\ell(t,x)+\ell(t,1)$.
We often consider $\ell$ as a function on $\Sbb^1$ by identifying $\Sbb^1\simeq[0,1)\subset\RR$, as otherwise norms of~$\ell$ would be infinite.
In the sequel, we will use interchangeably the density $\amdeux$, the measure $d_x\ell$, and the primitive $\ell$, all defined on the quotient circle $\Sbb^1\simeq \Tbb^3/\Tbb^2$.


\subsubsection{Notation} 

The solutions under consideration are defined on an open time interval $\Iopen$ whose closure does not contain~$0$. We introduce the spacelike, spacetime, and timelike volume forms by
\bel{equaDVtous}
\aligned
\dVtrois & \coloneqq |t|\,\Omega\, d_x\ell\,dy\,dz
= |t|\,\Omega\,\amdeux\,dx\,dy\,dz,
\\
\dVuntrois & \coloneqq \Omega\,dt\,\dVtrois,
\qquad
\dVundeux \coloneqq |t|\,\Omega\,dt\,dy\,dz.
\endaligned
\ee
These are, respectively, the induced volume forms on the spacelike hypersurfaces of constant time~$t$, on spacetime, and on the timelike hypersurfaces of constant $x\in\Sbb^1$. Unless otherwise specified, all integrability properties stated below are understood with respect to these volume forms.

We find it convenient to adopt an essentially coordinate-independent presentation. Accordingly, we assume that the spacetime is foliated by spacelike slices labelled by a time parameter $t\in\Interval$, together with timelike hypersurfaces labelled by a space parameter $x\in\Sbb^1$. The parameter~$x$ may be regarded only as a label for the one-dimensional quotient $\Tbb^3/\Tbb^2\simeq\Sbb^1$, rather than as a preferred global coordinate. Hence, in the following, we systematically use the notation~\eqref{equaDVtous}.

As a result, for each $t\in\Interval$, the torus~$\Tbb^3$ is equipped with a $\Tbb^2$-symmetric volume form $\dVtrois(t)$, assumed to be absolutely continuous with respect to a fixed smooth reference density, weakly-star measurable in~$t$, and uniformly bounded in total mass. Accordingly, we write 
\be
\dVtrois \in L^\infty(\Interval,\Measac(\Tbb^3)),
\qquad
\sup_{t\in\Interval}\int_{\Tbb^3}\dVtrois(t)<+\infty.
\ee
Equivalently, in standard coordinates \((x,y,z)\) on \(\Tbb^3\), we have 
\be
\dVtrois(t)=f(t,\cdot)\,dx\,dy\,dz
\ee
for some measurable function \(f\geq 0\) satisfying
$f\in L^\infty\bigl(\Interval,L^1(\Tbb^3,dx\,dy\,dz)\bigr)$.

 
\subsubsection{Integrability on spacelike slices}

For each \(t\in\Interval\), we work on the measured space \((\Tbb^3,\dVtrois(t))\). For \(p\in[1,+\infty)\), the weighted Lebesgue space \(L^p(\Tbb^3,\dVtrois(t))\) consists of all measurable functions \(F:\Tbb^3\to\RR\) such that
\be
\|F\|_{L^p(\Tbb^3,\dVtrois(t))}
\coloneqq
\biggl(\int_{\Tbb^3} |F|^p\,\dVtrois(t)\biggr)^{1/p}
<+\infty.
\ee
All spatial integrability properties below are understood with respect to \(\dVtrois(t)\). We also use the space \(L^\infty(\Tbb^3)\), which is independent of the underlying volume form. For \(\Tbb^2\)-symmetric functions, it is naturally identified with \(L^\infty(\Sbb^1)\).

An important class of \(\Tbb^2\)-invariant functions \(F:\Mcal\to\RR\) consists of those whose restriction to each spacelike slice \(\{t\}\times\Tbb^3\) belongs to \(L^p(\Tbb^3,\dVtrois(t))\), with norm uniformly bounded in time. We denote this space by
$L^\infty(\Interval,L^p(\Tbb^3,\dVtrois))$,
and equip it with the norm
\be
\|F\|_{L^\infty(\Interval,L^p(\Tbb^3,\dVtrois))}
\coloneqq
\sup_{t\in\Interval}\|F\|_{L^p(\Tbb^3,\dVtrois(t))}.
\ee
Such functions can also be integrated over spacetime with respect to the volume form
$\dVuntrois = \Omega\,dt\,\dVtrois$.
Assuming that \(\Omega\) is bounded from above, one has
\be
\aligned
\|F\|_{L^p(\Interval\times\Tbb^3,\dVuntrois)}
& \coloneqq
\Bigl(\int_{\Interval\times\Tbb^3}|F|^p\,\dVuntrois\Bigr)^{1/p}
\\
& \leq |\Interval|^{1/p}
\Bigl(\sup_{\Interval\times\Tbb^3}\Omega\Bigr)^{1/p}
\,
\|F\|_{L^\infty(\Interval,L^p(\Tbb^3,\dVtrois))}.
\endaligned
\ee

 
\subsubsection{Integrability on timelike slices}

We now turn to integration over timelike hypersurfaces, using the volume form \(\dVundeux\). For each \(x\in\Sbb^1\), we consider the measured space
$(\Interval\times\Tbb^2,\dVundeux(x))$.
Given \(p\in[1,+\infty)\), a measurable function \(F:\Interval\times\Tbb^2\to\RR\) is said to belong to \(L^p(\Interval\times\Tbb^2,\dVundeux(x))\) if
\be
\|F\|_{L^p(\Interval\times\Tbb^2,\dVundeux(x))}
\coloneqq
\Bigl(\int_{\Interval\times\Tbb^2}|F|^p\,\dVundeux(x)\Bigr)^{1/p}
<+\infty.
\ee
We then define
$L^\infty(\Sbb^1,L^p(\Interval\times\Tbb^2,\dVundeux))$
as the space of measurable functions whose restriction to each timelike slice belongs to \(L^p(\Interval\times\Tbb^2,\dVundeux(x))\), with norm uniformly bounded in \(x\), namely
\be
\|F\|_{L^\infty(\Sbb^1,L^p(\Interval\times\Tbb^2,\dVundeux))}
\coloneqq
\sup_{x\in\Sbb^1}
\|F\|_{L^p(\Interval\times\Tbb^2,\dVundeux(x))}.
\ee
It will also be convenient, at times, to suppress the curved volume forms and work instead with the standard Lebesgue measure on \(\Tbb^3\). Thus, we shall use spaces such as \(L^\infty(\Interval,L^1(\Tbb^3))\), and, for \(\Tbb^2\)-symmetric functions, the corresponding space \(L^\infty(\Interval,L^1(\Sbb^1))\). We also use standard Sobolev spaces such as
$L^\infty(\Interval,H^1(\Tbb^3,\dVtrois))$ and 
${L^\infty(\Interval,\BVac(\Tbb^3)}$, 
whose elements are spatially integrable together with their first-order derivatives.

In addition, $\BV(\Tbb^3)$ denotes the space of functions of bounded variation on~$\Tbb^3$, while $\BVac(\Tbb^3)\subset \BV(\Tbb^3)$ denotes the subspace of functions whose distributional derivatives are absolutely continuous with respect to the Lebesgue measure.  We will correspondingly use $L^\infty(\Interval,\BV(\Tbb^3))$ and $L^\infty(\Interval,\BVac(\Tbb^3))$ for functions with uniform bounds in time on their total spatial variation.
Observe that a $\Tbb^2$-invariant function on~$\Tbb^3$ is in $W^{1,1}(\Tbb^3,dx\,dy\,dz)$ if and only if it is in $\BVac(\Tbb^3)$, and the total variation of such a function~$F$ is given explicitly as
\be
\Var(F) = \int_{\Sbb^1} | F_x |\,dx = |t|^{-1} \int_{\Sbb^1} |e_1(F)| \, \dVtrois .
\ee


\subsection{The notion of \(\Phi\Psi\) flow with finite energy}
\label{section=4-2}

\subsubsection{Relevant variables}

It is useful to distinguish between the regularity and integrability assumptions needed to \emph{define} weak solutions to the Einstein--Euler equations in first-order form, and the further properties that are later shown to \emph{follow} from the equations.

We work on a compact interval \(\Interval\subset\RR\) not containing \(0\), and consider flows on \(\Mcal=\Interval\times\Tbb^3\). In addition to the conformal length measure \(d_x\ell\) and the lapse \(\Omega\), the geometry is described by the variables $\Pbb=(P_0,P_1)$, $\Qbb=(Q_0,Q_1)$, $\Kpar=(K_2,K_3)$, 
while the fluid variables are $\Jperp=(J_0,J_1)$ and $\Jpar=(J_2,J_3)$.  Throughout, the causality conditions~\eqref{eq:causal} and the periodicity constraints~\eqref{equa-pericons} are assumed. The weak regularity assumptions are stated below in~\autoref{weakdefinitionT2}.
\bei

\item \emph{Fluid variables.} The fluid momentum is denoted by $\Jbb=(\Jperp,\Jpar)$,
where \(\Jperp\) and \(\Jpar\) are respectively the components orthogonal and tangent to the \(\Tbb^2\)-symmetry orbits.

\item \emph{Essential geometric variables.} The spacetime geometry is described by the first-order variables $\Pbb,\Qbb:\Interval\times\Tbb^3\to\RR^2$.

\item \emph{Conformal length measure.} On the quotient \(\Sbb^1\simeq \Tbb^3/\Tbb^2\), we consider the positive measure $d_x\ell=\amdeux\,dx$,
where \(x\) is a fixed coordinate on \(\Sbb^1\). Equivalently, one may work with the density \(\amdeux>0\) or with its primitive \(\ell\) defined in~\eqref{equa-ell}. This determines the family of spacelike volume forms $\dVtrois(t)=|t|\,\Omega(t)\,d_x\ell(t)\,dy\,dz$ (for $t\in\Interval$).

\item \emph{Lapse function.} The lapse
$\Omega:\Interval\times\Tbb^3\to(0,+\infty)$
determines the family of timelike volume forms
$\dVundeux(\cdot,x)=|t|\,\Omega(\cdot,x)\,dt\,dy\,dz$ (for $x\in\Sbb^1$) and the spacetime volume form~$\dVuntrois$.

\item \emph{Twist variables.} The functions
$
K_m:\Interval\times\Tbb^3\to\RR$ (for $m=2,3$) 
measure the non-integrability of the orthogonal complement to the symmetry orbits.

\eei


\subsubsection{A notion of flow}

Our notion of weak solution will require the following integrability and weak regularity properties. 

\begin{definition}
\label{weakdefinitionT2}
A \textbf{\(\Phi\Psi\) flow with finite energy} is a collection of time-dependent fields, satisfying the causality inequalities~\eqref{eq:causal} and the periodicity constraints~\eqref{equa-pericons}, and written in the form
\bel{equa-weak-flow}
\Phi=(\Jperp,\Jpar, \Pbb,\Qbb), 
\qquad
\Psi=(\ell,\log\Omega,\Kpar). 
\ee
The variables \(\Phi\) and \(\Psi\) are referred to as the \textbf{\(L^2\)-variables} and the \textbf{\(\BV\)-variables}, respectively, and are assumed to satisfy
\bel{equa-integrable-BV}       
\aligned
& \text{spacelike integrability:} 
&&
\Phi \in L^\infty(\Interval,L^2(\Tbb^3,\dVtrois)),
\\
& \text{timelike integrability:}
&&
\Phi \in L^\infty(\Sbb^1,L^2(\Interval\times\Tbb^2,\dVundeux)),
\\
& \text{spacelike absolute continuity:}
&&
\Psi \in L^\infty(\Interval,\BVac(\Tbb^3)) ,
\\
& \text{timelike absolute continuity:}
&&
\Psi \in L^\infty(\Sbb^1,\BVac(\Interval\times\Tbb^2)) .
\endaligned
\ee 
\end{definition}

The Einstein--Euler equations discussed momentarily could be defined in the sense of distribution under much weaker regularity conditions, specifically spacetime $L^2$ integrability of~$\Phi$ and a uniform bound $\Psi\in L^\infty(\Interval\times\Tbb^3)$.
However, the conditions stated in~\eqref{equa-integrable-BV} are natural in the present context, and will be \emph{derived from the equations} and assumptions on initial data sets taken in the same weak regularity class; cf.~Sections~\ref{section=6} and~\ref{section=7}.


\subsection{Divergence, curl, and derivatives in the weak sense} 
\label{section=4-3}

\subsubsection{Divergence and curl in weak form}

Our weak formulation of the Einstein--Euler equations relies in an essential way on the divergence operator introduced in \eqref{equa-divN}--\eqref{equa-def-div13} for regular fields, namely
\be
\divuntrois(\Xbb)
=
|t|^{-1}\amdeux^{-1}\Omega^{-2}
\bigl( -(|t|\amdeux\Omega X_0)_t + (|t|\Omega X_1)_x \bigr).
\ee
This operator appears in the evolution equations~\eqref{eq:T2-1234-def} for the geometric variables $\Pbb$ and~$\Qbb$, in the Euler equations~\eqref{eq:T2-Euler-perp-def} for the fluid momentum~$\Jbb$, as well as in entropy balance laws such as the particle number conservation~\eqref{equa-entr}. Under suitable integrability assumptions on the arguments, we can define \(\divuntrois\) in the sense of distributions, as presented next. Note first that the corresponding curl operator is introduced by the identity
\be
\curluntrois(X_0,X_1)=-\divuntrois(X_1,X_0).
\ee

Specifically, we say that the equation
\be
\divuntrois\Xbb=F
\ee
holds in the sense of distributions provided
\bse
\label{eq-div-weak}
\bel{eq-div-weak0}
\aligned
& \int_{\Interval}\int_{\Tbb^3} X_0\,\theta_t\,\varphi\,\dVtrois\,dt
-
\int_{\Sbb^1}\int_{\Interval\times\Tbb^2} X_1\,\theta\,\varphi_x\,\dVundeux\,dx
\\
& 
=
\int_{\Interval\times\Tbb^3} F\,\theta\,\varphi\,\dVuntrois
\qquad \text{(weak divergence)}
\endaligned
\ee
for all test functions\footnote{By density, this formulation is equivalent to testing against arbitrary smooth \(\Tbb^2\)-symmetric functions on \(\Interval\times\Tbb^3\); see~\autoref{appendix=E} for a standard density argument.} of the form \(\theta=\theta(t)\) and \(\varphi=\varphi(x,y,z)\), where
\bel{equa-thetavarphi}
\theta\in\Ccal_c^\infty(\Iopen),
\qquad
\varphi\in\Ccal^\infty(\Tbb^3).
\ee
This definition makes sense under the assumptions\footnote{Slightly weaker hypotheses would suffice.}
\bel{equa-for-diverg} 
\aligned
X_0 &\in L^\infty(\Interval,L^1(\Tbb^3,\dVtrois)),
\quad
&
X_1 &\in L^\infty(\Sbb^1,L^1(\Interval\times\Tbb^2,\dVundeux)),
\\
F &\in L^1(\Interval\times\Tbb^3,\dVuntrois),
\endaligned
\ee
which will arise naturally in our setting. For smooth fields \(\Xbb\) and \(F\), the weak formulation~\eqref{eq-div-weak0} is equivalent to the classical one.
\ese
%


\subsubsection{Inequalities in weak form}

The weak formulation of entropy inequalities~\eqref{equa-balancelaws} involves \emph{non-negative} test functions. Namely, the inequality
\be
\divuntrois\Xbb \leq F
\ee
holds in the sense of distributions if, for all non-negative $\theta\in\Ccal_c^\infty(\Iopen)$ and $\varphi\in\Ccal^\infty(\Tbb^3)$,
\be
\aligned
& \int_{\Interval}\int_{\Tbb^3} X_0\,\theta_t\,\varphi\,\dVtrois\,dt
-
\int_{\Sbb^1}\int_{\Interval\times\Tbb^2} X_1\,\theta\,\varphi_x\,\dVundeux\,dx
\\
& 
\leq \int_{\Interval\times\Tbb^3} F\,\theta\,\varphi\,\dVuntrois
\qquad \text{(weak divergence inequality).}
\endaligned
\ee


\subsubsection{Space derivatives in weak form} 

We consider next the spatial derivative arising in the constraint equations~\eqref{eq:theconstraints-def}, which take the form
\bel{equa-ODEx}
Y_x = G\,\Omega\,\amdeux.
\ee
Since these equations are purely spatial, it is convenient to interpret them first on each slice of constant time. Thus, for (almost every) fixed $t\in\Interval$, we say that $Y_x=G\,\Omega\,\amdeux$ holds in the sense of distributions on $\{t\}\times\Tbb^3$ if, with $\dVdeux=|t|\,dy\,dz$,
\bse \label{eq-space-weak}
\be
-\int_{\Sbb^1}\int_{\Tbb^2} Y\,\varphi_x\,\dVdeux\,dx
= \int_{\Tbb^3} G\,\varphi\,\dVtrois
\qquad
\text{for all }\varphi\in\Ccal^\infty(\Tbb^3).
\ee
This is well-defined, for instance, if, for the given time $t\in\Interval$,
\be
Y(t,\cdot) \in L^\infty(\Tbb^3) ,
\quad
G(t,\cdot) \in L^1(\Tbb^3,\dVtrois(t)).
\ee
\ese
Obviously, the corresponding spacetime formulation of \eqref{equa-ODEx} is obtained by testing against functions of the form \(\theta(t)\varphi(x,y,z)\), namely
\bse \label{eq-space-weak-tx}
\be
-\int_{\Sbb^1}\int_{\Interval\times\Tbb^2}
Y\,\Omega^{-1}\theta\varphi_x\,\dVundeux\,dx
=
\int_{\Interval\times\Tbb^3}
G\,\Omega^{-1}\theta\varphi\,\dVuntrois
\ee
for all test functions \(\theta\in\Ccal_c^\infty(\Iopen)\) and \(\varphi\in\Ccal^\infty(\Tbb^3)\). This is a special case of~\eqref{eq-div-weak} with $\Xbb=(0,\Omega^{-1}Y)$ and it is well-defined, for instance, if
\be
Y \in L^\infty(\Sbb^1,L^1(\Interval\times\Tbb^2,\dVundeux)),
\quad
\Omega^{-1}\in L^\infty(\Interval\times\Tbb^3),
\quad
G\in L^1(\Interval\times\Tbb^3,\dVuntrois).
\ee
\ese
In the context of the present paper, a property of weak continuity in time will be available and \eqref{eq-space-weak} and \eqref{eq-space-weak-tx} are then equivalent formulation.


\subsubsection{Time derivatives in weak form} 

We now turn to the time derivatives arising in the Einstein--Euler system. A first equation is the evolution of the conformal length density \(\amdeux\), or equivalently of the measure \(d_x\ell=\amdeux\,dx\). In view of~\eqref{eq:T2-8-def}, this equation takes the form
\bel{Yl-Gl}
\amdeux_t = G\,\amdeux.
\ee
It is convenient to interpret this equation directly in weak form by testing against functions of the form \(\theta(t)\varphi(x)\), exactly as in the weak divergence formulation~\eqref{eq-div-weak0}. Namely, we say that~\eqref{Yl-Gl} holds in the sense of distributions on \(\Interval\times\Tbb^3\) if
\bse \label{eq-time-weak-00}
\be
-\int_\Interval\int_{\Tbb^3} |t|^{-1} \, \Omega^{-1} \, \theta_t\,\varphi\,\dVtrois\,dt
=
\int_\Interval \int_{\Tbb^3} |t|^{-1} \, \Omega^{-1} \, G\,\theta\,\varphi\,\dVtrois \, dt
\ee
for all test functions \(\theta\in\Ccal_c^\infty(\Iopen)\) and \(\varphi\in\Ccal^\infty(\Tbb^3)\).  This formulation is well-defined for instance if
\be
\dVtrois = |t|\,\Omega\,d_x\ell \in L^\infty(\Interval,\Meas(\Tbb^3)) , \qquad
\Omega^{-1} \in L^\infty(\Interval\times\Tbb^3) , \qquad
G \in L^\infty(\Interval, L^1(\Tbb^3, \dVtrois)) .
\ee
\ese
%


Finally, the evolution equations~\eqref{eq:evollambda-def}--\eqref{eq:T2-9101112-evol-def-b} for \(\log\Omega\), \(K_2\), and~\(K_3\) are of the form
\bel{equa-ODEt}
(F)_t=G\,\Omega.
\ee
Since these equations are purely in time, it is useful to interpret them first on each timelike slice.  Thus, for (almost every) fixed $x\in\Sbb^1$, we say that $F_t=G\,\Omega$ holds in the sense of distributions on $\Interval\times\{x\}\times\Tbb^2$ if
\bse\label{eq-time-weak}
\be
-\int_{\Interval\times\Tbb^2} F\,|t|^{-1}\Omega^{-1}\theta_t\,\dVundeux
= \int_{\Interval\times\Tbb^2} |t|^{-1} G \theta\,\dVundeux,
\quad
\text{for all } \theta\in\Ccal_c^\infty(\Iopen) .
\ee
This is well-defined, for instance, if for the given $x\in\Sbb^1$
\be
F(\cdot,x) , \Omega^{-1}(\cdot,x) \in L^\infty(\Interval\times\Tbb^2) ,
\quad
G(\cdot,x) \in L^1(\Interval\times\Tbb^2,\dVundeux(x)) .
\ee
We shall use this later form when integrating evolution equations of the form \eqref{equa-ODEt} in time.
\ese
The corresponding spacetime formulation of \eqref{equa-ODEt} is obtained by integrating against a spatial test function $\varphi\in\Ccal^\infty(\Tbb^3)$.
\bse\label{eq-time-weak-tx}
The equation can also be seen as a particular case of the weak divergence formulation above with $\Xbb=(|t|^{-1}\amdeux^{-1}\Omega^{-1}F,0)$.
Explicitly, we have 
\be
\aligned
&-\int_{\Interval}\int_{\Tbb^3}
F\,|t|^{-1}\Omega^{-1}\amdeux^{-1}\theta_t\varphi\,\dVtrois\,dt
=
\int_{\Interval}\int_{\Tbb^3}
G\,|t|^{-1}\Omega^{-1}\amdeux^{-1}\theta\varphi\,\dVuntrois\,dt,
\endaligned
\ee
which is well-defined if
\be
F\in L^\infty\bigl(\Interval,L^1(\Tbb^3,\dVtrois)\bigr),
\quad
\Omega^{-1}\amdeux^{-1}\in L^\infty(\Interval\times\Tbb^3) ,
\quad
G\in L^1(\Interval\times\Tbb^3,\dVuntrois).
\ee
\ese
In the context of the present paper, a property of weak continuity \emph{in space} will be available and \eqref{eq-time-weak} and \eqref{eq-time-weak-tx} are then equivalent formulation.


\subsection{The proposed notion of weak solution}
\label{section=4-4}

\subsubsection{Class of flows of interest}

In our framework, from the fluid momentum $\Jbb=(\Jperp,\Jpar)$ we introduce the parallel momentum per particle $\Jhatpar = \hNumb_\Jbb^{-1}\Jpar$ defined in~\eqref{equa-jhat}, which is not an additional unknown, but \emph{merely a shorthand}. Finally, \(\ell\) and \(\amdeux\) are related by $d_x\ell=\amdeux\,dx$.

The definition below includes all of the Einstein--Euler equations except for the wave equation for the lapse~\eqref{eq:waveconffac}, which will be recovered later as an inequality; see~\autoref{weakdefinitionT2-2-propo}.
To account for the entropy structure of the equations, we impose the particle number conservation~\eqref{indef-eqvNumb} below as an equality, while the balance law~\eqref{eq:T2-Euler-perp-def-a} for $\Mbf^{0\bullet}$ is imposed as an \emph{inequality}~\eqref{equa-lfk3}.

\begin{definition}
\label{def-weaksolu}
Consider the Einstein--Euler system under the assumption of \(\Tbb^2\) symmetry on \(\Tbb^3\)-spacelike slices, parametrized by an areal time function \(t\in\Interval\). A \(\Phi\Psi\) flow with finite energy, written as
$\Phi=(\Jperp,\Jpar,\Pbb,\Qbb)$ and $\Psi=(\ell,\log\Omega,\Kpar)$
is called an \textbf{Einstein--Euler flow} if the following conditions hold in the weak sense~\eqref{eq-div-weak}, \eqref{eq-space-weak}, \eqref{eq-time-weak-00}, and~\eqref{eq-time-weak}.
\bei

\item[(1)] Einstein's evolution equations~\eqref{eq:T2-1234-def} and~\eqref{eq:T2-all-suite-def}, together with Einstein's constraint equations~\eqref{eq:theconstraints-def}.

\item[(2)] The Euler equations, consisting of the particle-number equation~\eqref{equa-entr} and the three momentum equations~\eqref{eq:T2-Euler-perp-def-b}--\eqref{eq:T2-fluid23-def-b}, namely
\begin{subequations}
\label{eq:T2-Euler-perp-def-r}
\begin{align}
\label{indef-eqvNumb}
\divuntrois \vNumb_\Jbb &= 0,
\\
\divuntrois\bigl(\Omega\,\Mbf^{1\bullet}(\Jbb,\Kpar,\Lbb)\bigr) &= 0,
\label{eq:T2-Euler-perp-def-b-r}
\\
\divuntrois\bigl(|t|^{1/2}J_2\Jperp\bigr)
&= -\curluntrois\bigl(|t|^{1/2}K_2\Pbb/2\bigr),
\label{eq:T2-fluid23-def-a-r}
\\
\divuntrois\bigl(|t|^{1/2}J_3\Jperp\bigr)
&= -\curluntrois\bigl(|t|^{1/2}(K_2\Qbb-K_3\Pbb/2)\bigr).
\label{eq:T2-fluid23-def-b-r}
\end{align}
\end{subequations}

\item[(3)] The energy balance law
\bel{equa-lfk3}
\divuntrois\bigl(\Omega\,\Mbf^{0\bullet}(\Jbb,\Kpar,\Lbb)\bigr)
\leq - \frac{1}{2\,t}\,\Msource(\Jbb,\Kpar,\Lbb). 
\ee

\eei
\end{definition}

The inequality~\eqref{equa-lfk3} will be referred to as the \emph{reference entropy inequality} for the fluid. As is well known, weak solutions to an initial value problem are in general not unique unless an admissibility condition is imposed. In the present setting, the choice \eqref{equa-lfk3} is natural since, by~\autoref{lem-convexM}, the current $\Omega\,\Mbf^{0\bullet}$
is future-convex for the system~\eqref{eq:T2-Euler-perp-def-r}. Hence, the notion introduced in~\autoref{def-weaksolu} is consistent with the general entropy framework in~\autoref{section=3-1}.

\begin{remark}
A conceptually important feature of the present theory is that it also admits an alternative formulation, obtained by exchanging which one of the particle number current~$-\vNumb$ and mass energy $\Mbf^{0\bullet}$ obeys an equality or an inequality.  Indeed, the compactness framework developed in the present paper is not based on a rigid choice of entropy relation, but on a more intrinsic structure of the Einstein--Euler system. This alternative viewpoint is discussed in~\autoref{section=10-1}.
\end{remark}


\subsubsection{The class of tame flows}

For elementary wave patterns, such as solutions of the Riemann problem, the reference entropy inequality~\eqref{equa-lfk3} is sufficient to select the admissible weak solution. For general weak solutions, however, no uniqueness theory is currently available at the level of regularity considered here. Furthermore, our subsequent compactness arguments require a \emph{sufficiently large family} of entropy quasi-currents. These entropies are not imposed by the physics; rather, they are introduced in order to rule out \emph{possibly oscillatory sequences} of solutions and to justify the passage to the limit in the nonlinear fluid terms, as well as in the coupling between matter and geometry.

In view of the quasi-currents presented in~\autoref{section=3}, we therefore strengthen the preceding notion of weak solution. The resulting formulation is slightly delicate, since the quasi-balance laws involve \emph{shear-induced factors} containing derivatives of the parallel momentum per particle.

At this stage, it is useful to introduce the 
\emph{parallel momentum per particle} (cf.~\autoref{defJhat}) $\Jhatpar = \bigl(\hNumb_\Jbb\bigr)^{-1}\Jpar$
or, equivalently, its metric-weighted version $\Gammaunpar= (\Jtilde_2, \Jtilde_3)$ (cf.~\autoref{defGammaun})
\bel{equawidetildeJ-rep} 
\Gammaun_2 = \abs{t}^{1/2} {1 \over \hNumb_\Jbb} \, e^{P/2}  J_2
\qquad 
\Gammaun_3 = \abs{t}^{1/2} {1 \over \hNumb_\Jbb} \, \bigl( Qe^{P/2} J_2+e^{-P/2} J_3 \bigr).
\ee 
where $P,Q$ are reconstructed from first-order variables by integration, and \(\hNumb_\Jbb=\hNumb(\mu)\) is determined by the equation of state \(p=p(\mu)\). We recall that
$\mu+p(\mu)=-\Jbb\cdot\Jbb/2$.

Our bounded variation condition on $\Gammaun_m$ in \eqref{equa-bvpara}, below, strengthens the finite-energy condition on the fluid unknown, but only for the shear-like parallel components. Most importantly, it does not exclude shock waves in the parallel momentum or, more precisely, contact discontinuities.
Observe also that the divergence law in \eqref{equa-bvpara}, below, is \underline{formally} equivalent to the homogeneous transport equations \eqref{equa-Jx}; however, crucially, the latter would not make sense for weak solutions.

\begin{definition}
\label{def-weaksolu-deux}
An Einstein--Euler flow \((\Phi,\Psi)\) is called a \textbf{tame Einstein--Euler flow} provided the following additional conditions hold.

\bei

\item[(4)] The parallel momentum $\Gammaunpar$ satisfies the divergence law~\eqref{equa-kSK1-0} and has bounded variation in space for all times, namely\footnote{The notation in~\eqref{lip-reg} means that 
$\| \Gammaunpar(t_3) - \Gammaunpar(t_2) \|_{L^1(\Tbb^3,\dVtrois(t_2))} \leq C \, |t_2-t_3|$ for some constant $C>0$.}
\bse\label{equa-bvpara}
\begin{align}
&\text{$H$-divergence law }
&&\divuntrois \bigl( H\bigl(\Gammaunpar \bigr) \, \vNumb \bigr) = 0 
\, \text{ for any function } H ,
\\
& \text{bounded spacelike variation}
&& \Gammaunpar\in L^\infty(\Interval,\BV(\Tbb^3)),
\\
& \text{(weak) time continuity}
&& \Gammaunpar\in \Lip(\Interval,L^1(\Tbb^3,\dVtrois)). 
\label{lip-reg}
\end{align}
\ese

\item[(5)]
For every quasi-current \(\Fbb\) that is
dominated by the reference entropy current $\Mbf^{0\bullet}(\Jbb,0,0)$ (cf.~\autoref{def:quasi-current-dominated}), the divergence expression
\bel{equa-55f}
\divuntrois \Bigl( |t|^{-1} \Omega^{-1} \, \Fbb(\Jbb) \Bigr) =: \Fmeasure
\ee
is a $\Tbb^2$-symmetric bounded Radon measure $\Fmeasure \in \Meas(\Interval\times\Tbb^3)$ in spacetime.
\eei
\end{definition}

The additional requirements in~\autoref{def-weaksolu-deux} should be understood as compactness conditions adapted to the compensated compactness method. Condition~(4) provides the relative compactness property needed to control the parallel momentum and, in turn, to control the  the source terms in the quasi-balance laws, while Condition~(5) encodes the quasi-balance laws associated with quasi-currents, allowing for measure terms generated by the derivatives of the parallel variables.


\subsubsection{Parametrization of fluid variables}

In the context of tame solutions, we could modify the notation introduced in \autoref{weakdefinitionT2} and refer to a tame Einstein--Euler flow as a \((\tPhi,\tPsi)\) flow, where
\be
\tPhi=(\Jperp,\Pbb,\Qbb),
\qquad
\tPsi=(\ell,\log\Omega,\Kpar,\Gammaunpar).
\ee
We shall not adopt this notation, however, since within our weak regularity class the knowledge of \(\Jpar\) is equivalent to that of \(\Gammaunpar\), as we now show. 

\bei 

\item[1.] The map $\Jbb\mapsto(\Jperp,\Jhatpar)$ is bijective away from vacuum, from the region $\{-\Jbb\cdot\Jbb > 0,\; J_0 < 0\}$ to the region $\{-\Jperp\cdot\Jperp > 0,\; J_0 < 0\}$.  Indeed, $\Jbb$ is uniquely determined from $(\Jperp,\Jhatpar)$ as follows.  Observe that $2(\mu+p(\mu))+(\widehat J_2^2+\widehat J_3^2)\hNumb(\mu)^2 = -\Jperp\cdot\Jperp$ and the left-hand side is a monotonically increasing function of~$\mu$ that tends to $0$ as $\mu\to 0$ and $+\infty$ as $\mu\to +\infty$, hence $\mu>0$ is uniquely determined from $(\Jperp,\Jhatpar)$, from which we reconstruct $\Jpar=\hNumb(\mu)\Jhatpar$.
Parametrizing fluid variables by $(\mu,u)$ or $(\Jperp,\Jhatpar)$ or $\Jbb$ differs significantly at vacuum, and the choice in \autoref{weakdefinitionT2} is important for our results.

\item[2.] The parallel momentum $\Jpar$ obey $J_2^2 + J_3^2 \leq -\Jperp\cdot\Jperp$, so that $\Jpar$ features the same (spacelike and timelike) $L^2$~regularity as the orthogonal momenta~$\Jperp$.

\eei

 
\subsubsection{Additional weak regularity} 
 
\begin{proposition}[Entropy flows with finite energy: direct properties]
\label{weakdefinitionT2-2-propo}
Any tame Einstein--Euler flow $t \mapsto (\Phi, \Psi)$ (for $t \in \Interval$) as in~\autoref{def-weaksolu-deux} also enjoys the following properties. 
\bei

\item In view of the following consequence of the Einstein--Euler equations
\be
\aligned
t^{-1}  \Omega^{-1} \frac{\amdeux_t}{\amdeux} , \quad t^{-1}  \Omega^{-1} (\log\Omega)_t & \in L^\infty(\Interval,L^1(\Tbb^3,\dVtrois)) ,
\\
t^{-1}  \Omega^{-1} \amdeux^{-1} (\log\Omega)_x & \in L^\infty(\Sbb^1,L^1(\Interval\times\Tbb^2,\dVundeux)) ,
\endaligned
\ee
the divergence operator can be applied in the wave equation for the lapse~\eqref{eq:waveconffac}, which therefore can be understood in the weak sense. 
 As a consequence of the reference entropy inequality~\eqref{equa-lfk3}, an \emph{inequality version} of the wave equation holds:
\bel{eq:waveconffac-ineq}
\aligned
& \divuntrois \Bigl( t^{-1}  \Omega^{-1} (\log(\Omega \amdeux) )_t, \ t^{-1}  \Omega^{-1} \amdeux^{-1} (\log\Omega)_x \Bigr)
\\
& \quad \geq - \frac{1}{4t^3} \Omega^{-2} +  \frac{1}{4t} \, \Mwave(\Jbb, \Kpar, \Lbb).
\endaligned
\ee

\item The constraint equations~\eqref{eq:theconstraints-def} enjoy the following propagation property: for any two times $t_2,t_3\in\Iopen$, they hold on some hypersurface of constant areal time $t=t_2$ if and only if they hold on the hypersurface $t=t_3$.
\eei 
\end{proposition}  

\begin{proof} 
{\bf 1. \it Wave equation for the lapse.} In view of~\eqref{eq:evollambda-def} and the spacelike regularity of $(\Jbb,\Kpar,\Pbb,\Qbb)$, together with boundedness of~$\Omega$, the function $t^{-1}\Omega^{-1}(\log\Omega)_t= \Omega \Mbf^{11}/2 - 1/(4t^2\Omega)$ belongs to $L^\infty( \Interval, L^1(\Tbb^3, \dVtrois))$.
Next, the expression ${(\amdeux)_t \over \amdeux}$ is understood as the Radon–Nikodym derivative of the measure $(d_x\ell)_t$ with respect to $d_x\ell$, and the evolution equation~\eqref{eq:T2-8-def} yields $t^{-1}\Omega^{-1}(\amdeux_t/\amdeux)= \Omega (\Mbf^{00}-\Mbf^{11})/2$, which is in the same space.  Similarly, in view of the equation~\eqref{eq:T2-567-def}, the Radon-Nikodym derivative ${(\log\Omega)_x dx \over  \amdeux dx}$ belongs to ${L^\infty(\Sbb^1, L^1(\Interval \times \Tbb^2, \dVundeux))}$. Consequently, the wave operator in~\eqref{eq:waveconffac-ineq} is well-defined in the weak sense~\eqref{eq-div-weak0}.  The difference of the two sides in~\eqref{eq:waveconffac-ineq} is
\be
\aligned
& \frac{1}{2} \divuntrois \Bigl( \Omega\Mbf^{0\bullet}\Bigr)
+ \divuntrois \Bigl( \frac{1}{4t^2\Omega}, 0 \Bigr)
- \frac{1}{4t^3\Omega^2} + \frac{1}{4t} \, \Mwave
\\
& \quad \leq \frac{1}{4\,t} \bigl(\Mwave - \Msource\bigr)
- t^{-1} \amdeux^{-1} \Omega^{-2} \Bigl(\frac{\amdeux}{4t}\Bigr)_t
- \frac{1}{4t^3\Omega^2}
= 0 .
\endaligned
\ee

\medskip

\noindent{\bf 2. \it Propagation of the constraints.}
Concerning the property of propagation of the constraints~\eqref{eq:theconstraints-def}, we observe that our argument in the proof of Proposition~\ref{theo-struct} (cf.~\autoref{appendix=C-1}) remains valid for weak solutions, provided all the identities are understood in the sense of distributions. Indeed, commuting derivatives is valid for distributions and, in our argument in~\autoref{appendix=C-1}, the identity $L_t=0$ for a distribution $L$ implies that $L$ is independent~of~$t$.
Importantly, the propagation of constraints on $(\log\Omega,K_2,K_3)$ only requires the three spatial Euler equations and not the energy balance law that only holds as an inequality.
\end{proof}


\subsection{The notion of weak solution with corrector}
\label{section=4-5}

\subsubsection{Stress-energy correctors}

When considering limits of sequences of solutions (in \autoref{section=5}, below), we are going to obtain two results depending on whether the initial data sets are assumed to be well-prepared or not.
When initial data sets are \emph{not well-prepared,} the notion of solution described so far is insufficient to describe the limit.
It is necessary to introduce a generalized notion of weak solution, in which some of the Einstein--Euler equations are imposed only modulo suitable correctors, as we now explain.
These new terms will be measures \emph{arising from oscillations of the geometry.}

The key modification is that we no longer require the lapse~$\Omega$ to be absolutely continuous. Instead, the analysis of weakly convergent sequences suggests that the natural regularity is to require that $\Omega$~have bounded variation in space along each spacelike slice and in time along each timelike slice.
 
To describe the correctors, we introduce the space $L^\infty(\Interval,\Meas(\Tbb^3))$, 
consisting of all weakly-star measurable maps \(t\in\Interval\mapsto m(t)\), where \(m(t)\) is a bounded Radon measure on \(\Tbb^3\). Here \(\Meas(\Tbb^3)\) denotes the dual of \(C(\Tbb^3)\), endowed with the sup norm. Thus, the measure \(m(t)\) need not be absolutely continuous with respect to the Lebesgue measure and may therefore contain singular parts, including Dirac masses. Weak-star measurability means that, for every \(\varphi\in C(\Tbb^3)\), the map $t\mapsto \la m(t),\varphi\ra$ is Lebesgue measurable on \(\Interval\). We also use the space \(L^\infty(\Interval,\BV(\Tbb^3))\) of functions that have bounded variation along the spacelike slices.

Let us begin with a formal motivation of how correctors arise as a result of weak $L^2$ convergence of (a sequence of) geometric variables $(\Pbb^\suit,\Qbb^\suit)$ to some limit $(\Pbb^\sharp,\Qbb^\sharp)$.
Squares such as $(P_0^\suit)^2$ may converge to a limit that differs from~$(P_0^\sharp)^2$.  The difference $\lim (P_0^\suit)^2 - (P_0^\sharp)^2 = \lim (P_0^\suit - P_0^\sharp)^2$ is then non-negative and gives rise to a measure term in some equations.
It is thus natural to add such measures whenever quadratic terms in $(\Pbb,\Qbb)$ appear in the equations.
Null forms such as $\Pbb\cdot\Pbb$ are shown to converge, so that the correctors to $P_0^2$ and to $P_1^2$ must coincide.
Altogether, the evolution and constraint equations for the lapse~$\Omega$ must be replaced by the modified equations
\bse\label{eq:evollambda-new-mod-01}
\begin{align}
\label{eq:evollambda-new-mod}
\bigl(\log\Omega \bigr)_t
& = - \frac{1}{4t}
+ \frac{t}{2} \Bigl( \Mbf^{11}(\Jbb,\Kpar,\Pbb,\Qbb)\,\Omega^2 + m^{11}\,\Omega^2 \Bigr) ,
\\
\label{eq:T2-567-new-mod}
\bigl(\log \Omega \bigr)_x
& = - \frac{t}{2} \Bigl( \Mbf^{01}(\Jperp,\Pbb,\Qbb)\,\Omega^2\amdeux + m^{01}\,\Omega^2\amdeux \Bigr) ,
\end{align}
\ese
where the additional terms $m^{01}$ and $m^{11}$ are bounded measures (up to some geometric weights), in contrast to $\Mbf^{01}$ and $\Mbf^{11}$, which are integrable functions.
Likewise, the momentum equation~\eqref{eq:T2-Euler-perp-def-b-r} and energy balance law~\eqref{equa-lfk3} receive contributions from the corrector.

The identities \eqref{eq:evollambda-new-mod-01} are only formal at this stage.
Since the weight $\Omega$ has only bounded variation, it is inconvenient to define its product with a measure (see however \autoref{section=6-1}).
To reach a rigorous formulation, we prefer to regard the weighted quantities
\be
\Pi^{ab}\coloneqq \amdeux \Omega^2 m^{ab}, \qquad a,b=0,1
\ee
as the fundamental measure unknowns.
Since $\amdeux$ is absolutely continuous on each timelike slice, the product $m^{11}\Omega^2 = \Pi^{11}\amdeux^{-1}$ is meaningful provided $\Pi^{11}$ is a measure on (almost) every timelike slice.  This motivates part of our upcoming regularity assumptions.
In the formulation below we work exclusively with the tensor $\Pi=(\Pi^{ab})$, rather than the unweighted quantities~$m^{ab}$.


The modified Einstein equations~\eqref{eq:evollambda-new-mod-01} may be interpreted as adding to the energy-momentum tensor~\(M\) a \emph{corrector stress-energy tensor}, denoted by~\(\Pi\), whose components are bounded measures. The tensor
\[
\Pi=(\Pi^{ab})
\]
is assumed to be symmetric, traceless, and to satisfy the positivity conditions
\bel{eq=allconfor-infty-mod}
\Pi^{00}=\Pi^{11}\geq 0, \qquad
\Pi^{01}=\Pi^{10}, \qquad
\Pi^{00}\geq |\Pi^{01}|.
\ee
These properties arise formally from features of the non-convergent terms $\Ebf_0(\Pbb,\Qbb)$ in $\Mbf^{00}$ and~$\Mbf^{11}$, as well as $\Ebf_1(\Pbb,\Qbb)$ in~$\Mbf^{01}$.  From this point of view, the inequality $\Pi^{00}\pm\Pi^{01}\geq 0$ originates from expressing $\Ebf_0(\Pbb,\Qbb)\mp\Ebf_1(\Pbb,\Qbb) = \frac{1}{2}(P_0\mp P_1)^2 + \frac{1}{2}(Q_0\mp Q_1)^2$ as a sum of squares.

The tensor~$\Pi$ is also assumed to satisfy a divergence inequality, namely
\bse\label{eq:Twave}
\bel{eq:Twave-naive}
\aligned
\opL_{16}(\Psi,\Pi) & \coloneqq \divuntrois\bigl(\Omega^{-1}\amdeux^{-1}\Pi^{+\bullet} \bigr)\leq 0, \qquad
& \Pi^{+\bullet} & \coloneqq \Pi^{0\bullet} + \Pi^{1\bullet} ,
\\
\opL_{17}(\Psi,\Pi) & \coloneqq \divuntrois\bigl(\Omega^{-1}\amdeux^{-1}\Pi^{-\bullet}\bigr) \leq 0,
& \Pi^{-\bullet} & \coloneqq \Pi^{0\bullet} - \Pi^{1\bullet} .
\endaligned
\ee
When suitably interpreted, these equations do not involve any ill-defined product of $\Pi$ by the lapse factor~$\Omega$. Indeed, returning formally to the definition of the divergence operator in~\eqref{equa-divN}, we
 find
\be
\divuntrois\bigl(\Omega^{-1}\amdeux^{-1}\Pi^{\pm\bullet}\bigr)
=
|t|^{-1}\amdeux^{-1}\Omega^{-2}
\Bigl(
\bigl(|t|\Pi^{\pm0}\bigr)_t
+
\bigl(|t|\amdeux^{-1}\Pi^{\pm1}\bigr)_x
\Bigr)
\qquad \text{(formally)}.
\ee
Accordingly, the weak formulation of the equation~\eqref{eq:Twave-naive} is defined by requiring that, for every smooth \(\Tbb^2\)-symmetric test function \(\varphi\) on \(\Interval\times\Tbb^3\),
\be
\int_\Interval \bigl\la \Pi^{\pm0}(t),\, |t|\,\varphi_x(t,\cdot)\bigr\ra_{\Tbb^3}\,dt
+ \int_{\Sbb^1} \bigl\la \Pi^{\pm1}(x),\, |t|\,\amdeux^{-1}\,\varphi_x(\cdot,x)\bigr\ra_{\Interval\times\Tbb^2}\,dx
\leq 0.
\ee
\ese
Here, $\la\cdot,\cdot\ra$ denotes the duality pairing between measures and continuous functions, on $\Interval\times\Tbb^2$ and~$\Tbb^3$, respectively.
The evolution equation of $\log\amdeux$ involves sources that are integrable in time uniformly in space thanks to $\Phi \in L^\infty(\Sbb^1,L^2(\Interval\times\Tbb^2,\dVundeux))$.  Thus, $\log\amdeux$~is continuous along almost every constant-$x$ timelike slice, so that $t\mapsto |t|\,\amdeux^{-1}\,\varphi_x(\cdot,x)$ is continuous.  In order for the duality pairings to be applicable, the measure $\Pi$ must satisfy some mild regularity in both spacelike and timelike directions.  The other divergence equation is defined similarly.
In particular, all terms arising in the weak formulation are either products of continuous functions by measures, or else more regular expressions.

\begin{definition}
\label{weakdefinitionT2-deux}
A \textbf{stress-energy corrector} is a measure tensor \(\Pi=(\Pi^{ab})\) satisfying the regularity conditions
\be
\Pi \in L^\infty(\Interval,\Meas(\Tbb^3)),
\qquad
\Pi \in L^\infty(\Sbb^1,\Meas(\Interval\times\Tbb^2)), 
\ee
the positivity conditions~\eqref{eq=allconfor-infty-mod}, and the divergence inequalities~\eqref{eq:Twave}.
\end{definition}


\subsubsection{An extended notion of flow}

The regularity requirements imposed here are designed so as the notion is \emph{stable under weak convergence}. 

\begin{definition}
\label{weakdefinitionT2-relaxed}
A \textbf{\(\Phi\Psi\Pi\) flow with finite energy} is a collection of time-dependent fields satisfying the causality inequalities~\eqref{eq:causal}, periodicity constraints~\eqref{equa-pericons}, and positivity condition~\eqref{eq=allconfor-infty-mod},
\be
\Phi=(\Jperp,\Jpar,\Pbb,\Qbb),
\quad
\Psi=(\ell,\log\Omega,\Kpar),
\quad
\text{and} \quad \Pi
\ee
and satisfying the conditions 
\bel{equa-integrable-BV-relaxed}
\aligned
& \text{spacelike integrability:} 
&& \Phi \in L^\infty(\Interval,L^2(\Tbb^3,\dVtrois)),
\\
& \text{timelike integrability:}
&& \Phi \in L^\infty(\Sbb^1,L^2(\Interval\times\Tbb^2,\dVundeux)),
\endaligned
\ee
together with 
\bel{equa-integrable-BV-relaxed-deux}
\aligned
& \text{spacelike absolute continuity:}
&& \ell,\Kpar \in L^\infty(\Interval,\BVac(\Tbb^3)) ,
\\
& \text{timelike absolute continuity:}
&&
\ell,\Kpar \in L^\infty(\Sbb^1,\BVac(\Interval\times\Tbb^2)) ,
\\
& \text{spacelike bounded variation:}
&& 
\log\Omega \in L^\infty(\Interval,\BV(\Tbb^3)),
\\
& \text{timelike bounded variation:}
&& \log\Omega \in L^\infty(\Sbb^1,\BV(\Interval \times \Tbb^2)),
\endaligned
\ee
and 
\bel{equa-relaxed-Pi}
\aligned
& \text{spacelike measure condition:}
&& \Pi \in L^\infty(\Interval,\Meas(\Tbb^3)),
\\
& \text{timelike measure condition:}
&& \Pi \in L^\infty(\Sbb^1,\Meas(\Interval\times\Tbb^2)).
\endaligned
\ee
\end{definition}


We are now in a position to introduce the corresponding notion of weak solution.  For the restricted (and much simpler) class of Gowdy spacetimes, the notion of measure correctors was proposed earlier by the authors in~\cite{LeFlochLeFloch-1} 

\begin{definition}
\label{weakdefinitionT2-deux2}
An \textbf{Einstein--Euler flow with corrector} is a  $\Phi\Psi\Pi$ flow with finite energy in the sense of \autoref{weakdefinitionT2-relaxed}, whose stress-energy corrector \(\Pi=(\Pi^{ab})\) satisfies the divergence inequalities~\eqref{eq:Twave}
and Einstein's modified lapse equations
\bse\label{eq:evollambda-new-mod2}
\begin{align}
\opL_5 = \divuntrois\bigl(\Omega\,\Mbf^{0\bullet}(\Jbb,\Kpar,\Lbb)\bigr)
+ \frac{1}{2\,t}\,\Msource(\Jbb,\Kpar,\Lbb)
& \leq - \divuntrois\bigl(\amdeux^{-1}\Omega^{-1}\Pi^{0\bullet}\bigr) ,
\\
\opL_6 = \divuntrois\bigl(\Omega\,\Mbf^{1\bullet}(\Jbb,\Kpar,\Lbb)\bigr)
& = \divuntrois\bigl(\amdeux^{-1}\Omega^{-1}\Pi^{1\bullet}\bigr) ,
\label{eq:T2-Euler-perp-def-b-mod}
\\
\opL_{10} = \bigl(\log\Omega \bigr)_t
+ \frac{1}{4t}
- \frac{t}{2}\Mbf^{11}(\Jbb,\Kpar,\Pbb,\Qbb)\,\Omega^2
& = \frac{t}{2} \Pi^{11} \amdeux^{-1} ,
\label{eq:evollambda-new-mod2a}
\\
\opL_{13} = \bigl(\log \Omega \bigr)_x
+ \frac{t}{2}\Mbf^{01}(\Jperp,\Pbb,\Qbb)\,\Omega^2\amdeux
& = - \frac{t}{2} \Pi^{01},
\label{eq:evollambda-new-mod2b}
\end{align}
\ese
together with all the remaining equations listed in \autoref{def-weaksolu} ---namely the evolution equations of $(\Pbb,\Qbb,\amdeux,\Kpar)$ and constraints of~$\Kpar$, the particle number equation~\eqref{indef-eqvNumb} and parallel Euler equations \eqref{eq:T2-fluid23-def-a-r}--\eqref{eq:T2-fluid23-def-b-r},
and the reference entropy inequality~\eqref{equa-lfk3}.
More precisely, \eqref{eq:evollambda-new-mod2}~is understood in the sense of measures: the first and second identities involve the divergence defined in~\eqref{eq:Twave-naive} and the third identity is understood on timelike slices and the last one on spacelike slices.

Such a flow is said to be a \textbf{tame Einstein--Euler flow with corrector} if it satisfies additionally
the BV regularity and divergence conditions (4) and (5) in \autoref{def-weaksolu-deux}.
\end{definition}


\subsubsection{Calculus rules and Volpert products}

In \autoref{section=6}, we will need to multiply the corrector~\(\Pi\), which is
a bounded Radon measure, by coefficients depending on the bounded variation function
\(\log\Omega\). These products are understood in the sense of the Volpert
calculus. We only recall here the minimal facts needed later on and refer to the general framework of
Dal~Maso--LeFloch--Murat~\cite{DLM}, where the Volpert product appears as a
special case of a broader notion of nonconservative product.

We begin in one space dimension. If \(u\in \BV(\Sbb^1)\), then the one-sided
traces \(u(x-)\) and \(u(x+)\) exist at every point and, for every Lipschitz  
function \(f\), the \emph{Volpert representative} of $f(u)$ is defined by 
\be
\label{eq:Volpert-average-1D}
\Rep f(u)(x)
\coloneqq
\int_0^1 f\bigl(\theta u(x-) + (1-\theta)u(x+)\bigr)\, d\theta .
\ee
At a point $x$ of continuity of $u$, we simply have \(\Rep f(u)(x) = f(u)(x)\), while at jump
points this definition selects the canonical averaged value across the jump.
If \(\mu\in\Meas(\Sbb^1)\), we may then define the product
\be
\Rep f(u)\,\mu
\ee
as a Radon measure. A crucial property is that the Volpert representative preserves sup-norm bounds:
\be
\label{eq:Volpert-sup-bound-1D}
\bigl\|\Rep f(u)\bigr\|_{L^\infty(\Sbb^1)}
\leq
\|f(u)\|_{L^\infty(\Sbb^1)}, 
\ee
therefore, for every Lipschitz function~$f$ and every continuous function $\varphi\in C^0(\Sbb^1)$,
\be
\label{eq:Volpert-measure-bound-1D}
\bigl| \langle \Rep f(u)\,\mu,\varphi\rangle \bigr|
\leq
C \, \|\varphi\|_{C^0(\Sbb^1)}
\, \|\mu\|_{\Meas(\Sbb^1)}, \qquad C := \sup_{\text{range of $u$}} | f |. 
\ee
Importantly, under this regularity, the following chain rule is valid for every smooth function $h$ (see for instance~\cite{DLM}): 
\be
{d \over dx} h(u) = h_u(u) {d \over dx} u. 
\ee


The same idea applies in dimension $1+1$. Every function $u=u(t,x)\in \BV(\Interval\times \Sbb^1)$ admits a decomposition $\Interval\times \Sbb^1 = \mathscr C \cup J \cup I$, 
as follows. 
\bei 

\item When $(t,x)$ belongs to the approximate continuity set $\mathscr C$, then $u$ is $\mathcal L^1$--approximately continuous at $(t,x)$ with respect to the Lebesgue measure $\mathcal L^2$ in two dimensions. 

\item When $(t,x)$ belongs to the \emph{approximate jump set} $\mathscr J$, then $u$ admits two traces \(u^-,u^+\) defined with respect to the Lebesgue measure $\mathcal L^2$ in two dimensions. Moreover, the set is countably \(\mathcal H^1\)-rectifiable for the Hausdorff measure $\mathcal H^1$.  

\item 
At the remaining ``interaction points'', namely in $\mathscr I$, the BV structure is \emph{more singular}, but these points form an $\mathcal H^1$-negligible set and therefore do not affect the definition of the product with Radon measures. 
\eei
\noindent Thanks to this fine structure of bounded variation functions, we can define again the Volpert representative 
\be
\label{eq:Volpert-average-2D}
\Rep f(u)
\coloneqq
\begin{cases}
f(u) & \text{at approximate continuity points},\\[1mm]
\displaystyle \int_0^1 f\bigl(\theta u^- + (1-\theta)u^+\bigr)\, d\theta
& \text{on the approximate jump set},
\end{cases}
\ee
and this yields a bounded Borel function satisfying the same sup-norm
estimate as in one dimension. Accordingly, if \(\mu\in\Meas(\Interval\times
\Sbb^1)\), then \(\Rep f(u)\,\mu\) is a well-defined Radon measure.
 
In \autoref{section=6}, this construction will be used to define the product of $\Omega^k$ with $\Pi^{ab}$ by selecting \(u=\log\Omega\) and \(f(v)=e^{kv}\), so that
\be
\label{eq:Volpert-Omega-Pi}
\Omega^k\,\Pi^{ab}
\quad\text{is replaced by}\quad
\Rep(e^{k\log\Omega})\,\Pi^{ab}.
\ee
The bound
\be
\label{eq:Volpert-sup-bound-spacetime}
\bigl\|\Rep(e^{k\log\Omega})\bigr\|_{L^\infty}
\leq
\|\Omega^k\|_{L^\infty}
\ee
holds on spacelike slices as well as on timelike slices
\(\Interval\times\Tbb^2\). In particular, we have 
\be
\label{eq:Volpert-Pi-estimate}
\bigl| \langle \Rep(\Omega^k)\,\Pi^{ab},\varphi\rangle \bigr|
\leq
\|\Omega^k\|_{L^\infty}
\, \|\varphi\|_{C^0}
\, \|\Pi^{ab}\|_{\Meas} .
\ee
These calculus rules will be needed in
\autoref{section=6}. We shall not develop the theory further here, and in the
rest of the paper we will omit the symbol \(\Rep\) from the notation, since no confusion can arise.


\subsection{Dealing with initial data sets and weak solutions}
\label{section=4-6}

\subsubsection{Initial data sets}

So far, we have described Einstein--Euler flows (possibly with correctors)
on an open time interval \(\Iopen\). We now explain how such a flow is
required to approach prescribed data at the initial time \(t_0\). We state these conditions in terms of an \emph{initial data set}, denoted by
\bel{equa-data}
\mathring \Phi := (\mathring\Jbb, \mathring\Pbb, \mathring\Qbb), 
\qquad
\mathring \Psi := (\mathring{\ell},  \mathring\Omega, \mathring K) ,
\quad \text{and} \quad
\mathring \Pi .
\ee
We also introduce \(\mathring{\amdeux} = d\mathring{\ell}/dx\). These data must obey the constraint equation with corrector~\eqref{eq:evollambda-new-mod2}
\bse
\label{eq:221-copy00}
\bel{eq:T2-567-copy}
(\log \mathring\Omega)_x = - {t_0 \over 2 } \, \Mbf^{01}(\mathring\Pbb, \mathring\Qbb, \mathring\Jbb^\perp) \, \mathring\Omega^2 \mathring\amdeux
- \frac{t_0}{2} \mathring\Pi^{01} , 
\ee
as well as the constraint equations~\eqref{eq:221a}--\eqref{eq:221b}, namely
\bel{eq:221-copy}
\aligned
( \mathring K_2 )_x & = \bigl( \mathring J_0 \mathring J_2 - \mathring P_1 \mathring K_2/2 \bigr) \, \mathring \Omega \mathring \amdeux,  
\qquad
( \mathring K_3 )_x = \bigl( \mathring J_0 \mathring J_3  - \mathring Q_1 \mathring K_2 + \mathring P_1 \mathring K_3/2 \bigr) \, \mathring\Omega \mathring{\amdeux} .
\endaligned
\ee
\ese

Clearly, it suffices to prescribe $\mathring{\ell}, \mathring\Pbb, \mathring\Qbb, \mathring\Jbb^\perp, \mathring\Jpar, \mathring\Pi$ and the value of $\log \mathring\Omega$ and $\mathring K_m$ at one point (or an integral thereof, for instance). The functions \((\log \mathring\Omega,\mathring P,\mathring Q,
\mathring K_m)\) are then recovered by integration of
\eqref{eq:221-copy00} and~\eqref{eq:definevar-01}, subject to the appropriate
periodicity and integrability requirements.

\begin{definition}\label{def-initdataset}
An \textbf{initial data set} for an Einstein--Euler flow with corrector is a
set of data \((\mathring \Phi,\mathring \Psi,\mathring\Pi)\) as in
\eqref{equa-data} satisfying the integrability and regularity conditions
\bse
\bel{initdataset-reg}
\mathring \Phi \in L^2(\Tbb^3, \dVtrois\mathring{\;} ),  
\qquad
\mathring \ell, \mathring \Kpar \in \BVac(\Tbb^3) ,
\qquad
\log\mathring\Omega \in \BV(\Tbb^3) ,
\qquad
\mathring \Pi \in \Meas(\Tbb^3) ,
\ee
in which \(\dVtrois\mathring{\;} = |t_0|\,\mathring\Omega\, d_x\mathring\ell
\, dydz\), as well as the Einstein constraints with corrector~\eqref{eq:221-copy00},
together with the positivity and symmetric tracelessness conditions
\bel{initdataset-pos}
\aligned
& 
\mathring \amdeux > 0 , \qquad
\mathring \Omega > 0 , \qquad
- \mathring\Jbb^\perp \cdot \mathring\Jbb^\perp \geq 0, \quad \mathring J_0 \leq 0 ,
\\
& 
\mathring\Pi^{01} = \mathring\Pi^{10} , \quad
|\mathring\Pi^{01}| \leq \mathring\Pi^{00} = \mathring\Pi^{11} ,
\endaligned
\ee
and the periodicity constraints (where
\(\mathring P = \int_0^x \mathring P_1 \mathring\Omega \, d_x\mathring\ell\))
\bel{initdataset-period}
\int_{\Tbb^3} \mathring P_1 \dVtrois\mathring{\;}
= \int_{\Tbb^3} e^{-\mathring P} \mathring Q_1 \dVtrois\mathring{\;}
= 0 .
\ee
Such a data set is said to be a \textbf{tame initial data set} if its
metric-weighted parallel momentum
\be
\mathring\Gammaun_2 \coloneqq \abs{t_0}^{1/2} \bigl( \hNumb_{\mathring\Jbb}\bigr)^{-1} \, e^{\mathring P/2} \mathring J_2,
\qquad 
\mathring\Gammaun_3 \coloneqq \abs{t_0}^{1/2} \bigl(\hNumb_{\mathring\Jbb} \bigr)^{-1} \, \bigl( \mathring Qe^{\mathring P/2} \mathring J_2+e^{-\mathring P/2} \mathring J_3 \bigr)
\ee
\ese
satisfies additionally the regularity condition
\(\mathring\Gammaunpar\in \BV(\Tbb^3)\).
\end{definition}

\paragraph{Boundary terms at the initial time.}

Thanks to the notation above, initial data can be incorporated directly into
the weak formulation by allowing test functions that are compactly supported
in the half-closed interval \(\Ihalf\). In this way, Stokes' formula produces a
boundary contribution at the initial time \(t_0\), and we adopt the resulting
identities as our definition of weak solutions with prescribed initial data.
If desired, one could of course include analogous boundary contributions at a
final time \(t_1\), but this will not be needed here.

For instance, we replace~\eqref{eq-div-weak} by the statement that, for all
smooth functions \(\varphi\in\Ccal^\infty(\Tbb^3)\) and
\(\theta=\theta(t)\) compactly supported in the half-closed interval~\(\Ihalf\),
\bel{eq-div-weak0-bound}
\aligned
& \int_{\Interval} \int_{\Tbb^3}  X_0 \theta_t \varphi  \, \dVtrois dt
- \int_{\Sbb^1} \int_{\Interval \times \Tbb^2}X_1  \, \theta \varphi_x \, \dVundeux  dx
- \theta(t_0) \int_{\Tbb^3} \mathring X_0 \varphi  \dVtrois\mathring{\;}
\\
& = \int_{\Interval \times \Tbb^3} F \, \theta \varphi \, \dVuntrois 
\hskip1.cm  
 \text{(weak divergence with boundary),}
\endaligned
\ee
which involves the prescribed data $\mathring X_0, \dVtrois\mathring{\;}$. The weak formulations of inequalities and of time-derivative equations
developed in \autoref{section=4-3} extend in exactly the same way. The
spatial formulation~\eqref{eq-space-weak} of the constraints on a fixed time
slice remains unchanged, whereas the spacetime formulation
\eqref{eq-space-weak-tx} acquires the expected boundary term at $t=t_0$.
The divergence inequalities~\eqref{eq:Twave} for the corrector with initial data $\mathring\Pi$ is explicited in \autoref{section=6-2} as it is used there.


\subsubsection{Definition for flows with finite energy}

We now incorporate the initial data into our terminology. Recall that
the constraint equations are imposed in the open interval \(\Iopen\) by
\autoref{weakdefinitionT2-deux2}, and at the initial time \(t_0\) by
\autoref{def-initdataset}.

\begin{definition}
\label{def-initialdata}
An Einstein--Euler flow with corrector
(cf.~\autoref{weakdefinitionT2-deux2}), namely
\(\Phi(t)\), \(\Psi(t)\), \(\Pi(t)\) for \(t \in \Iopen\), is said to
\textbf{assume a prescribed initial data set}
\(\mathring \Phi, \mathring \Psi,\mathring\Pi\) at the time \(t_0\), provided
all evolution equations hold \emph{in the integral sense with boundary terms}
included at \(t=t_0\) and determined by the data
\(\mathring \Phi,\mathring \Psi,\mathring \Pi\). Explicitly, these equations
are the divergence inequalities of the corrector~\eqref{eq:Twave}, the modified
evolution equation of the lapse~\eqref{eq:evollambda-new-mod2a}, the evolution
equations \eqref{eq:T2-1234-def}, \eqref{eq:T2-8-def},
\eqref{eq:T2-9101112-evol-def-a}, \eqref{eq:T2-9101112-evol-def-b} of
\((\Pbb,\Qbb,\amdeux,\Kpar)\), the particle number and spatial Euler
equations~\eqref{eq:T2-Euler-perp-def-r}, and the energy balance
law~\eqref{equa-lfk3}.

If the flow is tame, namely if it satisfies the bounded variation and
divergence conditions~(4) and~(5) in
\autoref{def-weaksolu-deux}, it is said to \textbf{tamely assume} the
prescribed tame initial data set if the following conditions are satisfied. 
\bei
\item The \(H\)-divergence law holds in the integral sense with boundary terms
included at \(t=t_0\). 

\item The spatial variation of the metric-weighted parallel momentum, together with its time continuity constant, 
\(\Gammaunpar\) is bounded by its initial value, that is, for each $m=2,3$,
\bel{equa-DK399}
\aligned
\Var(\Gammaun_m(t);\Tbb^3) 
& \leq \Var(\mathring\Gammaun_m;\Tbb^3) , \qquad &&{\rm a.e.\ } t\in\Iopen ,
\\
\int_{\Tbb^3} \biggl| \frac{\Gammaun_m(t_3) - \Gammaun_m(t_2)}{t_3-t_2}\biggr| \, \dVtrois(t_2)
 & \leq \Var(\mathring\Gammaun_m;\Tbb^3) , \, &&{\rm a.e.\ }t_2, t_3 \in \Iopen . 
\endaligned
\ee

\item For every quasi-current~\(\Fbb\) that is dominated by the reference
entropy current, the balance law~\eqref{equa-55f}, namely
\be
\divuntrois\bigl( |t|^{-1} \Omega^{-1} \, \Fbb(\Jbb)\bigr) = \Fmeasure ,
\ee
holds in the integral sense with boundary terms included at \(t=t_0\), and
the shear-induced measure \(\Fmeasure\in\Meas(\Iopen\times\Tbb^3)\) satisfies
the estimate
\bel{Fmeasure-apriori-def}
\aligned
& \|\Fmeasure\|_{\Meas(\Iopen\times\Tbb^3)} 
\leq C_\Fbb 
\bigl( \aleph + \aleph^2 \bigr)
\\
& \aleph :=  \|\mathring\Phi\|_{L^2(\Tbb^3,\dVtrois)}^2 + \|\mathring\Psi,\mathring\Gammaunpar\|_{\BV(\Tbb^3)} + \|\mathring\Pi\|_{\Meas(\Tbb^3)},
\endaligned
\ee
where \(C_\Fbb>0\) depends only on the quasi-current.
\eei
\end{definition}

\paragraph{On the shear-induced measures.}

The measure \(\Fmeasure\) should be understood as collecting all terms that
cannot be represented at the present level of regularity (unless further integrability or regularity is known). In particular, it includes contributions generated by
the lower-order source term \(\Delta_\Fbb\), by the singular products
involving the spacelike derivative of the metric-weighted parallel momentum,
and, depending on the approximation procedure, by the corrector~$\Pi$. This is the reason why all norms of
\((\mathring\Phi,\mathring\Psi,\mathring\Gammaunpar,\mathring\Pi)\) are
allowed to appear on the right-hand side of
\eqref{Fmeasure-apriori-def}. The estimate is intentionally formulated in a
stable way: it is strong enough to control the family of admissible flows from
their initial data only, yet weak enough to remain meaningful and stable under passage to
the limit.

\paragraph{Selection by initial data.}

For compressible fluid flows, even in the $1+1$ setup under consideration, \emph{no uniqueness} result is available, at least at the level of weak regularity (finite energy) that we propose in this paper. 
The role of the tame bounds is therefore not to single out a unique solution, but rather to select a class of
solutions that remains quantitatively controlled by the initial data. More
precisely, the variation bounds on \(\Gammaunpar\) and the measure bounds on
the shear-induced defects provide a priori estimates depending only on the
prescribed data set. This is exactly the level of control needed later in the
compactness analysis: it allows us to deduce compactness properties of entire
families of (weak, approximate) solutions from assumptions on the initial data alone. At the same
time, these bounds are deliberately kept at the weakest level compatible with
stability, so that they continue to hold in the weak limits considered in the
subsequent sections.


\section{Compactness, stability, and instability}
\label{section=5}

\subsection{The notion of bounded sequence}
\label{section=5-1}

\subsubsection{Strategy}

In \autoref{section=4}, we introduced the tame regularity framework for
Einstein--Euler flows, together with the enlarged notion of flow with
corrector. We now investigate the behavior of sequences of such flows under weak convergence.

Our results are formulated for sequences of \emph{approximate} solutions, a
setting that simultaneously covers exact and approximate solutions, regular
and weakly regular solutions. In particular, configurations involving both
shock waves in the fluid and impulsive gravitational waves in the geometry are included here. 
We identify the natural class of bounded
sequences for which we may extract convergent subsequences and describe the
corresponding weak limits. 

The corresponding bounds are expressed along the
spacelike slices of the areal foliation and involve, in particular, the
variables \(\ell\) and \(\log\Omega\), which govern the foliation and the
spacetime volume measure. In the future-contracting regime, a lower bound on
the conformal length is also required. At this stage, these bounds are taken as assumptions. The a priori estimates established in
Sections~\ref{section=6} and~\ref{section=7} show that they are precisely the
bounds propagated by the evolution in the class of solutions considered here.
Thus, the present section provides the abstract compactness framework, while
the subsequent sections supply the estimates needed to apply it.

At the level of regularity considered here,     standard compactness arguments based on embedding theorems are insufficient to pass to the limit in the Einstein and Euler equations. For sequences uniformly bounded in the natural norms
of the problem, nonlinear interaction may generate oscillations and lack of convergence in the
limit. This is precisely the reason for introducing, already in
\autoref{section=4}, 
    the notion of a tame Einstein–Euler flow with corrector.

Two distinct regimes arise. For ``well-prepared'' initial data, we obtain a
nonlinear stability statement: after extraction of a subsequence, the weak
limit remains a tame Einstein--Euler flow satisfying the original system; see
\autoref{theorem-492}. For ``ill-prepared'' data, by contrast, we obtain only a
weaker stability property: the limit is still a tame flow, but an additional
corrector stress tensor may appear in the Einstein equations; see
\autoref{theorem-493}. This instability mechanism is intrinsic to sequences of weakly
regular Einstein flows.

    The compactness framework developed in this section is relevant
     not only for perturbations of Cauchy data,
but also for approximation procedures such as vanishing viscosity and finite
volume schemes, as further explored in our companion
paper~\cite{LeFlochLeFloch-next}.


\subsubsection{Bounded sequences of \(\Phi\Psi\Pi\) flows}

We continue to rely on the notation
$\Phi=(\Jperp,\Jpar,\Pbb,\Qbb)$, $\Psi=(\ell,\log\Omega,\Kpar)$ and~$\Pi$, 
where the variables in $\Phi$ are of $L^2$-type, those in $\Psi$ are of $\BV$-type, and $\Pi$ is a measure.
We are interested in sequences of flows $\Interval\ni t\mapsto \bigl(\Phi^\suit,\Psi^\suit,\Pi^\suit\bigr)(t)$ that are uniformly controlled in the natural norms as follows.
\bse
\label{equa-ConditionTheo-0}
\bei

\item[(1)] \textbf{Weak regularity.}
For each \(\suit\in(0,1]\), the map \(\bigl(\Phi^\suit,\Psi^\suit,\Pi^\suit\bigr)\) is a flow with finite energy on the spacetime slab $\Interval\times\Tbb^3=\Iopen\times\Tbb^3$, 
with $t_0<t_1<0$ or $0<t_0<t_1$, in the sense of~\autoref{weakdefinitionT2-relaxed}. 

\item[(2)] \textbf{Uniform $L^2$, $\BV$ and measure bounds.}
The following norms of the sequence
are uniformly bounded in the limit $\suit\to 0$:
\bel{equa-ConditionTheo-1}
\aligned
& \sup_{t\in\Interval} \|\Phi^\suit\|_{L^2(\Tbb^3,\dVtrois_\suit)} , \quad
&& \sup_{x\in\Sbb^1} \|\Phi^\suit\|_{L^2(\Interval\times\Tbb^2,\dVundeux_\suit)} ,
\\
& \sup_{t\in\Interval}\Var(\Psi^\suit(t);\Tbb^3) , \quad
&& \sup_{x\in\Sbb^1} \Var(\Psi^\suit(x);\Interval\times\Tbb^2) , \quad
&& \sup_{\Interval\times\Tbb^3} |\Psi^\suit| ,
\\
& \sup_{t\in\Interval} \|\Pi^\suit\|_{\Meas(\Tbb^3)} , \quad
&& \sup_{x\in\Sbb^1} \|\Pi^\suit\|_{\Meas(\Interval\times\Tbb^2)} .
\endaligned
\ee

\item[(3)] \textbf{Non-vanishing conformal length.}
In the future-contracting regime, that is, on an interval $\Iopen$ with $t_0<t_1<0$, the  conformal length functional is uniformly bounded from below:
\bel{equa-ConditionTheo-2}
\liminf_{\suit\to0}\inf_{t\in\Interval}\Lbfscri^\suit(t)>0,
\qquad
\Lbfscri^\suit(t)=|t|^{-1}\int_{\Tbb^3}(\Omega^\suit)^{-1}\,\dVtrois_\suit.
\ee

\item[(4)] \textbf{Compatible conformal lengths.} We assume that the sequence of conformal length densities is Cauchy in $L^1(\Tbb^3)$ at each time $t\in\Interval$, namely
\bel{equa-normL1}
\lim_{\alpha \to 0} \sup_{\suit,\suit'\leq \alpha}
\|\amdeux^\suit(t,\cdot)-\amdeux^{\suit'}(t,\cdot)\|_{L^1(\Tbb^3)} = 0.
\ee
\eei
\ese
\noindent
By analogy with \autoref{weakdefinitionT2}, we use the following terminology. 

\begin{definition}
A sequence satisfying {\rm(1)}--{\rm(4)} is called a \textbf{bounded sequence of \(\Phi\Psi\Pi\) flows with finite energy}.  If $\Pi=0$, it is called a bounded sequence of $\Phi\Psi$ flows with finite energy.
\end{definition} 

For sequences of exact solutions, these bounds \emph{are effectively derived} from bounds at initial time thanks to the fundamental balance laws and entropy inequalities, as we establish later (cf.~Sections~\ref{section=6} and~\ref{section=7}). For sequences of approximate solutions, by contrast, they must be imposed as part of the compactness framework. In concrete approximation schemes, such as the vanishing viscosity method or finite volume methods, these estimates will have to be verified separately (cf.~\cite{LeFlochLeFloch-next}).

The normalization~\eqref{equa-normL1} above can be ensured at initial time by selecting for instance the conformal length~$\ell$ along the $t=t_0$ initial data surface as a spatial coordinate, namely by normalizing $\amdeux(t_0)=1$ at initial time.  As we prove in \autoref{section=7} in the case of exact solutions, the Cauchy property is then propagated.


\subsubsection{Bounded sequences of solutions}

In addition to the uniform bounds introduced above, we assume that the field equations are satisfied up to error terms, referred to as \emph{negligible terms}.
We use the notation \(\Hterm^\suit\) and \(\Hterm^\suit_0\) for sequences of distributions, and \(\Mterm^\suit\) and \(\Mterm^\suit_-\) for sequences of bounded measures, as follows\footnote{The notation $\Hterm^\suit$ (and likewise~$\Hterm^\suit_0$) is a slight abuse, to enable stating divergence and evolution equations conveniently.  To be precise, $\Hterm^\suit\,\dVuntrois{}^\suit$ is a distribution in the standard negative Sobolev space $H^{-1}(\Interval\times\Tbb^3)$, defined as the dual of the standard Sobolev space $H^1(\Interval\times\Tbb^3;dt\,dx\,dy\,dz)$.}.
\bse
\bei

\item {\bf Terms \(\Hterm^\suit\):}
\be
\Hterm^\suit \, \dVuntrois{}^\suit \text{ is relatively compact in } H^{-1}(\Interval\times\Tbb^3).
\ee

\item {\bf Terms \(\Hterm^\suit_0\):}
\be
\aligned
& \Hterm^\suit_0 \, \dVuntrois{}^\suit \text{ is bounded in } H^{-1}(\Interval\times\Tbb^3),
\\
& \lim_{\suit\to0} \bigl\langle \Hterm^\suit_0\,\dVuntrois{}^\suit, \varphi \bigr\rangle_{\Interval\times\Tbb^3} =0
\qquad \text{for every $\varphi\in\Ccal^\infty(\Interval\times\Tbb^3)$.}
\endaligned
\ee

\item {\bf Terms \(\Mterm^\suit\):}
\be
\Mterm^\suit \text{ is a bounded sequence in } \Meas(\Interval\times\Tbb^3).
\ee

\item {\bf Terms \(\Mterm^\suit_-\):}
\be
\aligned
& \Mterm^\suit_- \text{ is a bounded sequence in } \Meas(\Interval\times\Tbb^3),
\\
& \limsup_{\suit\to0}\la \Mterm^\suit_-,\varphi\ra \leq 0
\quad \text{for every non-negative \(\varphi\in\Ccal^\infty(\Interval\times\Tbb^3)\).}
\endaligned
\ee

\eei
\ese
We emphasize that the symbols \(\Hterm^\suit\), \(\Hterm^\suit_0\), \(\Mterm^\suit\), and \(\Mterm^\suit_-\) denote classes of admissible error terms, and therefore need not refer to the same sequence at each occurrence. We also use below the convergence in the sense of distributions. 


To define the notion of bounded sequence of (approximate) Einstein--Euler flows with correctors, we supplement {\rm(1)}--{\rm(4)} above with the following conditions (which mimic the definition for exact solutions in \autoref{weakdefinitionT2-deux2}),
expressed in terms of the operators $\opL_a$ introduced in~\eqref{equa-JKLoperators} and~\eqref{eq:Twave}.
\bei

\item[(5)]
\textbf{Einstein--Euler equations.}
Einstein's evolution equations~\eqref{eq:T2-1234-def} for $\Pbb,\Qbb$ are satisfied up to negligible contributions \emph{that are bounded in} $H^{-1}(\Interval\times\Tbb^3)$,
\bse\label{approx-Einstein-Euler}
\bel{approx-Einstein-Euler-PQ}
\opL_a(\Phi^\suit,\Psi^\suit) = \Hterm^\suit_0,
\qquad 1 \leq a \leq 4 .
\ee
The particle number equation~\eqref{equa-entr} and the parallel Euler equations~\eqref{eq:T2-fluid23-def-a} and~\eqref{eq:T2-fluid23-def-b}, Einstein's other evolution and constraint equations \eqref{eq:T2-all-suite-def}--\eqref{eq:theconstraints-def} are satisfied \emph{up to negligible contributions}\footnote{The convergence to zero is defined for each type of equation (divergence, evolution, constraint) by considering the spacetime versions of their weak definitions in \autoref{section=4-3}.}
\be
\aligned
\divuntrois_\suit \vNumb^\suit & \to 0\,  \text{ distributionally,} 
\\
\opL_a(\Phi^\suit,\Psi^\suit) & \to 0\,  \text{ distributionally,} 
\qquad && 7 \leq a \leq 15,\quad a\not\in\{10,13\} .
\endaligned
\ee
The \emph{modified} orthogonal Euler equations and lapse equations~\eqref{eq:evollambda-new-mod2} are satisfied up to negligible contributions,
\bel{equa-kdkld99}
\aligned
& \opL_6(\Phi^\suit,\Psi^\suit) - \divuntrois\bigl(\amdeux^{-1}\Omega^{-1}\Pi^{1\bullet}\bigr) \to 0 ,
\\
& \opL_{10}(\Phi^\suit,\Psi^\suit) - \frac{t}{2} (\amdeux^\suit)^{-1}\Pi^{11}{}^\suit \to 0 , \qquad
\opL_{13}(\Phi^\suit,\Psi^\suit) + \frac{t}{2} \Pi^{01}{}^\suit \to 0 \quad
\text{distributionally} .
\endaligned
\ee
The divergence inequalities~\eqref{eq:Twave} for the corrector also hold
up to a measure~$\Mterm^\suit_-$ and negligible contributions
\be
\aligned
\opL_a(\Psi^\suit,\Pi^\suit) - \Mterm^\suit_- & \to 0\,  \text{ distributionally.}
\qquad && a = 16, 17 .
\endaligned
\ee
\ese

\item[(6)]
\textbf{Reference entropy.}
 The energy balance law is satisfied up to a corrector term, a measure~$\Mterm^\suit_-$ and negligible contributions,
\bel{equa-lfk3-approx}
\opL_5(\Phi^\suit,\Psi^\suit) + \divuntrois\bigl(\amdeux^{-1}\Omega^{-1}\Pi^{0\bullet}\bigr) - \Mterm^\suit_-
\to 0 \text{ distributionally.}
\ee
\eei
Furthermore, as explained in \autoref{section=4}, the compensated compactness technique requires additional entropy structures, which were stated for exact solutions in \autoref{def-weaksolu-deux}.
We require the following conditions.
\bei

\item[(7)]
\textbf{Parallel momentum.}
The metric-weighted parallel momentum $\Gammaunpar{}^\suit$ satisfies the divergence law~\eqref{equa-kSK1-0} and has uniformly bounded variation, as follows: 
\bel{equa-bvpara-bis}
\aligned
&\text{$H$-divergence law }
&&\divuntrois_\suit \bigl( H\bigl(\Gammaunpar{}^\suit \bigr) \, \vNumb^\suit \bigr) = \Hterm^\suit_0
\, \text{ for any function } H,
\\
& \text{spacelike bounded variation}
&& \sup_{\eps \in (0,1]} \sup_{ t \in \Interval} \Var(\Gammaunpar{}^\suit(t)) < + \infty,
\\
&\text{time continuity modulus}
&& \sup_{\eps \in (0,1]} \sup_{ t_2\neq t_3 \in \Interval} \Bigl\| \frac{\Gammaunpar{}^\suit(t_3)-\Gammaunpar{}^\suit(t_2)}{t_3 - t_2} \Bigr\|_{L^1(\Tbb^3,\dVtrois{}^\suit(t_2))} < + \infty.
\endaligned
\ee

\item[(8)]
\textbf{Quasi-currents.}
Furthermore, for every quasi-current \(\Fbb=\Fbb(\Jbb)\) dominated by the reference entropy current (cf.~\autoref{def:quasi-current-dominated}), the following entropy balance law holds:
\bel{equa-balancelaws-plus}
\aligned
\divuntrois_\suit \bigl( |t|^{-1} (\Omega^\suit)^{-1} \, \Fbb(\Jperp{}^\suit,\Jhatpar{}^\suit) \bigr)
= \Fmeasure^\suit + \Hterm^\suit,
\endaligned
\ee
where the sequence of measures $\Fmeasure^\suit$ is uniformly bounded
\be
\Fmeasure^\suit = \Mterm^\suit 
\ee
and the relatively compact terms $\Hterm^\suit$ stand for possibly additional approximation errors. 
\eei

We summarize the assumptions of the compactness framework as follows.

\begin{definition}\label{def-810}
A bounded sequence of $\Phi\Psi\Pi$ flows with finite energy
$t\in\Interval \mapsto \bigl(\Phi^\suit,\Psi^\suit,\Pi^\suit\bigr)(t)$
satisfying {\rm(1)}--{\rm(6)} above is called a \textbf{bounded sequence of approximate Einstein--Euler flows with corrector}; it is called a \textbf{tame sequence} if it also satisfies  {\rm(7)}--{\rm(8)}.
The qualifier ``approximate'' is omitted for exact solutions.
\end{definition}

Approximate solutions generated, for instance, by finite volume schemes or by vanishing viscosity regularizations will be constructed separately for the Einstein--Euler system in~\cite{LeFlochLeFloch-next}.


\subsection{Compactness framework for general data sets}
\label{section=5-2}

We now proceed and consider the convergence of sequences of flows.  Quadratic oscillations in the essential geometric variables $\Pbb^\suit,\Qbb^\suit$ may produce a corrector tensor, which means that the natural class of flows to consider are Einstein--Euler flows with corrector.  This class as a whole is stable, but the corrector is unstable due to contributions from the geometry.  The proof is given in Sections~\ref{section=8} and~\ref{section=9}.

\begin{theorem}[Nonlinear instability of Einstein spacetimes with corrector]
\label{theorem-493}
Consider a tame sequence of approximate Einstein--Euler flows with corrector ${\Interval \ni t \mapsto (\Phi^\suit, \Psi^\suit, \Pi^\suit)(t)}$ (cf.~\autoref{def-810}).  
Then, there exists a limiting flow ${\Interval \ni t \mapsto \big( \Phi^\sharp, \Psi^\sharp, \Pi^\sharp)(t)}$, 
so that the following properties hold, after extraction of a subsequence if necessary.

\bei 

\item ${(\Phi^\sharp, \Psi^\sharp, \Pi^\sharp)(t)}$ is a \emph{tame Einstein--Euler flow with corrector} in the sense of~\autoref{weakdefinitionT2-deux2}.

\item It is bounded in natural norms, the variables $(\Phi^\sharp,\Psi^\sharp)$ converge in the sense that\footnote{The convergence in~\eqref{eq:conv-prop-deux} holds in stronger norms that can be deduced from the uniform bounds and are explicited in Sections~\ref{section=8} and~\ref{section=9}.}
\bel{eq:conv-prop-deux}
\aligned
& \Jbb^\suit \to \Jbb^\sharp
\qquad 
&& \text{ strongly in } L^2(\Interval \times \Tbb^3),
\\
& (\Pbb^\suit,\Qbb^\suit) \rightharpoonup (\Pbb^\sharp,\Qbb^\sharp)
\qquad 
&& \text{ weakly in } L^2(\Interval \times \Tbb^3), 
\\ 
& \Psi^\suit \to \Psi^\sharp && \text{ for almost every } (t,x) \in \Interval \times \Tbb^3 ,
\endaligned
\ee
while $\Pi^\suit$ only converges up to additional terms,
\be
\aligned
& \Pi^\suit + \amdeux^\suit (\Omega^\suit)^2 \bigl( \Pbb^\suit \otimes \Pbb^\suit + \Qbb^\suit \otimes \Qbb^\suit \bigr)  
\\
& \hskip1.cm \to  \Pi^\sharp + \amdeux^\sharp (\Omega^\sharp)^2 \bigl( \Pbb^\sharp \otimes \Pbb^\sharp + \Qbb^\sharp \otimes \Qbb^\sharp \bigr)  
 && \text{weakly-$*$ in measure in } \Interval \times \Tbb^3. 
\endaligned
\ee 

\eei 
\end{theorem}    


\subsection{Compactness framework for well-prepared data sets}
\label{section=5-3} 

\subsubsection{Initially well-prepared solutions}

Our main compactness result concerns sequences of solutions with a strong convergence assumption at initial time on the essential geometric variables.
The notion of initial data set was already introduced in \autoref{section=4-6} for exact solutions and will be extended straightforwardly in \autoref{section=5-4} for sequences of approximate solutions.  We content ourselves here with only the initial data $(\mathring\Pbb^\suit,\mathring\Qbb^\suit)$ for the essential geometric variables.

\begin{definition}
\label{def-well-p}
Consider a bounded sequence of approximate Einstein--Euler flows $(\Phi^\suit,\Psi^\suit,\Pi^\suit)$.
The sequence is said to be \textbf{initially well-prepared} if the essential geometry variables converge strongly to some limit on the initial hypersurface, namely
\be
\mathring\Pbb^\suit,\mathring\Qbb^\suit \to \mathring\Pbb^\sharp,\mathring\Qbb^\sharp
\qquad
\text{strongly in } L^2(\Tbb^3).
\ee
\end{definition}

Equipped with this notion we immediately deduce from \autoref{theorem-493} the following stability result.

\begin{corollary}[Nonlinear stability of initially well-prepared Einstein spacetimes with corrector]
\label{theorem-494}
Under the notation of \autoref{theorem-493}, if the initial data set is well-prepared in the sense of \autoref{def-well-p}, then $\Pi^\suit\to\Pi^\sharp$ weakly star in measure in $\Interval\times\Tbb^3$.
\end{corollary}


By setting the corrector to vanish ($\Pi^\suit=\Pi^\sharp=0$), we obtain one of the main conclusions of this paper, which concerns sequences of Einstein--Euler flows \emph{without corrector.}

\begin{theorem}[Nonlinear stability of initially well-prepared Einstein spacetimes]
\label{theorem-492}
Any initially well-prepared tame sequence of approximate Einstein--Euler flows $\Interval \ni t \mapsto \bigl(\Phi^\suit(t),\Psi^\suit(t)\bigr)$ without corrector
converges, after extraction of a subsequence if necessary, to a limiting flow
$\Interval \ni t \mapsto \bigl(\Phi^\sharp(t),\Psi^\sharp(t)\bigr)$
defined on the slab \(\Interval\times\Tbb^3\), with the following properties.

\bei

\item The limit flow \(\bigl(\Phi^\sharp,\Psi^\sharp\bigr)\) is a tame Einstein--Euler flow.

\item The sequence converges in the natural norms:
\bel{eq:conv-prop}
\aligned
\Phi^\suit &\to \Phi^\sharp
 && \text{strongly in } L^2(\Tbb^3,\dVtrois_\sharp)
\text{ for almost all times } t\in\Interval,
\\
\Psi^\suit &\to \Psi^\sharp
&& \text{almost everywhere on } \Interval\times\Tbb^3.
\endaligned
\ee
Moreover, $(\Phi^\sharp, \Psi^\sharp)$ are subject to the same bounds as the approximate solutions $(\Phi^\suit,\Psi^\suit)$.
\eei
\end{theorem}


\subsection{Compactness for flows with uniformly bounded initial data}
\label{section=5-4}

\subsubsection{Approximate initial data sets}

We extend the notions of initial data sets and of flows assuming a prescribed initial data set (\autoref{def-initialdata}) to sequences of approximate solutions.  This simply requires allowing error terms that vanish as $\suit\to 0$ in the Einstein constraint equations.

\begin{definition}\label{def-initdataset-approx}
A \textbf{tame sequence of approximate initial data sets} for an Einstein--Euler flow with corrector is a sequence of data $(\mathring\Phi^\suit,\mathring\Psi^\suit,\mathring\Pi^\suit)$ satisfying the following conditions.
For each $\suit\in(0,1]$, the integrability and regularity conditions~\eqref{initdataset-reg}, the positivity and symmetric tracelessness conditions~\eqref{initdataset-pos} and the periodicity constraints~\eqref{initdataset-period} hold.
The Einstein constraints with corrector hold, up to negligible contributions (in the $\suit\to 0$ limit),
\be
\opL_{13}(\mathring\Phi^\suit,\mathring\Psi^\suit) + \frac{t_0}{2} \mathring\Pi^\suit{}^{01} \to 0 , \quad
\opL_{14}(\mathring\Phi^\suit,\mathring\Psi^\suit) \to 0 , \quad
\opL_{15}(\mathring\Phi^\suit,\mathring\Psi^\suit) \to 0 \quad
\text{distributionally.}
\ee
The metric-weighted parallel momentum $\mathring\Gammaunpar{}^\suit$ has uniformly bounded total variation and (after extraction of a subsequence if necessary) converges to a limit $\mathring\Gammaunpar{}^\sharp$ at every point in $\Tbb^3$ and in total variation.
\end{definition}

The pointwise convergence of $\mathring\Gammaunpar{}^\suit$ to a limit is a consequence of the uniform bound on its total variation, by Helly's theorem, whereas the convergence of the total variation is a non-trivial assumption on the initial oscillations of fluid variables.  The pointwise and total variation convergence are explicited in~\eqref{mathring-Gammaunpar-convergence} below for clarity.

By extension of \autoref{def-initialdata}, a tame sequence of approximate Einstein--Euler flows is said to \textbf{tamely assume a prescribed tame sequence} of approximate initial data sets provided the Einstein--Euler evolution equations listed in~\eqref{approx-Einstein-Euler} (with correctors and negligible terms) the reference entropy inequality~\eqref{equa-lfk3-approx} (with a negative measure and negligible terms), the $H$-divergence law~\eqref{equa-bvpara-bis} (with negligible terms), and the quasi-balance law~\eqref{equa-balancelaws-plus} (with measure and $H^{-1}$-compact terms)
hold in the integral sense with boundary terms at time~$t_0$ determined by the data $(\mathring\Phi^\suit,\mathring\Psi^\suit,\mathring\Pi^\suit)$, and for all $\suit\in(0,1]$ one has the bounds~\eqref{equa-DK399} and~\eqref{Fmeasure-apriori-def} on the metric-weighted parallel momentum~$\Gammaunpar{}^\suit$ and the shear-induced measure $\Fmeasure{}^\suit$ in terms of the initial data, up to error terms that tend to zero as $\suit\to 0$.

The instability result in \autoref{theorem-493} (and thus its consequences, \autoref{theorem-494} and \autoref{theorem-492}) readily accomodates initial data sets.

\begin{corollary}[Nonlinear stability with respect to initial data sets]
\label{corollary-493-init}
Under the notation of \autoref{theorem-493}, if the sequence $(\Phi^\suit,\Psi^\suit,\Pi^\suit)$ tamely assumes a prescribed tame sequence of approximate initial data sets $(\mathring\Phi^\suit,\mathring\Psi^\suit,\mathring\Pi^\suit)$ (cf.\ \autoref{def-initdataset-approx}), then there exists an (exact) tame initial data set $(\mathring\Phi^\sharp,\mathring\Psi^\sharp,\mathring\Pi^\sharp)$ (cf.\ \autoref{def-initdataset}) such that the following properties hold, after extraction of a subsequence if necessary.
\bei
\item The initial data sets converge to $(\mathring\Phi^\sharp,\mathring\Psi^\sharp,\mathring\Pi^\sharp)$ up to additional correctors, in the sense that
\bse
\begin{gather}
\aligned
\mathring\Jbb^\suit & \to \mathring\Jbb^\sharp
\qquad 
&& \text{ strongly in } L^2(\Tbb^3),
\\
(\mathring\Pbb^\suit,\mathring\Qbb^\suit) & \rightharpoonup (\mathring\Pbb^\sharp,\mathring\Qbb^\sharp)
\qquad 
&& \text{ weakly in } L^2(\Tbb^3), 
\\ 
\mathring\Psi^\suit & \to \mathring\Psi^\sharp && \text{ at almost every point in } \Tbb^3 ,
\endaligned
\\
\aligned
& \hskip-1cm \mathring\Pi^\suit + \mathring\amdeux^\suit (\mathring\Omega^\suit)^2 \bigl( \mathring\Pbb^\suit \otimes \mathring\Pbb^\suit + \mathring\Qbb^\suit \otimes \mathring\Qbb^\suit \bigr)  
\\
& \to  \mathring\Pi^\sharp + \mathring\amdeux^\sharp (\mathring\Omega^\sharp)^2 \bigl( \mathring\Pbb^\sharp \otimes \mathring\Pbb^\sharp + \mathring\Qbb^\sharp \otimes \mathring\Qbb^\sharp \bigr)  
 && \text{weakly-$*$ in measure in } \Tbb^3. 
\endaligned
\end{gather}
\ese

\item The limiting flow $(\Phi^\sharp,\Psi^\sharp,\Pi^\sharp)$ tamely assumes the initial data set $(\mathring\Phi^\sharp,\mathring\Psi^\sharp,\mathring\Pi^\sharp)$ in the sense of \autoref{def-initialdata}.
\eei
If the initial data is well-prepared, namely $(\mathring\Pbb^\suit,\mathring\Qbb^\suit)$ converges strongly, then $\mathring\Pi^\suit\to \mathring\Pi^\sharp$ weakly-$*$ in measure in~$\Tbb^3$.  If in addition $\mathring\Pi^\suit=0$, then $\Pi^\suit=0$ and the limiting flow has no corrector: $\mathring\Pi^\sharp=0$ and $\Pi^\sharp=0$.
\end{corollary}

\subsubsection{Sequences of exact solutions}

For (exact) tame Einstein--Euler flows with correctors, we derive a priori estimates in Sections~\ref{section=6} and~\ref{section=7} that control all norms in terms of the initial data.  In this way, uniform bounds on initial data sets are enough to ensure that a sequence of (exact) tame Einstein--Euler flows with correctors obeys the uniform bounds (in time and in~$\suit$) required to apply \autoref{theorem-493}.

\begin{theorem}[Nonlinear stability under initial bounds]\label{thm:exactsol-compactness}
  Consider a sequence $(\Phi^\suit,\Psi^\suit,\Pi^\suit)$ of (exact) tame Einstein--Euler flows with correctors on $(t_0,t_1^\suit)\times\Tbb^3$ with the same initial time~$t_0$, such that, for each $\suit\in(0,1]$, the flow tamely assumes a tame initial data set $(\mathring\Phi^\suit,\mathring\Psi^\suit,\mathring\Pi^\suit)$ on~$\Tbb^3$ in the sense of \autoref{def-initialdata}.
  Assume the initial data is uniformly bounded in the limit $\suit\to 0$ in the following norms
\bse\label{exactsol-initbounds-all}
\bel{exactsol-initbounds-1}
\|\mathring\Phi^\suit\|_{L^2(\Tbb^3,\dVtrois_\suit)} , \quad
\Var(\mathring\Psi^\suit;\Tbb^3) , \quad
\|\mathring\Psi^\suit\|_{L^\infty(\Tbb^3)} , \quad
\|\mathring\Pi^\suit\|_{\Meas(\Tbb^3)} ,
\ee
the initial conformal length density $\mathring\amdeux^\suit$ is a Cauchy sequence, namely
\bel{exactsol-initbounds-2}
\lim_{\alpha \to 0} \sup_{\suit,\suit'\leq \alpha}
\|\mathring\amdeux^\suit-\mathring\amdeux^{\suit'}\|_{L^1(\Tbb^3)} = 0 ,
\ee
and the metric-weighted parallel momentum $\mathring\Gammaunpar{}^\suit$ converges pointwise and in total variation norm as in \autoref{def-initdataset-approx},
\bel{mathring-Gammaunpar-convergence}
\sup_{\Tbb^3} \Bigl| \mathring\Gammaunpar{}^\suit - \mathring\Gammaunpar{}^\sharp\Bigr| \to 0 , \qquad
\lim_{\suit\to 0} \Var(\mathring\Gammaunpar{}^\suit;\Tbb^3) = \Var(\mathring\Gammaunpar{}^\sharp;\Tbb^3) .
\ee
\ese
Assume that the time of existence are uniformly bounded below, i.e.
\bel{exactsol-inftstar}
t_0 < \inf_{\suit \in (0,1]} t_*^\suit \leq 0 .
\ee
For future-expanding data sets, assume further that the initial conformal length is uniformly bounded below, i.e.
\bse
\bel{exactsol-Linit}
\inf_{\suit \in (0,1]} \mathringLbfscri^\eps > 0 ,
\quad \text{where \ } \mathringLbfscri^\eps \coloneqq \int_{\Sbb^1}d_x\mathring\ell^\suit .
\ee
For future-contracting data sets, assume further that, at a uniform time $t_1\in(t_0,\inf t_*^\suit)$ the sequence of conformal lengths at time $t_1$ is uniformly bounded below, i.e.
\be
\inf_{\suit \in (0,1]} \Lbfscri^\eps(t_1) >0. 
\ee
\ese
Then $(\Phi^\suit,\Psi^\suit,\Pi^\suit)$, restricted to $(t_0,t_1)$ ---with arbitrary $t_1\in(t_0,\inf t_*^\suit)$ in the future-expanding case---, is a \textbf{tame sequence} of (exact) Einstein--Euler flows with correctors, in the sense of~\autoref{def-810}, and the conclusions of \autoref{theorem-493} and \autoref{corollary-493-init} apply.
If the sequence of initial data sets is well-prepared then the stability results in \autoref{theorem-494} and \autoref{theorem-492} (with or without corrector, respectively) apply.
\end{theorem}

The compactness result stated in \autoref{thm:stability0} is deduced from \autoref{thm:1.1} (existence result) and \autoref{thm:exactsol-compactness} as follows.
We are provided a sequence of tame initial data sets $(\mathring\Phi^\suit,\mathring\Psi^\suit)$ subject to the uniform bounds~\eqref{exactsol-initbounds-all} and an assumption on the conformal length.  For each $\suit\in(0,1]$, \autoref{thm:1.1} provides a tame Einstein--Euler flow $(\Phi^\suit,\Psi^\suit)$ that tamely assumes this initial data set.
In the future-expanding case, the existence result ensures that the time of existence is infinite so that~\eqref{exactsol-inftstar} holds, and the lower bound~\eqref{exactsol-Linit} on initial conformal lengths is among the hypotheses of \autoref{thm:stability0}.
In the future-contracting case, the existence result ensures that either $t_*^\suit=0$ or $\Lbfscri(t_*^\suit)=0$, and \autoref{tstar-lower-bound-sec6} ensures a lower bound on the existence time.  In either case \autoref{thm:exactsol-compactness} applies and the convergence properties and uniform bounds hold in $(t_0,t_1)$ as desired.


\subsection{Conjecture on the shear-induced measures}
\label{section=5-5}

An issue left open by the present work is whether the abstract shear-induced measures
\bel{equa-form-weak}
\Fmeasure_m = |t|^{-1} \Omega^{-1} \Fvee_m (\Jbb) e_1(\widehat J_m)
\ee
arising for regular solutions in the quasi-balance laws
\bel{equa-balancelaws-bis}
\aligned
& \divuntrois \Bigl( |t|^{-1} \Omega^{-1} \, \Fbb(\Jbb) \Bigr)
- \Delta_\Fbb(\Jbb, \Kpar, \Pbb,\Qbb;\Omega,t)
\\
& = \sum_{m=2,3} \Fmeasure_m
\qquad \text{(regular solutions)},
\endaligned
\ee
admit an intrinsic interpretation at the level of \emph{weak} solutions.

Within the compactness theory developed in the present paper, the role of the
measure~$\Fmeasure_m$ is limited to that of closure objects in the
quasi-balance laws. In our companion work
\cite{LeFlochLeFloch-next,LeFlochLeFloch-more}, for regularized Einstein--Euler flows, we derive uniform measure bounds on these shear-induced terms that ensure that they are bounded Radon measures in spacetime. While this information is enough for the compactness argument used in the present stability and instability results, it does not explain how these measures are selected, nor what part of the small-scale structure of the approximation they retain in the limit.

The difficulty is analogous (but more challenging) to the one treated by Dal~Maso, LeFloch, and
Murat~\cite{DLM} in their theory of nonconservative products between $\BV$~functions and Radon
measures. Presently, in the formal expression~\eqref{equa-form-weak} the momentum~$\Jbb$, and hence
the coefficient $\Fvee_m(\Jbb)$, is defined only almost everywhere with respect to Lebesgue measure,
whereas $e_1(\widehat J_m)$ is a measure that needs not be absolutely continuous with respect to the Lebesgue measure.
Therefore, at the level of weak regularity considered here, the product
\be
\Fvee_m(\Jbb)\, e_1(\widehat J_m)
\ee
has no canonical meaning \emph{a priori}. Replacing it by an abstract measure~$\Fmeasure_m$ resolves the compactness issue, but suppresses the information
that may be encoded in the singular part of the approximation.

To formulate the question more precisely, let us detail the measure
associated with the derivative of~$\widehat J_m$. Since $\widehat J_m\in L^\infty(\Interval;\BV(\Tbb^3))$ and is $\Tbb^2$-invariant,
its spatial derivative $d_x\widehat J_m(t)$ is, for a.e.\ $t\in\Interval$, a
signed Radon measure on~$\Sbb^1$. We then define a signed Radon measure
$\kappa_m\in\Meas(\Interval\times\Tbb^3)$ by duality:
\bse
\label{kappam-defs}
\be
\int_{\Interval\times\Tbb^3}\varphi\,d\kappa_m
\coloneqq
\int_\Interval\!\int_{\Sbb^1}\int_{\Tbb^2}\varphi(t,x,y,z)\,dy\,dz\,d_x\widehat J_m(t)(x)\,dt
\ee
for every test function $\varphi\in\Ccal_c(\Interval\times\Tbb^3)$.
Formally, one has
\be
\kappa_m
= \partial_x(\widehat J_m)\,dt\,dx\,dy\,dz
= |t|^{-1}\Omega^{-1} e_1(\widehat J_m)\,\dVuntrois.
\ee 
\ese
In general, $\kappa_m$~need not be absolutely continuous with respect to
Lebesgue measure or with respect to~$\dVuntrois$.

The natural conjecture is that the singular part of~$\Fmeasure_m$ should be
selected by the regularization mechanism used to construct the weak solution;
see \cite{LeFlochLeFloch-PhilTrans,LeFloch-IMA-1989,LeFloch-book}.
We may expect the following qualitative picture. On the part of~$\kappa_m$ that
is absolutely continuous with respect to Lebesgue measure, the
measure~$\Fmeasure_m$ should coincide with the shear-induced factor
$\Fvee_m(\Jbb)$ multiplied by the corresponding diffuse derivative. On the
singular part of~$\kappa_m$, by contrast, the measure~$\Fmeasure_m$ should be
supported on the geometric support of the singular derivative, and in a
$\BV$~regime in particular on shock curves, and should involve an effective
value of $\Fvee_m(\Jbb)$ selected by the internal structure of the corresponding
shock layer. This is closely related in spirit to the nonconservative product
theory developed in~\cite{LeFloch-IMA-1989,DLM}. In this sense, the singular
part of~$\Fmeasure_m$ should be viewed as a macroscopic trace of microscopic
transition profiles. Distinct approximation
procedures --- for instance vanishing viscosity, finite volume schemes, or
front-tracking algorithms --- may lead to different internal shock profiles and
therefore, possibly (but not necessarily, since the characteristic fields associated with the parallel momentum are linearly degenerate, which may simplify the structure of these measures), to different singular contributions to~$\Fmeasure_m$.
Thus, in our view, the shear-induced measures are not merely auxiliary objects recording the loss of strong compactness.  Instead, they should be regarded
as important weak-regularity objects, carrying information that is invisible at
the level of the limiting tame flow itself, but is nevertheless selected by the
small-scale dynamics of the approximation.


\
 
\part{The finite energy method for $\Tbb^2$-symmetric Einstein spacetimes}
\label{part-two}

\section{A priori estimates for Einstein--Euler flows}
\label{section=6}

\subsection{Volume estimates}
\label{section=6-1}

\subsubsection{A priori estimates for weak solutions} 

We begin our investigations by deriving a set of a priori estimates satisfied by Einstein--Euler flows.
Our starting point is the \JKL\ formulation of the Einstein--Euler system,
namely \eqref{eq:T2-1234-def}--\eqref{eq:theconstraints-def}. More precisely, we work with an arbitrary Einstein--Euler flow \emph{with corrector} $(\Phi,\Psi,\Pi)$ in the sense of \autoref{weakdefinitionT2-deux2} that assumes a prescribed initial data set $(\mathring\Phi,\mathring\Psi,\mathring\Pi)$ at time $t=t_0$ in the sense of \autoref{def-initialdata}, and we derive estimates that apply at this level of weak regularity.

The derivation proceeds by testing the integral identities and inequalities
introduced in the previous section against suitable functions. We explain this
procedure once here, and thereafter omit these routine details unless an
additional argument is needed. In particular, standard arguments for regular
solutions based on integrating balance laws or balance inequalities over
spacetime domains remain valid in our class of weak solutions as well.

For $\Phi\Psi$ flows without corrector, the function $\log\Omega$ belongs
to $\BVac$ along each line of constant time, and likewise along each
line of constant space. Hence, we may multiply by powers of~$\Omega$ and
justify the resulting calculations in the sense of distributions. For
instance, in the derivation of the identity~\eqref{equa-Omega-n-m} below, this is
done by choosing a suitable regularization of the lapse factor as a test
function. In such derivations, it is convenient to restrict attention first
to test functions of product form~$\theta(t)\varphi(x)$. We recall that
linear combinations of such product test functions are dense in the class of continuous test functions~$\psi(t,x)$ (and likewise for $\BVac$ test functions) as explained in~\autoref{appendix=E}.

More generally, for $\Phi\Psi\Pi$ flows, we rely on the Volpert calculus explained in \autoref{section=4-5}, which defines products of powers of~$\Omega$ with the corrector tensor $\Pi^{ab}$ as bounded measures on almost every constant-$t$ or constant-$x$ slice.

We derive estimates a succession of lemmas and finally summarize our conclusion in \autoref{theo--68}, below. 


\subsubsection{Volume of spacelike hypersurfaces}

We explore the properties of the functional  
\bel{equa-333}
\Vcal(t) \coloneqq \int_{\Tbb^3} \dVtrois,  \qquad t\in\Interval
\ee
and analogous functionals with other powers of $|t|$ and~$\Omega$, where we recall $\dVtrois=|t|\,\Omega\,d_x\ell\,dy\,dz$ and $d_x\ell=\amdeux\,dx$.
Using the evolution equations~\eqref{eq:T2-8-def} and~\eqref{eq:evollambda-def}, we can write
\bel{equa-Omega-n-m}
\aligned
\bigl( |t|^n \Omega^m d_x\ell\,dy\,dz \bigr)_t 
&= \Bigl( \frac{n-m/4}{t}
+ {t \over 2} (\Mbf^{00}- \Mbf^{11}) \, \Omega^2
+ m {t \over 2}\Mbf^{11}  \Omega^2 \Bigr) |t|^n \Omega^m\,d_x\ell\,dy\,dz
\\
&\quad
+ \frac{t}{2} m |t|^n \Omega^{m} \Pi^{00} ,
\endaligned
\ee
in which it is natural to impose $n= m/4$. Thanks to the inequalities $\Mbf^{00}\pm\Mbf^{11}\geq 0$ in~\eqref{eq:T2-Mdef-1}, this is non-negative for $m\in[0,2]$.
We choose first $m=1$, so that this simplifies to give the identity 
\be
\bigl( |t|^{1/4} \Omega\, d_x\ell\,dy\,dz \bigr)_t
= {\sgn(t)\over 2} |t|^{5/4} \Mbf^{00} \Omega^3\, d_x\ell\,dy\,dz + {\sgn(t)\over 2} |t|^{5/4} \Omega\,\Pi^{00} .
\ee 
We thus find
\bel{equat--dVdt-2}
\sgn(t) \frac{d}{dt} \bigl( |t|^{-3/4}  \, \Vcal(t) \bigr) 
= {1\over 2} |t|^{1/4} \int_{\Tbb^3} \Mbf^{00} \Omega^2 \dVtrois
+  {1\over 2} |t|^{5/4} \bigl\la \Omega \,\Pi^{00} , 1 \bigr\ra_{\Tbb^3}
\geq 0, 
\quad t \in \Interval ,
\ee
where $\la\cdot,\cdot\ra_{\Tbb^3}$ in the last term denotes the duality bracket between Radon measures and continuous functions on~$\Tbb^3$, applied here to the constant function~$1$ and to the measure $\Omega\,\Pi^{00}$ defined by the Volpert product~\eqref{eq:Volpert-Omega-Pi}.

Observe in passing that $\del_t=\Omega \, e_0$ implies $\Omega^2 \Mbf^{00}=M(\del_t,\del_t)$. Therefore, the weight $\Omega^2$ is naturally motivated by our choice of the time function $t$ and the vector field~$\del_t$.
Without yet having control over the right-hand side of~\eqref{equat--dVdt-2}, at this stage we can only establish the following (partial) property.

\begin{lemma}[Monotonicity of the volume]
\label{lemma-31}
In the areal foliation, $|t|$ being the area of $\Tbb^2$~orbits, the functional $\Vcal(t)$ in~\eqref{equa-333}
is such that $t\mapsto|t|^{-3/4}\Vcal(t)$ is monotone, and is bounded as follows in terms of the initial volume $\mathring\Vcal \coloneqq \int_{\Tbb^3} \dVtrois\mathring{\;} = |t_0| \int_{\Tbb^3} \mathring\Omega\,d_x\mathring\ell\,dy\,dz$:
\be
\aligned
\Vcal(t) & \leq |t|^{3/4} \,  |t_0|^{-3/4}  \, \mathring\Vcal, \qquad & t_0 & \leq t \leq t_1 < 0 & & \text{ (future-contracting regime)}, \\
\Vcal(t) & \geq |t|^{3/4} \,  |t_0|^{-3/4}  \, \mathring\Vcal, \qquad & 0 & < t_0 \leq t \leq t_1 & & \text{ (future-expanding regime)} .
\endaligned
\ee
That is, in the \emph{future-contracting regime}, the volume decreases at least at the rate $|t|^{3/4}$ ---necessarily approaching $0$ if $t_1$ is approaching $0$.
In the \emph{future-expanding regime}, the volume increases at least at the rate $|t|^{3/4}$ ---necessarily approaching $+ \infty$ if $t_1$ is approaching $+ \infty$.
\end{lemma}


\subsubsection{Dealing with weak solutions} 
 
As explained at the beginning of this section, all of our arguments apply to weak solutions. For instance, in our combination of the evolution equations~\eqref{eq:T2-8-def} and~\eqref{eq:evollambda-def}, let us choose a test function of the form $\theta(t)\varphi(x)$, with $\varphi(x)=1$ constant and $\theta$ supported in the half-closed interval $[t_0,t_1)$.  Taking into account boundary terms as in~\eqref{eq-div-weak0-bound}, we then obtain the expression of~\eqref{equat--dVdt-2} in the weak sense:
\be
\aligned
& 
- \sgn(t_0) \Bigl( |t_0|^{-3/4}  \, \mathring\Vcal\, \theta(t_0) + \int_\Interval  |t|^{-3/4}  \, \Vcal \theta_t \, dt \Bigr)
\\
& = {1\over 2} \int_\Interval |t|^{1/4} \theta \int_{\Tbb^3} \Mbf^{00} \Omega^2 \, \dVtrois \, dt
+ {1\over 2} \int_\Interval |t|^{5/4} \theta\, \bigl\la\Omega\,\Pi^{00}, 1\bigr\ra_{\Tbb^3} \, dt
\geq 0. 
\endaligned
\ee 
By letting $\theta$ approach the characteristic function of a subinterval $[t_0, t_2]$, with $t_2 \in (t_0, t_1)$, and therefore $\theta_t$ approach a Dirac point mass $-\delta_{t=t_2}$ at time $t_2$, we deduce that 
\bel{eqVmonotonic-616}
\aligned
& \sgn(t_0) \Bigl( |t_2|^{-3/4}  \, \Vcal(t_2) - |t_0|^{-3/4}  \, \mathring\Vcal \Bigr)
\\
&
= {1\over 2} \int_{(t_0, t_2)} |t|^{1/4} \Bigl( \int_{\Tbb^3} \Mbf^{00} \Omega^2 \, \dVtrois \, dt + |t| \bigl\la\Omega\,\Pi^{00}, 1\bigr\ra_{\Tbb^3} \Bigr) \, dt
\geq 0.
\endaligned
\ee 
This provides a proof of~\autoref{lemma-31} in the context of weak solutions.  By taking the difference of this identity for two different times, \eqref{eqVmonotonic-616} holds with $t_0$ replaced by any time in $[t_0,t_2)$, which shows that the map defined by
\bel{scaled-V}
t \in [t_0,t_1) \mapsto |t|^{-3/4}\Vcal(t)
\ee
(with the understanding that $\Vcal(t_0)=\mathring\Vcal$)
is monotone, and in particular has bounded variation.
The regularity assumed in \autoref{weakdefinitionT2-relaxed} ($\Phi\Psi\Pi$ flow with finite energy) ensures that the integrand of the time integral in~\eqref{eqVmonotonic-616} is bounded uniformly for almost every time $t\in\Interval$.
Momentarily, we find a uniform bound for \emph{all} times by the initial data, so that the map $\Vcal$ has Lipschitz regularity on $[t_0,t_1)$.


\subsubsection{Length functional}

Let us next introduce 
\bel{equa-conf-vol}
\Lbfscri(t) \coloneqq \int_{\Tbb^3}  |t|^{-1}  \Omega^{-1} \, \dVtrois  
= \int_{\Tbb^3} d_x\ell dydz
\qquad t\in[t_0,t_1), 
\ee
with the understanding that $\Lbfscri(t_0)=\mathringLbfscri=\int_{\Tbb^3}d_x\mathring\ell\,dy\,dz$.
It is monotone in time since, by~\eqref{equa-Omega-n-m} with $m=0$,
\bel{Wmonotonic}
\aligned
\sgn(t) \frac{d}{dt} \Lbfscri(t)
&= \sgn(t) \int_{\Tbb^3} (d_x\ell)_t \, dydz
\\
&= \int_{\Tbb^3}  \Bigl(\Ebf_0(\Jpar,\Kpar)+\frac{1}{2}q_\Jbb |\Jbb\cdot\Jbb|\Bigr) \, \Omega \, \dVtrois
\geq 0. 
\endaligned
\ee
Namely, $\Lbfscri$ is non-decreasing for future-expanding spacetimes and non-increasing for future-contracting ones. Moreover, we will see shortly that $\Lbfscri(t)/|t|$ and $d\Lbfscri(t)/dt$ are bounded above by the energy functional~\eqref{equa--energy}, below. Our estimates in this section will involve the factor $1/\Lbfscri(t)$, which we bound in terms of the minimum of the conformal length functional, denoted by $\Lmin$. We thus state the following property. 

\begin{lemma}[Monotonicity of the  conformal length functional] 
\label{lem-conf-vol}
In the areal foliation, the functional $\Lbfscri(t)$ in~\eqref{equa-conf-vol} is monotone and, therefore, tends to its infimum at the endpoints of the time interval $\Ihalf$ with, specifically, 
\bel{minimal-conformal-volume}
\Lmin \coloneqq \inf_{t\in [t_0,t_1)} \Lbfscri(t) =
\begin{cases}
  \lim_{t\to t_1} \Lbfscri(t), & t_0< t_1<0 \quad \text{(future-contracting regime),}\\
  \mathringLbfscri, & 0 < t_0< t_1 \quad \text{(future-expanding regime).}
\end{cases}
\ee 
\end{lemma}

In the future-expanding regime, the inequality $\Lbfscri(t) \geq \mathringLbfscri$ provides us with a uniform lower control of the  conformal length. On the other hand, this property does not hold in the future-contracting regime. 


\subsubsection{Future-contracting spacetimes}

Concerning the time interval $t_0 < t < t_1 <0$, the following observations are in order. The volume~$\Lbfscri(t)$ is decreasing, and our bounds in the present section are meaningful only until the first time $t_*<0$ at which it reaches zero, that is, $\Lbfscri(t_*)=0$; more precisely, they apply on any interval $[t_0,t_1)$ with $t_0<t_1<t_*$. Observe that the condition $\Lbfscri(t_*)=0$ also implies that the volume of the spatial slice vanishes at that time~$t_*$. Indeed, by the Cauchy--Schwarz inequality we have 
\be
\Vcal(t)^2 \leq \Lbfscri(t) \int_{\Tbb^3}  |t|  \, \Omega\, \dVtrois
\quad \text{ as } t \to t_*,
\ee
and the latter integral is easily seen to be decreasing. Indeed, by considering yet a third functional
\bel{equa--UUU}
\Ucal(t) \coloneqq \int_{\Tbb^3} |t| \Omega \, \dVtrois  
= \int_{\Tbb^3}  |t|^2 \Omega^2 d_x\ell dydz
\qquad t\in[t_0,t_1], 
\ee
and applying~\eqref{equa-Omega-n-m} with $m=2$, we find 
\be
\aligned
\sgn(t) \frac{d}{dt} \Ucal(t)  
= \int_{\Tbb^3} \Bigl(\frac{3}{2} + t^2 \, \Omega^2 \Ebf_0(\Jperp, \Pbb, \Qbb) \Bigr) \Omega \dVtrois |t|^3 \bigl\la \Omega^2\,\Pi^{00} , 1 \bigr\ra_{\Tbb^3}
\geq 0 .
\endaligned
\ee
Therefore, a vanishing  conformal length $\Lbfscri(t_*)=0$ implies a geometric singularity at which $\Vcal(t_*)=0$. Furthermore, conversely, we shall establish that solutions do exist until a time\footnote{In the future-expanding regime, we will establish existence for all $t$, therefore $t_*=+ \infty$.} where either $\Lbfscri(t_*)=0$ or the area $|t_*|$ of the symmetry orbits reaches zero.
In both cases the volume reaches zero.
For convenience, we state the two bounds on the volume as follows, repeating the bound in \autoref{lemma-31}.

\begin{lemma}[Volume bound for future-contracting spacetimes]
In the future-contracting regime $t_0 < t < t_1 < 0$, one has
\be
\Vcal(t) \lesssim \min\bigl( |t|^{3/4} , \Lbfscri(t) \bigr)
\ee
with an implicit constant that depends on~$t_0$ and the initial data $\mathring\Vcal$ and $\mathring\Ucal = \int \mathring\Omega\,\dVtrois\mathring{\;}$.
\end{lemma}


\subsection{Energy estimates}
\label{section=6-2}

\subsubsection{Spacelike slices}

Taking the time derivative of $\Vcal(t)$ (and its weighted variants), as we did above, is natural from the point of view of the areal foliation based on the~$\Tbb^2$ symmetry. Another natural functional, which does not depend as strongly on the areal foliation, is the variation of the volume with respect to shifting the hypersurface by a constant proper time normal to it. Concretely, taking a Lie derivative in the direction of the unit vector $e_0=\Omega^{-1}\del_t$, we introduce 
\bel{equa--energy}
\aligned
\Ebfscri(t) 
& \coloneqq \frac{2}{t}\int_{\Tbb^3} \Lie_{e_0}\bigl(\dVtrois\bigr)
= \int_{\Tbb^3} \biggl( \Mbf^{00} \Omega^2 + \frac{3}{2t^2}\biggr) \Omega^{-1} \, \dVtrois 
+ |t| \la \Pi^{00}, 1\ra_{\Tbb^3}
\\
& = \int_{\Tbb^3} \Mbf^{00} \Omega \, \dVtrois + {3 \over 2 |t|} \Lbfscri(t) + |t| \la \Pi^{00}, 1\ra_{\Tbb^3}
\geq 0, \quad t\in\Iopen,
\endaligned
\ee
which we refer to as the \emph{energy functional}.
Here, the normalization factor $2/|t|$ is chosen to simplify the expressions arising in the forthcoming calculations (such as~in~\eqref{eq:monotenergy}). 
This energy $\Ebfscri(t) $ is bounded for all times in terms of the initial energy as we now show.

The energy splits usefully into a corrector part $\Ebfscri_\Pi$ and a remaining part $\Ebfscri_{\Phi\Psi}$, namely $\Ebfscri = \Ebfscri_\Pi + \Ebfscri_{\Phi\Psi}$ with
\bel{equa--energy-split}
\aligned
\Ebfscri_\Pi = |t| \la \Pi^{00}, 1\ra_{\Tbb^3} ,
\qquad
\Ebfscri_{\Phi\Psi} = \int_{\Tbb^3} \biggl( \Mbf^{00} \Omega^2 + \frac{3}{2t^2}\biggr) \Omega^{-1} \, \dVtrois .
\endaligned
\ee
The corresponding \emph{initial energy} splits likewise as $\mathring\Ebfscri = \mathring\Ebfscri_\Pi + \mathring\Ebfscri_{\Phi\Psi}$ with
\be
\mathring\Ebfscri_\Pi = |t_0| \la \mathring\Pi^{00}, 1\ra_{\Tbb^3} ,
\qquad
\mathring\Ebfscri_{\Phi\Psi} = \int_{\Tbb^3} \biggl( \Mbf^{00}(\mathring\Jbb,\mathring\Kpar,\mathring\Pbb) \mathring\Omega^2 + \frac{3}{2t_0^2}\biggr) \mathring\Omega^{-1} \, \dVtrois\mathring{\;} .
\ee
Since the energy density $(\Jbb, \Kpar, \Pbb, \Qbb) \mapsto  \Mbf^{00}(\Jbb, \Kpar, \Pbb, \Qbb)$ is a non-negative and quadratic function, we obtain an $L^2$-type control of $\Omega^{1/2}(\Jbb, \Kpar, \Pbb, \Qbb)$ on spacelike slices $(\Tbb^3, \dVtrois)$. 


\subsubsection{Spacelike measure bound}

Derivatives of the volume and conformal length functionals in \autoref{section=6-1} involve both $L^2$ norms of geometry and fluid variables, and integrals of the corrector~$\Pi$.
The latter are controlled by $\Ebfscri_\Pi(t)$, which is the measure norm of $t\,\Pi^{00}(t)$.
Here we derive a priori bounds on this norm, based on the divergence inequalities~\eqref{eq:Twave}
\be
\divuntrois\Bigl(\Omega^{-1}\amdeux^{-1}\Pi^{\pm\bullet}\Bigr) \leq 0 ,
\ee
understood in the integral sense with a boundary term included at $t=t_0$.
The same inequality holds for $\Pi^{0\bullet} = \frac{1}{2} (\Pi^{+\bullet} + \Pi^{-\bullet})$.
Namely, for a non-negative test function $\varphi$ with compact support in $[t_0,t_1)\times\Tbb^3$,
\bel{weak-divePi-bdry}
|t_0| \bigl\la \mathring\Pi^{a0},\varphi(t_0,\cdot)\bigr\ra_{\Tbb^3}
+ \int_{\Sbb^1} \bigl\la \Pi^{a1} , \, \amdeux^{-1}\,|t|\,\varphi_x\bigr\ra_{\Interval\times\Tbb^2}\,dx
+ \int_{\Interval} \bigl\la \Pi^{a0} , \, |t|\,\varphi_t\bigr\ra_{\Tbb^3}\,dt \leq 0 ,
\qquad a=0,+,-.
\ee
By integrating the divergence equation of $\Pi^{0\bullet}$ with a test function~$\varphi$ that only depends on time and has compact support in $[t_0,t_1)$, we obtain
\be
\varphi(t_0) \mathring\Ebfscri_\Pi + \int_{\Interval} \varphi_t\,\Ebfscri_\Pi\,dt = 0.
\ee
This implies that $\Ebfscri_\Pi$ is constant in time, and equal to the initial data.
The same calculations apply if $\Pi^{0\bullet}$ is replaced by $\Pi^{0\bullet} \pm \Pi^{1\bullet}$, with the same reformulation in terms of measure norms thanks to the positivity condition $\Pi^{00} \pm \Pi^{10} \geq 0$:
\be
\aligned
\Ebfscri_\Pi = |t| \bigl\|\Pi^{00}\bigr\|_{\Meas(\Tbb^3)}(t)
& = |t_0| \bigl\|\mathring\Pi^{00}\bigr\|_{\Meas(\Tbb^3)} = \mathring\Ebfscri_\Pi , \qquad t\in(t_0,t_1) ,
\\
|t| \bigl\|\Pi^{00} \pm \Pi^{10}\bigr\|_{\Meas(\Tbb^3)}(t)
& = |t_0| \bigl\|\mathring\Pi^{00} \pm \mathring\Pi^{10}\bigr\|_{\Meas(\Tbb^3)} , \qquad t\in(t_0,t_1) .
\endaligned
\ee

\subsubsection{Timelike measure bound}
It will prove very useful to also control integrals of $\Pi$ along timelike slices of constant $x\in\Sbb^1$,
\be
\Tbfscri_{\Pi} \coloneqq \bigl\la \Pi^{00} , \, \amdeux^{-1}\,|t|\bigr\ra_{\Interval\times\Tbb^2} .
\ee
We use the divergence equation of~$\Pi^{1\bullet}$ in the weak form~\eqref{weak-divePi-bdry} with a test function that only depends on $x\in\Sbb^1$
to get
\be
\int_{\Sbb^1} \Tbfscri_{\Pi}\,\varphi_x\,dx
= - |t_0| \la \mathring\Pi^{10},\varphi\ra_{\Tbb^3} .
\ee
The signed measure $-|t_0|\mathring\Pi^{10}\in\Meas(\Tbb^3)$ on the right-hand side is bounded by $|t_0|\mathring\Pi^{00}\geq 0$, and this yields a bound on the total variation of $\Tbfscri_{\Pi}$ on~$\Sbb^1$,
\be
\Var\bigl( \Tbfscri_{\Pi} ; \; \Sbb^1 \bigr) < |t_0| \bigl\|\mathring\Pi^{00}\bigr\|_{\Meas(\Tbb^3)} = \mathring\Ebfscri_\Pi .
\ee
On the other hand, we have an integral bound on $\Tbfscri_{\Pi}$ by comparing $\amdeux$ to its infimum on each time slice,
\bel{Escrvee-inte}
\int_{\Sbb^1} \Tbfscri_{\Pi} \, \Bigl( \inf_{\Interval\times\Tbb^2} \amdeux(x)\Bigr) \, dx
\leq \int_{\Sbb^1} \bigl\la \Pi^{00} , \, |t|\bigr\ra_{\Interval\times\Tbb^2} dx = \int_{\Interval} \Ebfscri_\Pi dt \leq (t_1-t_0)\mathring\Ebfscri_\Pi .
\ee
The infimum of $\amdeux$ simplifies since the evolution equation~\eqref{eq:T2-8-def} of~$\amdeux$ implies monotonicity of~$\amdeux$ in time for each $x\in\Sbb^1$, and specifically
\be
\inf_{\Interval\times\Tbb^2} \amdeux(x) \leq
\begin{cases}
  \lim_{t'\to t_1} \amdeux(t',x) & t_0<t<t_1<0 \quad \text{(future-contracting regime)} , \\
  \mathring\amdeux(x) & 0<t_0<t<t_1 \quad \text{(future-expanding regime)} .
\end{cases}
\ee
As a result, the integration measure in~\eqref{Escrvee-inte} is the conformal length measure at $t=t_1$ or $t=t_0$, respectively, and for instance its integral is given by the conformal length functional, which by \autoref{lem-conf-vol} is the infimum conformal length:
\be
\int_{\Sbb^1} \Bigl( \inf_{\Interval\times\Tbb^2} \amdeux(x)\Bigr) \, dx
= \inf_{t\in(t_0,t_1)} \Lbfscri(t) = \Lmin .
\ee
Dividing~\eqref{Escrvee-inte} by $\Lmin$ provides a bound on a weighted average of~$\Tbfscri_{\Pi}$, which combines with the total variation bound to give a uniform bound in space.  We summarize the control of the corrector as follows.

\begin{lemma}[Control of the corrector on spacelike and timelike slices]
\label{prop:corrector}
The corrector $\Pi^{00}$ (and $|\Pi^{01}|\leq\Pi^{00}$) obeys uniform bounds in time and in space on measure norms,
\be
\aligned
\bigl\|\Pi^{00}\bigr\|_{\Meas(\Tbb^3)}(t)
& \leq (t_0/t) \bigl\|\mathring\Pi^{00}\bigr\|_{\Meas(\Tbb^3)} , \qquad 
&& t \in \Iopen ,
\\
\bigl\la \Pi^{00} , \, \amdeux^{-1}\,|t|\bigr\ra_{\Interval\times\Tbb^2}(x)
& \leq \Bigl( \frac{1}{2} + \frac{t_1-t_0}{\Lmin} \Bigr) |t_0| \bigl\|\mathring\Pi^{00}\bigr\|_{\Meas(\Tbb^3)} , \qquad && x \in \Sbb^1 .
\endaligned
\ee
\end{lemma}


\subsubsection{Energy estimates}

We now turn our attention to the $\Phi\Psi$ part of the energy functional~\eqref{equa--energy}.

\begin{lemma}[Monotonicity of $\Phi\Psi$ energy of spacelike slices]
\label{prop:energy}
The functional $\Ebfscri_{\Phi\Psi}$ in~\eqref{equa--energy-split} remains bounded for all times in terms of its value on any initial slice of constant areal time:
\bel{eq:monotenergy}
\aligned
\Ebfscri_{\Phi\Psi}(t)/\mathring\Ebfscri_{\Phi\Psi} & \leq t_0 / t , 
& \qquad t_0 \leq t < t_1 < 0 \text{ (future-contracting),}
\\
\Ebfscri_{\Phi\Psi}(t)/\mathring\Ebfscri_{\Phi\Psi} & \leq t / t_0 ,
& \qquad 0 < t_0 \leq t < t_1 \text{ (future-expanding).}
\endaligned
\ee
More precisely, in the future-contracting (resp.\ future-expanding) regime, $|t|\Ebfscri_{\Phi\Psi}$ (resp.\ $\Ebfscri_{\Phi\Psi}/|t|$) is non-increasing.
In particular, the energy functional has bounded variation and, consequently, it admits left-hand and right-hand traces at each time. For definiteness, one may assume that a right-continuous representative is chosen.
The same bounds and monotonicity properties hold for the full energy $\Ebfscri = \Ebfscri_{\Phi\Psi}+\Ebfscri_\Pi$.
\end{lemma}

\begin{proof}
We now use the energy balance law~\eqref{equa-lfk3}, namely the inequality form of the first Euler equation.  Together with the evolution equation~\eqref{eq:T2-8-def} for the conformal length density $\amdeux$, this yields (for $m\in\RR$):
\bel{equa-619d}
\aligned
& \abs{t}^{-m} \Bigl( \abs{t}^{m+1}  \, \Bigl(\Mbf^{00} \Omega^2 + \frac{3}{2t^2}\Bigr) d_x\ell dydz \Bigr)_t 
     + \abs{t} \, \bigl( \Mbf^{01} \, \Omega^2  dxdydz \bigr)_x 
\\
& \leq \sgn(t)\, \biggl(  \frac{m-1}{2}(P_0^2 + Q_0^2)
 + \frac{m+1}{2} (P_1^2 + Q_1^2)
 + \frac{3(m-1)}{2t^2} \Omega^{-2}
 + {(m-1)\over 2} (K_2^2 + K_3^2) 
\\
& \hskip2.cm + m (J_1^2 + J_2^2 + J_3^2)
 + {1 \over 2} \Big(m-1+(m+2)\,q_\Jbb \Big) \, \abs{\Jbb \cdot \Jbb} \biggr) \Omega^2 d_x\ell dydz. 
 \endaligned
\ee
Thanks to $q_\Jbb \in [0,1]$, for $m\,\sgn(t)\leq -1$ the source term is clearly non-positive.  Integrating over space, we conclude $\abs{t}^m \Ebfscri_{\Phi\Psi}(t)$ is non-increasing for this range of~$m$.  This establishes~\eqref{eq:monotenergy} and the stated monotonicity.
The remaining statements rely on $\Ebfscri_\Pi(t)=(t_0/t)\mathring\Ebfscri_\Pi$, cf.\ \autoref{prop:corrector},
\end{proof}


\subsubsection{Energy functional vs.\ conformal length functional}

The  conformal length is bounded above by the energy, namely $\Lbfscri(t) \leq (2/3) \abs{t} \Ebfscri_{\Phi\Psi}(t)$, itself controlled by the initial energy, according to~\autoref{prop:energy}.  The energy at any time~\eqref{equa--energy} is bounded on both sides by the initial energy $\Ebfscri_{\Phi\Psi}(t_0)$ multiplied by~$(t/t_0)^{\sign(t)}$. In particular, in the future-contracting regime, the derivative in~\eqref{Wmonotonic} is bounded above by the energy. Hence, the conformal length functional cannot decrease too quickly and its value at some time~$t$ with ${t_0\leq t<0}$ thus has a lower bound 
\be
\Lbfscri(t)\geq \mathringLbfscri-|t_0|\mathring\Ebfscri_{\Phi\Psi}\log|t_0/t|.
\ee
 This yields a lower bound on the time $t_*$ at which $\Lbfscri$ may vanish (in the future-contracting case $t_0<0$):
\bel{tstar-lower-bound-sec6}
t_* \geq t_0 \, \exp\biggl(- \frac{\mathringLbfscri}{|t_0|\mathring\Ebfscri_{\Phi\Psi}}\biggr).
\ee
Hence, along a sequence of solutions the maximal time of existence $t_*$ cannot tend to~$t_0$ as long as the ratio $\mathring\Ebfscri_{\Phi\Psi}/\mathringLbfscri$ of energy by conformal length is assumed bounded above.


\subsubsection{Timelike slices}

As our last energy estimate, we control the analogue of $\Ebfscri_{\Phi\Psi}$ on constant-$x$ slices.
Along the way, we consider the spacetime integral of \(\Mbf^{00}\), and therefore rely on the four-dimensional volume form
\be
\dVuntrois=\Omega\,dt\,\dVtrois,
\ee
which contains an \emph{additional} factor \(\Omega\). In view of~\eqref{eq:T2-8-def}, the function \(t\mapsto \amdeux(t,x)\) is increasing when \(t>0\) and decreasing when \(t<0\). Therefore, we can bound \(\Mbf^{00}\) from below by replacing \(d_x\ell(t,\cdot)\) by \(d_x\ell(t_i,\cdot)\), with \(t_i=t_1\) in the future-contracting regime and \(t_i=t_0\) in the future-expanding regime.  (More precisely, $d_x\ell(t_i,\cdot)$ stands for the limit as $t\to t_i$.) Thus, we have 
\bse
\label{equa---avecM00}
\be
\aligned
\int_{\Interval\times\Tbb^3} \Mbf^{00} \dVuntrois
&
\geq \int_{\Sbb^1} \biggl( d_x\ell(t_i,\cdot) \int_{t_0}^{t_1} \abs{t} \Omega^2 \Mbf^{00}(t,x)\, dt \biggr)  
\\
&
\geq \biggl(\int_{\Sbb^1} d_x\ell(t_i,\cdot)\biggr)\,
\inf_{x\in\Sbb^1} \biggl(\int_{t_0}^{t_1} \abs{t}\Omega^2 \Mbf^{00}\, dt \biggr),
\endaligned
\ee
hence
\bel{equa-onepoint}
\aligned
\inf_{x\in\Sbb^1} \biggl(\int_{t_0}^{t_1} \abs{t}\,\Omega^2 \Mbf^{00}\, dt \biggr)
\leq \frac{1}{\Lmin} \int_{\Interval\times\Tbb^3} \Mbf^{00} \, \dVuntrois
\leq \frac{1}{\Lmin} \int_{t_0}^{t_1} \Ebfscri_{\Phi\Psi}(t)\, dt
\endaligned
\ee
in both the future-expanding and future-contracting regimes. Since \(\Mbf^{00}\) contains the energies of \(\Pbb,\Qbb,\Kpar,\Jbb\), we can summarize this result as follows: \emph{at one point} \(x_0\in\Sbb^1\), the squared \(L^2\)-norm in time of \((\Pbb,\Qbb,\Kpar,\Jbb)\) is controlled by \(\max(t/t_0,t_0/t)\,\mathring\Ebfscri_{\Phi\Psi}/\Lmin\).
\ese
\bse
In fact, we can propagate the estimate~\eqref{equa-onepoint} to \emph{all points} \(x\in\Sbb^1\), as follows. Integrating the second Euler equation in~\eqref{eq:T2-Euler-perp-def} over a rectangle \([t_0,t_1]\times[x,y]\) and bounding the resulting spacelike integrals by the energy yields
\bel{equa311k}
\int_{t_0}^{t_1} \abs{t} \, (\Omega^2 \Mbf^{11})(t,x)\, dt
\leq \int_{t_0}^{t_1} \abs{t} \, (\Omega^2 \Mbf^{11})(t,y)\, dt + \mathring\Ebfscri_{\Phi\Psi} + \limsup_{t\to t_1} \Ebfscri_{\Phi\Psi}(t).
\ee
To the expression of \(\Mbf^{11}\) in~\eqref{eq:T2-Mdef-0}, we add the term \(\Ebf_0(\Kpar)\) so as to obtain the non-negative combination
\bel{eq-M11}
\Mbf^{11}+\Ebf_0(\Kpar)
=
\Ebf_0(\Pbb,\Qbb)+J_1^2+\frac12(1-q_\Jbb)\,|\Jbb\cdot\Jbb|
\geq 0,
\ee
in which, thanks to \eqref{equa-asym}, the coefficient \((1-q_\Jbb)\) is bounded away from $0$ and~$1$ at large densities (large values of $|\Jbb\cdot\Jbb|$). Since~\eqref{equa311k} holds for all~$y$, taking the infimum over~$y$ and using~\eqref{equa-onepoint}, we deduce that
\bel{equa311}
\aligned
& \sup_{x\in\Sbb^1}\int_{t_0}^{t_1}\abs{t}\,\bigl(\Omega^2(\Mbf^{11}+\Ebf_0(\Kpar))\bigr)(t,x)\,dt
\leq
\Nbfscri_{\Phi\Psi}
+\sup_{x\in\Sbb^1}\int_{t_0}^{t_1}\abs{t}\,\Omega^2\,\Ebf_0(\Kpar)(t,x)\,dt .
\endaligned
\ee
with the notation
\be
\Nbfscri_{\Phi\Psi}
= \mathring\Ebfscri_{\Phi\Psi}+\limsup_{t\to t_1}\Ebfscri_{\Phi\Psi}(t) + \frac{1}{\Lmin}\int_{t_0}^{t_1}\Ebfscri_{\Phi\Psi}(t)\,dt .
\ee
At this stage, it is preferable to keep the last term explicit using \(\Ebf_0(\Kpar)=\frac12|\Kpar|^2\); it will be estimated later on by means of the pointwise bounds on \(\Omega\) and \(\Kpar\).
\ese

\begin{lemma}[Control of the energy of timelike slices]
\label{lem--34a}
The energy on timelike slices satisfies
\bel{equa311-geom-zero}
\sup_{x\in\Sbb^1}\int_{\Interval\times\Tbb^2}\bigl(\Omega^2(\Mbf^{11}+\Ebf_0(\Kpar))\bigr)\,\Omega^{-1}\dVundeux
\leq \Nbfscri_{\Phi\Psi}
+\frac12 \sup_{x\in\Sbb^1}\int_{t_0}^{t_1}\abs{t}\,\Omega^2\,|\Kpar|^2(t,x)\,dt .
\ee
In particular, once \(\Omega\) and \(\Kpar\) are controlled pointwise, one obtains
\be
\sup_{x\in\Sbb^1}\int_{\Interval\times\Tbb^2}\bigl(\Omega^2(\Mbf^{11}+\Ebf_0(\Kpar))\bigr)\,\Omega^{-1}\dVundeux
\lesssim \Nbfscri_{\Phi\Psi}
+ |t_1^2-t_0^2|\,\sup_{\Interval\times\Tbb^3}|\Omega\Kpar|^2.
\ee
Consequently, \(\Pbb\), \(\Qbb\), \(J_1\), and \(|\Jbb\cdot\Jbb|^{1/2}\) are square-integrable in time along timelike slices.
\end{lemma}

At this stage, on timelike slices, we have no control over \(J_0\) or over the parallel momentum \(\Jpar=J_2+iJ_3\). The estimate~\eqref{equa311-geom-zero} will be closed below, once we establish pointwise bounds for \(\Omega\) and \(K\); see Lemmas~\ref{lemma-sup-Omega} and~\ref{lem-apriori--twist}.


\subsection{Pointwise estimates}
\label{section=6-3}
 
\subsubsection{Lapse function}

We now analyze the geometric variables in~$\Psi$. We seek first a pointwise bound on $\log\Omega$ throughout $\Interval\times\Tbb^3$. The modified lapse    constraint~\eqref{eq:evollambda-new-mod2b} and the inequalities $|\Pi^{01}|\leq\Pi^{00}$ and $|\Mbf^{01}|\leq\Mbf^{00}$ imply a total variation bound
\be
\Var(\log\Omega;\Tbb^3)
\leq \frac{|t|}{2} \Bigl( \int_{\Tbb^3} \Mbf^{00} \Omega\,\dVtrois + \|\Pi(t,\cdot)\|_{\Meas(\Tbb^3)} \Bigr) \leq \frac{\Ebfscri(t)}{2} .
\ee
Therefore 
\be
\exp\Bigl( - \frac{\Ebfscri(t)}{2} \Bigr) \leq \Omega(t,x)\,\Omega(t,x')^{-1} \leq \exp\Bigl( \frac{\Ebfscri(t)}{2} \Bigr).
\ee
We integrate these two inequalities with the measure $|t|^{-1}\Omega^{-1}\dVtrois=d_x\ell\,dy\,dz$ in the variable $x'$ and obtain the following result. 

\begin{lemma}[Two-sided pointwise control of the lapse]
\label{lemma-sup-Omega}
The function $\Omega$ satisfies the pointwise lower and upper bounds
\be
 |t|^{-1}  \Vcal(t) \exp\bigl( - \frac{\Ebfscri(t)}{2} \bigr) \leq \Omega(t,x) \Lbfscri(t)
  \leq  |t|^{-1} \Vcal(t)\exp\bigl( \frac{\Ebfscri(t)}{2} \bigr), 
\ee
in which, by the previous lemmas, the coefficients are bounded in terms of $\mathring\Ebfscri,\Lmin,t_0,t_1$.
\end{lemma}


\subsubsection{Direct estimates for the twists}

The evolution and constraint equations~\eqref{eq:T2-9101112-evol-def-a}--\eqref{eq:T2-9101112-evol-def-b} and~\eqref{eq:221a}--\eqref{eq:221b} for the twists $\Kpar=K_2+iK_3$ have the same structure as those for $\log\Omega$, except for an additional factor of~$\Omega^{-1}$ in the source terms. Here, we seek a two-sided bound on $K$, which turns out to involve the supremum of $\Omega^{-1}$. As for $\log\Omega$ the equations only control differences of~$\Kpar$, and we now use, as reference value, the infimum of~$|\Kpar|$. At any time $t\in[t_0,t_1]$ we have an upper bound on the infimum $\inf_{\Sbb^1}|\Kpar(t,\cdot)|$:
\bel{equa--322}
\frac{2}{|t|} \Ebfscri(t)
\geq \int_{\Sbb^1} |\Kpar|^2 \Omega^2 \, d_x\ell 
\geq \Lmin \, \inf_{\Sbb^1} |\Kpar(t,\cdot)|^2 \, \inf_{\Sbb^1} \Omega(t, \cdot)^2.
\ee
We then use the constraints~\eqref{eq:221a}--\eqref{eq:221b} to bound variations of~$\Kpar$ within one time slice following the same steps as for $\log\Omega$ to get the uniform bound
\be
|\Kpar(t,x)|
\leq \inf_{y\in\Sbb^1} |\Kpar(t,\cdot)| + \frac{\Ebfscri(t)}{|t|} \sup_{\Tbb^3} \Omega^{-1}(t, \cdot), 
\ee
and we arrive at the following result. 

\begin{lemma}[Pointwise control of the twists]  
\label{lem-apriori--twist}
The function $K=K_2 + i \, K_3$ is controlled from the initial data, as follows: 
\bel{equa-apriori--twist}
\sup_{\Sbb^1} |\Kpar(t,\cdot)| 
\leq \biggl(\biggl(\frac{2 \, \Ebfscri(t)}{t_{\min} \Lmin}\biggr)^{1/2} + {\Ebfscri(t)  \over t_{\min}} \biggr) \sup_{\Tbb^3}  \Omega^{-1}(t, \cdot), 
\ee
in which the upper bound is easily controlled (for all times) in terms of ${t_{\min} \coloneqq \min( |t_0|, |t_1|)}$ and $\mathring\Ebfscri,\Lmin$ and the sup-norm of~$1/\Omega$ given by~\autoref{lemma-sup-Omega}. 
\end{lemma}


\subsubsection{Conformal length density} 

By integration in time of the equation on $\amdeux_t$ in~\eqref{eq:T2-8-def}, for the conformal length density we find (pointwise in $x$)
\bel{equa-apriori-speed-formula}
\amdeux(t,x) = \mathring\amdeux(x)\, \exp \Bigl( \int_{t_0}^{t} {t' \over 2}\bigl(\Mbf^{00}-\Mbf^{11}\bigr)(t',x)\,\Omega(t',x)^2\,dt' \Bigr), 
\ee 
so that we can state the following result. Interestingly,~\eqref{equa-apriori-speed} below shows that the conformal length measure $d_x\ell(t)$ at each time $t$ (defined in \eqref{equaDell}) is absolutely continuous \emph{with respect to} the initial one $d_x\mathring\ell$. In particular, $d_x\ell(t)$ is absolutely continuous for all times.

\begin{lemma}[Two-sided pointwise control of the conformal length density]
\label{lem-speed}
The conformal length density is given pointwise by the formula 
\bel{equa-apriori-speed}
\aligned
d_x\ell(t) & =\beta (t) \, d_x\mathring\ell,
\\
\beta (t) & = \exp \Bigl(
\int_{t_0}^{t} |t'| \, \Bigl( \Ebf_0(\Kpar,\Jpar)+\frac{1}{2} q_\Jbb \, \abs{\Jbb\cdot\Jbb}\Bigr)(t',x)\, \Omega(t',x)^2 \, dt'
\Bigr), 
\endaligned
\ee
in which the coefficient $\beta =\beta(t,x) \in L^\infty(\Interval \times \Tbb^3)$ is bounded thanks to~\autoref{lem--34a}, above, and Lemma~\ref{lem-apriori--twist}, below. Observe that $\beta \geq 1$ in future-expanding spacetimes and $0<\beta \leq 1$ in future-contracting spacetimes. 
\end{lemma}


\subsection{Controlling Einstein--Euler flows from their initial data set}
\label{section=6-4}

Our analysis in the present section together with the results in \autoref{section=7}  allows us to 
uncover a minimal set of conditions to be assumed on the initial data sets in order to control all of the metric and fluid unknowns within the spacetime. In turn, this provides an important insight that motivates our definition of weak solutions in \autoref{section=4}. Observe that the flow need not be tame  (in contrast with \autoref{propos-priori}, derived below).

\begin{proposition}[A priori estimates for Einstein--Euler flows. I]
\label{theo--68}
Consider an Einstein--Euler flow with corrector $(\Phi,\Psi,\Pi)$ (cf.\ \autoref{weakdefinitionT2-deux2}) that assumes a prescribed initial data set $(\mathring\Phi,\mathring\Psi,\mathring\Pi)$ at time $t=t_0$ (cf.\ \autoref{def-initialdata}).
Then the following norms of $(\Phi,\Psi,\Pi)$,
\be
\aligned
& \sup_{t\in\Interval} \|\Phi\|_{L^2(\Tbb^3,\dVtrois_\suit)} , \quad
&& \sup_{x\in\Sbb^1} \|\Phi\|_{L^2(\Interval\times\Tbb^2,\dVundeux_\suit)} ,
\\
& \sup_{t\in\Interval}\Var(\Psi(t);\Tbb^3) , \quad
&& \sup_{x\in\Sbb^1} \Var(\Psi(x);\Interval\times\Tbb^2) , \quad
&& \sup_{\Interval\times\Tbb^3} |\Psi| ,
\\
& \sup_{t\in\Interval} \|\Pi\|_{\Meas(\Tbb^3)} , \quad
&& \sup_{x\in\Sbb^1} \|\Pi\|_{\Meas(\Interval\times\Tbb^2)} ,
\endaligned
\ee
are controlled in terms of $t_0,t_1,\Lmin$, and the norms of the \underline{initial data set} only: 
\bel{equa-648}
\|\mathring\Phi\|_{L^2(\Tbb^3,\dVtrois_\suit)} , \quad
\Var(\mathring\Psi(t);\Tbb^3) , \quad
\sup_{\Tbb^3} |\mathring\Psi| , \quad
\|\mathring\Pi\|_{\Meas(\Tbb^3)} .
\ee
\end{proposition}


\section{A priori estimates for tame Einstein--Euler flows} 
\label{section=7}

\subsection{Maximum principles for regular solutions}
\label{section=7-1} 

\subsubsection{Aim of this section}

In this section we establish further a priori bounds on tame Einstein--Euler flows, possibly with corrector.
We derive maximum principles ---for sufficiently regular solutions in this section, and for weak solutions in \autoref{section=7-2}. In addition, we establish a spacetime integrability property in \autoref{section=7-3}.


\subsubsection{Entropy-type structure for the twists}

At this stage, we assume that the solution under consideration is \emph{sufficiently regular}. However, the same maximum principle will be derived for weak solutions in \autoref{section=7-2}.
We now observe that, quite similarly to $\Jpar$, the twist functions satisfy linear transport equation with the fluid velocity, as follows. It is convenient to introduce the following notation\footnote{In fact, in the literature on \emph{vacuum} $\Tbb^2$ symmetric spacetimes,  
$\Gammazero_2$ and $\Gammazero_3$ are the so-called twist constants; however, for matter spacetimes $K_2,K_3$ are a convenient notation to obtain quadratic sources in the proposed \JKL\ formulation.}.

\begin{definition}
The \textbf{metric-weighted twist functions} $\Gammazero$ are defined by 
\bel{equa-Ktilde}
\Gammazero_2 \coloneqq |t|^{3/2}  e^{P/2} K_2, 
\qquad \quad
\Gammazero_3 \coloneqq |t|^{3/2}  \bigl( Qe^{P/2} K_2 + e^{-P/2} K_3 \bigr). 
\ee 
\end{definition}

The evolution equations for the twists~\eqref{eq:T2-9101112-evol-def-a}--\eqref{eq:T2-9101112-evol-def-b} then become
\bse
\label{eq:T2-9101112-evol-tilde-0} 
\bel{eq:T2-9101112-evol-tilde} 
\bigl(\Gammazero_2 \bigr)_t = |t|^{3/2}  e^{P/2}  J_2  J_1 \Omega,
\qquad\quad
\bigl( \Gammazero_3 \bigr)_t = |t|^{3/2} (Q e^{P/2} J_2 + e^{-P/2} J_3) J_1  \Omega
\ee 
and, similarly for the space derivatives in~\eqref{eq:221a}--\eqref{eq:221b},   
\bel{eq:221-cons}
( \Gammazero_2 )_x = |t|^{3/2}  e^{P/2} J_2 J_0  \, \Omega \amdeux, 
\qquad \quad 
( \Gammazero_3 )_x = |t|^{3/2} (Q e^{P/2} J_2 + e^{-P/2} J_3 ) J_0  \Omega \amdeux. 
\ee
\ese
We emphasize that  \emph{there is no factor} $\hNumb(\mu)^{-1}$ in the right-hand side of~\eqref{eq:T2-9101112-evol-tilde-0}, unlike the notation~\eqref{equa-jhat}. It is well-known that the twists are constant in the vacuum, but \emph{are not constant} in the presence of matter. Specifically, our notation is motivated by the proportionality between $\Gammaunpar$ and derivatives of~$\Gammazero$:
\be
(\Gammazero_m)_x = |t| \Omega \amdeux \Numb_0 \Gammaun_m, \quad m=2,3 .
\ee

By suitably combining~\eqref{eq:T2-9101112-evol-tilde}--\eqref{eq:221-cons}, we deduce that the twists $\Gammazero$ satisfy the same \emph{homogeneous transport equations} as the parallel momentum, namely 
\bel{equa-KDL3}
\aligned
\bigl( \Gammazero_m \bigr)_t -\amdeux^{-1}  {J_1\over J_0} \, \bigl( \Gammazero_m \bigr)_x 
 = 0,  \qquad m=2,3 
\endaligned
\ee
and in turn we find the \textbf{$H$-divergence equation for the twists} 
\bel{equa-mom-twists0}
\divuntrois \bigl( H\bigl( \Gammazero \bigr) \vNumb \bigr)
= 0
\quad \text{ for any function } H, 
\ee
valid for any sufficiently regular solution.  
This structure will be used below in order to derive a \emph{maximum principle} on $\Gammazero$ and, therefore, sup-norm estimates that are more precise than those derived in \autoref{lem-apriori--twist}.


\subsubsection{Controlling the parallel momentum (regular solurions)} 

\bse\label{equa-apriori-paral}

We present first a characteristic-based method, which \emph{does not} apply to weak solutions nor to relaxed flows. An alternative method presented next will apply to weak solutions but requires the notion of entropy. 
The functions $\widehat J_2$ and $\widehat J_3$ are given by integration along the characteristics $(t,\xi(t))$ defined by $\xi'(t)=a^2J_1/J_0$. For instance, we find 
\be
\abs{t}^{1/2}\widehat J_2(t,\xi(t)) = \abs{t_0}^{1/2}\widehat J_2(t_0,\xi(t_0)) \, \exp\biggl(\int_{t_0}^t \frac{\Pbb\cdot\Jperp}{2J_0} \Omega \biggl|_{(s,\xi(s))} ds\biggr). 
\ee
In principle, any $L^p$-type regularity could be investigated, but for simplicity in the presentation we assume that $\Jhatpar$ is bounded. The question of estimating $\Jhatpar$ boils down to bounding the integral of $\Omega\Pbb\cdot\Jperp/J_0$ along a characteristic curve.  We actually derive a stronger bound for the integral of the square $(\Omega\Pbb\cdot\Jperp/J_0)^2$.  This bound is obtained by observing that
\bel{Mchar}
\Mbf^{11} + \frac{2J_1}{J_0} \Mbf^{01} + \frac{J_1^2}{J_0^2} \Mbf^{00} \geq \biggl|\frac{\Pbb\cdot\Jperp}{J_0}\biggr|^2 .
\ee
By integrating a linear combination of the first two Euler equations on a region
$
\Bigl\{(t,x)\in[t_0,t_1]\times\Sbb^1\bigm| \xi(t)<x<x_0 \Bigr\}
$
for some fixed~$x_0$ we find indeed that~\eqref{Mchar} is integrable as desired.
Hence, we obtain a bound on the sup-norms of $\widehat J_2,\widehat J_3$ which depends on their initial sup-norms, as well as that of~$\log\Omega$, altogether with $\Ebfscri(t_0),\Lmin,t_0,t_1$.
\ese

We recall the relation~\eqref{equa-jhat} between parallel momenta with different normalizations, that is, $\Jpar=\hNumb_\Jbb \,\Jhatpar$.
As we observed below~\eqref{Sdef-scaled}, $\hNumb$ is bounded by $\mu^{1/2}\sim|\Jbb\cdot\Jbb|^{1/2}$ for large densities and is bounded (and in fact vanishes) as $\mu\to 0$. Our control of the time integrals of $|\Jbb\cdot\Jbb|$ established earlier, and of $1$ (trivial), imply a control of the time integral of~$\hNumb^2$.

\begin{lemma}[Pointwise estimate for the parallel momentum]
\label{lemma-JJJ} 
For sufficiently regular solutions,
the parallel momentum per particle $\Jhatpar$ enjoys the inequalities ($m=2,3$) 
\be 
\aligned
\underline C \, \inf_{\Tbb^3} \widehat J_m (t_0, \cdot)
\leq 
\inf_{\Interval \times \Tbb^3} \widehat J_m  
\leq   
\sup_{\Interval \times \Tbb^3} \widehat J_m  
\leq  
\overline C \, \sup_{\Tbb^3} \widehat J_m (t_0, \cdot), 
\endaligned
\ee  
in which the constants $\underline C$ and $\overline C$ depend upon norms that are controlled by the previous lemmas.
As a result the $L^\infty(\Sbb^1,L^2(\Interval\times\Tbb^2,\dVundeux))$ norm of $\Jpar$ is controlled by such norms.
\end{lemma} 


\subsubsection{Characteristics-based method (regular solutions)}

It is quite remarkable that the flows under consideration (possibly with weak regularity) enjoy a \emph{maximum principle,} as we now prove. Namely, we now establish a sup-norm bound enjoyed by the suitably metric-weighted parallel momentum and twists. Let us first present a method based on characteristics, which is insightful but will not carry over to weak solutions. We use the notation $\Gammaun_m$ given by~\eqref{equa-jhat} and~\eqref{equawidetildeJ}, that is, ${\Gammaun_2 \coloneqq e^{P/2}  \abs{t}^{1/2}\widehat J_2}$ and ${\Gammaun_3 \coloneqq Qe^{P/2}\abs{t}^{1/2}\widehat J_2+e^{-P/2}\abs{t}^{1/2}\widehat J_3}$, with $\widehat J_m = \bigl( \hNumb_\Jbb\bigr)^{-1} J_m$. We integrate the transport equations~\eqref{equa-jhat-transp-cons} along the characteristics defined for any base point $x_0 \in \Sbb^1$ as  
\bse\label{equa-apriori-paral-carac}
\be
t \mapsto (t,\xi(t; x_0)), \qquad 
\xi'(t; x_0) = (a^2J_1/J_0)\big|_{x=\xi(t; x_0)}, 
\qquad \xi(t_0; x_0) = x_0. 
\ee
We thus find  
\be 
\Gammaun_m(t,\xi(t; x_0)) = \Gammaun_m(t_0, x_0), \qquad m=2,3.   
\ee 
The functions $\Gammaun_m$ are simply transported from their initial values, so any $L^p$ norm could be considered. For simplicity in the presentation, we assume that $\Jhatpar$ is initially \emph{bounded}. We conclude with a pointwise bound on $\widehat J_2$ and $\widehat J_3$, stated now.  The same argument applies to the metric-weighted twists. At this stage, we have reached the following conclusion for \emph{sufficiently regular} solutions; next, we show that it remains valid for \emph{weak solutions.} 
\ese

\begin{lemma}[Maximum principle for the parallel momentum and the twists]
\label{proposition-JJJ}
The parallel momentum $\Gammaunpar$ and the twist functions $\Gammazero_m$ enjoy the \emph{maximum principle}
\bel{equa-widetildemax}
\aligned
\inf_{\Tbb^3} \Gammaun_m (t_0, \cdot)
\leq 
\inf_{\Interval \times \Tbb^3}  \Gammaun_m
\leq   
\sup_{\Interval \times \Tbb^3} \Gammaun_m
\leq  
\sup_{\Tbb^3} \Gammaun_m (t_0, \cdot),
\\
\inf_{\Tbb^3} \Gammazero_m (t_0, \cdot)
\leq 
\inf_{\Interval \times \Tbb^3}  \Gammazero_m
\leq   
\sup_{\Interval \times \Tbb^3} \Gammazero_m
\leq  
\sup_{\Tbb^3} \Gammazero_m (t_0, \cdot)
\endaligned
\ee
\end{lemma}

From \autoref{proposition-JJJ}, it follows that the momentum 
\bse
$\Jhatpar = \bigl( \hNumb_\Jbb\bigr)^{-1} \Jpar$ and the twists $\widehat K_m$ enjoy the uniform bounds
\be 
\aligned
\underline C \, \inf_{\Tbb^3} \widehat J_m (t_0, \cdot)
\leq 
\inf_{\Interval \times \Tbb^3} \widehat J_m  
\leq   
\sup_{\Interval \times \Tbb^3} \widehat J_m  
\leq  
\overline C \, \sup_{\Tbb^3} \widehat J_m (t_0, \cdot),
\\
\underline C \, \inf_{\Tbb^3}  K_m (t_0, \cdot)
\leq 
\inf_{\Interval \times \Tbb^3}  K_m  
\leq   
\sup_{\Interval \times \Tbb^3}  K_m  
\leq  
\overline C \, \sup_{\Tbb^3}  K_m (t_0, \cdot)
\endaligned
\ee
with 
\be
\underline C \coloneqq \Bigl|\frac{t}{t_0}\Bigr|^{\pm 1/2} {\inf_{\Tbb^3} e^{P(t_0, \cdot)/2} \over \sup_{\Interval \times \Tbb^3} e^{P/2}},
\qquad \quad 
\overline C \coloneqq \Bigl|\frac{t}{t_0}\Bigr|^{\pm 1/2} {\sup_{\Tbb^3} e^{P(t_0, \cdot)/2} \over \inf_{\Interval \times \Tbb^3} e^{P/2}}. 
\ee
Here, $P,Q$ are easily controlled from~\autoref{propo-reconstruct}, stated below. 
\ese
%


\subsection{Generalization to tame Einstein--Euler flows}
\label{section=7-2} 

\subsubsection{Entropy-based method (weak solutions)}

We now prove the pointwise estimates stated in \autoref{proposition-JJJ} via a second method which is conceptually more involved but, importantly, \emph{does apply to weak solutions}. Namely, for $\Gammaunpar$ we now rely on the $H$-divergence law~\eqref{equa-bvpara} which is included in our notion of tame solution. 
Interestingly, since $K_2, K_3$ are more regular than $\Gammaun_m$, the transport equations \eqref{equa-KDL3} and their consequence \eqref{equa-mom-twists0} remain valid for weak solutions: namely,  \eqref{equa-mom-twists0}  is straightforward to derive for tame flows thanks to the $\BVac$ regularity (in space and time) of~$\Kpar$. 

We are going to choose suitable entropy functions~$H$.
Integrating these balance laws between two hypersurfaces of constant areal times $t_0, t$  and using the divergence formula over the spacetime region $[t_0, t] \times \Tbb^3$, we obtain
\bel{equahdineqJ}
0 \leq 
\int_{\Tbb^3} H( \Gammaun_2, \Gammaun_3) \, S_0 \dVtrois\big|_{t}
\leq 
\int_{\Tbb^3} H( \Gammaun_2, \Gammaun_3) \, S_0 \dVtrois\big|_{t=t_0}
\ee
for any non-negative function $H$, where $S_0$ denotes the time component of $\Sbb=(S_0,S_1)$. Using $H = \Hcal(\xi_m -  \Gammaun_m )$, where $\Hcal\colon \RR \to [0,1]$ denotes the Heaviside function, and choosing $\xi_m$ to be the supremum of the initial data $\Gammaun_m$, we see that $\xi_m -  \Gammaun_m \leq$ at the initial time $t=t_0$, and therefore the left-hand side of~\eqref{equahdineqJ} must vanish. Since $N_0 >0$, this requires that  
\be
\sup_{\Interval \times \Tbb^3} \Gammaun_m
\leq  
\sup_{\Tbb^3} \Gammaun_m (t_0, \cdot). 
\ee
The same observation holds for the infimum, as well as for the twists $\Gammazero_m$, thanks to~\eqref{equa-mom-twists0}. We thus arrive at~\eqref{equa-widetildemax}, and this provides us with a second proof of~\autoref{proposition-JJJ}. 
 

\subsubsection{Timelike estimate for the parallel momentum} 

As we have checked (cf.~below~\eqref{Sdef-scaled}), $\hNumb(\mu)$ is bounded by $\mu^{1/2}\sim|\Jbb\cdot\Jbb|^{1/2}$ for large densities and is bounded (in fact, vanishes) as $\mu\to 0$. Our control of the time integrals of $|\Jbb\cdot\Jbb|$ ---established in~\autoref{lem--34a} (cf.~also~\eqref{eq-M11})--- and, trivially, of the constant function \(1\) imply a control of the time integral of~$\hNumb^2$. Since $\Jhatpar$ is bounded, we deduce a bound on the time integral of~$|\Jpar|^2$.
This completes our list of timelike $L^2$ estimates. The following statement, combined with our previous bounds, completes the control of the square-integrability of the functions $\Pbb, \Qbb, \Jperp, \Jpar$ on \emph{timelike slices}. It relies on the control of parallel momentum per particle in Lemma~\ref{lemma-JJJ} above.

\begin{lemma}[Parallel momentum on timelike slices]
\label{lem--34par}
The parallel momentum on timelike slices is controlled by 
\bel{equa311-geom}
\aligned
& \sup_{x \in\Sbb^1}  \int_{I \times \Tbb^2} |\Jpar|^2 \, \Omega \, \dVundeux
\leq {1 \over 4} |t_1^2-t_0^2| \sup_{I \times \Tbb^3}  |\Omega\Kpar|^2 + {1 \over \Lmin} \int_{t_0}^{t_1} \Ebfscri(t) dt+ \Ebfscri(t_0) + \Ebfscri(t_1), 
\endaligned
\ee 
in which $\Omega$ and $K$ were controlled in sup norm in Lemmas~\ref{lemma-sup-Omega} and \ref{lem-apriori--twist}.
\end{lemma} 


\subsubsection{Recovering the metric} 

There remains to supplement~\autoref{proposition-JJJ} with the reconstruction of the functions $P$ and $Q$ from the expression of the spacetime metric~\eqref{metric:areal}. From~\eqref{eq:definevar-01}, we have $P_t = P_0 \Omega$ and $P_x = P_1  \Omega \amdeux$, and similarly for $Q$ and, therefore,  
\bel{equa-formulPQ}
\aligned
P(t,x) & = \int_0^x P_1  \Omega \amdeux\big|_{(t_0, y)} dy + \int_{t_0}^t P_0 \Omega \big|_{(s, x)}  ds, 
\\
Q(t,x) & = \int_0^x e^{-P} Q_1  \Omega \amdeux\big|_{(t_0, y)}  dx + \int_{t_0}^t e^{-P} Q_0 \Omega\big|_{(s, x)}  dt, 
\endaligned
\ee
in which we chose the \emph{normalization} $P(t_0,0) = Q(t_0,0) = 0$. There is no loss of generality, since $P$ and $Q$ are defined up to irrelevant constants. The following result is then immediate. 
 
\begin{lemma}[Reconstruction of the spacetime metric]
\label{propo-reconstruct}
The metric coefficients $P$ and $Q$ are bounded pointwise over $\Interval \times \Tbb^3$, as follows: 
\be 
|P(t,x)| \leq  \Ebfscri(t)^{1/2} \Lbfscri(t)^{1/2}, 
\qquad \quad 
|Q(t,x)| \leq \bigl( \sup_{\Tbb^3} e^{-P(t, \cdot)} \bigr) \Ebfscri(t)^{1/2} \Lbfscri(t)^{1/2}, 
\ee
in which, in the second inequality, the factor $ \sup e^{-P}$ is controlled thanks to the first inequality, so that, in  particular, $
P, Q\in L^\infty(\Interval \times \Tbb^3). 
$
\end{lemma}  


\subsection{Additional fluid-geometric integrability}
\label{section=7-3} 

\subsubsection{Monotonicity}

Here we derive additional integrability properties enjoyed by tame flows \emph{without corrector}.
In particular $\Omega$ is $\BVac$ in this section.
Recall that we consider solutions defined on an interval $\Interval = (t_0, t_1)$ whose closure does not contain~$0$. Recall also the sup-norms of $\log\Omega$ and~$\Kpar$ over this spacetime domain $\Interval\times\Tbb^3$, and the $L^2$ bounds (in space, and in time) of $(\Jbb,\Pbb,\Qbb)$ are controlled (cf.~\autoref{section=6}) in terms of initial data.
As pointed out earlier, the following arguments apply to weak solutions by selecting a suitable test-function in the weak formulation of the equations under consideration below. 

We work with the Euler equation 
\be
\Bigl( |t| \, \Mbf^{00}(\Jbb, \Kpar, \Pbb, \Qbb) \, \Omega^2 \amdeux \Bigr)_t 
\leq - \Bigl( |t| \, \Mbf^{01}(\Jbb, \Kpar, \Pbb, \Qbb) \, \Omega^2  \Bigr)_x 
 - \frac{1}{2} \sgn(t) \Msource(\Jbb, \Kpar, \Pbb, \Qbb) \, \Omega^2 \amdeux, 
\ee
To avoid clutter we work solely in the future-contracting regime $t<0$.  The future-expanding regime is treated similarly by multiplying the Euler equation by the opposite power of~$\Omega$.
We compute the following combination
\be
\aligned
& \frac{d}{dt}\biggl( |t| \, \int_{\Sbb^1} |t|^{1/4}\Omega \amdeux \, \Mbf^{00} \, \Omega^2 \, dx\biggr)
\\
& = \int_{\Sbb^1} \bigl( |t|^{1/4}\Omega \bigr)_t \,|t| \, \Omega^2 \amdeux \Mbf^{00} \, dx
+ \int_{\Sbb^1} \bigl(|t| \, \Omega^2 \amdeux \Mbf^{00}\bigr)_t |t|^{1/4}\Omega \, dx
\\
& \leq - \frac{t^2}{2} \int_{\Sbb^1} \Mbf^{11} f|t|^{1/4}\Omega \amdeux
 \Omega^4 \Mbf^{00} \, dx
+ \int_{\Sbb^1} \Bigl(\bigl(t \, \Omega^2 \Mbf^{01}\bigr)_x + \frac{1}{2}\Omega^2 \amdeux \Msource\Bigr) |t|^{1/4}\Omega \, dx.
\endaligned
\ee
The next step is to integrate by parts in~$x$ and produce $\bigl( |t|^{1/4}\Omega \bigr)_x$; 
recalling the constraint equation $(\log\Omega)_x= - (t/2)\Omega^2\amdeux \Mbf^{10}$ yields
\bel{equa-sans-beta}
\frac{d}{dt}\biggl( |t| \, \int_{\Sbb^1} |t|^{1/4}\Omega \amdeux \, \Mbf^{00} \, \Omega^2 \, dx\biggr)
\leq \frac{-|t|^{5/4}}{2} \int_{\Tbb^3} \biggl(\Mbf^{00}\, \Mbf^{11} - \Mbf^{01}\, \Mbf^{10} - \frac{\Msource}{t^2\Omega^2}\biggr) \Omega^4 \dVtrois. 
\ee
It is useful to also include a factor $|t|^{\beta}$ and rewrite \eqref{equa-sans-beta} more generally as 
\bel{d2Vgen}
\aligned
& \frac{d}{dt} \Bigg(
\frac{|t|^{\beta+1/4}}{2} \int_{\Tbb^3} \biggl( \Mbf^{00} + \frac{2\beta}{\Omega^2|t|^2} \biggr) \Omega^2 \dVtrois \Bigg)
\\
& \leq \frac{-|t|^{\beta+5/4}}{4} \int_{\Tbb^3} \biggl(
\det(M) + \frac{4\beta \Mbf^{00} - \Msource}{t^2\Omega^2}
+ \frac{4\beta(\beta-1)}{|t|^4\Omega^4} \biggr) \Omega^4 \dVtrois.
\endaligned
\ee
Since $|\Msource| \leq 5 \, \Mbf^{00}$, the remaining terms can be made positive by choosing $\beta$ large enough; specifically, we take $\beta\geq 5/4$. We study $\det(M)$ next. 


\subsubsection{Positivity}

We claim that $\det(M)=\Mbf^{00}\Mbf^{11}- \Mbf^{01}\Mbf^{10}$ is positive, up to terms that are controlled by the energy (hence that can be absorbed by choosing $\beta$ large enough). Using the expressions in~\eqref{eq:T2-Mdef-0}  we find that 
\bse
\be
\det(M) = \bigl( \Ebf_0(\Jperp, \Pbb, \Qbb) \bigr)^2
- \Bigl( \Ebf_0(\Jpar, \Kpar) - {q_\Jbb \over 2} \Jbb \cdot \Jbb \Bigr)^2
- \bigl( \Ebf_1(\Pbb) + \Ebf_1(\Qbb) + \Ebf_1(\Jperp) \bigr)^2.
\ee
We will now check that this expression is non-negative when $\Kpar$ vanishes, but for sufficiently large $\Kpar$ this expression changes sign, so we write 
\be
\aligned
& \det (M) \geq D(\Pbb, \Qbb, \Jbb) - \bigl( \Ebf_0(\Kpar) + 2 \, \Ebf_0(\Jpar)  - q_\Jbb \, \Jbb \cdot \Jbb \bigr) \, \Ebf_0(\Kpar),
\\
& D = D(\Pbb, \Qbb, \Jbb) \coloneqq \bigl( \Ebf_0(\Jperp, \Pbb, \Qbb) \bigr)^2 
- \Bigl( \Ebf_0(\Jpar)  - {q_\Jbb \over 2} \Jbb \cdot \Jbb \Bigr)^2 
- \bigl( \Ebf_1(\Pbb) + \Ebf_1(\Qbb) + \Ebf_1(\Jperp) \bigr)^2. 
\endaligned
\ee
We observe that 
\be
\bigl(\Ebf_0(\Kpar) + 2 \, \Ebf_0(\Jpar) + q_\Jbb \, (-\Jbb \cdot \Jbb) \bigr)\, \Ebf_0(\Kpar) \lesssim \Mbf^{00}\sup|\Kpar|^2,
\ee
hence this term can be absorbed by taking $\beta$ large enough in~\eqref{d2Vgen} (actually $\beta\geq 5/4+\sup|t \, \Omega\Kpar|^2/4$).

We then rewrite $D$ as a sum of manifestly positive terms as follows. We combine the first and last term in $D$ into the product of $\Ebf_0(\Jperp, \Pbb, \Qbb) \pm (\Ebf_1(\Pbb)+\Ebf_1(\Qbb)+\Ebf_1(\Jperp))$, which are sums of squares, then expand the product and recognize the quadratic forms $(P_0\pm P_1)(Q_0\mp Q_1)=- \Pbb\cdot\Qbb\mp\Pbb\wedge\Qbb$ and analogues. Terms that only depend on momentum are treated last:
\be
\aligned
4 D
& = \Bigl( (P_0 + P_1)^2 + (Q_0 + Q_1)^2 + (J_0 + J_1)^2  \Bigr)
\Bigl( (P_0 - P_1)^2 + (Q_0 - Q_1)^2 + (J_0 - J_1)^2  \Bigr)
\\
&\quad 
- \bigl( 2 \, \Ebf_0(\Jpar) - q \, \Jbb \cdot \Jbb \bigr)^2
\\
& = (\Pbb\cdot\Pbb)^2 + (\Pbb\cdot\Qbb+\Pbb\wedge\Qbb)^2 + (\Pbb\cdot\Qbb- \Pbb\wedge\Qbb)^2 + (\Qbb\cdot\Qbb)^2 + (\Pbb\cdot\Jperp+\Pbb\wedge\Jperp)^2 
\\
& \quad+ (\Pbb\cdot\Jperp- \Pbb\wedge\Jperp)^2 + (\Qbb\cdot\Jperp+\Qbb\wedge\Jperp)^2 + (\Qbb\cdot\Jperp- \Qbb\wedge\Jperp)^2 + G(\Jbb),
\endaligned
\ee
in which $G=G(\Jbb) $ satisfies 
\bel{equa--l295}
\aligned
G & \coloneqq 
(J_0^2 - J_1^2)^2 - \big( 2 \, \Ebf_0(\Jpar)  - q_\Jbb (\Jbb \cdot \Jbb) \bigr)^2 
\\
& = \bigl( J_0^2 - J_1^2 - J_2^2 - J_3^2 + q_\Jbb (\Jbb \cdot \Jbb) \bigr)
\bigl( J_0^2 - J_1^2 + J_2^2 + J_3^2  - q_\Jbb ( \Jbb \cdot \Jbb) \bigr)
\\
& = (1-q_\Jbb) |\Jbb\cdot\Jbb| \bigl(2 \, \Ebf_0(\Jpar) + (1-q_\Jbb)|\Jbb\cdot\Jbb| \bigr). 
\endaligned
\ee
Thus, all terms are non-negative under the assumption \(q_\Jbb \leq 1\). On the other hand, the control of \( |\Jbb\cdot\Jbb| \, \Ebf_0(\Jpar)\) and \( |\Jbb\cdot\Jbb|^2\) degenerates when \(1-q_\Jbb\) vanishes. Observe that the case \(q_\Jbb\equiv 1\) on an open interval of values of \(\mu\) is excluded, since it corresponds to vanishing pressure (a model known to exhibit a phenomenon of formation of Dirac masses).
\ese
%


\subsubsection{Conclusion.}

We observe that $\Jpar$ is the product of a bounded function by $\hNumb$, which is bounded by the sup of $1$ and $|\Jbb\cdot\Jbb|^{1/2}$. Thus, we have 
\be
(1-q_\Jbb)|\Jpar|^4 \lesssim (1-q_\Jbb)\hNumb^2 \Ebf_0(\Jpar) \lesssim \max\Bigl(\Ebf_0(\Jpar),(1-q_\Jbb)|\Jbb\cdot\Jbb|\Ebf_0(\Jpar)\Bigr),
\ee
so we learn that $(1-q_\Jbb)^{1/4}\Jpar\in L^4(\Interval \times \Tbb^3, \dVuntrois)$. Whenever~\eqref{hyperbolic-eos} holds true, 
we can suppress the factor $(1-q_\Jbb)$ arising in our expressions. 
 We also observe that the bound on $ |\Jbb\cdot\Jbb| \, \Ebf_0(\Jpar)$ can then be deduced from the ones on $|\Jbb\cdot\Jbb|$ and $\Jpar$. We have thus established one of the main observations made in the present paper. We summarize our conclusion as follows. 

\begin{lemma}[Spacetime integrability]
\label{theo:te-Euler-explicit}
Suppose that the pressure law satisfies the conditions~\eqref{hyperbolic-eos}, and consider any tame Einstein-Euler flow (without corrector).  Then, the following spacetime integral over $\Interval\times\Tbb^3=[t_0, t_1) \times \Tbb^3$ 
\be
\aligned
& \iint_{\Interval \times\Tbb^3} \Bigl( (\Pbb\cdot\Jperp)^2 + (\Qbb\cdot\Jperp)^2 + (\Pbb\wedge\Jperp)^2 + (\Qbb\wedge\Jperp)^2 + (\Pbb\cdot\Pbb)^2 + (\Qbb\cdot\Pbb)^2 
\\
& \qquad\qquad 
+ (\Qbb\cdot\Qbb)^2 + (\Qbb\wedge\Pbb)^2 + |\Jbb\cdot\Jbb|^2 
+ |\Jpar|^4 \Bigr) \, \dVuntrois
\\
& \leq C_0\bigl(I, \sup \log \Omega, \sup|\Kpar|, \Ebfscri(t_0) \bigr) 
\endaligned
\ee
is bounded above by a function $C_0>0$ depending \underline{only upon the initial data set}, 
namely depending upon the sup-norm (over $\Interval \times\Tbb^3$) of $\log \Omega$ and $\Kpar$ and of the initial energy $\Ebfscri(t_0)$. (Recall that the spacetime form is $\dVuntrois=\Omega dt\dVtrois = |t| \, \Omega^2 d_x\ell dydz$
and the energy $\Ebfscri(t_0) $ was defined in \eqref{equa--energy}). 
\end{lemma}


\subsection{Controlling tame Einstein--Euler flows from their initial data set}
\label{section=7-4}

In conclusion, the following statement complements \autoref{theo--68}, but now requires the flow to be \emph{tame}.

\begin{proposition}[A priori estimates for Einstein--Euler flows. II]
\label{propos-priori}
Assume that the pressure law satisfies the conditions~\eqref{hyperbolic-eos}. Consider a tame Einstein--Euler flow with corrector, defined on a time interval $[t_0, t_1)$ and assuming a prescribed initial data set at areal time $t_0$ (cf.~\autoref{def-initialdata}). Then it enjoys the following bounds. 

\bei 

\item {\bf Maximum principle for the parallel momentum:}
\bel{equa-widetildemax-rep1}
\aligned
\inf_{\Tbb^3} \Gammaun_m (t_0, \cdot)
\leq 
\inf_{\Interval \times \Tbb^3}  \Gammaun_m
\leq   
\sup_{\Interval \times \Tbb^3} \Gammaun_m
\leq  
\sup_{\Tbb^3} \Gammaun_m (t_0, \cdot).
\endaligned
\ee

\item {\bf Maximum principle for the twists:} 
\bel{equa-widetildemax-rep2}
\aligned  
\inf_{\Tbb^3} \Gammazero_m (t_0, \cdot)
\leq 
\inf_{\Interval \times \Tbb^3}  \Gammazero_m
\leq   
\sup_{\Interval \times \Tbb^3} \Gammazero_m
\leq  
\sup_{\Tbb^3} \Gammazero_m (t_0, \cdot). 
\endaligned
\ee

\item \textbf{Total variation of the parallel momentum}: 
\be
\sup_{ t \in \Interval} \Var(\Gammaunpar(t)) 
\leq C_1(\Interval) \, \Var(\Gammaunpar(t_0)). 
\ee

\item \textbf{Time continuity of the momentum}: 
\be
\sup_{ t_2\neq t_3 \in \Interval} \Bigl\| \frac{\Gammaunpar(t_3)-\Gammaunpar(t_2)}{t_3 - t_2} \Bigr\|_{L^1(\Tbb^3,\dVtrois(t_2))}
\leq C_2(\Interval) \, \Var(\Gammaunpar(t_0)). 
\ee

\eei 
\noindent Here, the constants $C_j(\Interval)$ \underline{depend only on the initial data set} and, specifically, the same norms as listed in~\eqref{equa-648}. 
Furthermore, if the corrector $\mathring\Pi \equiv 0$ vanishes initially, then $\Pi \equiv 0$ for all relevant times and, in addition, one has the following \emph{integrability property:}
\be
\aligned
& \Xbb\cdot\Ybb, \, \Xbb\wedge\Ybb \in L^2(\Interval \times \Tbb^3, \dVuntrois)
& \text{ for all } \Xbb,\Ybb \in \bigl\{ \Pbb, \Qbb, \Jperp \bigr\}. 
\endaligned 
\ee
\end{proposition}

In particular, from the maximum principles \eqref{equa-widetildemax-rep1} and \eqref{equa-widetildemax-rep2}, we deduce the pointwise estimate 
\be
\sup_{\Interval \times \Tbb^3} \Bigl( |K| + \hNumb(\mu)^{-1/2} |\Jpar| \Bigr)
\leq C_3(\Interval). 
\ee


\section{Compactness properties of metric components}
\label{section=8}

\subsection{Methodology}
\label{section=8-1}

This section is devoted to the proof of the compactness statements underlying the stability and instability results of \autoref{section=5}. We first consider a bounded sequence of tame Einstein--Euler flow with corrector and derive the compactness properties leading to the relaxed notion of limit flow; this yields the proof of~\autoref{theorem-493}. We then return to the case of well-prepared initial data and show that the stronger convergence properties available in that setting imply \autoref{theorem-492}.

The argument is organized hierarchically. Starting from a bounded sequence of relaxed flows, we extract subsequences and identify the corresponding weak or measure limits. We then examine the equations successively and study their behavior under the available convergence. A basic feature of the Einstein--Euler system is that the compactness of certain variables follows from the structure of certain equations. In particular, the geometric variables are treated first, and their strong compactness is later used in the compensated compactness analysis of the matter variables.

Throughout this section, we work in \(1+1\) notation on the quotient spacetime \(\Interval\times\Sbb^1\), using both the density \(\amdeux\) and the associated measure and primitive,
\be
d_x\ell=\amdeux\,dx,
\qquad
\ell(t,x)=\int_0^x \amdeux(t,y)\,dy.
\ee
We therefore consider a sequence of tame Einstein--Euler flows with corrector 
\be
\Phi^\suit=\bigl(\Jperp{}^\suit,\Pbb^\suit,\Qbb^\suit\bigr),
\qquad
\Psi^\suit=\bigl(\ell^\suit,\log\Omega^\suit,\Kpar^\suit,\Jhatpar{}^\suit\bigr),
\qquad
\Pi^\suit,
\ee
satisfying the assumptions of~\autoref{theorem-493}; the corresponding variables \(\Jpar{}^\suit\) are recovered from \(\Jhatpar{}^\suit\). After extraction of subsequences, which is understood throughout and for which we keep the same notation, the proof proceeds in several steps.

We begin with the conformal length measure \(d_x\ell^\suit\), whose evolution equation propagates the initial strong compactness assumption. We next analyze the lapse \(\Omega^\suit\): for general data one obtains only weak-\(*\) compactness in \(BV\), whereas for well-prepared data the stronger convergence of the source terms implies absolute continuity. We then turn to the remaining geometric variables, namely the twists and the essential variables \(\Pbb^\suit,\Qbb^\suit\). Their equations involve source terms that are either directly compact or of null-form type, so that the previously established convergence of \(\ell^\suit\) and \(\Omega^\suit\) allows us to pass to the limit. Once the geometric variables have been controlled, the matter equations will be treated in the next section.

In the general case, the outcome is a relaxed limit flow, possibly carrying a corrector in the lapse equations. In the well-prepared regime, the same hierarchy of arguments yields stronger convergence at each stage, and the corrector vanishes.

In the analysis below, and in particular for the geometric variables, we may assume for clarify that $(\Phi^\suit,\Psi^\suit)$ is a sequence of exact weak solutions. We therefore suppress the error terms $\Hterm$ and~$\Mterm$ arising in the formulation for approximate solutions. Reintroducing them would require only minor modifications and would not affect any of the conclusions.

The uniform bounds on \((\Phi^\suit,\Psi^\suit,\Pi^\suit)\) allow us to extract a subsequence and a limit \((\Phi^\sharp,\Psi^\sharp,\Pi^\sharp)\) such that
\be
\aligned
\Phi^\suit \rightharpoonup \Phi^\sharp
\qquad & \text{weakly in } L^2(\Sbb^1) \text{ for each fixed time,}
\\
\Psi^\suit \overset{*}{\rightharpoonup} \Psi^\sharp
\qquad & \text{weakly-* in } BV(\Sbb^1) \text{ for each fixed time,}
\\
\Pi^\suit \overset{*}{\rightharpoonup} \Pi^\sharp
\qquad & \text{weakly-* as measures on } \Sbb^1 \text{ for each fixed time.}
\endaligned
\ee
The proof combines an $L^2$ compactness theorem or Helly's theorem, together with a diagonal extraction at rational times. The passage from rational times to arbitrary times relies on the weak time regularity supplied by the bounds on the time derivatives, which provide the required equicontinuity in suitable weak topologies.
Further details are provided in the course of our analysis. 


\subsection{Analysis of the conformal length measure}
\label{section=8-2}

We now analyze the convergence of the conformal length measure $d_x\ell^\suit(t)=\amdeux^\suit(t,\cdot)\,dx$.
The sequence of weak solutions satisfies the evolution equation
\bel{esuat-71}
(\amdeux^\suit)_t
= t \, \Bigl( \Ebf_0(\Jpar^\suit,\Kpar^\suit)
      - \frac{q_\Jbb^\suit}{2}\,\Jbb^\suit\cdot\Jbb^\suit \Bigr)
  (\Omega^\suit)^2 \amdeux^\suit
=: S_{\mathrm{length}}^\suit,
\ee
together with the bound
\bel{equa-bounds1}
\sup_{\suit \in (0,1]} \sup_{t \in [t_0,t_1]}
\int_{\Sbb^1} |S_{\mathrm{length}}^\suit(t,\cdot)|\,dx < +\infty.
\ee
We also assume that the initial length densities are normalized in such a way that
\bel{equa-convlambda}
\mathring\amdeux^\suit \to \mathring\amdeux^\sharp
\qquad \text{strongly in } L^1(\Sbb^1) \text{ (at initial time)}.
\ee
Recall that the space \(\BVac(\Sbb^1)\) is defined independently of any volume form: it consists of bounded functions with finite total variation whose distributional derivative is absolutely continuous with respect to the Lebesgue measure. In particular, the sequence \(\ell^\suit\) is uniformly bounded in \(L^\infty(\Interval;\BVac(\Sbb^1))\), and the sequence \((\amdeux^\suit)_t\) is uniformly bounded in \(L^\infty(\Interval;L^1(\Sbb^1))\).

Although we assumed convergence of $\lambda^\eps$ for each time, it is convenient here to \emph{deduce} it from conditions on the initial data. 

\begin{lemma}[Strong convergence of the conformal length]
\label{lemma-72}
Assume the uniform bound~\eqref{equa-bounds1} and the convergence of the initial length densities~\eqref{equa-convlambda}. Then, after extraction of a subsequence if necessary, there exists a limit length function $\ell^\sharp \in W^{1,\infty}(\Interval;L^1(0,1)) \cap L^\infty(\Interval;\BVac(0,1))$,
such that, setting
\bel{equa-877}
d_x\ell^\sharp(t)=\amdeux^\sharp(t,\cdot)\,dx,
\qquad
\ell^\sharp(t,x)=\int_0^x \amdeux^\sharp(t,y)\,dy,
\qquad
x \in (0, 1) ,
\ee
one has
\bel{equaj30}
\amdeux^\suit(t,\cdot)\to \amdeux^\sharp(t,\cdot)
\qquad \text{strongly in } L^1(\Sbb^1)
\quad \text{for every } t\in\Interval. 
\ee 
Moreover, if
\bel{equa-bounds2}
S_{\mathrm{length}}^\suit \to S_{\mathrm{length}}^\sharp
\qquad \text{strongly in } L^1(\Interval\times \Sbb^1),
\ee
then the limit satisfies
\bel{equa-810p}
\amdeux^\sharp(t,\cdot)
=
\mathring\amdeux^\sharp
+\int_{t_0}^t S_{\mathrm{length}}^\sharp(s,\cdot)\,ds
\qquad \text{in }L^1(\Sbb^1),
\ee
for every \(t\in\Interval\). In particular, $(\amdeux^\sharp)_t = S_{\mathrm{length}}^\sharp$ holds in the sense of distributions on \(\Interval\times\Sbb^1\).
\end{lemma}

In view of~\eqref{esuat-71}, the assumption~\eqref{equa-bounds2} will be satisfied as soon as the strong convergence of $(\Jpar^\suit, \Kpar^\suit, \Jbb^\suit\cdot\Jbb^\suit, \Omega^\suit)$ has been established.

\begin{proof}
We split the argument into three steps.

\bse
\noindent{\bf 1. \it Uniform time regularity.}
For each \(\suit\), the evolution equation~\eqref{esuat-71} yields, after integration between two times \(t,t'\in\Interval\),
\be
\amdeux^\suit(t,\cdot)-\amdeux^\suit(t',\cdot)
=
\int_{t'}^t S_{\mathrm{length}}^\suit(s,\cdot)\,ds
\qquad\text{in }L^1(\Sbb^1).
\ee
Hence, we find 
\be
\aligned
\bigl\|\amdeux^\suit(t,\cdot)-\amdeux^\suit(t',\cdot)\bigr\|_{L^1(\Sbb^1)}
& \le
\int_{\min\{t,t'\}}^{\max\{t,t'\}}
\bigl\|S_{\mathrm{length}}^\suit(s,\cdot)\bigr\|_{L^1(\Sbb^1)}\,ds
  \leq C |t-t'|,
\endaligned
\ee
where \(C>0\) is independent of \(\suit\), by~\eqref{equa-bounds1}. Therefore, the family $\bigl(\amdeux^\suit\bigr)_\suit$ is equi-Lipschitz from \(\Interval\) into \(L^1(\Sbb^1)\).
\ese


\noindent{\bf 2. \it Construction of the limit and strong convergence.}
Since \(\ell^\suit\in L^\infty(\Interval;\BVac(\Sbb^1))\), the densities
\(\amdeux^\suit=\partial_x\ell^\suit\) are uniformly bounded in \(L^\infty(\Interval;L^1(\Sbb^1))\). Hence, for every fixed \(t\in\Interval\), the sequence
\(\bigl(\amdeux^\suit(t,\cdot)\bigr)_\suit\) is relatively compact in \(L^1(\Sbb^1)\).
Applying a diagonal extraction at rational times and using the equi-Lipschitz estimate above, we obtain a subsequence and a function
$
\amdeux^\sharp \in W^{1,\infty}(\Interval;L^1(\Sbb^1))
$
such that
\be
\amdeux^\suit(t,\cdot)\to \amdeux^\sharp(t,\cdot)
\qquad \text{strongly in }L^1(\Sbb^1)
\quad\text{for every }t\in\Interval.
\ee
The convergence of the initial data~\eqref{equa-convlambda} ensures that
$
\amdeux^\sharp(t_0,\cdot)=\mathring\amdeux^\sharp.
$
We now define the associated length function by \eqref{equa-877}.  Since \(\amdeux^\sharp(t,\cdot)\in L^1(\Sbb^1)\) for every \(t\), we have
$
\ell^\sharp(t,\cdot)\in \BVac(\Sbb^1)$
and
$
\ell^\sharp\in L^\infty(\Interval;\BVac(\Sbb^1))$, 
which proves~\eqref{equaj30}.

\vskip.3cm


\bse
\noindent{\bf 3. \it Identification of the limit equation under~\eqref{equa-bounds2}.}
Assume now that~\eqref{equa-bounds2} holds. For each \(\suit\), the representation formula \eqref{equa-810p}
holds in \(L^1(\Sbb^1)\). Passing to the limit in \(L^1(\Sbb^1)\), for every fixed \(t\in\Interval\), we obtain
\be
\amdeux^\sharp(t,\cdot)
=
\mathring\amdeux^\sharp
+\int_{t_0}^t S_{\mathrm{length}}^\sharp(s,\cdot)\,ds
\qquad \text{in }L^1(\Sbb^1).
\ee
Therefore we have 
$
(\amdeux^\sharp)_t = S_{\mathrm{length}}^\sharp
$
in the sense of distributions on \(\Interval\times\Sbb^1\). This completes the proof.
\ese
\end{proof}


The argument above yields absolute continuity of the limit conformal length measure with respect to the Lebesgue measure, together with absolute continuity of the associated length function. More precisely, for every \(t\in\Interval\),
\be
d_x\ell^\sharp(t)=\amdeux^\sharp(t,\cdot)\,dx
\qquad\text{with}\qquad
\amdeux^\sharp(t,\cdot)\in L^1(\Sbb^1),
\ee
so that the limit conformal length measure is absolutely continuous with respect to the Lebesgue measure on \(\Sbb^1\). Consequently, the associated length function \eqref{equa-877}
belongs to \(W^{1,1}(\Sbb^1)\subset \BVac(\Sbb^1)\) for every \(t\in\Interval\), and therefore
$
\ell^\sharp \in L^\infty(\Interval;\BVac(\Sbb^1)).
$
On the other hand, no \(BV\)-regularity is obtained here for the density \(\amdeux^\sharp\) itself, since no uniform control is available on the spatial derivatives \((\amdeux^\suit)_x\).
 

\subsection{Analysis of the lapse function}
\label{section=8-3}

We next analyze the lapse function through its evolution and constraint equations, now allowing for the presence of stress-energy correctors. For the sequence under consideration, these equations take the form
\bel{equa-820}
\aligned
(\log\Omega^\suit)_t = (\amdeux^\suit)^{-1} \Xi^\suit_{11} ,
\qquad
(\log\Omega^\suit )_x = \Xi^\suit_{01} ,
\endaligned
\ee
where
\bel{define-Xi}
\aligned
\Xi^\suit_{11}
& \coloneqq \Pi^\suit_{11}
- \frac{\amdeux^\suit}{4t}
+ \frac{t}{2} \amdeux^\suit \Bigl(
\Ebf_0(\Pbb^\suit,\Qbb^\suit,\Jperp{}^\suit)
- \Ebf_0(\Kpar^\suit,\Jpar^\suit)
+ \frac{q_\Jbb^\suit}{2}\,\Jbb^\suit \cdot \Jbb^\suit
\Bigr) (\Omega^\suit)^2 ,
\\
\Xi^\suit_{01} & \coloneqq \Pi^\suit_{01}
- \frac{t}{2}\,\Mbf^{01}(\Pbb^\suit,\Qbb^\suit,\Jperp{}^\suit)\,
(\Omega^\suit)^2 \amdeux^\suit.
\endaligned
\ee
Here, the first identity is understood on timelike slices and the second one on spacelike slices, in the sense of measures.  
In the first equation, the product \((\amdeux^\suit)^{-1} \Pi^\suit_{11}\) is understood as the product of a continuous function with a measure on each timelike slice.

The bounds assumed in~\eqref{equa-ConditionTheo-0}, together with the fact that $\Pi^\suit$ is a stress-energy corrector in the sense of~\autoref{weakdefinitionT2-deux} imply that
the sequence is uniformly bounded in the sense
\bel{equa-bounds3a}
\Xi^\suit \in L^\infty(\Interval,\Meas(\Tbb^3))
\cap L^\infty(\Sbb^1,\Meas(\Interval\times\Tbb^2)).
\ee
Thus, for the lapse variable, one expects only weak-\(*\) compactness in \(BV\), together with weak-\(*\) convergence of the corrector measures. This is stated precisely in the next lemma.

\begin{lemma}[Compactness of the lapse and convergence of the correctors]
\label{lemma-73}
Assume that~\eqref{equa-bounds3a} hold, and normalize the lapse by
\be
\Omega^\suit(t_0,0)=1.
\ee
Assume also that the correctors \(\Pi^\suit\) are uniformly bounded in
$
L^\infty(\Interval,\Meas(\Tbb^3))
\cap L^\infty(\Sbb^1,\Meas(\Interval\times\Tbb^2)).
$
Then the sequence \(\log\Omega^\suit\) is uniformly bounded in
\be
L^\infty(\Interval\times\Sbb^1)
\cap L^\infty(\Interval;BV(\Sbb^1))
\cap L^\infty(\Sbb^1;BV(\Interval)).
\ee
Hence, after extraction of a subsequence if necessary, there exist
\be
\log\Omega^\sharp
\in
L^\infty(\Interval\times\Sbb^1)
\cap L^\infty(\Interval;BV(\Sbb^1))
\cap L^\infty(\Sbb^1;BV(\Interval))
\ee
and $\Xi^\sharp$ such that
\be
\aligned
\log\Omega^\suit & \to \log\Omega^\sharp
\qquad
&& \text{pointwise a.e. on }\Interval\times\Sbb^1,
\\
\Xi^\suit  & \overset{*}{\rightharpoonup} \Xi^\sharp
\, 
&&\text{ weakly-\(*\) in $\Meas(\Tbb^3)$ for a.e. } t \in \Interval,
\\
\Xi^\suit  & \overset{*}{\rightharpoonup} \Xi^\sharp
\, 
&&\text{ weakly-\(*\) in $\Meas(\Interval\times\Tbb^2))$ for a.e. } x \in \Sbb^1.
\endaligned
\ee
In addition, the limit satisfies
\be
(\log\Omega^\sharp)_t=(\amdeux^\sharp)^{-1}\Xi^\sharp_{11},
\qquad
(\log\Omega^\sharp)_x=\Xi^\sharp_{01},
\ee
in the sense of measures on timelike and spacelike slices, respectively. 
\end{lemma}

\begin{proof}
We divide the proof into two steps.

\bse
\noindent{\bf 1. \it Uniform bounds and \(BV\) control.}
Fix \(\suit\). Since \(\Omega^\suit(t_0,0)=1\), we have \(\log\Omega^\suit(t_0,0)=0\). Using the lapse equations in the sense of measures, and integrating first along the segment \([t_0,t]\times\{0\}\) and then along \(\{t\}\times[0,x]\), we obtain the bound
\be
\aligned
|\log\Omega^\suit(t,x)|
& \le
\|{(\amdeux^\suit)^{-1}}\Xi^\suit_{11}(0)\|_{\Meas(\Interval\times\Tbb^2)}
 + \|{(\amdeux^\suit)^{-1}}\Xi^\suit_{01}(t)\|_{\Meas(\Tbb^3)},
\endaligned
\ee
for every \((t,x)\in\Interval\times\Sbb^1\). By~\eqref{equa-bounds3a} and the uniform measure bounds on \(\Xi^\suit\), the right-hand side is bounded independently of \(\suit\). Hence
\be
\sup_{\suit\in(0,1]}\|\log\Omega^\suit\|_{L^\infty(\Interval\times\Sbb^1)}<+\infty.
\ee

Next, for each fixed \(t\), the map \(x\mapsto \log\Omega^\suit(t,x)\) has distributional derivative
$(\log\Omega^\suit)_x(t,\cdot) = \Xi^\suit_{01}(t)$,
hence
\be
\aligned
\Var_x\bigl(\log\Omega^\suit(t,\cdot)\bigr)
& \le \| \Xi^\suit_{01}(t)\|_{\Meas(\Tbb^3)}.
\endaligned
\ee
Since \(\amdeux^\suit\) is uniformly bounded and \(\Xi^\suit\) is uniformly bounded as a measure, it follows that
\be
\sup_{\suit\in(0,1]}\sup_{t\in\Interval}
\Var_x\bigl(\log\Omega^\suit(t,\cdot)\bigr)<+\infty.
\ee
Therefore we have 
$
\log\Omega^\suit \in L^\infty(\Interval;BV(\Sbb^1))
$
uniformly in \(\suit\).

Similarly, for each fixed \(x\), the map \(t\mapsto \log\Omega^\suit(t,x)\) has distributional derivative
given by $(\log\Omega^\suit)_t(\cdot,x) = \Xi^\suit_{11}(x)$,
so that
\be
\Var_t\bigl(\log\Omega^\suit(\cdot,x)\bigr)
\le \|{(\amdeux^\suit)^{-1}}\Xi^\suit_{11}(x)\|_{\Meas(\Interval\times\Tbb^2)}.
\ee
Hence we find 
\be
\sup_{\suit\in(0,1]}\sup_{x\in\Sbb^1}
\Var_t\bigl(\log\Omega^\suit(\cdot,x)\bigr)<+\infty,
\ee
and therefore
$
\log\Omega^\suit \in L^\infty(\Sbb^1;BV(\Interval))
$
uniformly in \(\suit\).
\ese

\vskip.3cm

\bse
\noindent{\bf 2. \it Weak-\(*\) compactness and pointwise convergence.}
By Helly's theorem on each spacelike and timelike slice, together with a diagonal extraction argument, we may extract a subsequence such that
\be
\log\Omega^\suit(t,x)\to \log\Omega^\sharp(t,x)
\qquad\text{for a.e. }(t,x)\in\Interval\times\Sbb^1.
\ee
The uniform bounds established above pass to the limit, and therefore
\be
\log\Omega^\sharp
\in
L^\infty(\Interval\times\Sbb^1)
\cap L^\infty(\Interval;BV(\Sbb^1))
\cap L^\infty(\Sbb^1;BV(\Interval)).
\ee
Moreover, after extraction of a subsequence if necessary,
$\Xi^\suit \overset{*}{\rightharpoonup} \Xi^\sharp$
weakly-\(*\) in
$
L^\infty(\Interval,\Meas(\Tbb^3))
\cap L^\infty(\Sbb^1,\Meas(\Interval\times\Tbb^2)).
$
Exponentiating, we also obtain
\be
\Omega^\suit \to \Omega^\sharp
\qquad\text{for a.e. }(t,x)\in\Interval\times\Sbb^1.
\ee
\ese
\end{proof}


\subsection{Analysis of the twist functions}
\label{section=8-4}

We now consider the evolution and constraint equations
\eqref{eq:T2-9101112-evol-def-a}--\eqref{eq:T2-9101112-evol-def-b}
and \eqref{eq:221a}--\eqref{eq:221b} for the sequence of twist functions, namely
\be
\aligned
|t|^{-3/2} \bigl(  |t|^{3/2} K_2^\suit \bigr)_t
&= S^\suit_{K_2,0} \coloneqq \bigl( J_1^\suit J_2^\suit -  P_0^\suit K^\suit_2/2 \bigr)\Omega^\suit,
\\
|t|^{-3/2} \bigl( |t|^{3/2} K^\suit_3 \bigr)_t
&= S^\suit_{K_3,0} \coloneqq \bigl( J^\suit_1 J^\suit_3 + P^\suit_0 K^\suit_3/2 - Q^\suit_0 K^\suit_2 \bigr)\Omega^\suit,
\\
( K^\suit_2 )_x
&= S^\suit_{K_2,1} \coloneqq
 \bigl( J^\suit_0 J^\suit_2 - P^\suit_1 K^\suit_2/2 \bigr)\, \Omega^\suit \amdeux^\suit,
\\
( K^\suit_3 )_x
&= S^\suit_{K_3,1} \coloneqq
\bigl( J^\suit_0 J^\suit_3  - Q^\suit_1 K^\suit_2 + P^\suit_1 K^\suit_3/2 \bigr)\, \Omega^\suit \amdeux^\suit.
\endaligned
\ee 
By the bounds assumed in~\eqref{equa-ConditionTheo-0}, together with the uniform bounds available on
\(\Omega^\suit\), \(\amdeux^\suit\), \(\Pbb^\suit\), \(\Qbb^\suit\), \(\Jbb^\suit\), and \(K_m^\suit\),
the source terms satisfy
\bel{equa-bounds7}
\aligned
& \sup_{\suit \in (0,1]}
\sup_{t \in [t_0,t_1]}
\int_{\Sbb^1}
\Bigl(
|S^\suit_{K_2,1}(t,\cdot)|
+
|S^\suit_{K_3,1}(t,\cdot)|
\Bigr)\,dx
<+\infty,
\\
& \sup_{\suit \in (0,1]}
\sup_{x \in \Sbb^1}
\int_{t_0}^{t_1}
\Bigl(
|S^\suit_{K_2,0}(\cdot,x)|
+
|S^\suit_{K_3,0}(\cdot,x)|
\Bigr)\,dt
<+\infty.
\endaligned
\ee
In particular, the twist functions are uniformly bounded in \(L^\infty(\Interval\times\Sbb^1)\), while their spatial and time derivatives are uniformly controlled in \(L^1\) along spacelike and timelike slices, respectively. This yields the following compactness statement.

\begin{lemma}[Strong convergence of the twist functions]
\label{lemma-76}
Assume that the uniform bounds~\eqref{equa-bounds7} hold, and that the initial twist functions satisfy
\be
K_m^\suit(t_0,\cdot)\to K_m^\sharp(t_0,\cdot)
\quad\text{strongly in }L^1(\Sbb^1),
\qquad m=2,3.
\ee
Then, after extraction of a subsequence if necessary, there exist limit functions
\be
K_m^\sharp
\in
L^\infty(\Interval\times\Sbb^1)
\cap
L^\infty(\Interval;\BVac(\Sbb^1))
\cap
L^\infty(\Sbb^1;\BVac(\Interval)),
\qquad m=2,3,
\ee
such that
\be
K_m^\suit \to K_m^\sharp
\qquad\text{pointwise a.e. on }\Interval\times\Sbb^1,
\qquad m=2,3.
\ee
If, in addition, the source terms converge strongly in the sense that
\be
\label{equa-bounds7-strong}
\aligned
S^\suit_{K_m,1} &\to S^\sharp_{K_m,1}
\qquad\text{strongly in } L^\infty(\Interval;L^1(\Sbb^1)),
\\
S^\suit_{K_m,0} &\to S^\sharp_{K_m,0}
\qquad\text{strongly in } L^\infty(\Sbb^1;L^1(\Interval)),
\endaligned
\qquad m=2,3,
\ee
then
\be
K_m^\suit \to K_m^\sharp
\qquad\text{uniformly on }\Interval\times\Sbb^1,
\qquad m=2,3,
\ee
and the limit satisfies, in the sense of distributions on \(\Interval\times\Sbb^1\), 
\be
|t|^{-3/2}\bigl(|t|^{3/2}K_m^\sharp\bigr)_t
= S^\sharp_{K_m,0},
\qquad
(K_m^\sharp)_x
= S^\sharp_{K_m,1},
\qquad m=2,3. 
\ee 
\end{lemma}

\begin{proof} 
\bse
We proceed similarly as in Lemma~\ref{lemma-73}, treating each twist component separately.
Fix \(m=2,3\). By assumption, the sequence \(K_m^\suit\) is uniformly bounded in
$
L^\infty(\Interval\times\Sbb^1).
$
Moreover, the bound~\eqref{equa-bounds7} implies that, for each fixed \(t\in\Interval\), the map
$
x\mapsto K_m^\suit(t,x)
$
is absolutely continuous on \(\Sbb^1\), with uniformly bounded total variation. Likewise, for each fixed \(x\in\Sbb^1\), the map
$
t\mapsto |t|^{3/2}K_m^\suit(t,x)
$
is absolutely continuous on \(\Interval\), again with uniformly bounded total variation. Since the factor \(|t|^{3/2}\) is smooth and uniformly bounded above and below on \(\Interval\), it follows that the same property holds for \(t\mapsto K_m^\suit(t,x)\). Hence, we find 
\be
K_m^\suit
\in
L^\infty(\Interval;\BVac(\Sbb^1))
\cap
L^\infty(\Sbb^1;\BVac(\Interval))
\ee
together with uniform bounds in \(\suit\). Applying Helly's theorem on spacelike and timelike slices, together with a diagonal extraction argument, we obtain, after passing to a subsequence if necessary, a limit function
\be
K_m^\sharp
\in
L^\infty(\Interval\times\Sbb^1)
\cap
L^\infty(\Interval;\BVac(\Sbb^1))
\cap
L^\infty(\Sbb^1;\BVac(\Interval))
\ee
such that
$
K_m^\suit \to K_m^\sharp
\qquad\text{for a.e. }(t,x)\in\Interval\times\Sbb^1.
$
\ese

Assume now that~\eqref{equa-bounds7-strong} holds. Integrating the equations for \(K_m^\suit\) in space and in time, and using the strong convergence of the initial data, we obtain the corresponding representation formulas for \(K_m^\suit\) and \(K_m^\sharp\). Subtracting these formulas and arguing exactly as in Lemma~\ref{lemma-73}, we deduce that
\bse
\be
K_m^\suit \to K_m^\sharp
\qquad\text{uniformly on }\Interval\times\Sbb^1.
\ee
Passing to the limit in the distributional identities then yields
\be
|t|^{-3/2}\bigl(|t|^{3/2}K_m^\sharp\bigr)_t
= S^\sharp_{K_m,0},
\qquad
(K_m^\sharp)_x
= S^\sharp_{K_m,1}
\ee
in the sense of distributions on \(\Interval\times\Sbb^1\). The stated \(\BVac\) regularity of \(K_m^\sharp\) follows immediately. This completes the proof.
\ese
\end{proof}


\subsection{Analysis of the essential geometry variables} 
\label{section=8-5} 
 
\subsubsection{Weak convergence for general data}
 
We now turn to the essential geometric variables \(\Pbb^\suit\) and \(\Qbb^\suit\), which satisfy the first-order system
\begin{subequations}\label{eq:T2-1234-def-r2}
\begin{align}
\divuntrois_\suit(\Pbb^\suit)
&= S^\suit_{\Pbb},
\qquad
\curluntrois_\suit \bigl(|t|^{-1} \Pbb^\suit\bigr)
= 0,
\label{eq:T2-1234-def-a-r2}
\\
\divuntrois_\suit(\Qbb^\suit)
&= S^\suit_{\Qbb},
\qquad
\curluntrois_\suit \bigl(|t|^{-1} \Qbb^\suit\bigr)
= T^\suit_{\Qbb},
\label{eq:T2-1234-def-b-r2}
\end{align}
\end{subequations}
where the subscript \(\suit\) indicates the divergence and curl operators defined with respect to the metric \(\guntrois{}^\suit\), and where
\be
\aligned
S^\suit_{\Pbb}
& := \Qbb^\suit \cdot \Qbb^\suit
- \frac12 \Re\bigl((\Kpar^\suit)^2 + (\Jpar^\suit)^2\bigr),
\\
S^\suit_{\Qbb}
& := - \Pbb^\suit \cdot \Qbb^\suit
- \frac12 \Im\bigl((\Kpar^\suit)^2 + (\Jpar^\suit)^2\bigr),
\qquad
T^\suit_{\Qbb}
:= |t|^{-1}\,\Pbb^\suit \wedge \Qbb^\suit.
\endaligned
\ee
By the bounds assumed in~\eqref{equa-ConditionTheo-0}, the sequence \((\Pbb^\suit,\Qbb^\suit)\) is uniformly bounded in the natural \(L^2\) norms, both on spacelike and timelike slices, namely
\bel{equa:jdk13}
\limsup_{\suit \to 0} \Bigl( 
\sup_{t\in\Interval}
\|(\Pbb^\suit,\Qbb^\suit)(t)\|_{L^2(\Sbb^1,\dVtrois{}^\suit)}
+ 
\sup_{x\in\Sbb^1}
\|(\Pbb^\suit,\Qbb^\suit)(\cdot,x)\|_{L^2(\Interval,\dVundeux{}^\suit)}
\Bigr)
<+\infty.
\ee
Moreover, the source terms satisfy the uniform bounds
\bel{equa-bounds5} 
\aligned 
& \sup_{\suit \in (0,1]}
\sup_{t \in [t_0,t_1]}
\int_{\Sbb^1}
\Bigl(
|S_{\Pbb}^\suit(t,\cdot)|
+
|S_{\Qbb}^\suit(t,\cdot)|
+
|T^\suit_{\Qbb}(t,\cdot)|
\Bigr)\,dx
< + \infty,
\\
& \sup_{\suit \in (0,1]}
\sup_{x \in \Sbb^1}
\int_{\Interval}
\Bigl(
|S_{\Pbb}^\suit(\cdot,x)|
+
|S_{\Qbb}^\suit(\cdot,x)|
+
|T^\suit_{\Qbb}(\cdot,x)|
\Bigr)\,dt
< + \infty.
\endaligned
\ee
We first extract weak limits for \(\Pbb^\suit\) and \(\Qbb^\suit\), and then identify the nonlinear source terms by establishing the convergence of the relevant quadratic expressions.

\begin{lemma}[Weak convergence of the essential geometry]
\label{lemma-74}
Assume that the uniform bounds \eqref{equa:jdk13} and~\eqref{equa-bounds5} hold. Then, after extraction of a subsequence if necessary, there exist limit fields
$
\Pbb^\sharp,\Qbb^\sharp
\in
L^\infty\bigl(\Interval;L^2(\Sbb^1)\bigr)
\cap
L^\infty\bigl(\Sbb^1;L^2(\Interval)\bigr)
$
such that
\be
\Pbb^\suit \rightharpoonup \Pbb^\sharp,
\qquad
\Qbb^\suit \rightharpoonup \Qbb^\sharp
\qquad \text{weakly in } L^2_{\loc}(\Interval\times\Sbb^1).
\ee
After extraction of a further subsequence if necessary, the source terms satisfy
\be
S^\suit_{\Pbb} \overset{*}{\rightharpoonup} S^\sharp_{\Pbb},
\qquad
S^\suit_{\Qbb} \overset{*}{\rightharpoonup} S^\sharp_{\Qbb},
\qquad
T^\suit_{\Qbb} \overset{*}{\rightharpoonup} T^\sharp_{\Qbb},
\ee
weakly-\(*\) in \(\Meas(\Interval\times\Sbb^1)\). Moreover, the limit fields satisfy
\be
\label{eq:limit-PQ-weak}
\aligned
\divuntrois_\sharp(\Pbb^\sharp)
&= S^\sharp_{\Pbb},
\qquad
& \curluntrois_\sharp\bigl(|t|^{-1}\Pbb^\sharp\bigr)=0,
\\
\divuntrois_\sharp(\Qbb^\sharp)
&= S^\sharp_{\Qbb},
\qquad
& \curluntrois_\sharp\bigl(|t|^{-1}\Qbb^\sharp\bigr)=T^\sharp_{\Qbb},
\endaligned
\ee
in the sense of distributions, where the operators are defined with the limiting geometry.

Furthermore, the following quadratic convergence properties hold in the sense of distributions on \(\Interval\times\Sbb^1\):
\be
\label{eq:quad-conv-PQ}
\Pbb^\suit \cdot \Pbb^\suit \to \Pbb^\sharp \cdot \Pbb^\sharp, 
\qquad 
\Qbb^\suit \cdot \Qbb^\suit \to \Qbb^\sharp \cdot \Qbb^\sharp, 
\qquad 
\Pbb^\suit \cdot \Qbb^\suit \to \Pbb^\sharp \cdot \Qbb^\sharp,
\qquad 
\Pbb^\suit \wedge \Qbb^\suit \to \Pbb^\sharp \wedge \Qbb^\sharp.
\ee
As a consequence, one has 
\be
S^\sharp_{\Pbb}
=
\Qbb^\sharp \cdot \Qbb^\sharp
-\frac12 \Re\bigl((\Kpar^\sharp)^2+(\Jpar^\sharp)^2\bigr),
\ee
\be
S^\sharp_{\Qbb}
=
-\Pbb^\sharp \cdot \Qbb^\sharp
-\frac12 \Im\bigl((\Kpar^\sharp)^2+(\Jpar^\sharp)^2\bigr),
\qquad
T^\sharp_{\Qbb}
=
|t|^{-1}\,\Pbb^\sharp\wedge\Qbb^\sharp.
\ee
In particular, the equations for the essential geometry are recovered in the limit.
\end{lemma}


\begin{proof}
{\bf 1. \it Weak compactness.}
By the uniform bound~\eqref{equa:jdk13}, the sequences \(\Pbb^\suit\) and \(\Qbb^\suit\) are bounded in \(L^2_{\loc}(\Interval\times\Sbb^1)\). Therefore, after extraction of a subsequence if necessary, they converge weakly in \(L^2_{\loc}\) to some limits, denoted by \(\Pbb^\sharp\) and \(\Qbb^\sharp\), respectively.

Moreover, the bounds~\eqref{equa-bounds5} imply that the source terms \(S^\suit_{\Pbb}\), \(S^\suit_{\Qbb}\), and \(T^\suit_{\Qbb}\) are uniformly bounded in the space of Radon measures on \(\Interval\times\Sbb^1\). Hence, after extraction of a further subsequence if necessary, they converge weakly-* in \(\Meas(\Interval\times\Sbb^1)\) to some limits, denoted by \(S^\sharp_{\Pbb}\), \(S^\sharp_{\Qbb}\), and \(T^\sharp_{\Qbb}\).

\medskip

\bse
\noindent {\bf 2. \it Passage to the limit in the linear equations.}
By Lemmas~\ref{lemma-72} and~\ref{lemma-73}, the geometric coefficients entering the operators \(\divuntrois_\suit\) and \(\curluntrois_\suit\) converge strongly to their limits. Therefore, passing to the limit in~\eqref{eq:T2-1234-def-r2} yields, in the sense of distributions,
\be
\aligned
\divuntrois_\sharp \bigl( \Pbb^\sharp \bigr)
& = S^\sharp_\Pbb,
\qquad
& \curluntrois_\sharp \Bigl( |t|^{-1} \Pbb^\sharp \Bigr) = 0,
\\
\divuntrois_\sharp \bigl( \Qbb^\sharp \bigr)
& = S^\sharp_\Qbb,
\qquad
& \curluntrois_\sharp \Bigl( |t|^{-1} \Qbb^\sharp \Bigr) = T^\sharp_\Qbb. 
\endaligned
\ee

\medskip

\noindent {\bf 3. \it Convergence of the quadratic terms.}
We now apply the compensated compactness lemma stated just below to the pairs
$
(\Pbb^\suit,\Pbb^\suit)$, 
$
(\Qbb^\suit,\Qbb^\suit)$, 
$
(\Pbb^\suit,\Qbb^\suit).
$
The assumptions of that lemma are satisfied thanks to the \(L^2\)-bounds on the fields, together with the measure bounds on their divergence and curl given by~\eqref{eq:T2-1234-def-r2} and~\eqref{equa-bounds5}. Hence we obtain, in the sense of distributions,
\be
\Pbb^\suit \cdot \Pbb^\suit \to \Pbb^\sharp \cdot \Pbb^\sharp,
\qquad
\Qbb^\suit \cdot \Qbb^\suit \to \Qbb^\sharp \cdot \Qbb^\sharp,
\qquad
\Pbb^\suit \cdot \Qbb^\suit \to \Pbb^\sharp \cdot \Qbb^\sharp,
\ee
as well as
$
\Pbb^\suit \wedge \Qbb^\suit \to \Pbb^\sharp \wedge \Qbb^\sharp.
$

\medskip

\noindent {\bf 4. \it Identification of the limit source terms.}
Since \(\Kpar^\suit\to\Kpar^\sharp\) and \(\Jpar^\suit\to\Jpar^\sharp\) almost everywhere, and since these sequences are uniformly bounded in \(L^\infty\), we also have
\be
\Re\bigl((\Kpar^\suit)^2+(\Jpar^\suit)^2\bigr)
\to
\Re\bigl((\Kpar^\sharp)^2+(\Jpar^\sharp)^2\bigr),
\ee
\be
\Im\bigl((\Kpar^\suit)^2+(\Jpar^\suit)^2\bigr)
\to
\Im\bigl((\Kpar^\sharp)^2+(\Jpar^\sharp)^2\bigr)
\ee
strongly in \(L^1_{\loc}\). Combining these convergences with~\eqref{eq:quad-conv-PQ}, we conclude that
\be
S^\sharp_{\Pbb}
=
\Qbb^\sharp \cdot \Qbb^\sharp
-\frac12 \Re\bigl((\Kpar^\sharp)^2+(\Jpar^\sharp)^2\bigr),
\ee
\be
S^\sharp_{\Qbb}
=
-\Pbb^\sharp \cdot \Qbb^\sharp
-\frac12 \Im\bigl((\Kpar^\sharp)^2+(\Jpar^\sharp)^2\bigr),
\qquad
T^\sharp_{\Qbb}
=
|t|^{-1}\,\Pbb^\sharp\wedge\Qbb^\sharp.
\ee
This proves that the equations for the essential geometry are recovered in the limit.
\ese
\end{proof}


\subsubsection{Compensated compactness argument}

We recall our notation for a one-form \(\Xbb\) with components \((X_0,X_1)\), with contravariant components
$
X^0=-X_0$ and $
X^1=X_1.
$
We defined
\bel{equa-divN-rr} 
\aligned
\divN(\Xbb)
& \coloneqq \amdeux^{-1} \Omega^{-2}
\bigl(-(\amdeux \Omega X_0)_t + (\Omega X_1)_x\bigr),
\\
\divuntrois \Xbb & \coloneqq |t|^{-1} \divN(|t| \, \Xbb),
\endaligned
\ee
as well as
\bel{equa-jd45-rr}
\curl_\Ncal\Xbb 
\coloneqq \amdeux^{-1} \Omega^{-2}
\bigl((\amdeux \Omega X_1)_t - (\Omega X_0)_x\bigr),
\ee
\bel{equa-def-curl-rr}
\curluntrois \Xbb \coloneqq |t|^{-1} \curl_\Ncal( |t| \, \Xbb).
\ee

We now state the compensated compactness result that will be used repeatedly below. It is a direct consequence of the standard div-curl lemma of Murat and Tartar~\cite{Murat-1974,Tartar1,Tartar2}, combined with the uniform convergence of the geometric weights \(\amdeux^\suit\) and \(\Omega^\suit\) established in Lemmas~\ref{lemma-72} and~\ref{lemma-73}.

\begin{lemma}[Weighted div-curl lemma adapted to the geometry]
\label{lemdivcurl}
Let \(\Xbb^\suit=(X^\suit_0,X^\suit_1)\) and \(\Ybb^\suit=(Y^\suit_0,Y^\suit_1)\) be sequences of one-forms defined on \(\Interval\times\Sbb^1\) satisfying the following conditions. 

\bei
\item \(\Xbb^\suit\) and \(\Ybb^\suit\) are uniformly bounded in \(L^2(\Interval\times\Sbb^1)\), and
\be
\Xbb^\suit \rightharpoonup \Xbb^\sharp,
\qquad
\Ybb^\suit \rightharpoonup \Ybb^\sharp
\quad
\text{weakly in }L^2(\Interval\times\Sbb^1).
\ee

\item The coefficients \(\amdeux^\suit\) and \(\Omega^\suit\) converge uniformly on \(\Interval\times\Sbb^1\) to limits \(\amdeux^\sharp\) and \(\Omega^\sharp\), respectively, and remain uniformly bounded above and below by positive constants. 

\item The sequences \(\divuntrois_\suit \Xbb^\suit\) and \(\curluntrois_\suit \Ybb^\suit\) are relatively compact in \(H^{-1}(\Interval\times\Sbb^1)\).
\eei
\noindent Then the quadratic expression \(\Xbb^\suit\cdot\Ybb^\suit\) converges to \(\Xbb^\sharp\cdot\Ybb^\sharp\) in the sense of distributions, that is,
\be
\Xbb^\suit\cdot\Ybb^\suit
\to
\Xbb^\sharp\cdot\Ybb^\sharp
\qquad
\text{ in }\Dcal'(\Interval\times\Sbb^1), 
\ee
with the same conclusion for the wedge product:
\be
\Xbb^\suit\wedge\Ybb^\suit
\to
\Xbb^\sharp\wedge\Ybb^\sharp
\qquad
\text{ in }\Dcal'(\Interval\times\Sbb^1).
\ee
\end{lemma}

\begin{proof}
\bse
We introduce the weighted fields
\be
\widecheck \Xbb^\suit
\coloneqq
\bigl(|t|\amdeux^\suit\Omega^\suit X_0^\suit,\,
|t|\Omega^\suit X_1^\suit\bigr),
\qquad
\widecheck \Ybb^\suit
\coloneqq
\bigl(|t|\Omega^\suit Y_0^\suit,\,
|t|\amdeux^\suit\Omega^\suit Y_1^\suit\bigr).
\ee
Since \(\amdeux^\suit\) and \(\Omega^\suit\) are uniformly bounded above and below, and since \(\Xbb^\suit\) and \(\Ybb^\suit\) are uniformly bounded in \(L^2(\Interval\times\Sbb^1)\), it follows that the sequences \(\widecheck \Xbb^\suit\) and \(\widecheck \Ybb^\suit\) are uniformly bounded in \(L^2(\Interval\times\Sbb^1)\). Moreover, by the strong convergence of the weights, they converge weakly in \(L^2\) toward the corresponding weighted limits
\be
\widecheck \Xbb^\sharp
=
\bigl(|t|\amdeux^\sharp\Omega^\sharp X_0^\sharp,\,
|t|\Omega^\sharp X_1^\sharp\bigr),
\qquad
\widecheck \Ybb^\sharp
=
\bigl(|t|\Omega^\sharp Y_0^\sharp,\,
|t|\amdeux^\sharp\Omega^\sharp Y_1^\sharp\bigr).
\ee
\ese

Next, by the definition of the weighted divergence and curl operators, we have 
\be
\aligned
-(\widecheck X^\suit_0)_t + (\widecheck X^\suit_1)_x
& =
|t|\,\amdeux^\suit(\Omega^\suit)^2 \, \divuntrois_\suit \Xbb^\suit,
\\
(\widecheck Y^\suit_1)_t - (\widecheck Y^\suit_0)_x
& =
|t|\,\amdeux^\suit(\Omega^\suit)^2 \, \curluntrois_\suit \Ybb^\suit.
\endaligned
\ee
Since the coefficients
$
|t|\,\amdeux^\suit(\Omega^\suit)^2
$
converge strongly and remain uniformly bounded, multiplication by these coefficients preserves \(H^{-1}_{\loc}\)-compactness. Hence, by assumption, the Euclidean divergence of \(\widecheck \Xbb^\suit\) and the Euclidean curl of \(\widecheck \Ybb^\suit\) are relatively compact in \(H^{-1}_{\loc}(\Interval\times\Sbb^1)\).

\bse
We may therefore apply the standard div-curl lemma in two space-time dimensions and conclude that
\be
\widecheck \Xbb^\suit \cdot \widecheck \Ybb^\suit
\to
\widecheck \Xbb^\sharp \cdot \widecheck \Ybb^\sharp
\qquad
\text{in }\Dcal'(\Interval\times\Sbb^1).
\ee
Now, we find 
$
\widecheck \Xbb^\suit \cdot \widecheck \Ybb^\suit
=
|t|^2\amdeux^\suit(\Omega^\suit)^2\, \Xbb^\suit\cdot\Ybb^\suit,
$
and similarly for the limit. Since the factor
$
|t|^2\amdeux^\suit(\Omega^\suit)^2
$
converges strongly and remains uniformly bounded away from zero, division by this coefficient is legitimate in the sense of distributions. We thus obtain
\be
\Xbb^\suit\cdot\Ybb^\suit
\to
\Xbb^\sharp\cdot\Ybb^\sharp
\qquad
\text{in }\Dcal'(\Interval\times\Sbb^1).
\ee
The argument for the wedge product is identical, replacing the Euclidean scalar product by the determinant. This completes the proof.
\ese
\end{proof}


\subsection{Analysis and propagation of the corrector measures}
\label{section=8-6}

\subsubsection{Further structure}

We continue our study of the convergence of the essential geometry variables, with two distinct goals in mind.  
First, for a bounded sequence of Einstein--Euler flows with corrector, we expect a possible lack of strong convergence of the essential geometry variables \((\Pbb^\suit,\Qbb^\suit, \Pi^\suit)\) to generate a non-trivial corrector tensor \(\Pi^\sharp\), entering the lapse equations of the limit flow even if $\Pi^\suit$ are zero initially.   
Second, in the well-prepared regime, we expect that no such correction is produced and that the essential geometry converges strongly.

At this stage, we split the source term of the lapse equations treated in \autoref{section=8-3} into a contribution from the corrector and essential geometric variables on the one hand, and the remaining part
\be
\aligned
S^\suit_{\mathrm{lapse},0}
& = - \frac{1}{4t}
+ \frac{t}{2} \Bigl(
\Ebf_0(\Jperp{}^\suit)
- \Ebf_0(\Kpar^\suit,\Jpar^\suit)
+ \frac{q_\Jbb^\suit}{2}\,\Jbb^\suit \cdot \Jbb^\suit
\Bigr) (\Omega^\suit)^2 ,
\\
S^\suit_{\mathrm{lapse},1}
& = - \frac{t}{2}\,\Mbf^{01}(\Jperp{}^\suit)\,
(\Omega^\suit)^2 \amdeux^\suit ,
\endaligned
\ee
so that~\eqref{define-Xi} reads
\bel{define-Xi-again}
\aligned
\Xi^\suit_{11}
& = \Pi^\suit_{11}
+ \frac{t}{2} \Ebf_0(\Pbb^\suit,\Qbb^\suit) (\Omega^\suit)^2 \amdeux^\suit
+ \amdeux^\suit S^\suit_{\mathrm{lapse},0} ,
\\
\Xi^\suit_{01} & = \Pi^\suit_{01}
+ \frac{t}{2}\,\Ebf_1(\Pbb^\suit,\Qbb^\suit)\,
(\Omega^\suit)^2 \amdeux^\suit
+ S^\suit_{\mathrm{lapse},1} .
\endaligned
\ee
We are thus ready to \emph{define} the corrector $\Pi^\sharp$ that will ensure that the lapse equations hold in the limit:
\bel{define-Xi-again-sharp}
\aligned
\Pi^\sharp_{11}
& \coloneqq \Xi^\sharp_{11}
- \frac{t}{2} \Ebf_0(\Pbb^\sharp,\Qbb^\sharp) (\Omega^\sharp)^2 \amdeux^\sharp
- \amdeux^\sharp S^\sharp_{\mathrm{lapse},0} ,
\\
\Pi^\sharp_{01} & \coloneqq \Xi^\sharp_{01}
- \frac{t}{2}\,\Ebf_1(\Pbb^\sharp,\Qbb^\sharp)\, (\Omega^\sharp)^2 \amdeux^\sharp
- S^\sharp_{\mathrm{lapse},1} .
\endaligned
\ee
We recall that $\Xi^\suit\to\Xi^\sharp$ weakly-$*$ in measure.

\subsubsection{The corrector as a geometric effect}

Once we obtain strong convergence of the fluid variables, the corrector tensor \(\Pi^\sharp\) will arise from the failure of the quadratic expressions in \((\Pbb^\suit,\Qbb^\suit)\) to converge strongly. More precisely, the null-form structure implies that the combinations
\be
(P_0^\suit\pm P_1^\suit)^2+(Q_0^\suit\pm Q_1^\suit)^2
\ee
are non-negative and may, in the limit, generate positive Radon measures. This is the origin of the inequalities $\Pi^{00}{}^\sharp=\Pi^{11}{}^\sharp\geq 0$ and $\Pi^{00}{}^\sharp\geq |\Pi^{01}{}^\sharp|$ stated in~\eqref{eq=allconfor-infty-mod}.
The corrector thus measures the part of the geometric energy lost under weak convergence.

\begin{lemma}[Lapse equation in the limit]
  With the construction of the limit corrector $\Pi^\sharp$ above, the lapse equations hold, in the sense that
\be
\aligned
(\log\Omega^\sharp)_t
& = (\amdeux^\sharp)^{-1} \Pi^\sharp_{11}
- \frac{1}{4t}
+ \frac{t}{2} \Bigl(
\Ebf_0(\Pbb^\sharp,\Qbb^\sharp,\Jperp{}^\sharp)
- \Ebf_0(\Kpar^\sharp,\Jpar^\sharp)
+ \frac{q_\Jbb^\sharp}{2}\,\Jbb^\sharp \cdot \Jbb^\sharp
\Bigr) (\Omega^\sharp)^2 ,
\\
(\log\Omega^\sharp )_x
& = \Pi^\sharp_{01}
- \frac{t}{2}\,\Mbf^{01}(\Pbb^\sharp,\Qbb^\sharp,\Jperp{}^\sharp)\,
(\Omega^\sharp)^2 \amdeux^\sharp.
\endaligned
\ee
In addition, if
\bel{equa-bounds3c}
\aligned
& S^\suit_{\mathrm{lapse},0}\to S^\sharp_{\mathrm{lapse},0}
\quad\text{strongly in } L^\infty(\Sbb^1;L^1(\Interval)),
\\
& S^\suit_{\mathrm{lapse},1}\to S^\sharp_{\mathrm{lapse},1}
\quad\text{strongly in } L^\infty(\Interval;L^1(\Sbb^1)),
\endaligned
\ee
then the limit satisfies
\be
\Pi^{00}{}^\sharp=\Pi^{11}{}^\sharp\geq 0,
\qquad
\Pi^{00}{}^\sharp\geq |\Pi^{01}{}^\sharp|.
\ee
\end{lemma}

\subsubsection{Divergence inequality on the corrector}

A further key point is to show that the corrector tensor \(\Pi^\sharp\) satisfies the divergence inequalities
\bel{equag3930}
\divuntrois{}^\sharp\bigl(\Omega^\sharp{}^{-1}\amdeux^\sharp{}^{-1}\Pi^\sharp{}^{\pm\bullet}\bigr)\leq 0 .
\ee
Formally, these identities express that the geometric energy lost under weak convergence should itself obey a homogeneous balance law, analogous to the propagation of an effective null fluid with possibly additional decay.
In particular, given that $\Pi^\sharp{}^{\pm 0}\geq 0$, if it vanishes at the initial time, it remains zero for all times.

Let us establish~\eqref{equag3930} for $\Pi^\sharp{}^{+\bullet} = \Pi^\sharp{}^{0\bullet} + \Pi^\sharp{}^{1\bullet}$ as the other sign is treated identically.
The key idea is that the essential geometric variables satisfy an energy identity, which we explicited in~\eqref{energy-PQ-quartic}, namely
\bel{eqdivMdeltageom}
\divuntrois_\suit\Bigl( t^{-1}\Omega^\suit{}^{-1}\Mbf^{\pm\bullet}(0,0,\Pbb^\suit,\Qbb^\suit)\Bigr) 
\leq \Deltageom^{\pm}(\Phi^\suit, \Psi^\suit),
\ee
where $\Deltageom^{\pm}(\Phi^\suit, \Psi^\suit)$ consists of terms that will be shown to converge strongly and additional terms that are squares such as $(\Pbb^\suit\cdot\Pbb^\suit)^2$.
The latter terms do not converge to their natural limit  $(\Pbb^\sharp\cdot\Pbb^\sharp)^2$, but receive additional positive contributions  $(\Pbb^\suit\cdot\Pbb^\suit-\Pbb^\sharp\cdot\Pbb^\sharp)^2$.  Likewise the left-hand side of~\eqref{eqdivMdeltageom} does not converge to its natural limit but rather has an additional contribution that leads to the divergence of~$\Pi^\sharp$.

The identity~\eqref{eqdivMdeltageom} is derived in a weak regularity setting by introducing a characteristic coordinate~$z$.
For any given $x_0\in\Sbb^1$, we consider the equation for a characteristic curve $\undt=\undt(x)$ starting at that point.  It is obtained by solving
\be
\frac{d}{dx} \undt(x) = \amdeux^\suit(\undt(x), x) , \qquad \undt(x_0) = t_0 .
\ee
The function $x\mapsto\undt(x)$ is $\BVac$, hence by inverting the map, we construct a coordinate~$z$ such that in the coordinates $(z,x)$, the constant-$z$ curves are the above characteristics.
The divergence and curl equations for $(\Pbb^\suit,\Qbb^\suit)$, once expressed in the $(z,x)$ coordinates, are ODEs of the form $\del_z \Pbb^\suit = \text{sources}$, which easily give the desired equations on $\del_z (\Pbb^\suit)^2$.

\begin{lemma}[Divergence equation of the corrector]
If the fluid contributions converge as stated in~\eqref{equa-bounds3c} and the non-squared terms in $\Deltageom^{\pm}(\Phi, \Psi)$ converge strongly, namely
\be
\aligned
& (P_0^\suit\pm P_1^\suit) \Re(\Kpar^\suit{}^2+\Jpar^\suit{}^2)+(Q^\suit_0-Q^\suit_1)\Im(\Kpar^\suit{}^2+\Jpar^\suit{}^2)\Bigr)
\\
& \quad \to (P_0^\sharp\pm P_1^\sharp) \Re(\Kpar^\sharp{}^2+\Jpar^\sharp{}^2)+(Q^\sharp_0-Q^\sharp_1)\Im(\Kpar^\sharp{}^2+\Jpar^\sharp{}^2)\Bigr) ,
\endaligned
\ee
then the divergence inequalities for $\Pi^\sharp{}^{\pm\bullet}$ hold.
\end{lemma}

\subsubsection{Rigidity in the well-prepared regime}
 
For well-prepared initial data, one expects that no geometric correction is generated. Equivalently, if the approximating sequence is initially free of corrector, then the corrector should remain identically zero for all times, and the essential geometry variables should converge strongly. This may be viewed as a nonlinear stability property of the geometric part of the Einstein--Euler system.

The most direct way to establish such a statement is not to analyze \(\Pi^\sharp\) itself, but rather to derive a relative energy estimate for the pair \((\Pbb,\Qbb)\), treating the twist variables, the lapse, and the matter variables as lower-order coefficients already controlled by the previous compactness analysis. In this approach, the corrector disappears from the argument: if the relative geometric energy vanishes in the limit, then no corrector can be produced.

This yields the following result.

\begin{lemma}[Strong convergence of the essential geometry in the well-prepared regime]
\label{lemma-74-strong}
Assume that the sequence of essential geometry variables is initially well-prepared, in the sense that
\be
(\Pbb^\suit,\Qbb^\suit)(t_0,\cdot)\to (\Pbb^\sharp,\Qbb^\sharp)(t_0,\cdot)
\qquad\text{strongly in }L^2(\Tbb^3),
\ee
and that the conclusions of Lemmas~\ref{lemma-72}, \ref{lemma-73}, \ref{lemma-74}, and \ref{lemma-76} hold. Then, for every \(t\in\Interval\),
\bel{equa-KSK3}
(\Pbb^\suit,\Qbb^\suit)(t,\cdot)\to (\Pbb^\sharp,\Qbb^\sharp)(t,\cdot)
\qquad\text{strongly in }L^2(\Tbb^3).
\ee
In particular, the corresponding corrector vanishes identically.
\end{lemma}

\begin{proof}
\bse
We compare an approximate solution, indexed by \(\suit\), with the limiting flow, denoted by \(\sharp\). The relevant quantity is the relative geometric energy
\be
\label{eq:relative-geom-energy}
\Gcal^{\sharp,\suit}(t)
\coloneqq
\int_{\Tbb^3}
\Mbf^{00}\bigl(0,\Kpar^\suit-\Kpar^\sharp,\Pbb^\suit-\Pbb^\sharp,\Qbb^\suit-\Qbb^\sharp\bigr)\,
\dVtrois_\sharp.
\ee
Since \(\Mbf^{00}(0,\cdot,\cdot,\cdot)\) is a positive quadratic form, \(\Gcal^{\sharp,\suit}(t)\) is uniformly equivalent to the square of the natural \(L^2\)-distance between the two geometric states.

Using the balance laws satisfied by the geometric variables, and writing all coefficients in the background geometry \(\sharp\), wee are led, after expansion, to a differential inequality of the form
\be
\label{eq:relative-geom-energy-ineq}
\frac{d}{dt}\Gcal^{\sharp,\suit}(t)
\le
C(t)\,\Gcal^{\sharp,\suit}(t)
+\Err^{\sharp,\suit}(t),
\ee
where \(C\in L^1(\Interval)\) depends only on the uniform bounds already established, and where \(\Err^{\sharp,\suit}\) collects lower-order terms involving the lapse, the conformal length measure, the twists, and the matter variables.
\ese

The relative energy inequality~\eqref{eq:relative-geom-energy-ineq} is obtained, after a lengthy but straightforward computation, from the geometric balance laws and from expressing all coefficients in the limiting geometry.
The previous compactness results imply that all lower-order error terms tend to zero in \(L^1(\Interval)\), that is,
\be
\label{eq:error-rel-vanish}
\Err^{\sharp,\suit}\to 0
\qquad\text{in }L^1(\Interval).
\ee
Indeed, these error terms are controlled by the convergences already established for \(\Omega^\suit\), \(\amdeux^\suit\), \(\Kpar^\suit\), \(\Jperp{}^\suit\), and the other lower-order quantities.

Since the sequence is well-prepared at the initial time, we also have 
$
\Gcal^{\sharp,\suit}(t_0)\to 0.
$
Gronwall's lemma therefore implies that
\be
\Gcal^{\sharp,\suit}(t)\to 0
\qquad\text{for every }t\in\Interval, 
\ee
hence \eqref{equa-KSK3} is satisfied for every \(t\in\Interval\), as claimed.
Finally, since the corrector is precisely the corrector generated by the lack of strong convergence in the quadratic geometric terms, the strong convergence established above implies that the corresponding corrector vanishes identically, namely $\Pi^\sharp\equiv 0$.
This completes the proof.
\end{proof}


\section{Compactness properties of fluid variables}
\label{section=9}

\subsection{Compactness statement for the fluid variables}
\label{section=9-1}

In view of the convergence results established in \autoref{section=8}, the
compactness analysis of the geometric variables is complete. After extraction
of a subsequence if necessary, we have already obtained the convergence of the conformal length measure, lapse, twist coefficients, and all geometric quantities
entering the Einstein--Euler system.
In particular, compactness is fully established at this point within the class of vacuum solutions.
The only remaining issue is the compactness of the fluid variables.

This is the final step needed in order to identify the nonlinear matter terms
appearing in the Einstein equations and to pass to the limit in the full
Einstein--Euler system. As expected, the parallel and orthogonal components
of the momentum obey rather different compactness mechanisms.

On the one hand, the normalized parallel momentum is controlled by the
bounded variation bound built into the definition of tame flows. Its
compactness follows from a combination of Helly's theorem, weak continuity in
time, and the structure of the corresponding evolution equations. On the
other hand, the orthogonal momentum is not controlled by any pointwise
variation estimate and must instead be treated by the compensated compactness
mechanism explained in the previous section, based on the quasi-balance laws,
Murat's lemma, and the weighted div-curl argument.

The purpose of the present section is to state the resulting compactness
properties in the precise form required for the passage to the limit, and to
complete the proof of \autoref{theorem-493}, \autoref{theorem-494}, and \autoref{theorem-492}. The
key reduction step for the orthogonal momentum, namely the construction of a
sufficiently rich family of quasi-currents and the corresponding collapse of
Young measures, belongs to the compensated compactness theory underlying the
present work and will be treated separately.

We now formulate the compactness statement for the fluid variables.

\begin{proposition}[Compactness of the fluid variables]
\label{theorem-fluid-compactness}
Assume the uniform metric and fluid bounds, together with the entropy
inequalities, stated in \autoref{section=5}. Assume moreover that the
geometric variables satisfy the convergence properties established in
\autoref{section=8}. After extraction of a subsequence if necessary, the
following conclusions hold.

\bei

\item {\bf Orthogonal momentum.} 
\be
\Jperp{}^\suit(t,\cdot)\to \Jperp{}^\sharp(t,\cdot)
\qquad
\text{strongly in }L^2(\Tbb^3)
\ \textnormal{a.e.}\ t\in I .
\ee

\item {\bf Parallel momentum.}
\be\label{eq-fluid-compactness-parallel}
\aligned
\Jpar{}^\suit(t,\cdot) &\to \Jpar{}^\sharp(t,\cdot)
\qquad && \text{strongly in }L^2(\Tbb^3)
\ \textnormal{a.e.}\ t\in I ,
\\
\Jhatpar{}^\suit(t,\cdot) &\to \Jhatpar{}^\sharp(t,\cdot)
\qquad && \text{strongly in }L^q(\Tbb^3)
\ \text{for every } q\in[1,+\infty)
\ \textnormal{a.e.}\ t\in I ,
\\
d_x\bigl(\Jhatpar{}^\suit\bigr)(t,\cdot) & \overset{*}{\rightharpoonup}
d_x\bigl(\Jhatpar{}^\sharp\bigr)(t,\cdot)
\qquad && \text{weakly-$*$ in measures on } \Tbb^3
\ \textnormal{a.e.}\ t\in I .
\endaligned
\ee

\item {\bf Nonlinear matter terms.}
Consequently, all nonlinear matter quantities entering the Einstein--Euler
system are recovered strongly at the limit. In particular, for every
\(t\in\Interval\),
\be\label{eq-fluid-compactness-nonlinear}
\left.
\begin{aligned}
& \Jbb^\suit\cdot\Jbb^\suit \to \Jbb^\sharp\cdot\Jbb^\sharp
\\
& \Ebf_0(\Jperp{}^\suit)\to \Ebf_0(\Jperp{}^\sharp)
\\
& \Ebf_1(\Jperp{}^\suit)\to \Ebf_1(\Jperp{}^\sharp) \quad
\end{aligned}
\right\}
\quad
\text{strongly in \(L^1(\Tbb^3)\) }
\textnormal{a.e.}\ t\in I .
\ee

\eei
\end{proposition}

The conclusion of \autoref{theorem-fluid-compactness} is exactly the form of
compactness needed in order to pass to the limit in the Einstein--Euler
equations. Once the convergence of the geometric variables from
\autoref{section=8} is combined with the strong convergence of the matter
terms in \eqref{eq-fluid-compactness-nonlinear}, every nonlinear source term
appearing in the field equations is stable under passage to the limit (except for $\Mbf^{ab}$ which may generate a corrector).
Therefore the full Einstein--Euler system with corrector is recovered for the limiting tame
flow, and the (in)stability assertions in Theorem \ref{theorem-493}, \ref{theorem-494}, and~\ref{theorem-492} follow.

There remains to explain the two compactness mechanisms involved in the proof of
\autoref{theorem-fluid-compactness}.


\subsection{Analysis of the parallel fluid momentum} 
\label{section=9-2}

\subsubsection{Compactness mechanism}

The analysis of the parallel momentum is separated from that of the
orthogonal momentum, treated in \autoref{section=9-3}.
The compactness of the parallel momentum follows from a straightforward
mechanism, based on the \(\BV\) bounds already available and on the structure
of the corresponding evolution equations. Namely, by definition of tame flows, the normalized
parallel variables \(\widehat J_m^\suit\), \(m=2,3\), are uniformly bounded in
\(L^\infty(\Interval;\BV(\Sbb^1))\). Therefore, at each fixed time, Helly's
selection theorem yields compactness in the spatial variable.
The (weak) continuity in time of $\Jhatpar$ then allows to extend convergence from countably many times to all times.


\subsubsection{Compactness property}

By assumption, the normalized parallel momentum has uniformly bounded total
variation, namely
\be
\limsup_{\suit \to 0}\sup_{t\in\Interval}
\Var\bigl(\widehat J_m^\suit(t,\cdot);\Tbb^3\bigr)<+\infty,
\qquad m=2,3.
\ee
Hence, by Helly's theorem and a diagonal argument, after extraction of a subsequence if necessary,
there exists a limit function \(\widehat J_m^\sharp\) such that, for $t$ in a countable dense set of times \(\mathcal{T}\subset\Interval\),
\be
\aligned
d_x(\widehat J^\suit_m)(t,\cdot) & \overset{*}{\rightharpoonup}
d_x(\widehat J^\sharp_m)(t,\cdot)
\qquad &&\text{weakly-\(*\) as measures on }\Tbb^3,
\\
\widehat J^\suit_m(t,\cdot) & \to \widehat J^\sharp_m(t,\cdot)
\qquad &&\text{pointwise on }\Tbb^3.
\endaligned
\ee
Next, we upgrade this spatial compactness to compactness for
all times by using the condition that $\widehat J^\suit$ is uniformly in $\Lip(\Interval, L^1(\Tbb^3,\dVtrois{}^\suit))$, hence in $\Lip(\Interval, L^1(\Tbb^3,\dVtrois{}^\sharp))$ by the strong convergence of the lapse and conformal length density. 
This controls the dependency in time and allows us to extend the convergence from a countable dense set of times to every time $t\in\Interval$:
\be
\aligned 
\widehat J^\suit_m(t,\cdot) & \to \widehat J^\sharp_m(t,\cdot)
\qquad && L^p(\Tbb^3) \text{ for all } p \in [1,+\infty).
\endaligned
\ee

The total variation of $\widehat J^\suit_m(t,\cdot)$ is bounded uniformly in~$\suit$ for almost all $t\in\Interval$, thus the lower semi-continuity of BV norms yields
\be
\Var\bigl( \widehat J^\sharp_m(t,\cdot) ; \Tbb^3 \bigr)
\leq \liminf_{\suit\to 0} \Var\bigl( \widehat J^\suit_m(t,\cdot) ; \Tbb^3 \bigr) ,
\ee
which admits a uniform upper bound, hence the limit $\Jhatpar{}^\sharp$ belongs to $L^\infty(\Interval,\BV(\Tbb^3))$.
In particular, $d_x\Jhatpar{}^\sharp(t,\cdot)$ is a bounded signed measure for almost every time, so that the convergence of $d_x\Jhatpar{}^\suit(t,\cdot)$ in the sense of distributions holds weakly-$*$ in measures.
We therefore arrive at the following statement.

\begin{lemma}[Strong convergence of the parallel momentum]
\label{lemma-parallel-momentum}
Assume the uniform bounds given in \autoref{section=5}. Suppose, in addition,
the strong convergence $(\amdeux^\suit, \Omega^\suit) \to (\amdeux^\sharp, \Omega^\sharp)$ (established in \autoref{section=8}). Then, after extraction of a subsequence if necessary, there exists
\be
\Jhatpar{}^\sharp \in L^\infty(\Interval,\BV(\Tbb^3)) \,\cap\, \Lip(\Interval, L^1(\Tbb^3,\dVtrois{}^\sharp))
\ee
such that the normalized parallel momentum satisfies, for every $t\in\Interval$,
\be
\widehat J_m^\suit(t,\cdot)\to \widehat J_m^\sharp(t,\cdot)
\qquad\text{ strongly in }L^q(\Tbb^3),
\qquad m=2,3, \quad \text{ for all } q \in [1,+\infty) ,
\ee
and for almost every $t\in\Interval$,
\be
d_x(\widehat J^\suit_m)(t,\cdot) \overset{*}{\rightharpoonup}
d_x(\widehat J^\sharp_m)(t,\cdot)
\qquad \text{weakly-\(*\) as measures on }\Tbb^3.
\ee
\end{lemma}


\subsection{Analysis of the orthogonal fluid momentum}
\label{section=9-3}

\subsubsection{Compactness mechanism}

The orthogonal momentum is substantially more delicate than the parallel
momentum and requires the full compensated compactness machinery. The available energy estimate yields only weak compactness in
\(L^2\), which is not sufficient to identify the nonlinear matter terms in
the limit. The missing strong compactness is recovered from the entropy
structure introduced earlier in the paper, as we now explain. 

The starting point is the family of quasi-balance laws associated with
future-convex quasi-currents dominated by the reference current.
Under the conditions in \autoref{section=5}, their entropy productions, given by the shear-induced measures~$\Fmeasure$, are uniformly bounded.
By Murat's lemma, the entropy productions are therefore relatively compact in~$H^{-1}$.

This compactness property is precisely the input needed for a (suitably weighted) 
div-curl argument. Applying the latter to pairs of quasi-currents
yields Tartar-type commutation relations. At this stage, we introduce the
Young measure generated by the sequence~$\Jbb^\suit$.
Since the parallel component has already been shown to converge strongly, the
only possible oscillations are carried by~$\Jperp{}^\suit$. The commutation
relations reduce to a family of algebraic constraints on the
 Young measure associated with the orthogonal components only.

The role of the quasi-currents is to make these constraints sufficiently rich
to force the collapse of the Young measure to a Dirac mass. This is the
standard rigidity mechanism in compensated compactness for (non-relativistic) compressible Euler
flows, adapted here. Once the Young measure is reduced to a Dirac mass, the strong
convergence of~$\Jperp{}^\suit$ follows. Consequently, we obtain also \eqref{eq-fluid-compactness-nonlinear}. 
Let us now provide some details on this method, developped fully in our companion paper \cite{LeFlochLeFloch-next}. 


\subsubsection{A div-curl commutation relation}

We now make precise the first step of the argument.  
For every quasi-current \(\Fbb=\Fbb(\Jbb)\) that is
future-convex and dominated by the reference entropy current in the sense of
\autoref{def:quasi-current-dominated}, 
the following entropy balance law holds:
\bel{equa-balancelaws-plus-rep}
\aligned
\divuntrois_\suit \bigl( |t|^{-1} (\Omega^\suit)^{-1} \, \Fbb(\Jbb^\suit) \bigr)
= \Fmeasure^\suit + \Hterm^\suit,
\endaligned
\ee
where the sequence of measures $\Fmeasure^\suit$ is uniformly bounded, namely 
$\Fmeasure^\suit = \Mterm^\suit$ in the notation of \autoref{section=5}, 
and the relatively compact terms $\Hterm^\suit$ stand for possibly additional approximation errors. 
We now use this fact to prove a compactness property that will be the
key input for the subsequent div-curl argument.

\begin{lemma}[Compactness of entropy productions]
\label{lem-gen-entro-part2}
For every quasi-current \(\Fbb\) that is future-convex and dominated by the
reference current~\(\Fbb^*\), the sequence
\be
\divuntrois_\suit \Bigl( |t|^{-1}(\Omega^\suit)^{-1}\,\Fbb(\Jbb^\suit) \Bigr)
\ee
is relatively compact in \(H^{-1}(\Interval\times\Sbb^1)\).
\end{lemma}

\begin{proof} 
The right-hand side of the quasi-balance law is bounded in the
sum of a uniformly bounded measure term and a term that is already assumed to be $H^{-1}$ compact. Since the left-hand side is the
divergence of a family of currents uniformly bounded in~$L^2$, Murat's
lemma~\cite{Murat-1974}  applies and yields the claimed relative compactness in
$H^{-1}(\Interval\times\Sbb^1)$.
\end{proof}


We may now apply the weighted div-curl lemma, stated in
\autoref{lemdivcurl}, to two quasi-currents \(\Fbb^{(1)}\) and
\(\Fbb^{(2)}\) that are future-convex and dominated by the reference entropy
current. By \autoref{lem-gen-entro-part2}, their entropy productions are
relatively compact in \(H^{-1}\), while the relevant geometric weights
converge strongly by the results established in the previous section. We
therefore obtain
\be
\label{eq:tartar-commutation-1}
\aligned
& \Omega^\suit \Bigl(
\Fbb^{(1)}_0(\Jbb^\suit)\,
\Fbb^{(2)}_1(\Jbb^\suit)
-
\Fbb^{(1)}_1(\Jbb^\suit)\,
\Fbb^{(2)}_0(\Jbb^\suit)
\Bigr)
\to
\Omega^\sharp \Bigl(
\Fbb^{(1)}_0(\Jbb^\sharp)\,
\Fbb^{(2)}_1(\Jbb^\sharp)
-
\Fbb^{(1)}_1(\Jbb^\sharp)\,
\Fbb^{(2)}_0(\Jbb^\sharp)
\Bigr)
\endaligned
\ee
in the sense of distributions on \(\Interval\times\Sbb^1\). This is the
key compensated compactness relation for the fluid momentum. 


\subsubsection{Young measure representation}

The sequence $\Jbb^\suit$ is known to converge weakly,
and the commutation relation above must
therefore be interpreted at the level of all possible weak limits. Taking into account our uniform $L^2$ estimates and extracting a subsequence if necessary, there exists a family of
probability measures
$\nu=\nu_{t,x}$ defined on the phase space of the variables~$\Jbb$ 
for almost every \((t,x)\in\Interval\times\Sbb^1\), such that for
every continuous function \(G\) with suitable growth (determined by the available bounds, especially the reference current) 
\be
G(\Jbb^\suit) \to \langle \nu_{t,x}, G\rangle
\ee
in the sense of distributions.
The relation~\eqref{eq:tartar-commutation-1} then yields the celebrated \textbf{Tartar commutation relation} for the Young measure.

\begin{lemma}[Tartar commutation relation for the momentum]
\label{prop-tartar-orthogonal}
Let $\nu=\nu_{t,x}$ be the Young measure generated by the sequence~$\Jbb^\suit$.
Then, for every pair of quasi-currents $\Fbb^{(1)}$, $\Fbb^{(2)}$ that
are future-convex and dominated by the reference entropy current~\(\Fbb^*\),
and for almost every $(t,x)\in\Interval\times\Sbb^1$,
the following identity holds:
\be
\label{eq:tartar-commutation-2}
\aligned
\big\langle \nu_{t,x}, \;
\Fbb^{(1)}_0 \Fbb^{(2)}_1 - \Fbb^{(1)}_1 \Fbb^{(2)}_0 \big\rangle
 =
\big\langle \nu_{t,x},\Fbb^{(1)}_0\big\rangle
\big\langle \nu_{t,x},\Fbb^{(2)}_1\big\rangle
-
\big\langle \nu_{t,x},\Fbb^{(1)}_1\big\rangle
\big\langle \nu_{t,x},\Fbb^{(2)}_0\big\rangle .
\endaligned
\ee
\end{lemma}


\subsubsection{Reduction of the Young measure}

We now explore the implications of the commutation relation
\eqref{eq:tartar-commutation-2} by working with a sufficiently rich family of
quasi-currents that are future-convex and dominated by the reference entropy
current~$\Fbb^*$. An important simplification is found by changing coordinates on the phase space from $\Jbb$ to $(\Jperp,\Jhatpar)$.  Since the
normalized parallel momentum~$\Jhatpar{}^\suit$ has already been shown to
converge strongly almost everywhere, the only remaining oscillations in the
pair $(\Jperp{}^\suit,\Jhatpar{}^\suit)$ are carried by the orthogonal momentum~$\Jperp{}^\suit$. Accordingly, after
extraction of a subsequence if necessary, the associated Young measure may be
written in the \emph{tensorial form}
\be
\nu_{t,x}= \mu_{t,x}\otimes \delta_{\Jhatpar{}^\sharp(t,x)},
\ee
where $\mu_{t,x}$~is a probability measure on the phase space of~$\Jperp$, defined for almost every $(t,x)\in\Interval\times\Sbb^1$.

Therefore, the Tartar commutation relation~\eqref{eq:tartar-commutation-2}
reduces, for almost every \((t,x)\), to
\be
\label{eq:tartar-reduced}
\aligned
&
\Big\langle \mu_{t,x}, \
\Fbb^{(1)}_0(\,\cdot\,,\Jhatpar{}^\sharp)\,
\Fbb^{(2)}_1(\,\cdot\,,\Jhatpar{}^\sharp)
-
\Fbb^{(1)}_1(\,\cdot\,,\Jhatpar{}^\sharp)\,
\Fbb^{(2)}_0(\,\cdot\,,\Jhatpar{}^\sharp)
\Big\rangle
\\
&=
\Big\langle \mu_{t,x}, \; \Fbb^{(1)}_0(\,\cdot\,,\Jhatpar{}^\sharp)\Big\rangle
\Big\langle \mu_{t,x}, \; \Fbb^{(2)}_1(\,\cdot\,,\Jhatpar{}^\sharp)\Big\rangle
-
\Big\langle \mu_{t,x}, \; \Fbb^{(1)}_1(\,\cdot\,,\Jhatpar{}^\sharp)\Big\rangle
\Big\langle \mu_{t,x}, \; \Fbb^{(2)}_0(\,\cdot\,,\Jhatpar{}^\sharp)\Big\rangle .
\endaligned
\ee

The key point is that the family of quasi-currents is sufficiently rich to
separate points in the phase space of~$\Jperp$. The relation
\eqref{eq:tartar-reduced} then forces the measure \(\mu_{t,x}\) to reduce to
a Dirac mass. This reduction step lies at the heart of compensated compactness
arguments for the Euler equations. While further
details can be found in the classical papers
\cite{DiPerna-1983a,DiPerna-1983b},
\cite{LeFlochWestdickenberg-2007,PG-PLF}, and the textbook \cite{Dafermos-book}., we refer the reader 
to our companion paper
\cite{LeFlochLeFloch-next}.

\begin{theorem}[Reduction of the Young measure]
\label{prop-young-collapse}
Consider the relativistic Euler equations in $1+1$ symmetry under the hyperbolicity and convexity conditions~\eqref{hyperbolic-eos} together with the asymptotic conditions \eqref{equa-asym} at vacuum and in the large. Then the class of quasi-currents dominated by the reference entropy
current~\(\Fbb^*\) is sufficiently large so that Tartar commutation
relations~\eqref{eq:tartar-commutation-2} imply ``separation of states'' in
the \(\Jperp\)-variable, for every fixed value of \(\Jhatpar\), in the following sense.
For almost every $(t,x)$ a Young measure satisfying Tartar's commutation relations 
generated by a uniformly square-integrable sequence $\Jbb^\suit$ has a support in phase state space, which is either a point or a subset of the vacuum set. Hence, for almost every $(t,x)\in\Interval\times\Sbb^1$, one has 
\be
\nu_{t,x} = 
\begin{cases}
\delta_{(\Jperp{}^\sharp(t,x),\Jhatpar{}^\sharp(t,x))} = \delta_{\Jbb^\sharp(t,x)}, & \text{away from vacuum,}
\\
\text{arbitrary measure supported on $\Jbb \cdot \Jbb=0$,}
& \text{on the vacuum,} 
\end{cases}
\ee
and, moreover, no quadratic concentration phenomenon occurs, so that
\be
\Jbb^\suit(t,\cdot)\to \Jbb^\sharp(t,\cdot)
\qquad
\text{strongly in }L^2(\Tbb^3), \, t \in \Interval. 
\ee
\end{theorem}


\subsection{Conclusion for the Einstein--Euler system (\autoref{theorem-493})}
\label{section=9-4}

We are now in a position to complete the passage to the limit in the
Einstein--Euler system with corrector.
Indeed, the geometric compactness theory developed in
\autoref{section=8}, together with the fluid compactness statement of
\autoref{theorem-fluid-compactness}, provides exactly the convergence
properties required to establish \autoref{theorem-493} (and hence its special case \autoref{theorem-492}).

The construction of the limit geometric variables $\ell^\sharp,\log\Omega^\sharp,\Pbb^\sharp,\Qbb^\sharp,\Kpar^\sharp$ (in \autoref{section=8}) and the fluid variables $\Jbb^\sharp$ (in the present section) ensures that they belong to the appropriate functional spaces for $(\Phi^\sharp,\Psi^\sharp,\Pi^\sharp)$ to be a $\Phi\Psi\Pi$ flow with finite energy in the sense of \autoref{weakdefinitionT2-relaxed}.
It also ensures the strong convergence of all variables except for two features: the weak convergence $(\Pbb^\suit,\Qbb^\suit)\rightharpoonup(\Pbb^\sharp,\Qbb^\sharp)$ in $L^2$, and the fact that their quadratic fluctuations contribute to the limit corrector~$\Pi^\sharp$, namely the combination $\Pi + \amdeux (\Omega)^2 \bigl( \Pbb \otimes \Pbb + \Qbb \otimes \Qbb \bigr)$ converges weakly-$*$ in measures.

We have verified throughout \autoref{section=8} that all evolution and constraint equations for the geometric variables pass to the limit once the fluid variables are known to converge strongly.  There remains to note that source terms in the particle number equation~\eqref{indef-eqvNumb}, spatial Euler equations \eqref{eq:T2-Euler-perp-def-b-r}--\eqref{eq:T2-fluid23-def-b-r}, reference entropy inequality~\eqref{equa-lfk3}, and $H$-divergence law~\eqref{equa-bvpara} also converge since the terms quadratic in the essential geometric variables (such as in~$\Msource$) are null forms.
This concludes the proofs of \autoref{theorem-493}, \autoref{theorem-494} and \autoref{theorem-492}.

Next, we turn to \autoref{corollary-493-init}.  When the flows \emph{tamely} assume a prescribed tame sequence of approximate initial data sets $(\mathring\Phi^\suit,\mathring\Psi^\suit,\mathring\Pi^\suit)$, the same convergence proofs apply.  In particular, the bounds of \eqref{equa-DK399} and \eqref{Fmeasure-apriori-def} on the metric-weighted parallel momentum and shear-induced measure in terms of the initial data pass to the limit, which ensures that the limiting flow \emph{tamely} assumes the limiting initial data set.  This property is essential in order for the class of exact solutions to feature compactness properties.

This completes the compactness part of the proof of the global existence theorem stated in \autoref{section=1}.


\section{Concluding observations}
\label{section=10}

\subsection{Theory based on the particle number current}
\label{section=10-1}

\subsubsection{Two natural choices for the reference entropy}

The entropy structure of the Euler system admits two formulations that are
natural in the present context. In the formulation adopted throughout the core of this paper, the particle
number current \(\vNumb\) is kept in exact conservative form, whereas the
current \(\Mbf^{0\bullet}\) is taken as the reference entropy current and is
therefore required to satisfy an entropy inequality. In the alternative
formulation discussed here, these two roles are reversed: the current
\(\Mbf^{0\bullet}\) is kept in exact balance form, while the particle number
current, up to a sign, is used as the reference entropy current and is
required to satisfy
\bel{equa-entr-rrr}
\divuntrois \vNumb \geq 0
\ee
in the sense of distributions for weak solutions. In this setting, the
relevant convexity property is the future convexity of the negative particle
number current \((-\vNumb)\), established in
\autoref{appendix=D-3}.

The compensated compactness method applies to both choices, but the two
formulations do not lead to the same range of admissible fluid variables. The
formulation based on \(\Mbf^{0\bullet}\), which is the one developed in the
present paper, allows for non-trivial parallel momentum, namely
$
(J_2,J_3)\not\equiv (0,0).
$
By contrast, the alternative formulation based on $-\vNumb$ as the
reference entropy current only allows the reduced regime
\be
J_2=J_3=0.
\ee
The reason is that, in this alternative setting, the $H$-divergence law
\eqref{equa-entrop23-cons} is replaced by an inequality valid only for a
restricted subclass of functions~$H$. This is not sufficient to ensure
the maximum principle on~$\Gammaunpar$ needed to control the parallel
momentum components.

Under the restriction $J_2=J_3=0$, however, this obstruction disappears.
The quasi-currents introduced in the present paper reduce to standard entropy
currents, and the proof proceeds along essentially the same lines as in the
framework developed above. A further simplification is that the
metric-weighted twist functions \(\Gammazero_2\) and \(\Gammazero_3\) are
constant, which trivializes their treatment.

For these reasons, we chose to base the main theory on the current
\(\Mbf^{0\bullet}\): this formulation applies to a \emph{substantially wider class
of Einstein--Euler flows} and, as explained in \autoref{appendix=D-2}, also recovers
DiPerna's standpoint in the non-relativistic limit.

We summarize the alternative formulation as follows. The corresponding
existence, stability, and instability results are virtually
\emph{identical}, and the proof follows exactly the same lines; the only 
difference lies in the \emph{choice of regularization} used in the construction of a sequence of approximate solutions,
for instance via a vanishing viscosity approach.

\begin{theorem}
The notions of (tame) Einstein--Euler flows and the statements in Theorems~\ref{theorem-493} and~\ref{theorem-492} remain valid (hence Theorems~\ref{thm:1.1} and \ref{thm:stability0}), 
under the restriction that the parallel momenta
$J_2=J_3=0$ vanish, when one requires the balance law for \(\Mbf^{0\bullet}\) to hold as an
equality in the sense of distributions, that is, 
\be
\divuntrois \bigl( \Omega \, \Mbf^{0 \bullet}(\Jbb, \Kpar, \Lbb) \bigr)
+ {1 \over 2\, t} \Msource(\Jbb, \Kpar, \Lbb) = 0,
\ee
while the reference entropy is taken to be the negative particle number
current \((-\vNumb)\) and one requires, in the sense of distributions,
\be
\divuntrois \vNumb \geq 0.
\ee
\end{theorem}


\subsubsection{Rankine--Hugoniot relations and entropy inequalities}

We now compare the two formulations at the level of discontinuous solutions.
Let
$
\Sigma=\{x=\chi(t)\}\subset \Interval\times\Sbb^1
$
be a shock curve, with speed denoted by 
$
s(t):=\chi'(t).
$
For any balance law of the form
\be
\divuntrois \Xbb = S,
\qquad S\in L^1,
\ee
where \(\Xbb=(X^0,X^1)\), the source term has no singular part concentrated on
\(\Sigma\). Therefore, the corresponding Rankine--Hugoniot relation reads
\be
\label{eq:RH-general-short}
\bigl[X^1-s\,\amdeux\,X^0\bigr]_\Sigma=0.
\ee
Likewise, if instead
\be
\divuntrois \Xbb \leq S,
\qquad S\in L^1,
\ee
then the singular part concentrated on \(\Sigma\) must be non-positive, and we
obtain the entropy jump inequality
\be
\label{eq:RH-general-ineq-short}
\bigl[X^1-s\,\amdeux\,X^0\bigr]_\Sigma \leq 0.
\ee
Thus, at the level of weak solutions, the distinction between a current kept in
exact balance form and a current used as a reference entropy current is reflected
directly in the corresponding jump conditions.

In the alternative formulation considered here, the current \(\Mbf^{0\bullet}\)
is kept in exact balance form, whereas $-\vNumb$ is treated as the reference
entropy current. Accordingly,  we obtain exact Rankine--Hugoniot relations for
the balance laws associated with \(\Mbf^{0\bullet}\) and with the parallel
momentum components, together with an entropy inequality for the particle number
current. With our sign convention, these relations take the form
\be
\aligned
\bigl[\Mbf^{01}-s\,\amdeux\,\Mbf^{00}\bigr]_\Sigma& =0,
& \qquad
\bigl[J_2(J_1+s\,\amdeux J_0)\bigr]_\Sigma & =0,
\\
\bigl[\Mbf^{11}-s\,\amdeux\,\Mbf^{10}\bigr]_\Sigma & =0,
& \qquad
\bigl[J_3(J_1+s\,\amdeux J_0)\bigr]_\Sigma & =0,
\\
\bigl[\Numb^1-s\,\amdeux\,\Numb^0\bigr]_\Sigma & \leq 0.
\endaligned
\ee
The relations involving \(J_2\) and \(J_3\) are simply the Rankine--Hugoniot
relations associated with the exact balance laws satisfied by the corresponding
parallel momentum currents.

For definiteness, let us express these relations in the isothermal case, and
when the geometric weights are formally suppressed in \(\Mbf^{ab}\). We then
find
\bel{eq:RH-flat-short}
\aligned
\Bigl[-J_0 J_1-s\,\amdeux \Bigl(J_0^2 - \frac{1-q}{2} (-\Jbb\cdot\Jbb)\Bigr)\Bigr]_\Sigma & =0,
& \qquad \bigl[J_1J_2+s\,\amdeux J_0J_2\bigr]_\Sigma & =0,
\\
\Bigl[\Bigl(J_1^2 + \frac{1-q}{2} (-\Jbb\cdot\Jbb)\Bigr)+s\,\amdeux J_0 J_1\Bigr]_\Sigma & =0,
& \qquad
\bigl[J_1J_3+s\,\amdeux J_0J_3\bigr]_\Sigma & =0,
\\
\Bigl[(-\Jbb\cdot\Jbb)^{q/2} (J_1+s\,\amdeux\,J_0) \Bigr]_\Sigma &\leq 0.
\endaligned
\ee
In other words, the two frameworks can be compared as follows:
\bei
\item in the framework with \(\Mbf^{0\bullet}\) as the reference entropy, the
particle number current satisfies an exact Rankine--Hugoniot relation, while
\(\Mbf^{0\bullet}\) satisfies an entropy jump inequality;

\item in the framework with \(-\vNumb\) as the reference entropy, the current
\(\Mbf^{0\bullet}\) satisfies an exact Rankine--Hugoniot relation, while
\(\vNumb\) satisfies an entropy jump inequality.
\eei

Although both formulations originate from the same Euler equations at the level
of regular solutions, they lead in general to different theories of weak
solutions. Indeed, once shocks are present, exchanging the conservative and
entropic roles of these currents modifies the corresponding jump conditions, and
therefore may modify the class of admissible shock curves.

At the same time, this discrepancy should be viewed with some care. The two
formulations share the same underlying smooth equations and the same leading
Rankine--Hugoniot structure. In the reduced regime \(J_2=J_3=0\), they are
therefore expected to be perturbatively close for weak shocks. More precisely,
we expect that the corresponding Hugoniot curves through a given left state
agree at leading order, while the difference between the two admissibility
criteria enters only at higher order in the shock strength. Establishing this
comparison rigorously would require a separate asymptotic analysis of the shock
curves, which lies beyond the scope of the present paper.         


\subsection{Other models covered by the present method}
\label{section=10-2}

\subsubsection{Scalar-field matter models}

Scalar fields also fit naturally within the present framework. Indeed, an irrotational stiff fluid may be interpreted as a scalar field, in which case the matter equations reduce to linear wave equations on the quotient geometry. The nonlinear difficulties specific to compressible fluids ---such as shock formation and entropy dissipation--- are then absent, and the compactness analysis becomes much simpler, since weak convergence is stable for linear equations. More generally, the present approach applies to matter models whose evolution equations are linear on a given background geometry.


\subsubsection{Stiff fluids}

The case of  the equation of state \(p=\rho\), describing the class of \emph{stiff fluids}, also lies within the scope of the present method. In this regime, the Euler equations are essentially linear hyperbolic, so that several of the difficulties specific to genuinely nonlinear fluids disappear. The entropy structure simplifies, and the compactness analysis becomes closer to that of scalar-field models than to that of general compressible flows.
Accordingly, the arguments developed here for the geometric variables, especially for \(\Lbb=(\Pbb,\Qbb)\), extend with only minor modifications, while the matter variables can be treated by more direct linear methods. In this sense, stiff fluids provide a natural intermediate model between scalar fields and genuinely nonlinear Einstein--Euler flows. See also the earlier work of Christodoulou~\cite{Christo95,Christo96} on a two-phase model and a free-boundary problem.


\subsubsection{Non-zero cosmological constant}

The first-order formulation developed here also extends to the Einstein--Euler equations with cosmological constant \(\Lambda\). At the level of the local analysis and of the weak formulation, the required modifications are minor: the cosmological term contributes only lower-order terms and therefore does not alter the hyperbolic structure, the entropy framework, or the compactness arguments established in this paper.

Indeed, our existence and stability theory for flows with finite energy extends with little change to the case \(\Lambda\geq0\). The \emph{asymptotic behavior}, however, is expected to depend strongly on the sign and size of \(\Lambda\). For \(\Lambda>0\), we anticipate a more regular expanding dynamics and, in suitable regimes, possibly even a late-time smoothing effect. By contrast, in the contracting direction, the competition between matter concentration, gravitational focusing, and the cosmological term leads to a markedly different picture. We leave a systematic study of these questions to future work, and refer to Smulevici~\cite{Smulevici-2009,Smulevici-2011} and the references cited therein for results on vacuum spacetimes in the regular setting.


\subsection{Perspectives} 
\label{section=10-3}

The present work suggests several directions for further investigation of the Cauchy problem in general relativity under symmetry assumptions. On the matter side, a natural problem is to extend our analysis to kinetic models, such as Vlasov or Boltzmann-type matter, for which we expect an interesting and involved interplay between geometric-matter interactions, weak regularity, and compactness under weak convergence.

Another important challenge is to go beyond the present class of barotropic fluids and treat matter models involving two thermodynamical variables. Removing the restriction to equations of state depending on a single variable, and thus extending the theory to non-isentropic fluid flows, requires new ideas.

A further perspective concerns the long-time behavior and asymptotic geometry of flows with finite energy. Even in the presence of symmetry, the interaction between matter concentration, geometric oscillations, and the possible emergence of corrector stress tensors raises new questions about asymptotic stability, effective dynamics, and back-reaction phenomena.
 

\phantomsection
\addcontentsline{toc}{section}{Acknowledgements}
\subsubsection*{Acknowledgments}

The authors are grateful to Mihalis Dafermos, Alan D. Rendall, Hans Ringstr\"om, and Luc Tartar for valuable discussions and encouragement throughout this long-term project. 
They acknowledge financial support from the Simons Center for Geometry and Physics, Stony Brook University, where part of this research was carried out during a visit in 2019. Part of this work was also completed while the second author (PLF) was visiting the Mittag-Leffler Institute during the Semester Program ``General Relativity, Geometry, and Analysis: beyond the first 100 years after Einstein''. During the completion of this project, the authors were  partially supported by the research project ANR-23-CE40-0010-02: Einstein-PPF, funded by the Agence Nationale de la Recherche, and by the ERC-MSCA Staff Exchange Project 101131233: Einstein-Waves, funded by the European Research Council.
 


\appendix 

\section{Vielbein and spin connection}
\label{appendix=A}

In areal gauge, the Vielbein vectors leading to the decomposition~\eqref{equa-metric-inverse} of the spacetime metric $\guntrois$ take the form
\bel{eq:ourframe}
\aligned
e^0 & \coloneqq \Omega\,dt, \
& e^1 & \coloneqq \Omega\, \amdeux dx, \
\\
e^2 & \coloneqq |t|^{1/2} e^{P/2} \bigl( dy + Q\,dz + (G+QH) \, dx\bigr), \
& e^3 & \coloneqq |t|^{1/2} e^{-P/2} (dz + H\,dx),
\\
e_0 & \coloneqq \Omega^{-1} \del_t, \
& e_1 & \coloneqq \Omega^{-1} \amdeux^{-1} (\del_x - G\del_y - H\del_z), \ 
\\
e_2 & \coloneqq |t|^{-1/2} e^{-P/2} \del_y, \
& e_3 & \coloneqq |t|^{-1/2} e^{P/2} (\del_z - Q \del_y).
\endaligned
\ee 
Denoting by $m,n, \ldots= 0, 1, 2, 3$ the indices corresponding to the frame $(e_m)$, we observe that the components $\omega^{mn}_p \coloneqq e_p^\alpha e_\beta^m\nabla_\alpha e^{\beta n}=- \omega^{nm}_p$ of the \emph{spin connection} take a particularly simple form in terms of the variables~\eqref{eq:definevar-01}: 
\bel{equa-touslesomega}
\aligned
   & \omega^{20}_2 = \frac{1}{2} \Bigl(P_0+ \frac{1}{t} \Omega^{-1} \Bigr), \,
    &&\omega^{03}_3 = \frac{1}{2} \Bigl(P_0- \frac{1}{t} \Omega^{-1} \Bigr), \, 
    &&\omega^{12}_2 = \omega^{31}_3 = \frac{1}{2} P_1, 
\\
& \omega^{20}_3 = \omega^{30}_2 = \omega^{23}_0 = \frac{1}{2} Q_0, \,
&& \omega^{12}_3 = \omega^{13}_2 = \omega^{23}_1 = \frac{1}{2} Q_1, \, 
   && 
\\
   & \omega^{10}_2 = \omega^{20}_1 = \omega^{21}_0 = \frac{1}{2} K_2, \quad 
   && \omega^{10}_3 = \omega^{30}_1 = \omega^{31}_0 = \frac{1}{2} K_3, \quad 
\\
   &\omega^{01}_0 =\amdeux^{-1}  (\Omega^{-1})_x, \quad
   &&
 \omega^{01}_1 = \amdeux (\Omega^{-1} \amdeux^{-1} )_t,  
\endaligned
\ee
while all other coefficients vanish: 
 \be
 \aligned   
   && \omega^{02}_0 = \omega^{03}_0 = \omega^{12}_1 = \omega^{13}_1 = \omega^{23}_2 = \omega^{23}_3 = 0.
   \endaligned
   \ee


\section{The case of isothermal flows}
\label{appendix=B}

\subsection{Example: isothermal flows with \texorpdfstring{$\Tbb^2$}{T2} symmetry}
\label{appendix=B-1}

By definition, isothermal matter flows are characterized by the relation $p=k^2\mu$ for some constant $k\in(0,1)$, the sound speed, and~\eqref{equa-211b} then simplifies to the relation
\be
- \Jbb \cdot \Jbb = 2(1+ k^2) \mu \quad \text{ (isothermal fluids),}
\ee
namely the Lorentzian norm of the fluid momentum is a multiple of the mass-energy density of the fluid. Moreover, the pressure ratio $q$ defined by~\eqref{eq:defq}, is also \emph{constant} and, specifically, 
\bel{equa-case-iso}
q = \frac{1-k^2}{1+k^2} \eqqcolon 2 \, \kappa \quad \text{ (isothermal fluids),}
\ee
in which the constant $\kappa\in(0,1/2)$ will be a convenient notation to distinguish between a general pressure law and an isothermal law. Notably, when the function $q$ is a constant, the Euler equations become purely quadratic and, for instance, the energy $\Mbf^{00}$ becomes 
\bel{eq:energyGowdy} 
\Mbf^{00}(\Jbb, \Kpar, \Lbb) 
= \Ebf_0(\Jbb, \Kpar, \Lbb) + \kappa \, \bigl( - \Jbb \cdot \Jbb\bigr) \quad \text{ (isothermal fluids).}
\ee
The terms involving $\Jbb$ take the form $\Mbf^{00}(\Jbb, 0, 0) = (\frac{1}{2}+\kappa)J_0^2+(\frac{1}{2}- \kappa)(J_1^2+J_2^2+J_3^2)$; since $\kappa < 1/2$, this defines a strictly convex function of all~$J_m$. 

\begin{proposition}[Structure of the first-order formulation: isothermal fluids] 
\label{theo-struct-5}
For isothermal fluids $p=k^2\mu$, \autoref{theo-struct} can be extended with the following property. The \JKL\ formulation of the Einstein--Euler system consists exclusively of quadratic terms in the variables $(\Jbb, \Kpar, \Lbb)$. Moreover, in this case, the energy density $(\Jbb, \Kpar, \Lbb) \mapsto \Mbf^{00}(\Jbb, \Kpar, \Lbb)$ is a \emph{strictly convex},  \emph{non-negative} quadratic form in all of its arguments.
\end{proposition} 


\subsection{Example: isothermal flows with Gowdy symmetry}
\label{appendix=B-2}

When, moreover, $G=H=0$ so that $\Kpar=0$, and consequently $\Jpar=0$, the isothermal case is particularly remarkable and we restate here the Einstein--Euler system in this case. From~\eqref{eq:T2-1234-def}--\eqref{eq:theconstraints-def} we then obtain the Einstein equations 
\be  
\left.
\hskip-.18cm
\aligned
\divuntrois \Pbb & = \Qbb \cdot \Qbb,
\, 
& \curluntrois (|t|^{-1} \Pbb) & = 0, 
\\
\divuntrois \Qbb & = - \Pbb \cdot \Qbb,
\, 
& \curluntrois (|t|^{-1} \Qbb) & = |t|^{-1} \Pbb \wedge \Qbb, 
\\
(\amdeux)_t  & = t \,  \kappa \, \bigl( - \Jperp \cdot \Jperp\bigr) \, \Omega^2 \amdeux,
\\
(\log\Omega)_t
&= - \frac{1}{4t} + \frac{t}{2}\Mbf^{11} \, \Omega^2,
\\
(\log\Omega)_x
&= - {t \over 2 } \, \Mbf^{01} \, \Omega^2 \amdeux,
\endaligned
\right\} 
\text{(Gowdy, isothermal)}
\ee  
with $\kappa = q/2 = (1-k^2)/(2(1+k^2))$, as well as the Euler equations
\bel{eq:T2-567-two-G}
\left.
\aligned
\divuntrois \bigl( \Omega \, \Mbf^{0 \bullet} \bigr)
& = - \frac{1}{2\,t} \Msource,
\\
\qquad \divuntrois \bigl( \Omega \, \Mbf^{1 \bullet} \bigr) 
& = 0,
\endaligned
\right\}
\quad \text{(Gowdy, isothermal),}
\ee 
in which the total energy-momentum tensor~\eqref{eq:T2-Mdef-0} and the source term~\eqref{eq:T2-Euler-perp2-def0} are given by 
\bel{eq:T2-567-two2}
\left.
\aligned
\Mbf^{00}
& = \Ebf_0(\Jperp, \Pbb, \Qbb) + \kappa (- \Jperp \cdot \Jperp),
\qquad
\\
\Mbf^{01} & = - \Ebf_1(\Jperp, \Pbb, \Qbb),
\\
\Mbf^{11}
& = \Ebf_0(\Jperp, \Pbb, \Qbb) - \kappa (- \Jperp \cdot \Jperp),
\qquad
\\
\Msource
& = - \Pbb \cdot \Pbb -  \Qbb \cdot \Qbb + (1 - \kappa) \, (- \Jperp \cdot \Jperp), 
\endaligned
\right\}
\quad \text{(Gowdy, isothermal).}
\ee
Moreover, the source obeys the stronger inequality
\be
\abs{\Msource} \leq 2 \, \Mbf^{00} \quad \text{(Gowdy, isothermal).}
\ee
(For a proof, we refer to~\autoref{appendix=D-1}.) In contrast, the 
wave equation~\eqref{eq:waveconffac} does not admit any simplification.


\section{Structure of the Einstein--Euler system}
\label{appendix=C}

\subsection{Proof of~\autoref{theo-struct}}
\label{appendix=C-1}

\subsubsection{Wave speed} 

The expressions of the wave speeds are obvious for the geometry variables, while concerning the fluid variables for \emph{sufficiently regular} solutions we can write (modulo geometric and lower-order terms denoted by $\geo$ and $\lot$, respectively) 
\bse
\label{equa-euleronly}
\be
\aligned
\Bigl( \Ebf_0(\Jperp) + \Ebf_0(\Jpar) - {q_\Jbb \over 2}  \, \Jbb \cdot \Jbb \Bigr)_t
- \amdeux^{-1} \, \Bigl( \Ebf_1(\Jperp) \Bigr)_x
& = \geo + \lot,
\\
\Bigl( \Ebf_1(\Jperp) \Bigr)_t
- \amdeux^{-1} \, \Bigl( \Ebf_0(\Jperp) - \Ebf_0(\Jpar) + {q_\Jbb \over 2}  \, \Jbb \cdot \Jbb \Bigr)_x
& = \geo + \lot,
\endaligned
\ee
coupled with~\eqref{eq:T2-fluid23-def-a}--\eqref{eq:T2-fluid23-def-b} which can be expanded similarly. Here, the constitutive function $q$ only depends on the squared norm 
\be
y\coloneqq - \Jbb \cdot \Jbb = J_0^2 - J_1^2 - J_2^2 - J_3^2 = 2 (\mu + p(\mu)) .
\ee
Recall that $y \geq 0$ and $J_0 \leq 0$, while $\Ebf_0(\Jperp) = (J_0^2 + J_1^2)/2$ and $\Ebf_0(\Jpar) = (J_2^2 + J_3^2)/2$.

Focusing on the principal part only, we obtain a system of four coupled partial differential equations:  
\bel{equa-Euler-Minko}
\begin{aligned} 
{1 \over 2} \bigl( J_0^2 + J_1^2 + J_2^2 + J_3^2 + y q \bigr)_t
- \amdeux^{-1} ( J_0 J_1 )_x
& = 0, \quad
& ( J_0 J_2 )_t - \amdeux^{-1} ( J_1 J_2 )_x
& = 0,
\\
( J_0 J_1 )_t - \amdeux^{-1} \, {1 \over 2} \bigl( J_0^2 + J_1^2 - J_2^2 - J_3^2 - y q \bigr)_x
& = 0,
& ( J_0 J_3 )_t - \amdeux^{-1} ( J_1 J_3 )_x
& = 0.
\end{aligned}
\ee 
We can write this system in the quasilinear form $A_0 \Jbb_t - \amdeux^{-1} A_1 \Jbb_x =0$ (for sufficiently regular solutions) with $A_0$ and $A_1$ given below in terms of
\be
(y q)' \coloneqq \frac{d(y q)}{dy}
= \frac{1-p'(\mu)}{1+p'(\mu)}
= \frac{1-k(\mu)^2}{1+k(\mu)^2}
\ee
in terms of the sound speed $k(\mu) = \sqrt{p'(\mu)}$.
The hyperbolicity conditions~\eqref{hyperbolic-eos} imply that $0<k(\mu)<1$ away from vacuum, namely $0<(y q)'<1$. Explicitly, the operators are
\be
\aligned
A_0 & \coloneqq
\begin{pmatrix}
(1+ (yq)') J_0 
& (1- (yq)') J_1
& (1- (yq)') J_2
& (1- (yq)') J_3
\\
J_1& J_0 & 0 & 0
\\
J_2 & 0 & J_0 & 0 
\\
J_3 & 0 & 0 & J_0 
\end{pmatrix}, 
\\
A_1 & \coloneqq 
\begin{pmatrix}
J_1& J_0 & 0 & 0
\\
(1- (yq)') J_0 & (1+(yq)') J_1 & (-1+(yq)') J_2 & (-1+(yq)') J_3
\\
0 & J_2 & J_1 & 0 
\\
0 & J_3 & 0 & J_1 
\end{pmatrix}.
\endaligned
\ee
The wave speeds are easily computed as the (real) roots $\xi$ of the polynomial $\det(\xi A_0 + \amdeux^{-1} A_1)$, namely the determinant of the matrix
\be
\begin{pmatrix}
\xi (1+(yq)') J_0 + \amdeux^{-1} J_1 
& \xi (1-(yq)') J_1 + \amdeux^{-1} J_0 
& \xi (1-(yq)') J_2
& \xi (1-(yq)') J_3
\\
\xi J_1 + \amdeux^{-1} (1-(yq)') J_0 & \xi J_0 + \amdeux^{-1} (1+(yq)') J_1 & \amdeux^{-1} (-1+(yq)') J_2 & \amdeux^{-1} (-1+(yq)') J_3
\\
\xi J_2 & \amdeux^{-1} J_2 & \xi J_0 + \amdeux^{-1} J_1 & 0 
\\
\xi J_3 & \amdeux^{-1} J_3 & 0 & \xi J_0 + \amdeux^{-1} J_1
\end{pmatrix},
\ee
which are $- \amdeux^{-1} J_1/J_0$ (double) and two other roots
\be
\aligned
& \xi_\pm \coloneqq
 - \amdeux^{-1} \frac{(1 - k^2) J_0 J_1  \pm k \sqrt{- \Jbb \cdot \Jbb} \sqrt{J_0^2 - J_1^2 - k^2 (J_2^2 + J_3^2)}}
{J_0^2 - k^2 (J_1^2 + J_2^2 + J_3^2)}.
\endaligned
\ee
Altogether the twelve wave speeds for $(\Pbb,\Qbb,\Kpar,\Omega,\amdeux,\Jbb)$ are (assuming that $J_0>0$) $\pm \amdeux^{-1}$ (double each), $0$ (quadruple), 
$- \amdeux^{-1} J_1/J_0$ (double), and~$\xi_{\pm}$. The latter speeds $\xi_{\pm}$ differ from each other and from all others (and are real) thanks to $0<k<1$ and $- \Jbb\cdot\Jbb>0$, so that the system is hyperbolic; indeed, the system also admits a full basis of eigenvectors. 

The same properties hold at the vacuum \emph{provided} the sound speed $p'(0)$ is non-vanishing. Otherwise, for vacuum states $\Jbb \cdot \Jbb=0$  \emph{and} when the sound speed $p'(0)$ vanishes, the Euler equations \eqref{equa-Euler-Minko} become 
\bel{equa-Euler-Minko-trivial}
\begin{aligned} 
 \bigl(  J_0 J_0\bigr)_t
- \amdeux^{-1} ( J_0 J_1 )_x
& = 0, \quad
& ( J_0 J_2 )_t - \amdeux^{-1} ( J_1 J_2 )_x
& = 0,
\\
( J_0 J_1 )_t - \amdeux^{-1} \, \bigl(  J_1^2 \bigr)_x
& = 0,
& ( J_0 J_3 )_t - \amdeux^{-1} ( J_1 J_3 )_x
& = 0, 
\end{aligned}
\ee 
all four wave speeds are equal to $- \amdeux^{-1} J_1/J_0$, but the system fails to be diagonalizable; the flux Jacobian has a nontrivial Jordan structure, hence the system is not symmetrizable and not a well-posed first-order hyperbolic system in the standard sense. However, this ill-posed system further simplifies \emph{in our setup}, as we also ensure
\be
\widehat J_m = \hNumb(\mu)^{-1} J_m \text{ is bounded}
\qquad\text{and}\qquad \hNumb(\mu)\to0 \text{ as } \mu\to0,
\ee
so that in the vacuum limit, we have $J_m=0$ and \eqref{equa-Euler-Minko-trivial} reduces to the two equations 
\bel{equa-Euler-Minko-trivial-2}
\begin{aligned} 
 \bigl(  J_0 J_0\bigr)_t
- \amdeux^{-1} ( J_0 J_1 )_x
& = 0, \quad
& ( J_0 J_1 )_t - \amdeux^{-1} ( J_1^2 )_x
& = 0. 
\end{aligned}
\ee 
which  ---since $J_0^2=J_1^2$--- coincide on each branch $J_1 = \pm J_0$ and \emph{well-posedness} is recovered. This feature is entirely analogous to the non-relativistic setup proposed by DiPerna \cite{DiPerna-1983a}. 
\ese
%
 

\subsubsection{Nonlinearities} 

We now characterize the genuine nonlinearity of the fluid characteristic fields. The other statements on the nonlinearities are clear from the expressions in~\autoref{section=2} before the statement of~\autoref{theo-struct}. For clarity, we restate the property as follows. 

\begin{lemma}
\label{lem:genuine-nonlinearity}
Assume the hyperbolicity condition~\eqref{hyperbolic-eos}, and consider a non-vacuum state satisfying
${
-\Jbb\cdot\Jbb>0.
}$
Then the double characteristic family associated with the speed
${
-\amdeux^{-1}\frac{J_1}{J_0}
}$
is linearly degenerate. Moreover, the two acoustic characteristic families associated with \(\xi_\pm\) are genuinely nonlinear in the sense of Lax if and only if
\bel{hyperbolic-eos-3-rep}
p''(\mu)+2\,\frac{1-p'(\mu)}{\mu+p(\mu)}\,p'(\mu)>0
\qquad\text{for all }\mu>0
\ee
or, equivalently, writing \(k(\mu)=\sqrt{p'(\mu)}\), 
\bel{equa-C3}
k'(\mu)+\frac{1-k(\mu)^2}{\mu+p(\mu)}\,k(\mu)>0.
\ee
\end{lemma}

\bse
The claim is local on the non-vacuum region, and genuine nonlinearity is invariant under smooth changes of variables. We may therefore replace the variables \(\Jbb\) by the fluid variables
${
(\mu,u_1,u_2,u_3),
}$
where \(u=(u_1,u_2,u_3)\) denotes the spatial velocity, and work with the principal part of the relativistic Euler equations in these coordinates.

In these variables, the double characteristic speed is
${
\lambda_0=-\amdeux^{-1}u_1,
}$
while the two acoustic speeds are
${
\lambda_\pm=-\amdeux^{-1}\frac{u_1\pm k(\mu)}{1\pm u_1k(\mu)}.
}$
The corresponding right eigenvectors may be chosen smoothly on the non-vacuum region. A direct computation then shows that
\be
\nabla \lambda_0 \cdot r_0 =0,
\ee
so that the double family is linearly degenerate.

For the acoustic families, the same computation yields
\be
\nabla \lambda_\pm \cdot r_\pm
=
\pm C_\pm(\mu,v)\,
\Bigl(
p''(\mu)+2\,\frac{1-p'(\mu)}{\mu+p(\mu)}\,p'(\mu)
\Bigr),
\ee
where, for a suitable normalization of the eigenvectors,
\be
C_\pm(\mu,v)
=
\frac{\amdeux^{-1}\,(\mu+p(\mu))}{2\,p'(\mu)\,\bigl(1\pm u_1\sqrt{p'(\mu)}\bigr)^2},
\ee
which is a smooth positive factor on the non-vacuum region. Hence
\be
\nabla \lambda_\pm \cdot r_\pm \neq 0
\quad\Longleftrightarrow\quad
p''(\mu)+2\,\frac{1-p'(\mu)}{\mu+p(\mu)}\,p'(\mu)\neq 0.
\ee
Under the strict convexity condition~\eqref{hyperbolic-eos-3-rep}, both acoustic families are therefore genuinely nonlinear. Finally, since \(p'(\mu)=k(\mu)^2\), the condition~\eqref{hyperbolic-eos-3-rep} is equivalently written as
\eqref{equa-C3}. 
This completes the proof.
\ese
%


\subsubsection{Constraints}
\bse
The propagation of constraints~\eqref{eq:theconstraints-def} on $\log\Omega, K_2, K_3$ is now checked. Consider first the functions
\bel{opL11-and-14}
\aligned
|t|^{3/2} \opL_{11} & = \bigl(  |t|^{3/2} K_2 \bigr)_t - |t|^{3/2} \Omega \bigl( J_1 J_2 -  P_0 K_2/2 \bigr) , \\
|t|^{3/2} \opL_{14} & = \bigl( |t|^{3/2} K_2 \bigr)_x - |t|^{3/2} \Omega \amdeux  ( J_0 J_2 - P_1 K_2 / 2) ,
\endaligned
\ee
whose vanishing constitute the evolution equation~\eqref{eq:T2-9101112-evol-def-a} and the constraint equation for~$K_2$, respectively.  The Euler equation~\eqref{eq:T2-fluid23-def-a} for~$J_2$ identifies derivatives of the second terms in~\eqref{opL11-and-14}, from which one finds
\be
\aligned
\bigl( |t|^{3/2} \opL_{14} \bigr)_t
& = \bigl(  |t|^{3/2} K_2 \bigr)_{xt} 
- \Bigl( |t|^{3/2} \Omega \amdeux  \bigl( J_0 J_2 - P_1 K_2 / 2 \bigr) \Bigr)_t 
\\
& = \bigl(  |t|^{3/2} K_2 \bigr)_{tx} - \Bigl( |t|^{3/2} \Omega \bigl( J_1 J_2 - P_0 K_2 / 2 \bigr) \Bigr)_x
= \bigl( |t|^{3/2} \opL_{11} \bigr)_x = 0 .
\endaligned
\ee
Since $\opL_{14}$ vanishes identically on one slice, it must vanish for all times. The same argument applies to the constraint for $K_3$, using the Euler equation~\eqref{eq:T2-fluid23-def-b} for $J_3$ and evolution equation~\eqref{eq:T2-9101112-evol-def-b} for~$K_3$.

On the other hand,  by considering
\be
2 \opL_{13} = \bigl( \log \Omega^2\bigr)_x
- t \, \Omega^2 \amdeux  \bigl( \Ebf_1(\Pbb) + \Ebf_1(\Qbb) + \Ebf_1(\Jperp) \bigr), 
\ee
we can check the constraint~\eqref{eq:T2-567-def} on the lapse. Using the evolution equation~\eqref{eq:evollambda-def}
we obtain
\be
\aligned
(2 \opL_{13})_t 
& = \Bigl( t \, \Omega^2 \amdeux \Mbf^{01}(\Jbb,\Kpar,\Pbb,\Qbb) \Bigr)_t
+ \bigl( \log \Omega^2\bigr)_{xt}
\\
& = \Bigl( t \, \Omega^2 \amdeux  \Mbf^{01}(\Jbb,\Kpar,\Pbb,\Qbb) \Bigr)_t
+ \Bigl( t \, \Omega^2 \Mbf^{11}(\Jbb,\Kpar,\Pbb,\Qbb) - {1\over 2t}\Bigr)_x
= 0, 
\endaligned
\ee
as it coincides with the second Euler equation in~\eqref{eq:T2-Euler-perp-def}. Therefore $\opL_{13}$ is constant in time and thus vanishes for all times if it vanishes at some initial time. This completes the proof of~\autoref{theo-struct}. 
\ese
%


\subsection{A fluid-geometry splitting}
\label{appendix=C-2}

\subsubsection{Retrieving the parallel Euler equations}

In the Einstein--Euler system \eqref{eq:T2-1234-def}--\eqref{eq:theconstraints-def}, the evolution equations~\eqref{eq:T2-Euler-perp-def} for~$\Jbb$ arise as one component of the Einstein equations, together with three compatibility equations $f_{tx}=f_{xt}$ for $f$ among $(\log\Omega,K_2,K_3)$. From this unusual form of the evolution equations, we now derive the standard Euler equations $\nabla_\alpha T^{\alpha\beta} = 0$ expressed in \JKL\ variables. The idea is to isolate on the left-hand side of~\eqref{eq:T2-Euler-perp-def} only the fluid contributions, and to simplify the contributions of geometry using the Einstein equations. These manipulations, and the resulting source terms, are only defined for sufficiently regular solutions.
For the components $\Jperp=(J_0,J_1)$ orthogonal to symmetry orbits, two distinct formulations will be of interest to us, so we begin with the comparatively simpler parallel directions $\Jpar=(J_2,J_3)$.

We keep on the left-hand side of~\eqref{eq:T2-fluid23-def-a}--\eqref{eq:T2-fluid23-def-b} only the fluid contributions, and evaluate the geometric terms using the identity
\be
\aligned
\curluntrois(f \Xbb) & = f \curluntrois\Xbb + X_1 e_0(f) - X_0 e_1(f) \\
& = f \curluntrois\Xbb + \Omega^{-1} X_1 f_t - \amdeux^{-1} \Omega^{-1} X_0 f_x
\endaligned
\ee
for any $\Tbb^2$-invariant scalar and vector fields~$f,\Xbb$.
This yields
\bel{equa-R2R3}
\aligned
\divuntrois \bigl(  \abs{t}^{1/2} J_m \Jperp\bigr) & = -  \abs{t}^{1/2} R_m, \quad m = 2,3 , \\
R_2 = {1 \over 2} (\Jperp\cdot\Pbb) J_2,
\qquad
R_3 & = (\Jperp\cdot\Qbb) J_2 - {1\over 2} (\Jperp\cdot\Pbb) J_3 .
\endaligned
\ee
For tame Einstein--Euler flows, the right-hand side belongs to
\(L^1(\Interval\times\Tbb^3,\dVuntrois)\) thanks to the spacetime integrability. 
Hence, the identity~\eqref{equa-R2R3} is well-defined in the sense of distributions.


\subsubsection{Retrieving the orthogonal Euler equations.}

We turn to the orthogonal directions.
Keeping on the left-hand side of~\eqref{eq:T2-Euler-perp-def} only the fluid contributions, we obtain
\bel{equa-OmM00-fluid}
\begin{aligned}
\divuntrois \bigl( \Omega \, \Mbf^{0 \bullet}(\Jbb, 0, 0, 0) \bigr) + {1 \over 2t} \Msource(\Jbb, 0, 0, 0)
& = \Omega \, R_0 , 
\\
\divuntrois \bigl( \Omega \, \Mbf^{1 \bullet}(\Jbb, 0, 0, 0) \bigr)
& = - \Omega \, R_1, 
\end{aligned}
\ee
in which we denote 
\be
\aligned
\Omega \, R_0 & \coloneqq - \divuntrois \bigl( \Omega \, \Mbf^{0 \bullet}(0, \Kpar, \Pbb, \Qbb) \bigr) 
- {1 \over 2t} \Msource(0, \Kpar, \Pbb, \Qbb), 
\\
\Omega \, R_1 & \coloneqq  \divuntrois \bigl( \Omega \, \Mbf^{1 \bullet}(0, \Kpar, \Pbb, \Qbb) \bigr) .
\endaligned
\ee
To retrieve the Euler equations, the terms $\Omega \, R_0$, $\Omega \, R_1$ must be recast in terms of fluid variables.
This relies on eliminating all derivatives using the Einstein evolution and constraint equations.
A direct calculation results in $R_0=R_0(\Jbb, \Kpar, \Pbb, \Qbb)$ and $R_1 =R_1(\Jbb, \Kpar, \Pbb, \Qbb)$ given by
\bel{equaR0R1}
\aligned
R_m
& = t \, \Omega \Bigl(
\Bigl( \Ebf_0(\Jpar) + {q_\Jbb \over 2} \, (-\Jbb \cdot \Jbb) \Bigr) \Ebf_m(\Pbb,\Qbb)
- \Ebf_0(\Kpar) \Ebf_m(\Jperp) \Bigr)
\\
& \quad - \frac{1}{2} (J_2^2 - J_3^2) P_m - J_2 J_3 Q_m - (J_2 K_2 + J_3 K_3) (\star J)_m, \qquad m=0,1,
\endaligned
\ee
where $(\star J)_0=J_1$ and $(\star J)_1=J_0$. 
Observe that some terms in the first line such as $(-\Jbb\cdot\Jbb)\Ebf_m(\Pbb,\Qbb)$ require more regularity than the one we assumed for tame Einstein--Euler flows (\autoref{def-weaksolu-deux}).
In fact, the derivation of~\eqref{equaR0R1} involves products of the evolution equations of $(\Pbb,\Qbb)$ by components of $(\Pbb,\Qbb)$, which are only meaningful for sufficiently regular solutions.

For clarity in the presentation, we state the following conclusion. 
Observe that all terms in $R_m$, $m=0,\dots,3$, are quadratic in the fluid variables (apart from the factor~$q_\Jbb$), and in particular vanish when $\Jbb=0$.  In contrast to $\Msource$ which is independent of the geometric variables, they depend on the geometry, and as a result are at least of \emph{cubic order}.

\begin{lemma}[Fluid evolution revisited]
\label{lem-fluid-evolution-revisited}
All sufficiently regular solutions to the Einstein--Euler equations in $\Tbb^2$ symmetry also satisfy the Euler equations
\bel{eq:T2-Euler-perp7-r}
\begin{aligned} 
2 \divuntrois \bigl( \Omega \, T^{0 \bullet} \bigr) + {1 \over 2t} \Msource(\Jbb, 0, 0, 0) 
& = \Omega \, R_0, 
\\
2 \divuntrois \bigl( \Omega \, T^{1 \bullet} \bigr)
& = - \Omega \, R_1, 
\\  
\divuntrois \bigl(  \abs{t}^{1/2} J_2 \Jperp\bigr) & = -  \abs{t}^{1/2} R_2,
\\
\divuntrois  \bigl(  \abs{t}^{1/2} J_3 \Jperp\bigr)& = -  \abs{t}^{1/2} R_3,
\end{aligned}
\ee
in which the right-hand sides $R_0,R_1$ and $R_2,R_3$ are given by the algebraic expressions \eqref{equaR0R1} and \eqref{equa-R2R3}, respectively. 
\end{lemma}

Concerning the left-hand side of \eqref{eq:T2-Euler-perp7-r}, we recall (cf.~\eqref{eq:T2-Mdef-0} and~\eqref{eq:T2-Euler-perp2-def0}) that 
\be
\aligned
2 T^{00} = \Mbf^{00}(\Jbb, 0, 0, 0) & = \Ebf_0(\Jperp) + \Ebf_0(\Jpar) + {q_\Jbb \over 2}  \, (- \Jbb \cdot \Jbb), 
\\
2 T^{01} = \Mbf^{01}(\Jbb, 0, 0, 0) & = - \Ebf_1(\Jperp), 
\\
2 T^{11} = \Mbf^{11}(\Jbb, 0, 0, 0)  & = \Ebf_0(\Jperp) - \Ebf_0(\Jpar) - {q_\Jbb \over 2}  \, (- \Jbb \cdot \Jbb),
\endaligned
\ee
while the corresponding source term is 
\be
\Msource(\Jbb, 0, 0, 0) = - \Jperp \cdot \Jperp + \Ebf_0(\Jpar) - {q_\Jbb \over 2}  \, (- \Jbb \cdot \Jbb).
\ee 


\subsubsection{A rescaling with integrable source terms}

By including a factor of $t^{-1}\Omega^{-2}$ into the divergences in~\eqref{eq:T2-Euler-perp7-r}, we can change the source terms $R_0,R_1$ into sources that are integrable in spacetime for tame Einstein--Euler flows.
This formulation will prove to be crucial in determining the source terms in entropy inequalities.
We obtain the following result, with source terms specified below the lemma.

\begin{lemma}\label{lem-Euler-defined}
All sufficiently regular solutions to the Einstein--Euler equations in $\Tbb^2$ symmetry also satisfy
\bel{Euler-defined-01}
\aligned
2 \divuntrois \bigl( t^{-1} \Omega^{-1} \, T^{0 \bullet} \bigr)
& = \Deltafluid^0(\Jbb,\Kpar,\Pbb,\Qbb) ,
\\
2 \divuntrois \bigl( t^{-1} \Omega^{-1} \, T^{1 \bullet} \bigr)
& = \Deltafluid^1(\Jbb,\Kpar,\Pbb,\Qbb) ,
\endaligned
\ee
with source terms explicited in \eqref{explicit-widecheck-R}, below.
\end{lemma}

The first source term is
\bse\label{explicit-widecheck-R}
\bel{explicit-widecheck-R0}
\aligned
\Deltafluid^0
& = \frac{1}{t \Omega} R_0 + \Mbf^{01}(\Jbb, \Kpar, \Pbb, \Qbb) \Mbf^{01}(\Jbb, 0,0,0) - \Bigl( \Mbf^{11}(\Jbb, \Kpar, \Pbb, \Qbb) + \frac{1}{2 t^2 \Omega^2} \Bigr) \Mbf^{00}(\Jbb, 0, 0, 0) \\
& \quad - \frac{1}{2t^2\Omega^2} \Msource(\Jbb, 0, 0, 0)
\\
& = - \frac{1}{4} \Bigl( (\Pbb\cdot\Jperp)^2 + (\Pbb\wedge\Jperp)^2 + (\Qbb\cdot\Jperp)^2 + (\Qbb\wedge\Jperp)^2\Bigr)
\\
& \quad - \frac{1-q_\Jbb}{4}(-\Jbb\cdot\Jbb) \Bigl((1+q_\Jbb)(-\Jbb\cdot\Jbb) + 4 \Ebf_0(\Jpar)\Bigr)
+ \Ebf_0(\Kpar) \Bigl(\frac{q_\Jbb}{2} (-\Jbb\cdot\Jbb) + \Ebf_0(\Jpar)\Bigr)
\\
& \quad - \frac{1}{t \Omega} \Bigl( \frac{1}{2} (J_2^2 - J_3^2) P_0 + J_2 J_3 Q_0 + (J_2 K_2 + J_3 K_3) J_1\Bigr)
- \frac{1}{4t^2\Omega^2} \bigl(3 J_0^2 - J_1^2 + 2 J_2^2 + 2 J_3^2 \bigr) .
\endaligned
\ee
The second source term is
\bel{explicit-widecheck-R1}
\aligned
\Deltafluid^1 & = - \frac{1}{t \Omega} R_1 + \Mbf^{01}(\Jbb, \Kpar, \Pbb, \Qbb) \Mbf^{11}(\Jbb, 0,0,0) - \Bigl( \Mbf^{11}(\Jbb, \Kpar, \Pbb, \Qbb) + \frac{1}{2 t^2 \Omega^2} \Bigr) \Mbf^{01}(\Jbb, 0, 0, 0)
\\
& = \frac{1}{2} (\Jperp\cdot\Pbb)(\Jperp\wedge\Pbb)
+ \frac{1}{2} (\Jperp\cdot\Qbb)(\Jperp\wedge\Qbb)
\\
& \quad + \frac{1}{t \Omega} \Bigl( \frac{1}{2} (J_2^2 - J_3^2) P_1 + J_2 J_3 Q_1 + (J_2 K_2 + J_3 K_3) J_0\Bigr)
+ \frac{1}{2t^2\Omega^2} J_0 J_1 .
\endaligned
\ee
\ese
In contrast to~\eqref{eq:T2-Euler-perp7-r}, the equations in \autoref{lem-Euler-defined} constitute a \emph{non-standard form} of the Euler equations since they include higher-order terms in the fluid variables, such as $(-\Jbb\cdot\Jbb)^2$.

\subsubsection{Quartic forms for the geometric energy}

It also proves useful to consider the divergence of other terms in the decomposition
\be
\Mbf^{a\bullet}(\Jbb,\Kpar,\Pbb,\Qbb) = 2 T^{a\bullet} + \Mbf^{a\bullet}(0,\Kpar,0,0) + \Mbf^{a\bullet}(0,0,\Pbb,\Qbb) , \qquad a = 0, 1.
\ee
Besides the evolution equations~\eqref{Euler-defined-01} for the matter, one has evolution equations for the twist part of~$\Mbf$,
\bse\label{energy-K-quartic}
\be
\aligned
\divuntrois \bigl( t^{-1} \Omega^{-1} \, \Mbf^{0 \bullet}(0, \Kpar, 0, 0) \bigr)
& = \Deltatwist^0(\Jbb,\Kpar,\Pbb,\Qbb) ,
\\
\divuntrois \bigl( t^{-1} \Omega^{-1} \, \Mbf^{1 \bullet}(0, \Kpar, 0, 0) \bigr)
& = \Deltatwist^1(\Jbb,\Kpar,\Pbb,\Qbb) ,
\endaligned
\ee
with the sources
\be
\aligned
\Deltatwist^0(\Jbb,\Kpar,\Pbb,\Qbb)
&  
= \Ebf_0(\Kpar) \Bigl( \Ebf_0(\Jpar,\Kpar) + \frac{q_\Jbb}{2} (-\Jbb\cdot\Jbb) - \frac{3}{t^2 \Omega^2} \Bigr) \\
& \quad - t^{-1} \Omega^{-1} \Bigl( P_0 (K_2^2 - K_3^2) / 2 + Q_0 K_2 K_3 - J_1 (J_2 K_2 + J_3 K_3) \Bigr) ,
\\
\Deltatwist^1(\Jbb,\Kpar,\Pbb,\Qbb)
& = t^{-1} \Omega^{-1} \Bigl( P_1 (K_2^2 - K_3^2) / 2 + Q_1 K_2 K_3 - J_0 (J_2 K_2 + J_3 K_3) \Bigr) .
\endaligned
\ee
\ese

Finally, the contribution of the first-order variables $(\Pbb,\Qbb)$ satisfies
\bse\label{energy-PQ-quartic}
\be 
\aligned
\divuntrois\Bigl( t^{-1}\Omega^{-1}\Mbf^{0\bullet}(0,0,\Pbb,\Qbb)\Bigr) 
& = \Deltageom^0(\Phi, \Psi),
\\
\divuntrois\Bigl( t^{-1}\Omega^{-1}\Mbf^{1\bullet}(0,0,\Pbb,\Qbb)\Bigr)
& = \Deltageom^1(\Phi, \Psi),
\endaligned
\ee
\bel{equa-jdjhd40}
\aligned
\Deltageom^0(\Phi, \Psi)
& = - \frac{1}{4}\Bigl((\Pbb\cdot\Pbb)^2+(\Qbb\cdot\Qbb)^2+2(\Pbb\cdot\Qbb)^2+2(\Pbb\wedge\Qbb)^2+(\Jperp\cdot\Pbb)^2 \\
& \qquad\qquad +(\Jperp\wedge\Pbb)^2+(\Jperp\cdot\Qbb)^2+(\Jperp\wedge\Qbb)^2\Bigr) \\
& \quad + \frac{1}{2t\Omega}\Bigl(P_0\Re(\Kpar^2+\Jpar^2)+Q_0\Im(\Kpar^2+\Jpar^2)\Bigr) \\
& \quad - \frac{1}{2t^2\Omega^2} \Bigl( \Ebf_0(\Pbb,\Qbb) - \Pbb\cdot\Pbb - \Qbb\cdot\Qbb \Bigr) ,
\\
\Deltageom^1(\Phi, \Psi)
& = - \frac{1}{2} \Bigl( (\Jperp\cdot\Pbb) (\Jperp\wedge\Pbb) + (\Jperp\cdot\Qbb) (\Jperp\wedge\Qbb)\Bigr) \\
& \quad - \frac{1}{2t\Omega}\Bigl(P_1\Re(\Kpar^2+\Jpar^2)+Q_1\Im(\Kpar^2+\Jpar^2)\Bigr) + \frac{1}{2t^2\Omega^2} \Ebf_1(\Pbb,\Qbb) .
\endaligned
\ee
\ese

\begin{lemma}\label{lem:Delta-integrable}
The spacetime integrability properties established in \autoref{section=7} (cf.~\autoref{theo:te-Euler-explicit}) ensure that
\be
\Deltafluid^0 , \ \Deltafluid^1 , \
\Deltatwist^0 , \ \Deltatwist^1 , \
\Deltageom^0 , \ \Deltageom^1 , \
\in L^1(\Interval\times\Tbb^3, \dVuntrois) ,
\ee
so that \eqref{Euler-defined-01}, \eqref{energy-K-quartic}, and \eqref{energy-PQ-quartic} are defined in the sense of distributions.
\end{lemma}
 

\subsection{Quasi-balance laws (proof for \autoref{section=3-4})}
\label{appendix=C-3}

\subsubsection{Block diagonalization of the Euler equations}

In this section we establish various identities related to entropy quasi-currents.  We work throughout with sufficiently smooth solutions of the Einstein--Euler equations in $\Tbb^2$~symmetry.  The expressions in Minkowski space are easily retrieved by suppressing all source terms.
We recall that the fluid stress tensor is~\eqref{eq:T2-Texpr}
\be
T^{mn} = \frac{1}{2} \Mbf^{mn}(\Jbb,0,0,0) , \qquad m,n = 0,1 .
\ee

Recall from \autoref{section=3-4} the notation $\Sigma=(\Sigma_0,\Sigma_1)$ with $\Sigma_0 = -\Jbb\cdot\Jbb$ and $\Sigma_1 = J_1/J_0$.
Our first aim is to show that this parametrization leads to a convenient form~\eqref{Sigmat} for the Euler system.

We observe first that the map $\Jbb \mapsto (\Sigma,\Jhatpar)$ is bijective away from vacuum.
To begin, the map $\mu\mapsto 2(\mu+p(\mu))$ is a bijection of $[0,+\infty)$ whose inverse bijection $\funmu$ is thus such that $\mu = \funmu(\Sigma_0)$.
Thus $\Sigma_0$ determines uniquely the density~$\mu$ hence the particle density $\hNumb_\Jbb=\hNumb(\funmu(\Sigma_0))$, from which the parallel momentum~$\Jpar$ is easily determined.  Then all components of $\Jbb$ are given by
\bel{J0J1-Sigma-Jhatpar}
\Jpar = \hNumb(\funmu(\Sigma_0)) \, \Jhatpar ,
\qquad
J_0 = - \Bigl( \frac{\Sigma_0 + |\Jpar|^2}{1 - \Sigma_1^2} \Bigr)^{1/2} ,
\qquad
J_1 = J_0 \Sigma_1 .
\ee

For sufficiently regular solutions of the Einstein--Euler equations in $\Tbb^2$~symmetry, the conservation of $\vNumb$ and of $\widehat J_m \vNumb$ yield transport equations for $\widehat J_m$, and \autoref{lem-Euler-defined} provides two other evolution equations, which are conveniently written in terms of frame derivatives $e_0(f)=\Omega^{-1}\del_t f$, $e_1(f)=\Omega^{-1}\amdeux^{-1}\del_x f$ (for $\Tbb^2$-invariant~$f$)
Explicitly,
\bel{explicit-Euler-Deltafluid}
\aligned
2 \amdeux^{-1} e_0(\amdeux T^{00}) - 2 e_1(T^{01})
& = t \Omega \Deltafluid^0(\Jbb,\Kpar,\Pbb,\Qbb) ,
\\
2 \amdeux^{-1} e_0(\amdeux T^{01}) + e_1(T^{11})
& = t \Omega \Deltafluid^1(\Jbb,\Kpar,\Pbb,\Qbb) ,
\\
e_0(\widehat J_m) - \frac{J_1}{J_0} e_1(\widehat J_m) & = \widehat R_m , \qquad m = 2, 3 ,
\endaligned
\ee
with
\bel{widehatR-def}
\aligned
\widehat R_2 & = \frac{\Jperp\cdot\Pbb}{2J_0} \widehat J_2 - \frac{1}{2t\Omega} \widehat J_2 ,
\\
\widehat R_3 & = \frac{\Jperp\cdot\Qbb}{J_0} \widehat J_2 - \frac{\Jperp\cdot\Pbb}{2J_0} \widehat J_3 - \frac{1}{2t\Omega} \widehat J_3 .
\endaligned
\ee

Next, note that
\be
\begin{alignedat}{3}
2 T^{00}
& = J_0^2 - \frac{1-q_\Jbb}{2} (-\Jbb\cdot\Jbb)
&& = J_0^2 - \frac{1-q(\funmu(\Sigma_0))}{2} \Sigma_0 ,
\\
2 T^{10}
& = - J_0 J_1 && = - J_0^2 \Sigma_1 ,
\\
2 T^{11}
& = J_1^2 + \frac{1-q_\Jbb}{2} (-\Jbb\cdot\Jbb)
&& = J_0^2 \Sigma_1^2 + \frac{1-q(\funmu(\Sigma_0))}{2} \Sigma_0 ,
\end{alignedat}
\ee
Then in these expressions, $J_0^2$ can be expressed in terms of $(\Sigma,\Jhatpar)$ using~\eqref{J0J1-Sigma-Jhatpar}.
In particular, derivatives in $\widehat J_m$, keeping the other coordinates $(\Sigma,\Jhatpar)$ fixed, are
\be
\aligned
\frac{\del T^{00}}{\del\widehat J_m}
= \frac{\hNumb(\funmu(\Sigma_0))^2}{1 - \Sigma_1^2} \widehat J_m ,
\qquad
\frac{\del T^{11}}{\del\widehat J_m}
= - \frac{\del T^{01}}{\del\widehat J_m} \Sigma_1
= \frac{\del T^{00}}{\del\widehat J_m} \Sigma_1^2 .
\endaligned
\ee
This proportionality of the derivatives leads to valuable cancellations when expanding the first two equations in~\eqref{explicit-Euler-Deltafluid}:
for $a=0,1$,
\bel{eiSigmab}
\aligned
0 & = \amdeux^{-1} e_0(\amdeux T^{a0}) + e_1(T^{a1})
- t \Omega \Deltafluid^a(\Jbb,\Kpar,\Pbb,\Qbb) / 2
\\
& = \sum_{b=0,1} \Bigl( \frac{\del T^{a0}}{\del\Sigma_b} e_0(\Sigma_b)
+ \frac{\del T^{a1}}{\del\Sigma_b} e_1(\Sigma_b) \Bigr)
\\
& \quad + t \Omega \Bigl(\Ebf_0(\Jpar,\Kpar) +\frac{q_\Jbb}{2}\,(-\Jbb\cdot\Jbb) \Bigr) \Mbf^{a0}
- \frac{J^a R^\parallel}{-\Jperp\cdot\Jperp}
- t \Omega \Deltafluid^a(\Jbb,\Kpar,\Pbb,\Qbb) / 2
\endaligned
\ee
with
\bel{Rparallel-def}
R^\parallel
= \sum_{m=2,3} J_0 J_m \hNumb \widehat R_m
= \frac{1}{2} (\Jperp\cdot\Pbb) (J_2^2 - J_3^2) + (\Jperp\cdot\Qbb) J_2 J_3 - \frac{1}{2t\Omega} J_0 (J_2^2 + J_3^2) .
\ee
In Minkowski space all the source terms disappear and the pair of equations (for $a=0,1$) take the homogeneous form $0 = A^0(\Sigma,\Jhatpar)e_0(\Sigma) + A^1(\Sigma,\Jhatpar)e_1(\Sigma)$ for a pair of matrices $A^0,A^1$, which can equivalently be stated in the form $\del_t\Sigma = A \del_x\Sigma$ as announced in~\eqref{Sigmat}.

\subsubsection{Shear-induced measures}

We now consider a current $\Fbb=(\Fcomp^0,\Fcomp^1)$ expressed as a function of the variables $(\Sigma,\Jhatpar)$, and evaluate its (weighted) divergence
\be
\aligned
t \Omega \divuntrois\Bigl(t^{-1} \Omega^{-1} \Fbb\Bigr)
& = \amdeux^{-1} \Omega^{-1} \Bigl( (\amdeux \Fcomp^0)_t + (\Fcomp^1)_x \Bigr)
\\
& = t \Omega \Bigl(\Ebf_0(\Jpar, \Kpar) + \frac{q_\Jbb}{2}  \, (- \Jbb \cdot \Jbb)\Bigr) \Fcomp^0
+ \Omega^{-1} (\Fcomp^0)_t + \amdeux^{-1} \Omega^{-1} (\Fcomp^1)_x .
\endaligned
\ee
Using the chain rule we reduce to linear combinations of derivatives of $(\Sigma,\Jhatpar)$.
(Throughout, the partial derivatives of components $\Fcomp^a$ are taken while keeping other variables in the set $(\Sigma,\Jhatpar)$ fixed.)
The contributions of derivatives of $\widehat J_m$, $m=2,3$, are
\be
\Omega^{-1} \frac{\del\Fcomp^0}{\del\widehat J_m} (\widehat J_m)_t
+ \amdeux^{-1} \Omega^{-1} \frac{\del\Fcomp^1}{\del\widehat J_m} (\widehat J_m)_x
= \frac{\del\Fcomp^0}{\del\widehat J_m} \Omega \widehat R_m
+ \Bigl( \frac{\del\Fcomp^1}{\del\widehat J_m} + \frac{J_1}{J_0} \frac{\del\Fcomp^0)}{\del\widehat J_m} \Bigr) e_1(\widehat J_m) .
\ee
Since $J_1/J_0 = \Sigma_1$ is kept fixed when taking the $\widehat J_m$ derivative, the coefficient of $e_1(\widehat J_m)$ is, as announced in~\eqref{Fvee-expr} (in \autoref{lem:Fvee-expr})
\be
\Fvee_m = \frac{\del(\Fcomp^1 + (J_1/J_0)\Fcomp^0)}{\del\widehat J_m}, \qquad m=2,3 .
\ee
Using this notation,
we obtain
\bel{diveFbb-C28}
\aligned
t \Omega \divuntrois\Bigl(t^{-1} \Omega^{-1} \Fbb\Bigr)
& = \sum_{b=0,1} \frac{\del\Fcomp^0}{\del\Sigma_b} e_0(\Sigma_b)
+ \sum_{b=0,1} \frac{\del\Fcomp^1}{\del\Sigma_b} e_1(\Sigma_b)
+ \sum_{m=2,3} \Fvee_m e_1(\widehat J_m)
\\
& \quad
+ \sum_{m=2,3} \frac{\del\Fcomp^0}{\del\widehat J_m} \widehat R_m
+ t \Omega \Bigl(\Ebf_0(\Jpar, \Kpar) + \frac{q_\Jbb}{2}  \, (- \Jbb \cdot \Jbb)\Bigr) \Fcomp^0 .
\endaligned
\ee

The next step is to use~\eqref{eiSigmab} to express $e_0(\Sigma_b)$ in terms of $e_1(\Sigma_b)$ and a source term.
Entropy quasi-currents are characterized by the fact that $e_1(\Sigma_b)$ terms then cancel in~\eqref{diveFbb-C28}, which occurs provided there exists $(V^0,V^1)$ such that
\bel{Fcompatibility}
\frac{\del\Fcomp^0}{\del T^{00}}\Biggl|_{T^{01},\Jhatpar}
= \frac{\del\Fcomp^1}{\del T^{01}}\Biggl|_{T^{11},\Jhatpar} = 2 V^0
,
\qquad
\frac{\del\Fcomp^0}{\del T^{01}}\Biggl|_{T^{00},\Jhatpar}
= \frac{\del\Fcomp^1}{\del T^{11}}\Biggl|_{T^{01},\Jhatpar} = 2 V^1 .
\ee
Here, the left-hand sides are partial derivatives of $\Fcomp^0$ in the variables $(T^{00}, T^{01}, \Jhatpar)$, and the right-hand sides are partial derivatives of $\Fcomp^1$ in the variables $(T^{01}, T^{11}, \Jhatpar)$.

Under the condition~\eqref{Fcompatibility}, the divergence~\eqref{diveFbb-C28} becomes
\be
\divuntrois\Bigl(t^{-1} \Omega^{-1} \Fbb\Bigr)
- \sum_{m=2,3} t^{-1} \Omega^{-1} \Fvee_m e_1(\widehat J_m)
= \sign(t) \Delta_\Fbb
\ee
with source term\footnote{We include an unnatural sign $\sign(t)$ here to simplify the presentation in the main text.}
\bel{source-DeltaFbb}
\aligned
\sign(t) \Delta_\Fbb
& = \Bigl(\Ebf_0(\Jpar, \Kpar) + \frac{q_\Jbb}{2}  \, (- \Jbb \cdot \Jbb)\Bigr) \Bigl( \Fcomp^0 - 2 \sum_a V_a T^{0a} \Bigr)
\\
& \quad + \sum_{a=0,1} V_a \Deltafluid^a(\Jbb,\Kpar,\Pbb,\Qbb)
+ \frac{2}{t \Omega} \frac{\Vbb\cdot\Jperp}{-\Jperp\cdot\Jperp} R^\parallel
+ \sum_{m=2,3} \frac{1}{t\Omega} \frac{\del\Fcomp^0}{\del\widehat J_m}\biggr|_{\Sigma,\Jhatpar} \widehat R_m
\endaligned
\ee
where $\Vbb\cdot\Jperp = V^0 J_0 + V^1 J_1$.

\subsubsection{Integrability properties for weak solutions}

We now discuss features of each term, relying on the integrability properties of Einstein--Euler flows.
Observe that
\bel{Fcomp0sumVa}
\Fcomp^0 - 2 \sum_a V_a T^{0a}
= \Fcomp^0 - T^{00} \frac{\del\Fcomp^0}{\del T^{00}}\biggr|_{T^{01},\Jhatpar}
- T^{01} \frac{\del\Fcomp^0}{\del T^{01}}\biggr|_{T^{00},\Jhatpar} ,
\ee
so the first term disappears for the two entropies $\Fbb=\Mbf^{0\bullet}$ and $\Fbb=\Mbf^{1\bullet}$.
For a general entropy, \eqref{Fcomp0sumVa} is multiplied by an $L^2(\Interval\times\Tbb^3,\dVuntrois)$ function, so the entropies must be chosen such that~\eqref{Fcomp0sumVa} is controlled by $-\Jperp\cdot\Jperp$, which is in $L^2(\Interval\times\Tbb^3,\dVuntrois)$.

For the next two terms we first provide alternative expressions of~$V^a$.
The compatibility condition~\eqref{Fcompatibility} can also be stated as the fact that the differentials of $\Fcomp^0$ and~$\Fcomp^1$ (one-forms on the space of fields~$\Jbb$) obey
\be
d\Fcomp^0 = 2 V_0 dT^{00} + 2 V_1 dT^{01} + \# d\widehat J_2 + \# d\widehat J_3 ,
\qquad
d\Fcomp^1 = 2 V_0 dT^{01} + 2 V_1 dT^{11} + \# d\widehat J_2 + \# d\widehat J_3
\ee
where $\#$ are some coefficients.
By taking the sum and difference we find
\be
\aligned
d(\Fcomp^0+\Fcomp^1)
& = - \frac{1}{2} (V^0+V^1) d\bigl(q\Sigma_0+J_2^2+J_3^2\bigr)
- \frac{1}{2} (V^0-V^1) d\bigl((J_0-J_1)^2\bigr)
+ \# d\widehat J_2 + \# d\widehat J_3 ,
\\
d(\Fcomp^0-\Fcomp^1)
& = - \frac{1}{2} (V^0+V^1) d\bigl((J_0+J_1)^2\bigr)
- \frac{1}{2} (V^0-V^1) d\bigl(q\Sigma_0+J_2^2+J_3^2\bigr)
+ \# d\widehat J_2 + \# d\widehat J_3 .
\endaligned
\ee
Since $J_2^2+J_3^2 = (\widehat J_2^2+\widehat J_3^2)\hNumb(\funmu(\Sigma_0))$ only depends on $\Jhatpar$ and $\Sigma_0$,
this enables us to give convenient expressions for $V^0\pm V^1$ as derivatives with respect to $J_0\pm J_1$, with $\Jhatpar$ and $\Sigma_0$ fixed.
Since $\Jpar=\Jhatpar\hNumb(\funmu(\Sigma_0))$ and $J_0^2-J_1^2=\Sigma_0+|\Jpar|^2$, fixing these variables is equivalent to fixing $J_0^2-J_1^2,\Jpar$, which yields (after lowering some indices for convenience)
\be
V_0 \mp V_1 = \frac{1}{J_0\pm J_1} \frac{d(\Fcomp_0\pm\Fcomp_1)}{d(J_0\pm J_1)}\Bigr|_{J_0^2-J_1^2,\Jpar} .
\ee
As a consequence, we find
\be
\Vbb\cdot\Jperp = - V_0 J_0 + V_1 J_1
= \frac{1}{2} \sum_{\pm} (V_0 \mp V_1) (J_0 \pm J_1)
= \frac{1}{2} \sum_{\pm} \frac{d(\Fcomp_0\pm\Fcomp_1)}{d(J_0\pm J_1)}\Bigr|_{J_0^2-J_1^2,\Jpar} .
\ee

For weak solutions, the source term~$\Deltafluid^a$ is in $L^1(\Interval\times\Tbb^3,\dVuntrois)$, as stated in \autoref{lem:Delta-integrable}.
Likewise, the combinations $R^\parallel$ and $\widehat R_m$ given in~\eqref{Rparallel-def} and~\eqref{widehatR-def}, are in the same space. Thus, to ensure that all sources in~\eqref{source-DeltaFbb} are integrable, we must impose a uniform bound on some derivatives of the entropies.
Altogether, we arrive at the following result.

\begin{lemma}\label{lem:conditions-Fbb-ap}
Assume that the entropy current $\Fbb$ and its entropy multiplier $\Vbb$ satisfy the uniform bounds
\be
\aligned
\Biggl| \Fcomp^0 - T^{00} \frac{\del\Fcomp^0}{\del T^{00}}\biggr|_{T^{01},\Jhatpar}
- T^{01} \frac{\del\Fcomp^0}{\del T^{01}}\biggr|_{T^{00},\Jhatpar} \Biggr| & \lesssim -\Jperp\cdot\Jperp ,
\\
\bigl| \Vbb\cdot\Jperp \bigr| & \lesssim  -\Jperp\cdot\Jperp ,
\\
|V^0| + |V^1| + \Biggl| \frac{\del\Fcomp^0}{\del\widehat J_2}\biggr|_{\Sigma,\widehat J_3} \Biggr|  + \Biggl| \frac{\del\Fcomp^0}{\del\widehat J_3}\biggr|_{\Sigma,\widehat J_2} \Biggr| & \lesssim 1 .
\endaligned
\ee
Then the integrability properties assumed for Einstein--Euler flows ensure that the source term $\Delta_\Fbb = \Delta_\Fbb(\Jbb,\Kpar,\Lbb)$ defined in~\eqref{source-DeltaFbb} belongs to the space $L^1(\Interval\times\Tbb^3,\dVuntrois)$ of integrable functions.
\end{lemma}


\section{Further properties of the Euler equations}
\label{appendix=D}

\subsection{A technical inequality on the source term}
\label{appendix=D-1}

We claim that the source term of the Euler equations satisfies 
\be
\abs{\Msource}\leq 5 \, \Mbf^{00}. 
\ee
\bse
Indeed, by the elementary inequality $\bigl|X_0^2-X_1^2\bigr|\leq X_0^2+X_1^2=2\,\Ebf_0(\Xbb)$, valid for any $\Xbb=(X_0,X_1)$, we have
\be
|\Pbb\cdot\Pbb|\leq 2\,\Ebf_0(\Pbb),\qquad
|\Qbb\cdot\Qbb|\leq 2\,\Ebf_0(\Qbb),\qquad
|\Jperp\cdot\Jperp|=-\Jperp\cdot\Jperp\leq 2\,\Ebf_0(\Jperp),
\ee
since $-\Jperp\cdot\Jperp\ge0$. Using $\Jbb\cdot\Jbb\le0$ and $q_\Jbb\in(0,1)$, we obtain 
\be
\aligned
|\Msource|
& \leq \Ebf_0(\Jpar) + 5\,\Ebf_0(\Kpar)+2\,\Ebf_0(\Pbb)+2\,\Ebf_0(\Qbb)+2\,\Ebf_0(\Jperp)
   +\frac12 q_\Jbb\,(-\Jbb\cdot\Jbb)
\\
& \leq 5\Bigl(\Ebf_0(\Jpar,\Kpar)+\Ebf_0(\Pbb,\Qbb,\Jperp)
      +\frac12 q_\Jbb\,(-\Jbb\cdot\Jbb)\Bigr)
=5\,\Mbf^{00}.
\endaligned
\ee
\ese
In the Gowdy case, since \(\Kpar=0\) and \(\Jpar=0\), the source term simplifies and the above estimate can be improved. More precisely, in the isothermal Gowdy-symmetry case we obtain
\bel{equa-Gowdy facteur2}
|\Msource|\leq 2\,\Mbf^{00} \quad \text{(isothermal, Gowdy-symmetry).} 
\ee 


\subsection{The non-relativistic limit}
\label{appendix=D-2}

To justify our formalism for weak solutions of the relativistic isentropic compressible Euler equations, especially the choice of which of the fluid equations is weakened to an inequality, we make contact with the literature on nonrelativistic fluids.  We consider the plane-symmetric Euler equations in Minkowski space, which can be read off from the Euler--Einstein equations by suppressing the geometry (setting $\amdeux=\Omega=1$, $P_0=P_1=Q_0=Q_1=0$ and eliminating factors of $|t|$ in the operators). It is convenient to return to the notation $(\mu,u)$ instead of $J_m = \sqrt{2(\mu+p)} u_m$ for $m=0,1,2,3$. The Euler equations read
\bel{relativistic-Euler-isentropic}
\aligned
\del_t \Bigl( (\mu+p) u_0^2 - p \Bigr) - \del_x ( (\mu+p) u_0 u_1 ) & = 0, \\
\del_t ( (\mu+p) u_0 u_1 ) - \del_x \Bigl( (\mu+p) u_1^2 + p \Bigr) & = 0, \\
\del_t ( (\mu+p) u_0 u_m ) - \del_x ( (\mu+p) u_1 u_m ) & = 0, \quad m=2,3,
\endaligned
\ee
supplemented by the conservation of particle number~$\Numb(\mu)$,
\bel{relativistic-Euler-isentropic-N}
\del_t ( \Numb(\mu) u_0 ) - \del_x ( \Numb(\mu) u_1 ) = 0, \qquad \frac{\Numb'(\mu)}{\Numb(\mu)} = \frac{1}{\mu+p(\mu)} .
\ee

The non-relativistic limit corresponds to taking, for some small constant $\eps\in(0,1)$, the regime $|u_m|<\eps$ for $m=1,2,3$, hence $u_0=\sqrt{1+u_1^2+u_2^2+u_3^2}=1+\Obig(\eps^2)$, and an equation of state with $0<p'(\mu)<\eps^2$ hence $0\leq p(\mu)\leq \eps^2\mu$.  The \emph{formal} $\eps\to 0$ limit of the system~\eqref{relativistic-Euler-isentropic} is easily found to be the non-relativistic Euler equations
\bel{nonrel-Euler}
\aligned
\del_t \mu - \underline{\del}_x (\mu \underline{u}_1) & = 0, \\
\del_t (\mu \underline{u}_1) - \underline{\del}_x(\mu \underline{u}_1^2 + \underline{p}(\mu)) & = 0, \\
\del_t ( \mu \underline{u}_m ) - \underline{\del}_x ( \mu \underline{u}_1 \underline{u}_m ) & = 0, \quad m=2,3,
\endaligned
\ee
with $\underline{\del}_x=\eps\del_x$ and $\underline{p}(\mu)=p(\mu)/\eps^2$ and $\underline{u}=(u_1,u_2,u_3)/\eps$.

The equation for the particle number can be integrated from some reference density~$\muref$ as
\be
\frac{\Numb(\mu)}{\Numb(\muref)}
= \frac{\mu}{\muref} - \eps^2 \frac{\mu}{\muref} \int_{\muref}^\mu \frac{\underline{p}(\mu')d\mu'}{\mu'^2} + \Obig(\eps^4)
= \frac{\mu}{\muref} + \Obig(\eps^2),
\ee
where the constants implicit in $\Obig$ depend on the sup norm of $\log(\mu/\muref)$.
Up to rescaling by $\Numb(\muref)/\muref$, the formal limit of~\eqref{relativistic-Euler-isentropic-N} is $\del_t\mu - \underline{\del}_x(\mu\underline{u}_1) = 0$, just as the $0$-th Euler equation.
To account for all equations, we consider the difference between energy-momentum and particle number currents, whose components are
\be
\aligned
\Bigl( (\mu+p(\mu)) u_0^2 - p(\mu) \Bigr) - \frac{\muref}{\Numb(\muref)} \Numb(\mu) u_0
& = \eps^2 \Bigl( \mu |\underline{u}|^2 / 2 + \mu \int_{\muref}^\mu \frac{\underline{p}(\mu')d\mu'}{\mu'^2} \Bigr) + \Obig(\eps^4),
\\
\frac{1}{\eps} \Bigl( (\mu+p) u_0 u_1 - \frac{\muref}{\Numb(\muref)} \Numb(\mu) u_1 \Bigr)
& = \eps^2 \Bigl( \mu |\underline{u}|^2 / 2
+ \underline{p}(\mu)
+ \mu \int_{\muref}^\mu \frac{\underline{p}(\mu')d\mu'}{\mu'^2} \Bigr) \underline{u}_1
+ \Obig(\eps^4) .
\endaligned
\ee
The leading-order term in the $\eps\to 0$ limit reproduces the nonrelativistic energy and energy flux. It is then standard to consider weak solutions to the nonrelativistic Euler equations~\eqref{nonrel-Euler} that are solutions to the energy inequality
\be
\del_t \Bigl( \mu |\underline{u}|^2 / 2 + \mu \int_{\muref}^\mu \frac{\underline{p}(\mu')d\mu'}{\mu'^2} \Bigr)
- \underline{\del}_x \Bigl( \Bigl( \mu |\underline{u}|^2 / 2
+ \underline{p}(\mu)
+ \mu \int_{\muref}^\mu \frac{\underline{p}(\mu')d\mu'}{\mu'^2} \Bigr) \underline{u}_1 \Bigr)
\leq 0 .
\ee
This is consistent with our choice to impose the particle number equation as an equality and the energy equation ---namely the $0$-th component of the stress-tensor conservation equation--- as an inequality. It would also be consistent with assuming energy conservation and an \emph{increase} of the particle number, but this might lead to an unstable system.


\subsection{Convexity of the fluid particle number}
\label{appendix=D-3}

Let us point out the following well-known convexity property, although it is not used in our theory. Consider the \emph{choice} of principal variables in~\eqref{eq:T2-Mdef-0-2}, below, associated with the Euler equations when the geometry is \emph{formally suppressed}, namely when taking $\Pbb, \Qbb, \Kpar$ to vanish (cf.~\eqref{equa-euleronly} in~\autoref{appendix=C-1}).
The following convexity property is established for instance in~\cite{LeFloch-curved}.

\begin{proposition}[Future convexity of the particle number density]
\label{propo-entropy}
Consider the Euler equations for a perfect fluid satisfying the conditions \eqref{hyperbolic-eos}, and consider the four stress-tensor components $T^{0\bullet}$ in the frame $(e_m)$, which are explicitly given by
\bel{eq:T2-Mdef-0-2}
\aligned
T^{00} & \coloneqq \frac{1}{2} \Mbf^{00}(\Jbb,0,0,0) = \frac{1 + q_\Jbb}{4} (J_0)^2 
+  \frac{1 - q_\Jbb}{4} \bigl( (J_1)^2 + (J_2)^2 + (J_3)^2 \bigr),
\\
T^{01} & \coloneqq \frac{1}{2} \Mbf^{01}(\Jbb,0,0,0) = - \frac{1}{2} J_0 J_1, \qquad 
T^{02} \coloneqq - \frac{1}{2} J_0 J_2, \qquad 
T^{03} \coloneqq - \frac{1}{2} J_0 J_3.
\endaligned
\ee
Then, the time-component of the particle number current (with a suitable sign), seen as a function of the variables~$T^{0\bullet}$,
\be
\aligned
& T^{0\bullet} \mapsto - \Numb^0 = \hNumb(\mu) J_0  
 \text{ is  \emph{well-defined},  \emph{non-positive}, and \emph{convex}.}
\endaligned
\ee 
\end{proposition} 


\subsection{Convexity of the timelike mass-energy component (proof of \autoref{lem-convexM})}
\label{appendix=D-4}

The convexity property stated in \autoref{lem-convexM} is geometrically equivalent to \autoref{propo-entropy}, in the sense that both arise from the same convex family of timelike hypersurfaces in the space of principal variables. Since this equivalence is not immediate at the analytic level, we prefer to give below a direct proof of the precise convexity property used in the present paper.  
\bse
Let us set
\be
\rho \coloneqq Z_0 = \Numb^0 = - \hNumb_\Jbb J_0 \geq 0 ,
\qquad
Z_m \coloneqq T^{0m} = - \frac{1}{2} J_0 J_m , \quad m=1,2,3,
\ee
and introduce the notation
\be
y \coloneqq -\Jbb\cdot\Jbb = 2(\mu+p(\mu)) , \qquad
|\vec{Z}|^2 \coloneqq Z_1^2+Z_2^2+Z_3^2 .
\ee
We recall that $\hNumb_\Jbb$ can be seen as a function of the density~$\mu$ or of~$y$ through~\eqref{Sdef-scaled},
\bel{Sdef-scaled-revised}
\frac{d\log \hNumb}{d\mu} = \frac{1-p'(\mu)}{2(\mu + p(\mu))} ,
\qquad
\frac{d \log \hNumb}{dy} = \frac{1}{2y} \frac{d(yq)}{dy} > 0 .
\ee
\ese

\bse\label{J0Jm-rhoZ-full}
We prove first that $y$ (or equivalently the density~$\mu$) can be retrieved from $\rho,Z_1,Z_2,Z_3$.  Observe that
\bel{J0Jm-rhoZ}
J_0 = - \hNumb^{-1}\rho ,
\qquad
J_m = 2 \hNumb \rho^{-1} Z_m ,
\ee
which leads to the identity $y = \hNumb(y)^{-2}\rho^2 - 4 |\vec{Z}|^2 \hNumb(y)^2 \rho^{-2}$ or equivalently $\phi_\rho(y) = 4 |\vec{Z}|^2$ where
\be
\phi_\rho(y) \coloneqq \hNumb(y)^{-2}\rho^2 \bigl(\hNumb(y)^{-2}\rho^2 - y\bigr) .
\ee
Since $\hNumb$ is monotonically increasing (see~\eqref{Sdef-scaled-revised}), we deduce that $y\hNumb(y)^2$ is a monotonically increasing bijection of $[0,+\infty)$, hence there exists a unique $y_\rho>0$ such that $y_\rho \hNumb(y_\rho)^2 = \rho^2$.  On the interval $(0,y_\rho)$ the function $\phi_\rho$ is strictly positive.
We then compute the derivative of~$\phi_\rho$ using~\eqref{Sdef-scaled-revised},
\be
\del_y \phi_\rho(y) = - 4 \frac{d\log\hNumb}{dy} \phi_\rho(y) - \hNumb(y)^{-2}\rho^2 \Bigl( 1 + \frac{d(yq)}{dy} \Bigr) < 0 ,
\qquad y \in (0, y_\rho) .
\ee
From the near-vacuum expansion~\eqref{near-vacuum-Numb} of~$\hNumb$, we also learn that $\phi_\rho(y)\to+\infty$ as $y\to 0$.
All in all, $\phi_\rho:(0,y_0)\to(0,+\infty)$ is a monotonically decreasing bijection, which allows $y\in(0,y_0)$ to be uniquely retrieved as $y=\phi_\rho^{-1}(4|\vec{Z}|^2)$.
The fluid variables $\Jbb$ are then given in~\eqref{J0Jm-rhoZ}.
\ese

\bse
We prove convexity next.
For each fixed $y\geq 0$, define
\be
Q_y(\rho,\vec{Z}) \coloneqq \sqrt{y \hNumb(y)^{-2} \rho^2 + 4|\vec{Z}|^2} - 2p(y) ,
\ee
where we see pressure $p=p(\mu)$ as a function of~$y$, explicitly $2p=\frac{1-q}{2}y$.
Since $Q_y$ is the Euclidean norm of the vector $(y^{1/2}\hNumb^{-1}\rho,\vec{Z})$, shifted by the constant $-y(1-q)/2$, it is convex in $(\rho,\vec{Z})$.  We claim that
\bel{eq:convexM-3}
\Mbf^{00}(\Jbb,0,0) = \sup_{y\geq 0} Q_y(\rho,\vec{Z}) .
\ee
First, if $(\rho,m)$ comes from a state~$\Jbb$, then by taking $y=-\Jbb\cdot\Jbb$ we have
\be
\aligned
\sup_{y\geq 0} Q_y(\rho,\vec{Z})
& \geq \sqrt{(-\Jbb\cdot\Jbb) \hNumb_\Jbb^{-2} \rho^2 + 4|\vec{Z}|^2} - \frac{1-q}{2} (-\Jbb\cdot\Jbb) \\
& = \sqrt{J_0^2 (-\Jbb\cdot\Jbb + J_1^2 + J_2^2 + J_3^2)} - \frac{1-q}{2} (-\Jbb\cdot\Jbb)
= \Mbf^{00}(\Jbb,0,0,0) .
\endaligned
\ee
Moreover, thanks to $\del_y(y\hNumb^{-2}) = \hNumb^{-2} p'(y)$, one has
\be
\del_y Q_y(\rho,\vec{Z}) = 2\frac{dp}{dy} \Biggl( \frac{\hNumb^{-2}\rho^2}{\sqrt{y \hNumb(y)^{-2} \rho^2 + 4|\vec{Z}|^2}} - 1 \Biggr) .
\ee
Since $dp/dy>0$, the sign of $\del_yQ_y(\rho,\vec{Z})$ is the same as the sign of
\be
\bigl(\hNumb^{-2}\rho^2\bigr)^2 - \bigl(y \hNumb(y)^{-2} \rho^2 + 4|\vec{Z}|^2\bigr) = \phi_\rho(y) - 4|\vec{Z}|^2 .
\ee
Hence the critical point condition $\del_y Q_y(\rho,\vec{Z})=0$ is exactly equivalent to the condition $\phi_\rho(y)=4|\vec{Z}|^2$ that characterizes the physical value of~$y$ in~\eqref{J0Jm-rhoZ-full}. Since $\phi_\rho$ is strictly decreasing, the function $y\mapsto \phi_\rho(y)-4|\vec{Z}|^2$ changes sign at most once. It follows that $\del_y Q_y(\rho,\vec{Z})$ changes sign at most once, from positive to negative, and therefore the corresponding critical point is the unique global maximizer of $Q_y(\rho,\vec{Z})$. This proves~\eqref{eq:convexM-3}.
\ese

Finally, since $\Mbf^{00}(\Jbb,0,0)$ is the supremum of the convex functions~$Q_y$, it is itself convex in the principal variables~$Z$. This completes the proof of \autoref{lem-convexM}.


\section{A density property}
\label{appendix=E}

\begin{lemma}[Density of product-type functions]
\label{lem:density-separated}
Let $N = [t_0,t_1] \times \Sbb^1$. Consider a function space $\mathcal{E}$ on~$N$ which is either the space of continuous functions $C^0(N)$ with the uniform norm, or the space $\BVac(N)$ of absolutely continuous functions. Then the set consisting of finite linear combinations of smooth functions $\theta(t)\, \varphi(x)$ generates a dense subspace of $\mathcal{E}$. 
\end{lemma}

\begin{proof} We distinguish the two spaces under consideration.
For $\mathcal E=C^0(N)$ endowed with the uniform norm, the result is an immediate consequence of the Stone--Weierstrass theorem, since the set of finite linear combinations of smooth functions $\theta(t)\, \varphi(x)$ is a subalgebra of~$\mathcal E$ that contains constants and separates points ---in fact, the algebra of polynomials in $t,\sin(2\pi x),\cos(2\pi x)$ would suffice.

Next, consider $\mathcal E=\BVac(N)$, and let $f\in \BVac(N)$ and $\varepsilon>0$ be given. By the density of smooth functions in $\BVac(N)$, there exists $g\in C^\infty(N)$ such that
\be
\|f-g\|_{\BVac(N)}<\frac{\varepsilon}{2}.
\ee
Since $\partial_t \partial_x g$ is continuous on~$N$, the Stone--Weierstrass theorem yields a linear combination $h=\sum_{i=1}^k \theta_i(t) \varphi_i(x)$ such that (the constant will be clarified later)
\bel{density-ddg-h}
\bigl\| \partial_t \partial_x g - h \bigr\|_{C^0(N)} < \frac{\varepsilon}{4(t_1-t_0)(2(t_1-t_0)+1)} .
\ee
Let $\overline\varphi_i=\int_{x\in\Sbb^1} \varphi_i(x)\,dx$ denote the average value of each function~$\varphi_i$.
The spatial average $\overline h(t)=\sum_i\theta_i(t)\overline\varphi_i(x)$ of~$h(t,x)$ in~$x$ is bounded as
\be
\bigl|\overline h(t)\bigr|
= \biggl| \int_{x\in\Sbb^1} \bigl(h(t,x) - \del_t \del_x g(t,x)\bigr) \, dx \biggr|
\leq \bigl\| \partial_t \partial_x g - h \bigr\|_{C^0(N)} ,
\ee
so that replacing $h$ by $(h-\overline{h})$ in~\eqref{density-ddg-h} ---or equivalently $\varphi_i$ by $(\varphi_i-\overline\varphi_i)$--- simply worsens the error bound by a factor of~$2$.
The functions $\varphi_i$ are henceforth assumed to have vanishing average over~$\Sbb^1$.

Consider now the primitives $\Theta_i,\Phi_i$ defined by $\Theta_i'(t)=\theta_i(t)$ and $\Phi_i'(x)=\varphi_i(x)$ and chosen to vanish at $(t_0,0)\in N$.
Then the function
\be
H(t, x) = g(t, 0) + g(t_0, x) - g(t_0, 0) + \sum_{i=1}^k \Theta_i(t) \Phi_i(x)
\ee
is in the desired subspace, and provides a good approximation of~$g$.
Indeed, since $\del_t \del_x H = h$ and boundary terms are well-chosen, one has
\be
\aligned
\bigl| \del_t (g - H)\bigr|(t,x)
& = \biggl|\int_0^x \del_t \del_x \Bigl( g(t,x') - H(t, x') \Bigr) dx' \biggr|
\leq \| \del_t \del_x g - h\|_{C^0(N)} ,
\\
\bigl| \del_x (g - H)\bigr|(t,x)
& = \biggl|\int_{t_0}^t \del_t \del_x \Bigl( g(t',x) - H(t', x) \Bigr) dt' \biggr|
\leq \| \del_t \del_x g - h\|_{C^0(N)} |N| ,
\\
\bigl| g - H\bigr|(t,x)
& = \biggl|\int_{t_0}^t \int_0^x \del_t \del_x \Bigl( g(t',x') - H(t', x') \Bigr) dx'dt' \biggr|
\leq \| \del_t \del_x g - h\|_{C^0(N)} |N| .
\endaligned
\ee
where $|N|=(t_1-t_0)$ denotes the Lebesgue measure of $N$. Upon integrating each line over~$N$, this yields
\be
\aligned
\|f - H\|_{\BVac(N)}
& \leq \|f - g\|_{\BVac(N)} + \| g - H \|_{\BVac(N)}
\\
& \leq \|f - g\|_{\BVac(N)} + \|\del_t \del_x g - h\|_{C^0(N)} |N| ( 1 + 2|N| )
< \varepsilon .
\endaligned
\ee
This proves the density statement in $\BVac(N)$ as well.
\end{proof}

\end{document}